%% file: HOIF_GTM_arXiv.tex
\documentclass[letterpaper,11pt]{article}
\usepackage[margin=1in]{geometry}
\input{preamble.tex}
\usepackage[T1]{fontenc}
\usepackage{lmodern}
\providecommand{\keywords}[1]{\textbf{Keywords:} #1}
\usepackage{bbold}
\usepackage[mathcal]{euscript}
\usepackage{bibunits}
\usepackage{makecell}
\def\sfN{\mathsf{N}}
\usepackage{orcidlink}
\usepackage{epstopdf}
\def\bbeta{\bm{\beta}}

\def\bt{\mathbf{t}}
\def\bT{\mathbf{T}}

\def\nuis{\mathsf{nuis}}
\def\test{\mathsf{test}}
\def\sfwbar{\bar{\sfw}}
\def\sfubar{\bar{\sfu}}
\def\H{\mathrm{H}}
\def\F{\mathrm{F}}
\def\G{\mathrm{G}}
\def\P{\mathbb{P}}

\def\p{\mathsf{p}}

\def\T{\mathrm{T}}
\def\proj{\mathsf{proj}}

\def\T{\mathrm{T}}

\def\ini{\mathsf{init}}

\def\BK{\mathrm{BK}}

\newcommand{\Lin}[1]{{\color{blue}{#1}}}
\newcommand{\Rev}[1]{{\color{red}{#1}}}

\newcommand{\Yulin}[1]{{ \color{blue}{#1}}}

\begin{document}
	
\title{Higher-Order Debiased Estimators for General Treatment Models\footnote{LL and ZZ co-advised the project, are listed in alphabetical order (equal contribution), and serve as co-corresponding authors. The authors would like to thank the Co-Editor, Matias Cattaneo, the Associate Editor and two reviewers for their helpful comments that significantly improved our paper. The authors are also grateful for insightful discussions with Jamie Robins, \href{https://www.matteobonvini.com/}{Matteo Bonvini} (on dose-response estimation) and with \href{https://zijguo.github.io/}{Zijian Guo} (on computation). The research related to this paper is supported by  funds from the National Key R\&D Program of China [grant number 2022YFA1008300] (Zheng Zhang), NSFC Grant No. 12471274 (Lin Liu), the Fundamental Research Funds for the Central Universities, and the Research Funds
of Renmin University of China [project number 23XNA025] (Zheng Zhang), and Science and Technology Talent and Platform Program of Yunnan Province Grant No. 202605AF350070 (Lin Liu).}}

\author[1]{Yulin Zhang\orcidlink{0009-0002-7094-8377}\thanks{E-mail: \href{yulinzhang@ruc.edu.cn}{yulinzhang@ruc.edu.cn}}}
\author[2]{Lin Liu\orcidlink{0000-0002-9883-7962}\thanks{E-mail: \href{linliu@sjtu.edu.cn}{linliu@sjtu.edu.cn}}} 
\author[1]{Zheng Zhang\thanks{E-mail: \href{zhengzhang@ruc.edu.cn}{zhengzhang@ruc.edu.cn}}}
\affil[1]{Institute of Statistics \& Big Data, Renmin University of China}
\affil[2]{Institute of Natural Sciences, MOE--LSC, School of Mathematical Sciences, CMA--Shanghai, SJTU--Yale Joint Center for Biostatistics and Data Science, Shanghai Jiao Tong University}
	
\date{}
	
\maketitle

\vspace{-2.5em}

\begin{bibunit}[plainnat]

\begin{abstract} 
We have witnessed tremendous progress in developing the foundation for econometrics and causal inference in the past decades. The most popular paradigm in the current literature is the classical (first-order) semiparametric theory, in which a key building block is the (first-order) influence function. However, it is now well known that estimators based on influence functions can be sub-optimal in terms of convergence rates in various settings. To address this issue, higher-order influence functions (HOIF) are developed, generalizing the classical semiparametric theory. However, most existing results in this regard focus on treatment effect parameters defined in explicit forms, such as average treatment effects (ATE). In applications, economists are often confronted with tasks of inferring more complex parameters, such as quantile treatment effects (QTE) or effects of complicated treatment regimes/policy. These more complex parameters can often only be implicitly defined as the solution to nonlinear estimating equations, which correspond to $M$/$Z$-estimation problems. Our current understanding of these problems is mainly limited to the classical semiparametric theory. Given the foundational role of HOIF for estimating explicit parameters such as ATE, a modest step toward enriching the statistical foundation of econometrics and causal inference is to develop the corresponding higher-order estimators for those more complex parameters. To this end, we consider parameters of a class of non-separable structural models in the econometrics literature and develop a class of higher-order estimators for the target parameters. Statistical properties of these higher-order estimators are derived using recent advances in $U$-processes theory. Our proposed higher-order estimators relax complexity-reducing assumptions, quantified by \Holder{} smoothness, imposed on the nuisance parameters compared to existing alternative estimators for many important parameters in this class, including QTE and quantile dose-response functions, among others. Numerical experiments, including simulation studies and a real data analysis, are also conducted to corroborate our theoretical claims and illustrate how higher-order estimators can be used in practice.
\end{abstract}


\allowdisplaybreaks

\section{Introduction}
\label{sec:intro}

\subsection{Motivation}

Recent decades have witnessed surged interests in the econometrics and causal inference literature on the construction of nearly rate-optimal estimators for parameters/estimands/(statistical) functionals, such as the average treatment effect (ATE). One particular fruitful econometric paradigm is through the lens of semiparametric theory \citep{schick1986asymptotically, newey1990semiparametric}, which delivers, among other things, the celebrated doubly-robust Augmented Inverse Probability Weighting (AIPW) estimator of the ATE \citep{robins1994estimation, hahn1998role, chernozhukov2018double}, based on the \emph{(first-order) influence function} of the target parameter. The asymptotic variance of such estimators is known to achieve the \emph{semiparametric efficiency bound}, provided that sufficient  complexity-reducing assumptions are imposed on \emph{the nuisance parameters} (the outcome regression and the propensity score in the case of ATE).

\subsubsection*{Sub-optimality of first-order estimators}

Despite the popularity of estimators based on influence functions (henceforth, first-order estimators), growing evidence suggests that they can be sub-optimal in certain settings. Specifically, even when the parameter of interest is \emph{estimable} at the parametric $n^{-1/2}$ rate, first-order estimators may fail to achieve $n^{-1/2}$-consistency, letting alone attain the semiparametric efficiency bound. This sub-optimality can arise in a variety of contexts, including, e.g., when the nuisance parameters are assumed to belong to \Holder{}-type smoothness classes  \citep{liu2017semiparametric, robins2023minimax}, structure-agnostic classes for certain types of parameters \citep{robins1997toward, balakrishnan2026fundamental, liu2024assumption, jin2025sharp}, or hybrid smoothness--structure-agnostic classes \citep{bonvini2024doubly}; or when the nuisance parameters are in the many-covariate or small bandwidth settings, where a direct plug-in estimator based on first-order influence functions possibly may lead to a ``leave-in'' bias \citep{cattaneo2018kernel, cattaneo2019two}.

Using ATE and classical nonparametric nuisance models as an illustration, when both the outcome regression and the propensity score belong to \Holder{} smoothness classes with smoothness index $s$ and variable dimension $d$, standard\footnote{When we add the qualification ``standard'', we mean first-order estimators where the nuisance parameter estimates are computed from an independent nuisance sample as in \citet{chernozhukov2018double} or \citet{kallus2024localized}.} first-order estimators of ATE are $n^{-1/2}$-consistent when $\frac{2 s / d}{1 + 2 s / d} > 0.5$, i.e., $s / d > 0.5$. However, it is well known that $s / d > 0.25$ is both sufficient and necessary (hence minimal) for ATE to be $n^{-1/2}$-estimable \citep{robins2009semiparametric, robins2023minimax}. Indeed, estimators based on \emph{higher-order influence functions} (HOIFs) \citep{robins2008higher, robins2016technical, liu2017semiparametric}, a generalization of the classical first-order influence functions, have been shown to be $n^{-1/2}$-consistent for ATE when $s / d > 0.25$  \citep{robins2023minimax, liu2017semiparametric}. In the remainder of this paper, we refer to these estimators as the ``higher-order estimators'' for convenience. More importantly, to the best of our knowledge, no other simpler estimators are $n^{-1/2}$-consistent under this minimal \Holder{} smoothness condition, despite years of research efforts. These results also prompt recent advances in applying the HOIF framework to statistical problems similar to but more complex than ATE, such as conditional average treatment effects \citep{kennedy2024minimax} and average dose-response functions \citep{colangelo2026double, bonvini2022fast}. Our approach is also connected to the literature addressing ``leave-in bias'' in the many-covariate and small-bandwidth settings \citep{cattaneo2018kernel, cattaneo2019two}; see Section~\ref{sec:papers}, Remark~\ref{rem:connection} and Supplementary Material Section~\ref{app:connection} for an extended discussion. The present paper extends this line of work to a broader class of parameters central to modern econometric analysis, which will be made precise next.


\subsubsection*{Challenges due to implicitly-defined parameters}

The existing HOIF framework primarily concerns parameters that are explicitly defined, including but not limited to ATE and parameters used in conditional independence tests \citep{shah2020hardness} as prominent examples. However, many causal parameters of interest in modern economics applications lack closed-form expressions and are instead defined as solutions to moment conditions or optimization problems. The quantile treatment effect (QTE) provides a canonical example. In policy evaluation, when outcomes are heterogeneous or heavy-tailed, QTE offers a more informative summary of treatment effects than the mean \citep{manski2004statistical, firpo2007efficient,chernozhukov2005iv,powell2020quantile}. Yet, QTE can only be characterized as the solution to a variational problem ($M$-estimation) or to a moment equation ($Z$-estimation) (see Example~\ref{qte} later). Another prominent example in the recent econometrics literature is the $\alpha$-expected shortfall (ES) \citep{fan2025policy}, which measures the average outcome in the worst $\alpha$ fraction of the distribution. This quantity is of particular interest in risk management and policy evaluation under tail risk, where policymakers are concerned with extreme rather than average outcomes. Borrowing a terminology from \citet{robins2016technical}, such parameters are \emph{implicitly defined}. These implicitly defined parameters can pose significant challenges to estimation and inference. 

Another important challenge arises from the non-separability between the nuisance parameter and the parameter of interest \citep{su2019non, chernozhukov2025linear}. Using QTE with binary treatment as an example, as shown in Example~\ref{qte} and Remark~\ref{rem:binary QTE} later, one of the nuisance parameters, the generalized outcome regression model, depends on the true but unknown QTE parameter. To address this problem, \citet{kallus2024localized} develop the \emph{localized debiased machine learning} (LDML) estimator, which first produces an initial estimator of the QTE using only the propensity score independent of the QTE itself. With this initial estimator, LDML proceeds to estimate the generalized outcome regression model and then constructs a rate doubly-robust estimator based on the influence function of the QTE. However, as shown in \citet{kallus2024localized} and for the sake of completeness in Section~\ref{sec:dml}, for the LDML estimator to be $n^{-1/2}$-consistent, it demands certain strong assumptions on the rate of convergence of the initial QTE estimator, hence also on the regularity of the propensity score; see Proposition~\ref{prop:dml}. \citet{ai2021unified} bypass the two-step procedure and estimate the QTE by using only the propensity score. But unlike the LDML estimator, the QTE estimator of \citet{ai2021unified} fails to be rate double-robust. It is then natural to raise the following question:

\begin{quote}
\emph{Can estimators of QTE or other more general implicitly-defined causal parameters be constructed that improve upon existing ones based on influence functions, in terms of convergence rates and the required assumptions on the initial estimators?}
\end{quote}

We provide a positive answer to the above inquiry by going beyond the classical first-order semiparametric theory, in a sense formalized later in the paper. For implicitly defined parameters, however, the literature beyond first-order estimators is rather scarce, with the only exception being Section~6 of \citet{robins2016technical}, which touches on the issue without providing a treatment with further technical details. Also, the regression slope in a partially linear model was considered in \citet{robins2016technical} as a leading example, which has an explicit form\footnote{Partial-linear semiparametric regression model reads as $Y = \beta A + b (X) + \varepsilon$, where $\varepsilon$ has mean zero and $\varepsilon \indep A, X$. Then the regression slope $\beta \equiv \bbE [\cov (A, Y | X)] / \bbE [\var (A | X)]$ is simply the ratio between the expected conditional covariance of $A, Y$ given $X$ and the expected conditional variance of $A$ given $X$, two well-studied examples of parameters, often referred to as doubly-robust functionals \citep{rotnitzky2021characterization}.}. 
Our paper partially fills this gap. To strike a balance between the generality and practical utility of our results, we opt for considering a broad class of parameters defined via generic moment equations with possibly non-smooth, non-separable generalized residual functions (see \eqref{ident}). This formulation accommodates discrete, continuous, and more complicated treatment types in a unified way, and it aligns with a growing econometrics literature that represents causal targets through generic moments.  For example, \citet{ao2021multivalued} use a generic moment function applied to counterfactual distributions under multivalued treatments. Within this unified formulation, we cover a wide range of parameters encountered in applied causal inference, including ATE, QTE, and more generally the General Treatment Models (GTMs) recently coined in the econometrics literature \citep{ai2021unified, chen2025local}. 
Extending the HOIF framework to parameters of GTMs is thus instrumental towards building a broadly applicable point and interval estimation scheme for parameters encountered in econometrics.

\subsubsection*{Summary of main results and contributions}
	
Compared to the HOIF theory for explicitly defined parameters, the following challenges remain to be addressed, which constitute the main contributions of our paper.

\begin{enumerate}[label = (\arabic*)]
\item We derive HOIFs for causal parameters implicitly defined through GTMs, facilitating the construction of improved higher-order estimating equations, compared to state-of-the-art estimators based on the first-order theory, such as the LDML estimator. While we mainly analyze  estimators solving second-order estimating equations, we refer to them as higher-order estimators to avoid introducing additional terminology and we also present the full $m$-th order construction in Supplementary Material Section~\ref{app:truncation}. Among the parameters of GTMs, the QTE with binary treatment is one of the most commonly encountered examples in practice. We obtain a $n^{-1/2}$-consistent estimator for QTE under the weakest \Holder{} smoothness assumptions on the nuisance parameters to date, almost matching the minimal smoothness assumptions for ATE to be $n^{-1/2}$-estimable.

\item State-of-the-art first-order estimators require an initial estimator of the target parameter to estimate the nuisance parameters. This extra step incurs an additional rate-condition on the initial estimator and, consequently, on certain nuisance parameter(s). As we will see in Section~\ref{sec:examples}, for the QTE with binary treatment, our proposed higher-order estimators can drastically relax this rate-condition, highlighting another advantage of higher-order estimators (see Table~\ref{tab:smoothness} in Section~\ref{sec:examples}).


\item The original formulation of GTMs by \citet{ai2021unified} assumes a parametric specification for the target parameters. When the treatment takes finitely many discrete values, finite-dimensional GTMs are effectively nonparametric. However, in many modern applications, treatment variables are often continuous. It may then be  necessary to consider infinite-dimensional GTMs to avoid model misspecification bias; or otherwise, the finite-dimensional GTMs can at best be interpreted as a projection of the true causal or structural parameters onto a parametric model. To our knowledge, \citet{colangelo2026double} and \citet{bonvini2022fast} are among the first to consider estimating the average dose-response function (see Example~\ref{adrf} later) using higher-order estimators when the treatment is continuous. To incorporate such scenarios, we 
extend our higher-order estimators to nonparametric GTMs, thereby generalizing the results of \citet{colangelo2026double} and \citet{bonvini2022fast} to include, for example, the quantile dose-response function (see Example~\ref{qdrf} later). 
\end{enumerate}

Before proceeding, we illustrate the value of incorporating HOIF theory for GTM parameters through a simulation study; full details are provided in Sections~\ref{sec:review} and~\ref{sec:simulation}. For clarity of exposition, we present the discussion at a high level and defer some technical rigor to later sections. Figure~\ref{fig:QTE_Bias_comparison} reports the simulation results. The parameter of interest $\beta^{*}$ is the QTE at the $25\%$ quantile,  with binary treatment and a one-dimensional covariate that suffices to control for confounding. The histograms display the simulated distribution of the centered estimator, $\hat{\beta}-\beta^{*}$.
In the left panel of Figure~\ref{fig:QTE_Bias_comparison}, we adopt the aforementioned state-of-the-art LDML estimator \citep{kallus2024localized}, with nuisance parameters (see Example~\ref{qte} later) estimated by neural networks \citep{chen2024causal}. We design the simulation such that the nuisance parameters have low smoothness in \Holder{} sense (with smoothness close to $s = 0.4$). As a result, even LDML with neural networks fails to deliver an estimator that is close to being unbiased, which is not too surprising given the empirical evidence on the difficulty of fitting functions of low regularity by neural networks in practice \citep{xu2022deepmed}. In contrast, the right panel of Figure~\ref{fig:QTE_Bias_comparison} shows that, after centering by the true QTE, our HOIF-based estimator is well centered around zero, though with slightly higher variance than the LDML estimator. This empirical result demonstrates that the HOIF framework can be a useful tool even when modern neural networks are employed for nonparametric first-step estimation, complementing state-of-the-art methods based on classical semiparametric theory.

\begin{figure}[htbp]
  \centering
  \includegraphics[width=0.8\textwidth]{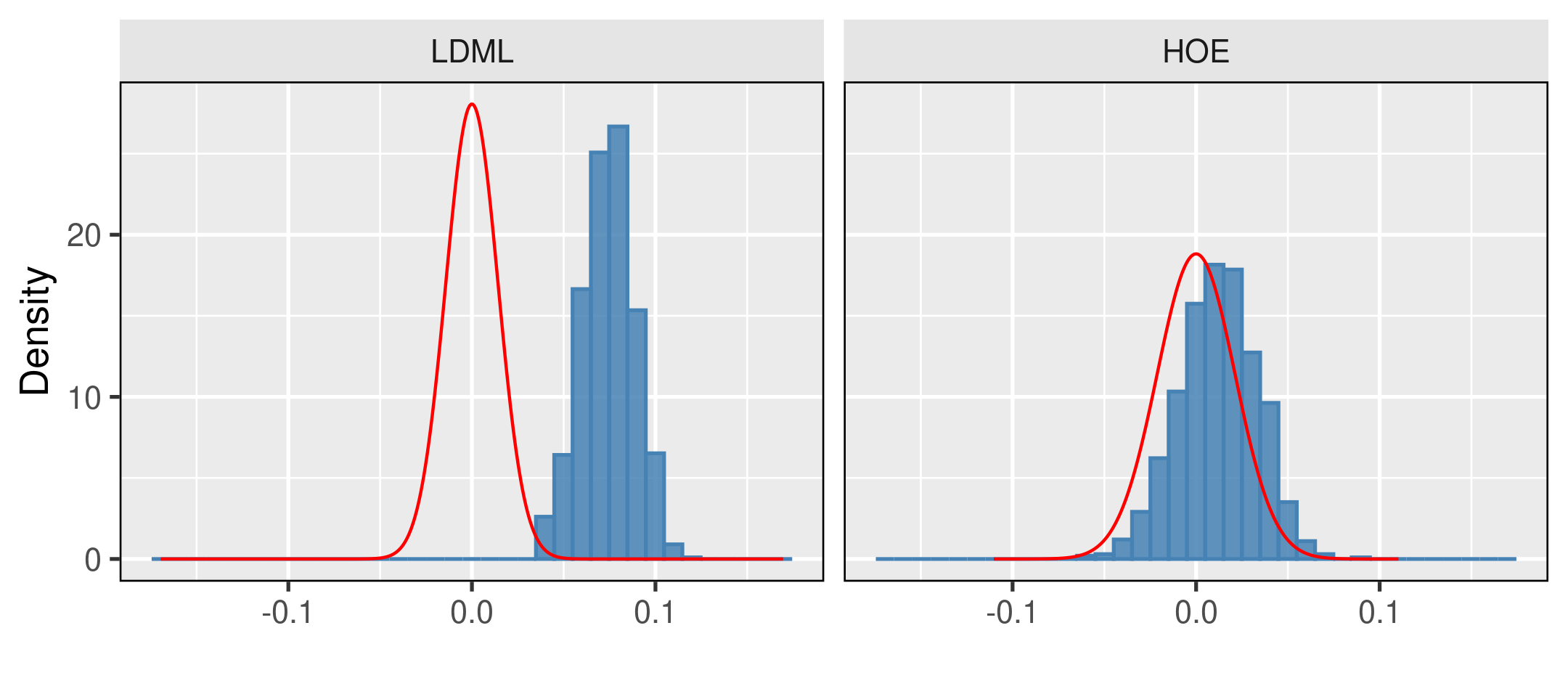} 
  \caption{Comparison between the first-order (LDML) estimator and the higher-order estimator (HOE) of QTE. The histograms over 1000 Monte Carlo simulations are displayed, and the red curves are the corresponding normal probability density function with the same mean and variance as the Monte Carlo mean and variance.}
  \label{fig:QTE_Bias_comparison}
\end{figure}

\subsection{A survey of related works}
\label{sec:papers}

Our work is related to several strands of literature in econometrics and statistics. We now give a brief (and very likely incomplete) survey.

\begin{enumerate}[label = (\arabic*)]
\item  \textbf{Two-step procedures.} The problem studied in this paper is related to classical and modern semiparametric two-step estimation in econometrics, where some first-stage nuisance components may be indexed by the target parameter. Classical semiparametric theory typically imposes conditions under which the first-stage bias is asymptotically negligible, so that the limiting distribution of the target estimator is not affected; see, among others, \citet{newey1990semiparametric, newey1994large,chen2003estimation} and the references therein. 

More recent work develops estimation and inference procedures that remain valid when first-stage bias is non-negligible. For example, \citet{cattaneo2018kernel} and \citet{cattaneo2019two} study two-step estimators where the first-stage bias can enter the asymptotic distribution of the target parameter, using kernel-based and linear-regression-based nuisance estimators, respectively. Specifically, \citet{cattaneo2019two} consider a setting with many covariates, where the first-stage nuisance estimator is obtained by linear regression. They show that a non-negligible leave-in bias arises when the same sample is used for both steps, and propose jackknife-type methods to consistently estimate and remove this bias, yielding a bias-corrected estimator. \citet{cattaneo2018kernel} develop a small-bandwidth asymptotic framework for kernel-based first-step estimators, where the smoothing bias is not assumed negligible. They show that the standard nonparametric bootstrap automatically corrects for this bias, so that percentile confidence intervals achieve correct coverage. When the asymptotic bias cannot be consistently estimated, \citet{cavaliere2024bootstrap} propose a distributional adjustment to restore bootstrap validity. Their method first obtains a bootstrap p-value, which is not asymptotically uniform due to the bias. They show that this p-value converges to a non‑uniform limit that does not depend on the bias. By applying a probability integral transformation estimated via a second bootstrap, they obtain a corrected $p$-value that is asymptotically uniform, thus restoring valid inference without directly estimating the bias. \cite{cattaneo2025higher} develop higher-order distributional approximations for kernel-based estimators of the density-weighted average derivative. By employing Edgeworth expansions, they refine the large-sample approximation of the estimator's distribution, which leads to improved finite-sample inference, particularly enhancing the coverage accuracy of confidence intervals.
Another broad class of approaches focuses on estimators with a small bias property (SBP) \citep{newey2004twicing}, meaning that the bias of the target estimator vanishes faster than the  bias of the first-stage nuisance estimators. \citet{newey2004twicing} shows that a class of semiparametric estimators based on twicing kernels satisfy this property. Tackling a similar problem and motivated by the HOIF theory, \citet{chen2024method} proposed moment-based $\sqrt{n}$-consistent estimators of a class of causal parameters even when the nuisance parameters, assumed to be generalized linear models, are not consistently estimable.

\item \textbf{Influence functions and semiparametric inference. } Our work builds on and extends semiparametric and nonparametric inference based on influence functions, which play a central role in semiparametric efficiency theory. Influence-function-based estimation is a classical topic in semiparametric theory and has been studied for decades. For example, \citet{schick1986asymptotically} develops a general method for constructing asymptotically efficient estimators in semiparametric models. \citet{newey1990semiparametric} provides an introduction to semiparametric efficiency bounds, discussing their nature, calculation, and application in constructing semiparametric estimators and deriving their limiting distributions. For a more detailed and comprehensive treatment of influence functions for general semiparametric models and statistical inference, see \citet{bickel1998efficient}. As mentioned, our results can be viewed as a further development of the higher-order generalization of the classical influence-function-based approaches \citep{robins2008higher}.

For binary treatment, \citet{robins1994estimation} and \citet{hahn1998role} derived the efficient influence functions for ATE and ATT (average treatment effect on the treated). \citet{cattaneo2010efficient} considered a general treatment effect model with possibly over-identified and non-smooth moment conditions under multi-valued treatments, covering a broad class of treatment-effect parameters including average and quantile treatment effects, and established the corresponding efficiency bounds. Building on these efficiency results, the relevant literature focuses on developing practical estimators for treatment effects; see, for example, \citet{hirano2003efficient} and the references therein. Recent work also leverages influence functions to obtain valid inference with
high-dimensional or machine-learned nuisance components \citep{chernozhukov2018double}.  \citet{colangelo2026double} further extend DML-type ideas to kernel-based inference on the average dose--response function under continuous treatments. \citet{rotnitzky2021characterization} and \citet{chernozhukov2022locally} characterized a large class of explicitly-defined parameters, including ATE and ATT, for which DML-type estimators can be constructed based on their influence functions. \citet{newey2018cross} developed refined DML-type estimators that achieve faster convergence rates under smoothness classes with different nuisance parameters estimated from different independent subsamples. 

Papers mentioned above mainly concern explicitly-defined parameters, including ATE and ATT. In terms of implicitly-defined parameters,  \citet{firpo2007efficient} studies efficient influence functions for QTE with binary treatment. \citet{ao2021multivalued} develop a unified moment-based framework for multi-valued treatments, covering means, quantiles, and distribution functions, and propose an efficient influence function based estimator for multivalued treatment effects for the treated. Their approach also enables decomposition analysis, separating wage structure effects from composition effects. Applying their method to evaluate the Workforce Investment Act (WIA) program, they find that heterogeneity in participation levels is an important dimension for evaluating social programs, demonstrating the practical relevance of their methods for policy evaluation.  Complementarily, \citet{ai2021unified} provide a unified framework for continuous treatments and derive efficient influence functions for general treatment models.
As mentioned, more recently, \citet{kallus2024localized} propose the LDML estimator for the QTE, along the line of DML estimators.

\end{enumerate}

\subsection{Notation} 
\label{sec:notation}

Before proceeding, we collect some frequently used notation. Let $L_{2, 0}$ denote the space of mean zero $L_{2}$ functions with respect to a dominating measure $\nu$. Without loss of generality, we take $\nu$ to be the Lebesgue measure throughout the paper. We write $a_n \lesssim b_n$ (resp., $a_n \gtrsim b_n$) if $a_n \leq C b_n$ (resp., $a_{n} \geq C b_{n}$) for some universal constant $C > 0$ independent of $n$;  $a_n \asymp b_n$ if $a_n \lesssim b_n$ and $b_n \gtrsim a_n$; and $a_n \gg b_n$ (resp., $a_n \ll b_n$) if $b_n/a_n \to 0$ (resp., $b_n/a_n \to \infty$). For $a,b \in \bbR$, let $a \vee b = \max\{a,b\}$ and $a \wedge b = \min\{a, b\}$. For a vector $\bar{a} \in \bbR^{q}$, let $\| \bar{a} \| \coloneqq (\bar{a}^{\top}\bar{a})^{1/2}$ denote its Euclidean norm and $\bar{a}^{\otimes 2} = \bar{a}\,\bar{a}^{\top} $ denote  the tensor (Kronecker) product. 
For a matrix $A\in\bbR^{p\times q}$ we write its operator norm as
$
\Vert A\Vert_{\op}
\coloneqq
\inf\{c>0:\,\Vert Av\Vert \le c\Vert v\Vert
\ \text{for every}\ v\in\bbR^{q}\}
$. 
Given any function $f: \calO \to \bbR^{m}$ with fixed integer $m \geq 1$, we denote its $L_{\infty}$-norm as $\|f (\cdot)\|_{\infty} \coloneqq \sup_{o \in \calO} \|f (o)\|$.

The higher-order estimators introduced later involve $L_{2} (\P)$ projections onto the span of a size-$k$ dictionary $\bar{\phi}_{k} \coloneqq (\phi_{1}, \ldots, \phi_{k})$, where each $\phi_{j}: \calO \rightarrow \bbR$. For any $f: \calO \rightarrow \bbR$, this projection is denoted by $\Pi (f \mid \bar{\phi}_{k}) (\cdot) \coloneqq \bbE \{f (O) \bar{\phi}_{k} (O)\}^{\top} \Sigma_k^{-1} \bar{\phi}_{k} (\cdot)$, where $\Sigma_{k} \coloneqq \bbE \{\bar{\phi}_{k} (O)^{\otimes 2}\}$ is the population Gram matrix. The orthocomplement of this projection is $\Pi^{\perp} (f \mid \bar{\phi}_{k}) (\cdot) \coloneqq f (\cdot) - \Pi (f \mid \bar{\phi}_{k}) (\cdot)$.
	
For any measurable function $h: \calO^{m} \rightarrow \bbR$, let $\bbU_{n, m}$ be the $m$-th order $U$-statistic operator, 
\begin{align*}
\bbU_{n, m} h \equiv \bbU_{n, m} h (O_{1}, \cdots, O_{m}) = \frac{(n - m)!}{n!} \sum_{1 \leq i_{1} \neq \cdots \neq i_{m} \leq n} h (O_{i_{1}}, \cdots, O_{i_{m}}).
\end{align*}
$\bbU_{n, 1} \equiv \bbP_{n}$ is understood to be the empirical mean.
 
We next review some basic concepts from empirical processes theory; see \citet{van2023weak} for further details. For a probability measure $Q$ and a constant $p > 1$, let $\| \cdot \|_{Q,p}$ denote the $L_{p} (Q)$-seminorm, i.e., $\|f\|_{Q, p} \coloneqq \{\int |f (x)|^{p} \diff Q (x)\}^{1/p}$ for finite $p$ while $\|f\|_{Q,\infty}$ denotes the essential supremum of $|f|$ with respect to $Q$. Let $\calF$ be a class of symmetric measurable functions $f : \calO^{m} \rightarrow  \bbR$ with envelope $\F$, that is, $\sup_{f\in\calF} |f| \leq \F $. If $\| \F \|_{Q,p} > 0$, we write $\sfN ( \calF,\|\cdot\|_{Q,p} , \epsilon\|\F\|_{Q,p})$ for the minimal number of $\|\cdot\|_{Q,p}$-balls of radius $\epsilon\|\F\|_{Q,p}$ needed to cover $\calF$. We finally recall the following  definition.
\begin{definition}[VC type class]
\label{def:VC-type}
A function class $\calF$ with envelope $\F$ is said to be of VC type with characteristic $(A, v)$ if $\sup_{Q} \sfN \left( \calF, \| \cdot \|_{Q, 2}, \epsilon \| \F \|_{Q, 2} \right) \leq (A / \epsilon)^{v}$ for all $\epsilon \in (0, 1]$, where $\sup_{Q}$ denotes the supremum over all discrete finite probability measures.
\end{definition}

For vector-valued functions $\calF \coloneqq \{f: \mathcal O^{m} \to \bbR^{q}\}$ (with fixed $q$), these definitions are interpreted with $|\cdot|$ replaced by the Euclidean norm $\|\cdot\|$ on $\bbR^{q}$. All mentioned entropy and maximal-inequality results  remain valid (up to universal constants). 


\subsection{Organization}
\label{sec:organization}

The remainder of the paper is organized as follows. Section~\ref{sec:setup} introduces the observation scheme and the general treatment model (GTM) and reviews the state-of-the-art estimators for GTM based on (first-order) influence functions from the standard semiparametric theory. Section~\ref{sec:hoif} presents the main results: In Section~\ref{sec:review}, using the QTE as an illustration, we motivate the importance of introducing higher-order estimators to reduce bias and provide the intuition of why HOIFs can be used for bias correction;  Section~\ref{sec:intuition} introduces the HOIF-based estimator for finite-dimensional GTM parameters, and Section~\ref{sec:properties} establishes its statistical properties; Section~\ref{sec:examples} specializes the results obtained in Section~\ref{sec:properties} to QTE.
Section~\ref{sec:dose} addresses the estimation of infinite-dimensional GTM parameters, allowing greater modeling flexibility. Section~\ref{sec:empirical} presents empirical studies of the HOIF estimator. In particular, Section~\ref{sec:simulation} investigates the finite-sample performance of the HOIF estimator in low-regularity settings using simulation experiments, followed by a real data analysis in Section~\ref{sec:real}. Section~\ref{sec:conclusion} concludes the paper with a discussion of caveats and future directions. Proofs and less essential technical results are collected in the Supplementary Material.

\section{Problem Setup and Review of Existing Results}
\label{sec:setup}

\subsection{Causal parameters defined via general treatment models}
\label{sec:parameters}

Let $T \in \calT$ be a general treatment variable that can be discrete, continuous, or a mix of both. Let $Y (t) \in \calY\subset\mathbb{R}$ denote the real-valued potential outcome when $T$ is intervened/set to $t$, and let $Y$ be the observed outcome. To control for confounding, we observe a vector of compactly supported baseline covariates  $X \in \calX \equiv [0, 1]^{d} \subseteq  \bbR^{d}$, for some integer $d \geq 1$.  Let $\bbP$ denote the joint distribution of the observed data $(X, T, Y)$ and $\bbE$ the corresponding expectation.
Let $\Gamma: \calY \times \calT \times \calB \rightarrow \bbR^{p}$ be a known generalized residual function, which can be nonsmooth and non-separable. We are interested in the parameter $\bbeta^*=(\beta^*_0,\beta^*_1,...,\beta^*_{p-1})^{\top}\in \calB \subset \bbR^{p}$, defined as the \emph{unique} solution to the following integral moment equations:
\begin{equation}
\label{ident} 
\int_{\mathcal{T}} \bbE \{\Gamma
(Y (t), t, \bbeta^{\ast})\} \omega (t) \diff t = 0,
\end{equation}
where  $\omega(\cdot): \calT \rightarrow \bbR$ is a user-specified weighting function.
The general formulation  \eqref{ident} can incorporate most causal parameters in the existing literature with flexible specifications of $\Gamma(\cdot)$ and $\omega(\cdot)$; see Examples \ref{ate}--\ref{qdrf} later in this section. We refer to \eqref{ident} as the \emph{general treatment model} (GTM), a unified moment-based framework for binary, multi-valued, and continuous treatments that covers means, quantiles, distribution functions, and various treatment effects. Similar frameworks are also studied in \citet{cattaneo2010efficient,   ao2021multivalued, ai2021unified}. Our goal is to conduct statistical inference on $\bbeta^{\ast}$ using $n$ independent and identically distributed (i.i.d.) observations $\{O_{i} \equiv (X_{i}, T_{i}, Y_{i}), i = 1, \cdots, n\} \overset{\mathrm{i.i.d.}}{\sim} \P$, under weak complexity-reducing assumptions on the nuisance parameters (e.g., the smoothness/rate conditions imposed in Assumption~\ref{as:Holder}), to be defined soon.

Since \eqref{ident} involves the potential outcome $Y (t)$, we need to impose the following identification assumptions to turn \eqref{ident} into a functional of the observed data distribution. Let $\p_{T\mid X}(t\mid x)$ denote the conditional density/mass function of $T$ given $X=x$,  referred to as the \emph{generalized propensity score}.

\begin{assumption}\label{as:identification}
We assume the following: 
(i)~(Consistency) $Y \equiv Y (T)$;
(ii)~(Positivity) The generalized propensity score satisfies $0 < c_1 \leq \p_{T|X}(t|x) \leq c_2 < \infty$ for all $t \in \mathcal{T}$ with $\omega(t) > 0$, for constants $c_1, c_2$;
(iii)~(Unconfoundedness) $Y(t) \indep T \mid X$ for all $t \in \mathcal{T}$.
\end{assumption}
Assumption~\ref{as:identification} comprises the standard identification conditions for causal effects in observational studies. Notably, the stated positivity assumption allows practitioners to choose weight tailored for regions of $\calT$ where positivity holds. Under Assumption~\ref{as:identification}, $\bbeta^{\ast}$, which is defined through the potential-outcome moment equations \eqref{ident}, is identified as the unique solution to the following weighted moment equations based on the observed data distribution:
\begin{equation}\label{newmodel}
\bbE \{\xi (X, T) \Gamma (Y, T, \bbeta^*) \}= 0,
\end{equation}
where the function $\xi (x, t) \coloneqq {\omega (t)}/{\p_{T | X} ( t |x)}$ is also called the \textit{stabilized weight} in the epidemiology literature \citep{hernan2000marginal}. In this paper, we focus on two types of $\omega$: (1) $\omega (\cdot) \equiv \p_{T} (\cdot)$, the marginal distribution of the treatment $T$, which yields a finite-dimensional causal parameter  (see Section~\ref{sec:dml} and Section~\ref{sec:hoif}); (2) $\omega (\cdot) \equiv \p_T(\cdot)\delta_t(\cdot)$, the point mass at a particular value $t$ of $T$, which yields a nonparametric causal curve as a function of $t$  (see Section~\ref{sec:dose}).  
Our framework can be readily extended to the case where $\omega (\cdot)$ is any known function of $T$ or additionally depends on $X$ (see Example \ref{sto}), but we focus on these two choices to avoid unnecessary complications.

We present several commonly encountered examples of treatment effect parameters defined via GTM \eqref{ident}, demonstrating the importance of developing new and improved statistical procedures for this general class.
\begin{example}[Average Treatment Effect (ATE)]
\label{ate}
When $\Gamma (Y(t),t,\bbeta)  = ( (Y(t) - \beta_0)\cdot (1 - t),  (Y(t) -  \beta_1)\cdot t )^{\top}$, $\omega(t) = \p_{T}(t)$ and  $\calT = \{0, 1\}$, where $\bbeta = (\beta_{0}, \beta_{1})^{\top}$, the solution to \eqref{ident} is $\beta^*_{0} = \bbE [Y (0)]$ and $\beta^*_{1} = \bbE[Y (1)]$. ATE has been widely studied in the econometric literature; see \cite{hahn1998role,hirano2003efficient,chernozhukov2018double}, among others.
\end{example}

\begin{example}[Quantile Treatment Effect (QTE)]
\label{qte}
 When  $\Gamma(Y(t),t,\bbeta) = ( (\tau - \mathbbm{1} \{Y(t) \leq \beta_0\})(1 - t), (\tau - \mathbbm{1} \{Y(t) \leq \beta_1\})t)^{\top}$, $\omega(t) = \p_{T}(t)$ and  $\calT = \{0, 1\}$, where again $\bbeta =(\beta_0,\beta_{1})^{\top}$, the solution to \eqref{ident} is $\beta^*_0 = \inf\{q: \mathbb{P}(Y (0) \leq q) \geq \tau\}$ and $\beta^*_{1} = \inf\{q: \mathbb{P}(Y (1) \leq q) \geq \tau\}$, which are the $\tau^{th}$ quantiles of the potential outcomes.  QTE has been studied in \cite{chernozhukov2005iv,firpo2007efficient,kallus2024localized}, among others. We will revisit QTE in Sections~\ref{sec:intuition} and~\ref{sec:examples} later as the most important concrete example of GTM parameters.
\end{example}


\begin{example}[Distributional Treatment Effect (DTE)]
\label{dte}
For any $y\in\mathbb{R}$, let $\Gamma (Y(t),t,\bbeta)  = ( (\mathbbm{1}\{Y(t)\leq y\} - \beta_0 )( 1- t)  , (\mathbbm{1}\{Y(t)\leq y\} - \beta_1) t)^{\top}$, $\omega(t) = \p_{T}(t)$, and $\calT = \{0, 1\}$, where $\bbeta = (\beta_{0}, \beta_{1})^{\top}$. The solution to \eqref{ident} is $\beta^*_{0} = \P (Y (0) \leq y)$ and $\beta^*_{1} = \P (Y (1) \leq y)$. DTE has been studied in  \cite{chernozhukov2013inference} and \citet{donald2014estimation}, among others.
\end{example}

\begin{example}[$\alpha$-Expected Shortfall (ES)]
\label{es}
 Fix $\alpha\in(0,1)$,  let $\calT=\{0,1\}$ and  set $\omega(t)\equiv \p_T(t)$.  Define
\begin{equation*}
\Gamma\big(Y(t),t,\bbeta\big) = t \cdot
\begin{pmatrix}
\alpha-\mathbbm{1}\{Y(t)\le \beta_0\}\\[3pt]
\dfrac{\mathbbm{1}\{Y(t) \le \beta_0\}}{\alpha}\big(Y(t)-\beta_1\big)
\end{pmatrix},
\end{equation*}
with $\bbeta =(\beta_0,\beta_{1})^{\top}$. Then the solution to \eqref{ident} is given by
\begin{align*}
\beta_0^* = \inf\{q: \mathbb{P}(Y (1) \leq q) \geq \alpha\}, 
\qquad 
\beta_1^* = \E (Y(1)\mid Y(1) \le \beta_0^*).
\end{align*}
Thus, $\beta_0^*$ is the $\alpha$-quantile of $Y(1)$ and $\beta_1^*$ is the associated expected shortfall.  Expected shortfall and related tail-risk functionals are widely used in policy learning in econometrics and related areas; see, e.g., \citet{patton2019dynamic, fan2025policy}.
\end{example}

Up to this point, all the examples have set $\omega (\cdot) \equiv \p_{T} (\cdot)$. In the next example, $\omega (\cdot)$ is taken as a known function of $T$.

\begin{example}[Stochastic Intervention Effect]
\label{sto}
The weighting function $\omega (t)$ can be further generalized to depend on the covariates, i.e., $\omega (t, x)$. When $\omega (t, x)$ is a known, user-specified density function of $T$,
$\bbeta^{\ast}$ corresponds to the stochastic (dynamical) intervention effect \citep{kennedy2019nonparametric}. 
\end{example}


\begin{example}[Marginal Structural Models]
\label{msm}
In its most common formulation, marginal structural models (MSMs) \citep{hernan2000marginal} postulate a parametric model for the marginal mean of potential outcomes:
\begin{equation}
\label{main:msm}
\bbE Y (t) = g (t; \bbeta^{\ast}).
\end{equation}
Model~\eqref{main:msm} can be rewritten as Model~\eqref{ident} by letting $\Gamma (Y (t), t, \bbeta^{\ast}) = Y (t) - g (t; \bbeta^{\ast})$. The weighting function $\omega$ needs to be arbitrary if \eqref{main:msm} is taken as part of the modeling assumptions. See Supplementary Material Section~\ref{app:msm} for further discussions.
\end{example}

If we set the weighting function $\omega (\cdot)$ in \eqref{ident} as the product of the marginal density function of $T$ and the Dirac delta function $\p_T(\cdot)\delta_t (\cdot)$ for each possible treatment value $t \in \calT$, then the target parameter  $\bbeta^* \equiv \beta^{\ast} (\cdot)\in \mathbb{R}^{\infty}$ becomes  a function of $t$, that is, an infinite-dimensional ($p=\infty$) parameter.  In Section~\ref{sec:dose}, we further develop our approach for such infinite-dimensional GTM parameters. Notably, this framework
 encompasses two important infinite-dimensional parameters: the average and quantile dose–response functions (ADRF and QDRF).

\begin{example}\label{adrf}[Average Dose-Response Function (ADRF)]
   When $\Gamma (Y(t),t,\beta)  = Y(t) - \beta$ and $\omega (\cdot) = \p_T(\cdot)\delta_t (\cdot)$, the solution to \eqref{ident} is  $\beta^*(t) = \bbE [ Y(t)] $. The ADRF has been widely studied in the literature; see \cite{kennedy2017non}, \cite{su2019non} and \cite{colangelo2026double}, among others.
\end{example}

\begin{example}\label{qdrf}[Quantile Dose-Response Function (QDRF)]
When  $\Gamma(Y(t),t,\beta) = \tau - \mathbbm{1} \{Y(t) \leq \beta\}$ and $\omega (\cdot) = \p_T(\cdot)\delta_t (\cdot)$,  the solution to \eqref{ident} is $\beta^*(t) = \inf\{q : \P (Y(t)\leq q) \geq \tau\}$, the $\tau^{th}$ quantile of the potential outcome $Y(t)$. The QDRF has been studied by \cite{galvao2015uniformly} and \cite{su2019non}.
\end{example}

\begin{remark}
\label{ade}
Finally, though not our main focus, we note that other unknown weighting functions $\omega$ can also be considered. For example, when $T \in \mathcal{T} = \mathbb{R}$ is a continuous treatment variable with probability density function $\p_T(t)$, let $\Gamma(Y(t), t, \beta_0) = Y(t) + \beta_0 t$ and $\omega(t) = \partial_t \p_T(t)$, the derivative of $\p_T(t)$. Under Assumption~\ref{as:identification} and using integration by parts, the solution to \eqref{ident} becomes $\beta_0^* =\int_{\calT} \partial_t \bbE \{\p_{T \mid X}^{-1}(T \mid X)Y \mid T=t\} \p_T^2(t) \diff t$. This is the density-weighted average derivative effect (DWADE), which has been studied in recent works such as \citet{cattaneo2025higher} under a randomized experiment, i.e., with a known propensity score function $\p_{T \mid X}$; see Remark~\ref{rem:cattaneo2025higher} for further discussion. Extending the analysis to observational studies is a topic worth pursuing in future work.
\end{remark}

\subsection{Review of existing results based on influence functions}
\label{sec:dml}

As mentioned, we take $\omega (\cdot) \equiv \p_{T} (\cdot)$ here and in Section~\ref{sec:hoif}, so that $\xi (x,t) = \p_{T}(t)/ \p_{T|X} (t|x)$. Before presenting our new results, we first review state-of-the-art estimators of $\bbeta^{\ast}$ based on influence functions from the (first-order) semiparametric efficiency theory. For convenience, we define $b_{\bbeta} (x, t) \coloneqq \bbE \{\Gamma (Y, T, \bbeta) | X = x, T = t\}$. Since $\Gamma (Y, T, \bbeta)$ essentially plays the same role as the outcome $Y$ in ATE, we refer to  $b_{\bbeta} (\cdot,\cdot)$ as the \emph{generalized outcome regression model}.  We denote by $\theta_{\bbeta} \coloneqq (\xi, b_{\bbeta})$ the collection of nuisance parameters for a fixed $\bbeta$. It is worth noting that $\bbeta$ may be non-separable from the nuisance function $b_{\bbeta}$ and can enter in a nonlinear manner.  From \eqref{newmodel} and by the law of iterated expectations, $\bbeta^{\ast}$ solves the equation 
\begin{equation}
\label{est eq}
\psi (\bbeta) \equiv \psi (\theta_{\bbeta}) \coloneqq  \bbE \{\xi (X, T) \Gamma (Y, T, \bbeta)\} = \bbE \{\xi (X, T) b_{\bbeta} (X, T)\} = 0.
\end{equation}

\citet{ai2021unified} developed the semiparametric efficiency theory for $\bbeta^*$. Since the rest of this paper relies on this preliminary result, we reproduce it below in Proposition~\ref{thm:eif}:
\begin{proposition}
\label{thm:eif}
Under Assumption~\ref{as:identification}, the efficient influence function (EIF) of $\psi (\bbeta)$ is
\begin{align*}
\IF_{\psi}^{(1)}(\bbeta) = & \ \xi (X, T) \left\{ \Gamma (Y, T, \bbeta) - b_{\bbeta} (X, T) \right\} + \left\{ \int_{ \calT} b_{\bbeta} (X, t) \p_{T} (t) \diff t - \bbE \{\xi (X, T) b_{\bbeta} (X, T)\} \right\} \notag \\
& + \left\{ \int_{\calX} b_{\bbeta} (x, T) \p_{X} (x) \diff x - \bbE \{\xi (X, T) b_{\bbeta} (X, T)\} \right\}. 
\end{align*}
Furthermore, the EIF of $\bbeta^{\ast}$ that solves $\psi (\bbeta) \equiv 0$ is
\begin{equation}
\label{eif_beta}
\IF_{\bbeta^*}^{(1)} = H_{\bbeta^{\ast}}^{-1} \cdot \IF_{\psi}^{(1)}(\bbeta^{\ast}), \text{ where } H_{\bbeta} \coloneqq - \nabla_{\bbeta} \psi (\bbeta) = - \nabla_{\bbeta} \bbE \{\xi (X, T) \Gamma (Y, T, \bbeta)\}.
\end{equation}
\end{proposition}
For certain results in this paper, we impose the following \Holder{} smoothness assumptions on the nuisance parameters $\xi$ and $b_{\bbeta}$.
\begin{assumption}
\label{as:Holder}
$\xi (\cdot, \cdot)$ and $b_{\bbeta} (\cdot, \cdot)$ are assumed to be \Holder{} smooth functions with smoothness indices $s_{1}$ and $s_{2}$ respectively, for every $\bbeta \in \calB$. We let $s \coloneqq (s_{1} + s_{2}) / 2$ denote the average smoothness of the two nuisance parameters. The nuisance estimators $\hat{\xi}$ and $\hat{b}_{\bbeta}$ converge to $\xi$ and $b_{\bbeta}$ respectively at the minimax optimal rates ($n^{- \frac{s_{1}}{d + 2 s_{1}}}$ and $n^{- \frac{s_{2}}{d + 2 s_{2}}}$) in $L_{2} (\P)$-distance.
\end{assumption}

The \Holder{} smoothness conditions in Assumption~\ref{as:Holder} are standard for quantifying the complexity of nuisance functions in nonparametric and semiparametric theory. The stated minimax rates are attainable by classical series, sieve, and kernel estimators under appropriate choice of tuning parameters such as the bandwidth or the number of basis functions; see, e.g., \citet{stone1982optimal}. Assumption~\ref{as:Holder} is used to translate the high-level rate conditions on the nuisance estimators into primitive smoothness requirements; see the discussion following Proposition~\ref{prop:dml}, Remark~\ref{rk:improved rate}, and the discussion below Assumption~\ref{as:sigma}. In the binary QTE example in Section~\ref{sec:examples}, this condition is specialized as Assumption~\ref{as:Holder QTE} and is used in the proof of Theorem~\ref{thm:QTE}.

As noted in the Introduction, a major statistical challenge in applying $\IF_{\psi}^{(1)}$ to construct a first-order estimator of $\bbeta$ is that the nuisance parameter $b_{\bbeta}$ itself depends on the unknown parameter $\bbeta$. Directly applying Proposition~\ref{thm:eif} therefore requires estimating the entire nuisance function, meaning that $b_{\bbeta}$ needs to be estimated for every candidate $\bbeta$. One simple way to bypass this difficulty is to first obtain an initial estimator $\bbeta_{\mathsf{init}}$ close to $\bbeta^{\ast}$, for example, by solving the empirical version of \eqref{est eq} and then substitute $b_{\bbeta_{\mathsf{init}}}$ into the influence function $\IF_{\psi}^{(1)}$. This approach is referred to as localized debiased machine learning (LDML), recently developed in \citet{kallus2024localized}.

To formally construct this first-order estimator $\hat{\bbeta}^{(1)}$, we define
\begin{align*}
\hat{\Xi} (O, O'; \bbeta) \coloneqq \hat{\xi} (X, T) \{\Gamma (Y ,T , \bbeta) - \hat{b}_{\bbeta_{\mathsf{init}}} (X, T)\} +  \hat{b}_{\bbeta_{\mathsf{init}}} (X, T') ,
\end{align*}
where $O'=(Y',T',X')$ is an independent copy of $O = (Y,T,X)$.

We construct the first-order estimator $\hat{\psi}^{(1)}(\bbeta)$ of  $\psi(\bbeta)$ by the following second-order $U$-statistic (similar construction also appeared in \citet{tchetgen2010doubly, bonvini2022sensitivity}):
\begin{align}
\label{eq:sample_est_eq}
\hat{\psi}^{(1)}(\bbeta) \coloneqq \bbU_{n, 2} \{\hat{\Xi} (O_{1}, O_{2}; \bbeta)\} =\frac{1}{n(n-1)} \sum_{i = 1}^{n}\sum_{i \neq j, j= 1}^{n} \hat{\Xi} (O_{i}, O_{j}; \bbeta) , 
\end{align}
and define $\hat{\bbeta}^{(1)}$ as the solution to $\hat{\psi}^{(1)}(\bbeta) \equiv 0$. For convenience, we assume that the nuisance estimators $\hat{\xi}$, $\hat{b}_{\bbeta_{\mathsf{init}}}$ and the initial estimator $\bbeta_{\mathsf{init}}$ of $\bbeta^*$ are all computed using an independent separate sample of size $n$, referred to as the \emph{nuisance sample}. The $U$-statistic operator $\bbU_{n, m}$, for any $m \geq 2$, is taken only over the main sample throughout the paper.
All results are stated conditional on the nuisance sample. For brevity, we often suppress this dependence in our notation, particularly the dependence on $\bbeta_{\ini}$. To improve efficiency, one can follow the standard cross-fitting procedure of \citet{chernozhukov2018double} by flipping the roles of the main and nuisance samples to construct a cross-fit version of the estimator $\hat{\psi}^{(1)} (\bbeta)$. However, to avoid notation clutter, we will stick to the estimator without cross-fitting.

\begin{remark}
\label{rem:binary QTE}
In the case of a binary treatment $T \in \{0,1\}$, we have $\xi(X,T) = T/\P(T = 1 | X) + (1 - T)/\P(T = 0 | X)$, $\bbeta=(\beta_0,\beta_1)^{\top}$ and $b_{\bbeta} (X, T) = T \cdot \bbE \{\Gamma(Y,T,\beta_1) \mid X, T=1\} +  (1 - T) \cdot \bbE \{\Gamma(Y,T,\beta_0) \mid X, T = 0\}$. The influence function of $\IF_{\psi}^{(1)}(\bbeta)$ then reduces to 
\begin{align*}
\IF_{\psi}^{(1)}(\bbeta) = \left( \begin{matrix} 
T \cdot \xi(X,T) \left\{ \Gamma(Y,1,\beta_1) - b_{\bbeta} (X, T = 1) \right\} + b_{\bbeta} (X, T=1) \\
(1-T)\cdot \xi(X,T) \left\{ \Gamma(Y,0,\beta_0) - b_{\bbeta} (X, T=0)   \right\} + b_{\bbeta} (X, T = 0) 
\end{matrix} \right).
\end{align*}
This expression  was used to construct the first-order LDML estimating equation for QTE, as proposed in Section 1.1 of  \citet{kallus2024localized}; see Section~\ref{sec:intuition} for a further discussion. For a continuous treatment $T$, the influence function leads to a second-order $U$-statistic estimator. 
\end{remark}

The statistical properties of $\hat{\bbeta}^{(1)}$ are formally stated in Proposition~\ref{prop:dml} below. This result is analogous to Theorem~1 of \citet{kallus2024localized}.  A proof sketch is provided in Section~\ref{sec:review}, which also serves to motivate the introduction of our higher-order estimators.

\begin{proposition}
\label{prop:dml}
Suppose that Assumption~\ref{as:identification} holds and there exist three diminishing sequences $r_{n, \xi}, r_{n, b},  r_{n, \bbeta_{\ini}}\rightarrow 0$ as $n \rightarrow \infty$ such that
\begin{equation*}
\begin{split}
\Vert \hat{\xi} - \xi \Vert_{\P,2} \lesssim r_{n, \xi},\  \Vert \hat{b}_{\bbeta_{\ini}} - b_{\bbeta_{\ini}} \Vert_{\P,2} \lesssim r_{n, b} \text{ and } \Vert \bbeta_{\ini} - \bbeta \Vert \lesssim r_{n, \bbeta_{\ini}}.
\end{split}
\end{equation*}
If we further assume that conditions~(i)--(vi) of Theorem~3 in \cite{kallus2024localized} hold  after the substitutions  
$\theta\mapsto\bbeta,\; \mu\mapsto b,\; U+V\mapsto\psi$ and that $r_{n, \xi} ( r_{n, b} + r_{n, \bbeta_{\ini}}) = o (n^{- 1 / 2})$, then
\begin{align*}
\sqrt{n} (\hat{\bbeta}^{(1)} - \bbeta^{\ast}) \overset{d}{\rightarrow} N (0, \sigma^{2}),
\end{align*}
where $\sigma^{2} = \var (\IF_{\bbeta^*}^{(1)})$ and $\IF_{\bbeta^*}^{(1)}$ is the EIF of $\bbeta^*$ defined in \eqref{eif_beta}. 
\end{proposition}

Proposition~\ref{prop:dml} shows that, for $\hat{\bbeta}^{(1)}$ to be $n^{-1/2}$-consistent and asymptotically normal, the initial estimator $\bbeta_{\ini}$ must converge to $\bbeta^{\ast}$ fast enough to accurately estimate the nuisance function $b_{\bbeta}$. If $\bbeta_{\ini}$ is obtained by solving the empirical version of \eqref{est eq} using estimated stabilized weights $\hat{\xi}$ that converge to the true $\xi$ at rate $r_{n,\xi}$, then $r_{n,\bbeta_{\ini}} = O(r_{n,\xi})$. Thus, unless $r_{n,\bbeta_{\ini}}$ converges to zero at an even faster rate than $r_{n,\xi}$ (a highly unlikely scenario), we require both $r_{n,\xi}$ and $r_{n,b}$ to be $o_{\mathbb{P}}(n^{-1/4})$.

Under Assumption~\ref{as:Holder}, these rate conditions imply that $\frac{2s_{1}/d}{1 + 2 s_{1}/d} > 0.5$ and $\frac{s_{1}/d}{1 + 2 s_{1}/d} + \frac{s_{2}/d}{1 + 2 s_{2}/d} > 0.5$.  In the special case $s_{1}=s_{2}=s$, this further reduces to $s/d > 0.5$, under which both nuisance parameters even satisfy the Donsker condition. However, inspired by the known optimal rate for ATE estimation under \Holder{}-type nuisance assumptions \citep{robins2023minimax}, it is reasonable to conjecture that $s/d \geq 0.25$ suffices to guarantee a $n^{-1/2}$‐consistent estimator for finite-dimensional parameters of GTM under reasonable regularity conditions on the generalized residual function $\Gamma$. The remainder of this section is devoted to constructing a higher-order estimator of $\bbeta^{\ast}$ requiring much relaxed \Holder{}-type smoothness assumptions on the nuisance parameters and the initial $\bbeta_{\mathsf{init}}$, by leveraging the HOIF framework.


\section{Higher-Order Estimators: The Finite-Dimensional Case}
\label{sec:hoif}

In this section, we present our main methodological and theoretical results on new higher-order estimators for finite-dimensional GTM parameters. We first motivate the need for higher-order estimators beyond the first-order estimator $\hat{\bbeta}^{(1)}$ through a simple numerical experiment in Section~\ref{sec:review}, accompanied by a brief and intuitive introduction to HOIF theory. Section~\ref{sec:intuition} then discusses how to construct higher-order estimators for GTM parameters, and Section~\ref{sec:properties} studies their statistical properties. Finally, Section~\ref{sec:examples} specializes these general results to the QTE.

\subsection{Motivation for higher-order estimators: The QTE example}
\label{sec:review}
Before proceeding, we provide a more technically rigorous motivation for constructing higher-order estimators for GTMs, complementing the empirical motivation based on the simulation results in Figure~\ref{fig:QTE_Bias_comparison} in the Introduction. We focus on $\beta_{1}^{\ast}$, the $\tau$-quantile of the potential outcome $Y(1)$, which was shown to be a GTM parameter in Example~\ref{qte}. As discussed in Remark~\ref{rem:binary QTE} and Proposition~\ref{prop:dml} (following \citet{kallus2024localized}), we can analyze the bias of the LDML estimator $\hat{\beta}^{(1)}$ by examining the bias of the first-order moment equation $\hat{\psi}^{(1)}(\beta_1)$ for $\psi(\beta_1)$ at a given $\beta_1$:
\begin{equation*}
\begin{split}
\E \{\hat{\psi}^{(1)} (\beta_{1}) - \psi (\beta_{1})\} & = \E \left[ \left\{ \frac{\P (T = 1 \mid X)}{\hat{\P} (T = 1 \mid X)} - 1 \right\} \{b_{\beta_{1}} (X, T = 1) - \hat{b}_{\beta_{1, \mathsf{init}}} (X, T = 1)\} \right] \\
& = \E \left[ \{\hat{\xi} (X, T) - \xi (X, T)\} \{b_{\beta_{1}} (X, T) - \hat{b}_{\beta_{1, \mathsf{init}}} (X, T)\} \right],
\end{split}
\end{equation*}
where we recall that $\xi (x, t) = t / \P (T = 1 \mid X = x)$, $b_{\beta_{1}} (x, t) = t \cdot \E [\tau - \mathbbm{1} \{Y \leq \beta_{1}\} \mid X = x, T = t]$, and $\hat{\xi}$ and $\hat{b}_{\beta_{1}}$ are their respective estimates.  Recall that $\hat{\xi}$ and $\hat{b}_{\beta_{1,\ini}}$ are constructed on the nuisance sample, whereas the expectation is taken with respect to the main sample. Hence, conditional on the nuisance sample, $\hat{\xi}$ and $\hat{b}_{\beta_{1,\ini}}$ can be treated as fixed functions, and the bias $\E [ \{\hat{\xi} (X, T) - \xi (X, T)\} \{b_{\beta_{1}} (X, T) - \hat{b}_{\beta_{1, \mathsf{init}}} (X, T)\} ]$ captures only the approximation bias of these fixed estimators, not the sampling variability in the nuisance estimates.

Decomposing  $b_{\beta_{1}} - \hat{b}_{\beta_{1, \mathsf{init}}}$ as $(b_{\beta_{1}} - b_{\beta_{1, \mathsf{init}}}) + (b_{\beta_{1, \mathsf{init}}} - \hat{b}_{\beta_{1, \mathsf{init}}})$  and applying the triangle inequality, we can bound the bias of $\hat{\psi}^{(1)}(\beta_1)$ under suitable regularity conditions on $\beta_1 \mapsto b_{\beta_1}$ as follows:
\begin{align}
\label{bias1}
\Vert \hat{\xi} - \xi \Vert_{\P, 2} (\Vert \hat{b}_{\beta_{1, \mathsf{init}}} - b_{\beta_{1, \mathsf{init}}} \Vert_{\P, 2} + \Vert \beta_{1, \mathsf{init}} - \beta_{1} \Vert).
\end{align}
Since $\beta_{1, \mathsf{init}}$ depends on $\hat{\xi}$, the bias of $\hat{\psi}^{(1)}(\beta_1)$ is consequently bounded by
\begin{align}
\label{bias2}
\Vert \hat{\xi} - \xi \Vert_{\P, 2} \Vert \hat{b}_{\beta_{1, \mathsf{init}}} - b_{\beta_{1, \mathsf{init}}} \Vert_{\P, 2} + \Vert \hat{\xi} - \xi \Vert_{\P, 2}^{2}.
\end{align}
If $\hat{\xi}$ or $\hat{b}_{\beta_{1, \mathsf{init}}}$ does not converge sufficiently fast to their true values $\xi$ or $b_{\beta_{1, \mathsf{init}}}$, the bias of $\hat{\psi}^{(1)} (\beta_{1})$ might not be sufficiently small, as indicated by the bounds \eqref{bias1} and \eqref{bias2}.

The higher-order estimator $\hat{\beta}_{1}^{(2)}$ is motivated by the HOIF theory initially developed in \citet{robins2008higher}, which is essentially a bias reduction technique. To apply the HOIF theory to estimate $\beta_{1}^{\ast}$, we first fix a set of basis functions of $(X, T)$ of dimension $k$, denoted by $\bar{\phi}_{k}$. Since $T$ is binary, we have $\bar{\phi}_{k} (x, t) = t \sfzbar_{k} (x)$ for some basis functions $\sfzbar_{k}$ over $X$. Given $\bar{\phi}_{k}$, we introduce the following projections of the residuals $\hat{\xi} - \xi$ and $\hat{b}_{\beta_{1, \mathsf{init}}} - b_{\beta_{1}}$ onto the space spanned by $\phi_{k}$ (defined in Section~\ref{sec:notation}):
\begin{align*}
& \Pi (\hat{\xi} - \xi \mid \bar{\phi}_{k}) (x, t) = \bar{\phi}_{k} (x, t)^{\top} \Sigma_{k}^{-1} \E [\{\hat{\xi} (X, T) - \xi (X, T)\} \bar{\phi}_{k} (X, T)], \\
& \Pi (\hat{b}_{\beta_{1, \mathsf{init}}} - b_{\beta_{1}} \mid \bar{\phi}_{k}) (x, t) = \bar{\phi}_{k} (x, t)^{\top} \Sigma_{k}^{-1} \E [\{\hat{b}_{\beta_{1, \mathsf{init}}} (X, T) - b_{\beta_{1}} (X, T)\} \bar{\phi}_{k} (X, T)],
\end{align*}
where $\Sigma_{k} = \E \{\bar{\phi}_{k} (X, T)^{\otimes 2}\} = \E \{T \sfzbar_{k} (X)^{\otimes 2}\}$.

Although it is generally impossible to estimate the bias of an estimator without imposing strong assumptions, the following term can be estimated unbiasedly by a second-order $U$-statistic computable from the data (if $\Sigma_{k}$ is not known, we can plug-in some estimator $\hat{\Sigma}_{k}$ of $\Sigma_{k}$):
\begin{align*}
\mathrm{B}_{\psi, k} (\beta_{1}) & \coloneqq \E \{\Pi (\hat{\xi} - \xi \mid \bar{\phi}_{k}) (X, T) \cdot \Pi (b_{\beta_{1}} - \hat{b}_{\beta_{1, \mathsf{init}}} \mid \bar{\phi}_{k}) (X, T)\} \\
& = \E [\{\hat{\xi} (X, T) - \xi (X, T)\} \bar{\phi}_{k} (X, T)^{\top}] \Sigma_{k}^{-1} \E [\{b_{\beta_{1}} (X, T) - \hat{b}_{\beta_{1, \mathsf{init}}} (X, T)\} \bar{\phi}_{k} (X, T)] \\
& = \E [\{\hat{\xi} (X, T) - 1\} \sfzbar_{k} (X)^{\top}] \Sigma_{k}^{-1} \E [\{T (\tau - \mathbbm{1} \{Y \leq \beta_{1}\}) - \hat{b}_{\beta_{1, \mathsf{init}}} (X, T)\} \sfzbar_{k} (X)],
\end{align*}
as it is a product of two expectations of random variables that can be directly evaluated from the data. The last equality follows from $\bar{\phi}_k(x,t) = t \sfzbar_k(x)$ and $\xi(x,t) = t / \P (T = 1 \mid X = x)$, which give $\bbE \{\xi(X,T) \bar{\phi}_k(X,T)\} = \bbE \{\sfzbar_k(X)\}$, and by the law of iterated expectations:
\begin{align*}
\E \{b_{\beta_{1}} (X, T) \bar{\phi}_{k} (X, T)\} = \E \{T \cdot \E (\tau - \mathbbm{1} \{Y \leq \beta_{1}\} \mid X, T) \cdot \sfzbar_k(X)\} = \E \{T (\tau - \mathbbm{1} \{Y \leq \beta_{1}\})  \sfzbar_{k} (X)\}.
\end{align*}
Furthermore, the bias of $\hat{\psi}^{(1)} (\beta_{1})$ can be decomposed as
\begin{align*}
\E \{\hat{\psi}^{(1)} (\beta_{1}) - \psi (\beta_{1})\} = \mathrm{B}_{\psi, k} (\beta_{1}) + \E \{\Pi^{\perp} (\hat{\xi} - \xi \mid \bar{\phi}_{k}) (X, T) \cdot \Pi^{\perp} (b_{\beta_{1}} - \hat{b}_{\beta_{1, \mathsf{init}}} \mid \bar{\phi}_{k}) (X, T)\},
\end{align*}
where $\Pi^{\perp}$ denotes the orthocomplement of the projection and following \citet{robins2008higher}, we denote the second term in the decomposition as $\TB_{\psi, k} (\beta_{1})$. $\TB_{\psi, k} (\beta_{1})$ is referred to as the \emph{truncation bias} in \citet{robins2008higher}, because it represents the remaining bias, after ``truncating'' the potentially infinite-dimensional residual functions $\hat{\xi} - \xi$ and $b_{\beta_{1}} - \hat{b}_{\beta_{1, \mathsf{init}}}$ up to a finite dimension $k$ by the basis $\bar{\phi}_{k}$. Under appropriate assumptions on  $\hat{\xi} - \xi$ and  $b_{\beta_{1}} - \hat{b}_{\beta_{1, \mathsf{init}}}$ and a suitable choice of  $\bar{\phi}_{k}$ (to be clarified later), $\TB_{\psi, k} (\beta_{1})$ can be made negligible compared to sampling variability. 

The key idea of our higher-order estimation approach is to decompose the bias of the first-order estimator $\hat{\psi}^{(1)} (\beta_{1})$ of $\psi (\beta_{1})$ into two parts $\mathrm{B}_{\psi, k} (\beta_{1}) + \TB_{\psi, k} (\beta_{1})$, so that $\mathrm{B}_{\psi, k} (\beta_{1})$ can be estimated at a sufficiently fast rate by an estimator $\hat{\mathrm{B}}_{\psi, k}$ and $\TB_{\psi, k} (\beta_{1})$ is negligible. We then solve a bias-reduced estimating equation $\hat{\psi}^{(1)} - \hat{\mathrm{B}}_{\psi, k}$ (typically a $U$-statistic) to construct the higher-order estimator $\hat{\beta}_{1}^{(2)}$. We summarize the construction of $\hat{\beta}_{1}^{(2)}$ in Figure~\ref{fig:alg}.

\begin{figure}
    \centering
        \begin{tikzpicture}[
    >=Latex,
    font=\small,
    node distance=1.8cm and 2.0cm,
    box/.style={
        draw,
        semithick,
        rectangle,
        minimum width=2.6cm,
        minimum height=1.0cm,
        align=center
    },
    plug/.style={->,semithick},
    est/.style={
    ->,
    semithick,
    decorate,
    decoration={snake, amplitude=0.9pt, segment length=5pt},
    shorten >=-1.1pt
},
    every node/.style={inner sep=2pt}
]

\node[box] (nuis) at (0,1) {nuisance sample};
\node[box, below=0.5cm of nuis] (main) {main sample};

\node[left=0.8cm of nuis] (xi) {$\hat{\xi}$};
\node[above=0.5cm of nuis] (betainit) {$\beta_{1,\mathrm{init}}$};
\node[right=1.2cm of nuis.north east, anchor=west, yshift=-0.1cm] (bhat) {$\hat{b}_{\beta_{1,\mathrm{init}}}$};
\node[right=0.8cm of nuis.south east, anchor=west,, yshift=0.1cm] (sigma)
    {$\hat{\Sigma}_k$};

\node[ right=2.2cm of main] (grp) {$
\left\{
\begin{array}{c}
\hat{\psi}^{(1)}\\
\hat{\mathrm{B}}_{\psi,k}
\end{array}
\right\}
$};

\node[right=1.3cm of grp] (phi2) {$\hat{\psi}^{(2)}_k = \hat{\psi}^{(1)} - \hat{\mathrm{B}}_{\psi,k}$};
\node[right=1.3cm of phi2] (beta2) {$\hat{\beta}^{(2)}_1$};

\draw[est] (nuis.west) -- (xi.east);
\draw[est] (nuis.north) -- (betainit.south);

\draw[plug] (xi.north) |- (betainit.west);
\draw[est] ([, yshift=-0.1cm]nuis.north east) -- (bhat.west);
\draw[est] ([, yshift=0.1cm]nuis.south east) -- (sigma.west);
\draw[plug] (betainit.east) -| ([xshift=-0.8cm]bhat.north);
\draw[plug] (bhat.south) --([yshift=-1.6cm]bhat.south);
\draw[plug] (sigma.south) --([yshift=-0.8cm]sigma.south);
\draw[plug]
    ([yshift=-0.1cm]xi.south)
    -- ([yshift=-2.4cm]xi.south)
    -- ([yshift=-2.8cm]bhat.south)
    -- ([yshift=-1.7cm]bhat.south);

\draw[plug] (main.east) -- (grp.west);
\draw[plug] (grp.east) -- (phi2.west);
\node[right=0.4cm of phi2] (arr) {\large$\Longrightarrow$};
\end{tikzpicture}
    \caption{A schematic illustration on how to construct the higher-order estimator $\hat{\beta}_{1}^{(2)}$ of $\beta_{1}^*$ in the example of QTE. $\hat{\mathrm{B}}_{\psi, k}$ is a second-order $U$-statistic estimator of $\mathrm{B}_{\psi, k}$.}
    \label{fig:alg}
\end{figure}
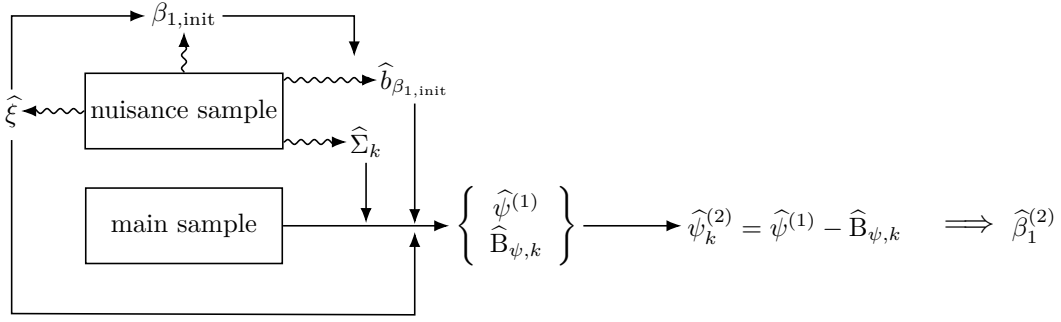

As mentioned in the Introduction, Figure~\ref{fig:QTE_Bias_comparison} displays the histograms of different QTE estimators after centered on the true QTE value (i.e., the estimate minus the truth) in a simulation setting where nuisance estimation errors $\|\hat{\xi} - \xi\|_{\P,2}$ and $\|\hat{b}_{\beta_{1,\ini}} - b_{\beta_{1,\ini}}\|_{\P,2}$ are designed to be large. In the left panel, the LDML estimator exhibits a noticeable bias relative to its sampling variability. In contrast, the right panel shows that the histogram of our higher-order estimators is much more centered around zero, indicating a negligible bias relative to its sampling variability. The variance of the higher-order estimator is only slightly larger than that of the LDML estimator. These results demonstrate that the proposed higher-order estimator can be useful not only in theory but also in practice.

\subsection{Higher-order estimators: The general case}
\label{sec:intuition}

In the general case, we begin by noting that the bias of the first-order estimator $\hat{\psi}^{(1)} (\bbeta)$ for the moment equation $\psi(\bbeta)$ can also be written as the mean of a product of two nuisance estimation residuals:
\begin{align}\label{eq:bias}
\bbE \{\hat{\psi}^{(1)}(\bbeta) - \psi(\bbeta)\} = \bbE [\{\hat{\xi} (X, T) - \xi (X, T)\} \{b_{\bbeta} (X, T) - \hat{b}_{\bbeta_{\mathsf{init}}} (X, T)\}].
\end{align}
To reduce this bias, we project each residual onto a size-$k$ dictionary $\bar{\phi}_{k} \coloneqq \{\phi_{1}, \cdots, \phi_{k}\}$ functions.
Specifically, we select $\bar{\phi}_{k} \equiv \sfzbar_{k_{x}} \otimes \sfwbar_{k_{t}}$ as the tensor product of a  size-$k_{x}$ dictionary $\sfzbar_{k_{x}} \coloneqq \{\mathsf{z}_{1}, \cdots, \mathsf{z}_{k_{x}}\}$ over $\calX$ and a size-$k_{t}$ dictionary $\sfwbar_{k_{t}} \coloneqq \{\mathsf{w}_{1}, \cdots, \mathsf{w}_{k_{t}}\}$ over $\calT$, as in general $T$ is no longer binary. Using the same projection arguments as in Section~\ref{sec:review}, the bias decomposes as
\begin{align*}
\bbE \{\hat{\psi}^{(1)} (\bbeta) - \psi (\bbeta)\} 
& = \bbE \{\Pi (\hat{\xi} - \xi \mid \bar{\phi}_{k}) (X, T) \cdot \Pi (b_{\bbeta} - \hat{b}_{\bbeta_{\mathsf{init}}} \mid \bar{\phi}_{k}) (X, T)\} \\
&\quad + \bbE \{\Pi^{\perp} (\hat{\xi} - \xi \mid \bar{\phi}_{k}) (X, T) \cdot \Pi^{\perp} (b_{\bbeta} - \hat{b}_{\bbeta_{\mathsf{init}}} \mid \bar{\phi}_{k}) (X, T)\} \\
& \eqqcolon \mathrm{B}_{\psi, k} (\bbeta)+ \mathrm{TB}_{\psi, k}  (\bbeta),
\end{align*}
where $\Pi$ denotes projection onto the span of $\bar{\phi}_k$ and $\Pi^{\perp}$ is its orthocomplement. Moreover,
\begin{equation}
\label{EB}
\mathrm{B}_{\psi, k} (\bbeta) = \bbE [\{\hat{\xi} (X, T) - \xi (X, T)\} \bar{\phi}_{k} (X, T)^{\top}] \Sigma_{k}^{-1} \bbE [\bar{\phi}_{k} (X, T) \{b_{\bbeta} (X, T) - \hat{b}_{\bbeta_{\mathsf{init}}} (X, T)\}].
\end{equation}
A key insight underlying the improved second-order estimator is that the bias term $\mathrm{B}_{\psi,k}(\bbeta)$ admits an unbiased sample analogue, if $\Sigma_k$ is known. First, by the law of iterated expectations,
\begin{align*}
\mathbb{E}[\bar{\phi}_k (X,T) \{b_{\bbeta}(X,T)-\hat{b}_{\bbeta_{\mathsf{init}}}(X,T)\}] = \mathbb{E}[\bar{\phi}_k(X,T)\{\Gamma(Y,T,\bbeta)-\hat{b}_{\bbeta_{\mathsf{init}}}(X,T)\}].
\end{align*}
Next, recall that $\xi(X,T)=\p_T(T)/\p_{T\mid X}(T\mid X)$. Weighting by $\xi$ transforms the joint distribution of $(X,T)$ into the product of their marginal distributions. Indeed,
\begin{align*}
\mathbb{E}[\xi(X,T)\bar{\phi}_k(X,T)] = \iint \bar{\phi}_k(x,t)\,\p_T(t)\,\p_X(x) \diff t \diff x.
\end{align*}
This identity is crucial: it shows that the expectation of $\xi(X,T)\bar{\phi}_k(X,T)$ can be estimated by averaging $\bar{\phi}_k(X_i,T_j)$ over independent pairs $(X_i,T_j)$ from the main sample. Combining these two observations, we obtain an unbiased estimator of $\mathrm{B}_{\psi,k}(\bbeta)$ as the following third-order $U$-statistic:
\begin{align*}
\tilde{\mathrm{B}}_{\psi,k}(\bbeta) = \mathbb{U}_{n,3}\Big[\big\{\hat{\xi}(X_1,T_1)\bar{\phi}_k(X_1,T_1) - \bar{\phi}_k(X_1,T_3)\big\}^{\top} \Sigma_k^{-1} \bar{\phi}_k(X_2,T_2) \big\{\Gamma(Y_2,T_2,\bbeta)-\hat{b}_{\bbeta_{\mathsf{init}}}(X_2,T_2)\big\}\Big].
\end{align*}
Recall that the $U$-statistics operator $\mathbb{U}_{n,3}$ is taken over the main sample. We then construct the improved oracle
second-order estimator of $\psi (\bbeta)$ as
\begin{align*}
\tilde{\psi}_{ k}^{(2)} (\bbeta) \coloneqq \hat{\psi}^{(1)} (\bbeta) - \tilde{\mathrm{B}}_{\psi,k} (\bbeta).
\end{align*} 
We call $\tilde{\psi}_{ k}^{(2)}$ an oracle estimator because it requires the knowledge of $\Sigma_{k}$, which is generally unknown. The corresponding oracle second-order estimator $\tilde{\bbeta}^{(2)}$ for  $\bbeta^*$ is defined as the solution to $\tilde{\psi}_{k}^{(2)} (\bbeta) = 0$. When $\Sigma_{k}$ is replaced by an estimator $\hat{\Sigma}_{k}$ (see Remark~\ref{rem:unknown Sigma}) computed from the nuisance sample, we obtain the feasible estimators $\hat{\mathrm{B}}_{\psi, k}$ and $\hat{\psi}_{k}^{(2)} (\bbeta) \coloneqq \hat{\psi}^{(1)} (\bbeta) - \hat{\mathrm{B}}_{\psi, k}(\bbeta)$. The feasible estimator $\hat{\bbeta}^{(2)}$ for $\bbeta^*$ is then defined as the solution to $\hat{\psi}_k^{(2)}(\bbeta) = 0$.

Let $\bar{\psi}_{k} (\bbeta) \coloneqq \bbE [\tilde{\psi}_{k}^{(2)} (\bbeta)]$ denote the truncated parameter of $\psi (\bbeta) $ \citep{robins2016technical}. The bias of the second-order estimator $\tilde{\psi}_{k}^{(2)} (\bbeta)$ then reduces to
\begin{align*}
\bbE \{\tilde{\psi}_{k}^{(2)} (\bbeta) - \psi (\bbeta)\} = \bar{\psi}_{k} (\bbeta)- \psi (\bbeta) = \TB_{\psi, k} (\bbeta),
\end{align*}
which can be made much smaller than the original bias. Remark~\ref{rk:improved rate} discusses why this yields an improved convergence rate for  $\tilde{\psi}_{k}^{(2)} (\bbeta)$ under \Holder{}-type assumptions in Assumption~\ref{as:Holder}. Formally analyzing the statistical properties of $\tilde{\bbeta}^{(2)}$ and $\hat{\bbeta}^{(2)}$ requires a detailed examination of the impact of the initial estimator $\bbeta_{\ini}$ and heavy use of $U$-processes theory, since  they are third-order $Z$-estimators. We will make these  intuitive arguments precise in Section~\ref{sec:properties}.

\begin{remark}\label{rk:bias}
The initial nuisance estimators $\hat{\xi}$ and $\hat{b}_{\bbeta_{\mathsf{init}}}$ can be obtained by any flexible method (e.g., kernel smoothing, series estimation, or machine learning) and are treated as fixed functions conditional on the nuisance sample. The dictionary $\bar{\phi}_k$ is introduced solely for the purpose of bias correction: we project the residual functions $\hat{\xi}-\xi$ and $b_{\bbeta}-\hat{b}_{\bbeta_{\mathsf{init}}}$ onto the linear span of $\bar{\phi}_k$, and the bias correction term $\tilde{\mathrm{B}}_{\psi,k}(\bbeta)$ estimates the inner product of these projections. The improvement comes from choosing $\bar{\phi}_k$ with sufficient approximation power for these residuals;  in particular, the dictionary used for bias correction should be richer than or different from any dictionary used in the initial nuisance estimation. 
\end{remark}

\begin{remark}\label{rk:improved rate}
We now provide intuition for the improved convergence rate of the debiased estimator $\tilde{\psi}_{k}^{(2)} (\bbeta)$. For simplicity, we ignore the impact of $\bbeta_{\ini}$ by assuming $\bbeta_{\ini} = \bbeta$.  Under Assumption~\ref{as:Holder}
and with an appropriate choice of basis  $\bar{\phi}_{k}$ \citep{belloni2015some}, the truncation bias satisfies
 $\TB_{\psi, k}(\bbeta) \asymp k^{- 2 s/d}$, where $s$ is the average smoothness defined in Assumption~\ref{as:Holder}. The variance of $\tilde{\psi}_{k}^{(2)} (\bbeta)$ is of order $\frac{k}{n^{2}} + \frac{1}{n}$. Balancing the truncation bias against the variance by choosing $k \asymp n^{\frac{2 d}{d + 4 s}}$ yields
\begin{align*}
\left[ \bbE \{\tilde{\psi}_{k}^{(2)} (\bbeta) - \psi (\bbeta)\}^{2} \right]^{1 / 2} \lesssim \left\{ \begin{array}{ll}
n^{-1 / 2} & s \geq \frac{d}{4}, \\
n^{- \frac{4 s}{d + 4 s}} & s < \frac{d}{4}.
\end{array} \right.
\end{align*}
Thus, compared to the first-order estimator $\hat{\psi}^{(1)}$, the debiased estimator $\tilde{\psi}_k^{(2)}$ relaxes the smoothness requirement for $\sqrt{n}$-consistency. Specifically, $\tilde{\psi}_k^{(2)}$ achieves the $\sqrt{n}$ rate whenever $s \ge d/4$, whereas $\hat{\psi}^{(1)}$ fails in the regime $d/4 < s < d/2$ (see the discussion following Proposition~\ref{prop:dml}).
\end{remark}

\begin{remark}
\label{rem:unknown Sigma}
Following \cite{mcgrath2024nuisance} and \cite{chen2024method}, we focus on the case where $\Sigma_{k}^{-1}$ is known or some estimator $\hat{\Sigma}_{k}^{-1}$ of $\Sigma_{k}^{-1}$ is sufficiently close to $\Sigma_{k}^{-1}$, which generally requires extra smoothness assumptions on the marginal density $\p_{X}$ of $X$. How the minimax-optimal rate depends on the smoothness of $\p_{X}$ is a long-standing open problem \citep{richardson2014causal} and is beyond the scope of this paper. For completeness, in Supplementary Material Section~\ref{app:truncation}, we derive the HOIFs of $\psi (\bbeta)$ (or more precisely, of $\bar{\psi}_{k} (\bbeta) $) and construct higher-order estimators of arbitrary order $m$, facilitating future work in this direction. When $k = o (n)$, we briefly discuss in Supplementary Material Section \ref{app:unknown} how to estimate $\bbeta$ using higher-order estimators with $m \asymp \log n$ without any smoothness assumption on $\p_{X}$, where $\Sigma_{k}$ is estimated by its sample average estimator.
\end{remark}


\subsection{Asymptotic properties of higher-order estimators: The general case}
\label{sec:properties}

\emph{En route} to deriving statistical properties of our proposed second-order estimators, we further require
the following regularity conditions. When stating these conditions, we introduce positive and finite constants $c_{3}, \ldots, c_{8}$ that do not vary with $n$. 

\begin{assumption}
\label{as:bounded}
 The support $\mathcal{T}$ of the treatment variable $T$ is a compact subset of $\bbR$. 
\end{assumption}

\begin{assumption}\label{as:uniqueness}
(i) The parameter space $\calB$ is a compact subset of $\bbR^{p}$ and the true parameter $\bbeta^{*}$ is in the interior of $\calB$. (ii)
$\bbeta^{\ast}$ is the unique solution to the population first-order estimating equation $\psi (\bbeta) = 0$ and for every $\epsilon > 0$, $\inf_{\|\bbeta - \bbeta^{\ast}\| \geq \epsilon} \|\psi (\bbeta)\| > 0$.
\end{assumption}
\begin{assumption}\label{as:smoothness}
(i) $\psi (\bbeta)$ is differentiable with respect to $\bbeta$ at any $\bbeta \in \calB$; denote its derivative with respect to $\bbeta$ as $\nabla_{\bbeta} \psi (\bbeta) \in \bbR^{p \times p}$. (ii) Each component of $\nabla_{\bbeta} \psi (\bbeta)$ is continuous at $\bbeta^{\ast}$. (iii) The eigenvalues of $\nabla_{\bbeta} \psi (\bbeta^{\ast})$ are bounded between constants $c_3$ and $c_4$. 
\end{assumption}


\begin{assumption}\label{as:nuisance_est}
(i) $\sup_{\bbeta \in \calB} \|b_{\bbeta} (\cdot, \cdot)\|_{\infty} \lesssim 1$,  $\sup_{\bbeta \in \calB} \|\hat{b}_{\bbeta} (\cdot, \cdot)\|_{\infty} \lesssim 1$ and $ \|\hat{\xi} (\cdot, \cdot) \|_{\infty}  \lesssim 1$. (ii) For any $\bbeta_{1}, \bbeta_{2} \in\calB$, we have $   \|  {b}_{\bbeta_{1}}(X,T)-{b}_{\bbeta_{2}}  (X,T) \|_{\P,2} \lesssim \| \bbeta_{1} - \bbeta_{2}\|.$ (iii) For any $(x, t) \in \calX \times \calT$, we have $ \ c_5 \leq \p_{X}(x)\leq c_6 $ and $ \  c_5 \leq \p_{T}(t)\leq c_6$.
\end{assumption}

\begin{assumption}
\label{as:criterion}
(i) The function class \ $\left\{ \Gamma (\cdot, \cdot, \bbeta): \bbeta \in \calB \right\}$ is of VC type with fixed characteristic $(A, v)$. (ii) $\sup_{\bbeta \in \calB} \| \Gamma (\cdot,\cdot,\bbeta) \|_{\infty} \lesssim 1$. (iii) There exists some constant $\alpha_{0}\in(0,1]$ such that, for any $\bbeta_{1}, \bbeta_{2} \in\calB$, any small $\delta>0$,  $\bbE^{1/2} [\sup_{||\bbeta_{1}-\bbeta_{2}||\lesssim\delta} \left\| \Gamma(Y,T,\bbeta_1)-\Gamma(Y,T,\bbeta_2) \right\|^2 ]\lesssim \delta^{\alpha_{0}}$. 
\end{assumption} 

\begin{assumption}
\label{as:basis}
(i) For each $j$, the $j^{th}$ basis functions $\mathsf{z}_{j}(\cdot)$ and $\mathsf{w}_{j}(\cdot)$ are supported on sets of probability at most $c_7/k_x$ and $c_7/k_t$, respectively. Let $\calS_{X,j}$ and $\calS_{T,j}$ denote the supports of  $\mathsf{z}_{j}(\cdot)$ and $\mathsf{w}_{j}(\cdot)$, respectively. (ii) At any point $(x,t)\in\calX \times \calT$, at most $c_8$ elements of $\sfzbar_{k_{x}} (x)$ and $\sfwbar_{k_{t}} (t)$
 are simultaneously non-zero.  (iii) For every $k_{x}$ and $k_{t}$, and for any unit vectors $\bma=(a_1,\ldots,a_{k_x})\in \bbR^{k_x}$ and $\bmb=(b_1,\ldots,b_{k_t})\in \bbR^{k_t}$, we have $\bma^{\top}\int_{\calS_{X,j}} \sfzbar_{k_{x}} (x)\sfzbar_{k_{x}} (x)^{\top} \diff x \bma \gtrsim a_j^2 $ and $\bmb^{\top}\int_{\calS_{T,j}}  \sfwbar_{k_{t}} (t) \sfwbar_{k_{t}} (t)^{\top}  \diff t \bmb \gtrsim b_j^2$.
(iv) There exist
two sequences of  constants $\zeta_{x}(k_{x})$ and $\zeta_{t}(k_{t})$
such that $\Vert \sfwbar_{k_{t}} (\cdot) \Vert_{\infty} \leq \zeta_{t}(k_{t}) \lesssim \sqrt{k_t}$, $\Vert \sfzbar_{k_{x}} (\cdot) \Vert_{\infty} \leq \zeta_{x}(k_{x}) \lesssim \sqrt{k_x}$, and   $\zeta(k)\coloneqq\zeta
_{x}(k_{x})\zeta_{t}(k_{t})$. (v) The dimension $k=k_x k_t$ satisfies  $k \to \infty$  and $k = o (n^2)$ as $n \to \infty$.
\end{assumption}

 Assumption~\ref{as:bounded} requires that the treatment $T$ be bounded. This condition is imposed to simplify the analysis, as most existing technical tools for $U$-statistics and $U$-processes require the $U$-statistic kernel to be bounded. We expect that it could be weakened to light-tailed conditions \citep{chakrabortty2025tail}, but this is not the main focus of our paper. Assumptions~\ref{as:uniqueness} and~\ref{as:smoothness}  are standard identification and regularity conditions for $Z$-estimators. Assumption~\ref{as:nuisance_est} imposes sufficient regularity conditions on the nuisance functions and their estimators.  In particular, Assumption~\ref{as:nuisance_est}(iii), together with Assumption~\ref{as:identification}(ii), implies that $\xi(\cdot,\cdot)$ is bounded away from zero and infinity, which is commonly required in the existing literature; see \cite{ai2021unified} and \cite{kennedy2017non}, among others.

Assumption~\ref{as:criterion}(i) is a high-level complexity condition on the generalized residual function that defines GTM treatment effect parameters in \eqref{ident}. In applications, this assumption shall be verified case-by-case.   
Assumptions~\ref{as:criterion}(ii)-(iii) control the envelope functions, a standard requirement in $M/Z$-estimation when the criterion is non-smooth (see, e.g., \citealt{chen2003estimation,ai2021unified}). Notably, Assumption~\ref{as:criterion}(iii) is weaker than the Lipschitz continuity condition and accommodates various continuous and discontinuous functions. Fortunately, in various examples, these assumptions are straightforward to verify. In Supplementary Material Section~\ref{app:Verification}, we show that they hold for both ATE and QTE from Examples~\ref{ate} and~\ref{qte}.

Assumptions~\ref{as:basis}(i)--(ii) ensure control over the sup norm of series projection estimators, which are essential for analyzing the $U$-process and establishing variance bounds.  These assumptions are satisfied by many commonly used sieve bases, including Cohen–Daubechies–Vial wavelet series, B-splines and local polynomial partition series;  see \citealp{belloni2015some,  liu2017semiparametric,cattaneo2020large} for detailed discussions. Assumption~\ref{as:basis}(iii) imposes a mild local non-collinearity requirement on the basis functions and is verified by the basis functions  mentioned above. This condition is also adopted in \citet{cattaneo2020large}.
Assumption~\ref{as:basis}(iv) is standard in the nonparametric regression literature and is verified for the basis functions  mentioned above.    Assumption~\ref{as:basis}(v) restricts the basis dimension $k$ from growing too fast (slower than $n^{2}$), a necessary condition to achieve asymptotic normality of the estimators.

We now establish the consistency and convergence rate of the oracle second-order estimator $\tilde{\bbeta}^{(2)}$.
\begin{theorem}
\label{th:consistency}
Suppose that Assumptions~\ref{as:identification} and \ref{as:bounded}--\ref{as:basis} hold. 
Let $\tilde{r}_{n,\bbeta} \to 0$ be a diminishing sequence such that
\begin{align*}
   \bigg( \frac{\sqrt{k}}{n} \vee \frac{1}{\sqrt{n}}\bigg)\cdot \log n  
+ \sup_{\bbeta\in\calB}\| \Pi^{\perp} [\hat{\xi} - \xi \mid \bar{\phi}_{k}] \|_{\P,2} \cdot \|  \Pi^{\perp} [b_{\bbeta} - \hat{b}_{\bbeta_{\mathsf{init}}} \mid \bar{\phi}_{k}]  \|_{\P,2} \lesssim \tilde{r}_{n,\bbeta}, 
\end{align*}
where $\Pi^{\perp}[f \mid \bar{\phi}_{k}]$ denotes the orthocomplement of the projection of $f$ onto the linear span of $\bar{\phi}_{k}$ (see Section~\ref{sec:notation}). Then $\tilde{\bbeta}^{(2)}$ is consistent and satisfies $\bbE[\|\tilde{\bbeta}^{(2)} - \bbeta^{*}\|] \lesssim \tilde{r}_{n,\bbeta}$.
\end{theorem}

The proof of Theorem~\ref{th:consistency} is presented in Section~\ref{app:beta-consistency} of the Supplementary Material. The convergence rate $\tilde{r}_{n,\bbeta}$ comprises two components: the first term captures variance, and the second term represents the bias arising from approximating the residuals of the estimated and true nuisance parameters using the basis $\bar{\phi}_{k}$.

 The next result characterizes the asymptotic distribution of $\tilde{\bbeta}^{(2)}$.
\begin{theorem}
\label{th:beta-dist}
Suppose Assumptions~\ref{as:identification} and \ref{as:bounded}--\ref{as:basis} hold. Further, assume $\tilde{r}_{n,\bbeta}^{\alpha_{0}}\log n \rightarrow 0$, where $\alpha_{0}$ is the constant from Assumption~\ref{as:criterion}, and that
\begin{equation} \label{eq:tbias-beta0}
\frac{n}{\sqrt{k + n}} \sup_{ \|\bbeta - \bbeta^{\ast}\| \lesssim \tilde{r}_{n,\bbeta} } \| \Pi^{\perp} [\hat{\xi} - \xi | \bar{\phi}_{k}] \|_{\P,2} \cdot \|  \Pi^{\perp} [b_{\bbeta} - \hat{b}_{\bbeta_{\mathsf{init}}} | \bar{\phi}_{k}]  \|_{\P,2} \rightarrow 0.
\end{equation}
Then 
\begin{align*}
\frac{n}{\sqrt{k + n}} (\tilde{\bbeta}^{(2)} -\bbeta^{*}) \stackrel{d}{\to} N\left(0,   \{\nabla_{\bbeta}\psi (\bbeta^{*})\}^{-1} \left\{ V_{1} ( \bbeta^{\ast} )  + \frac{1}{2}  V_2 ( \bbeta^{\ast} ) \right\} \{\nabla_{\bbeta}^\top\psi (\bbeta^{*})\}^{-1} \right),
\end{align*}
where $V_1(\bbeta)$ and $V_2(\bbeta)$ are defined in \eqref{eq:V1}--\eqref{eq:V2} in Section~\ref{app:beta-dist} of the Supplementary Material. In particular, the asymptotic variance depends on the relative scaling between $k$ and $n$: $V_2(\bbeta^{\ast})=0$ when $k\ll n$, while $V_1(\bbeta^{\ast})=0$ when $k\gg n$. Consequently, 
\begin{itemize}
\item [(i)] if $k \ll n$, then
\begin{align*}
\sqrt{n} (\tilde{\bbeta}^{(2)} -\bbeta^{*}) & \stackrel{d}{\to} N \left( 0, \{\nabla_{\bbeta}\psi (\bbeta^{*})\}^{-1}V_1  (\bbeta^{\ast}) \{\nabla_{\bbeta}^{\top}\psi (\bbeta^{*})\}^{-1} \right);
\end{align*}

\item [(ii)] if $k \gg n$ (recall Assumption~\ref{as:basis} requires $k = o (n^{2})$), then
\begin{align*}
\frac{n}{\sqrt{k}} (\tilde{\bbeta}^{(2)} -\bbeta^{*}) & \stackrel{d}{\to} N \left(0, \{\nabla_{\bbeta}\psi (\bbeta^{*})\}^{-1} \frac{1}{2} V_2 (\bbeta^{\ast}) \{\nabla_{\bbeta}^{\top}\psi (\bbeta^{*})\}^{-1} \right);
\end{align*}

\item [(iii)] and if $k / n \rightarrow \tau$, for some $\tau \in (0, \infty)$, 
\begin{align*}
\sqrt{n} (\tilde{\bbeta}^{(2)} -\bbeta^{*}) \stackrel{d}{\to} N\left(0,   (1+\tau) \{\nabla_{\bbeta}\psi (\bbeta^{*})\}^{-1} \left\{ V_{1} ( \bbeta^{\ast} )  + \frac{1}{2}  V_2 ( \bbeta^{\ast} ) \right\} \{\nabla_{\bbeta}^{\top}\psi (\bbeta^{*})\}^{-1} \right).
\end{align*}
\end{itemize}
\end{theorem}

The complete proof of Theorem~\ref{th:beta-dist} is presented in Supplementary Material Section~\ref{app:beta-dist}. We provide a proof sketch as follows. We begin by expanding $\psi (\bbeta^{*})$ around  $\tilde{\bbeta}^{(2)}$:
\begin{align*}
0 = \psi (\bbeta^{*}) = \psi (\tilde{\bbeta}^{(2)}) + \nabla_{\bbeta}\psi (\bbeta^{\dag}) \cdot (\bbeta^{\ast} - \tilde{\bbeta}^{(2)}),
\end{align*}
where $\bbeta^{\dag}$ lies between $\bbeta^{*}$ and $\tilde{\bbeta}^{(2)}$.  Next,  we decompose $\psi (\tilde{\bbeta}^{(2)})$ as follows:
\begin{align*}
\frac{n}{\sqrt{k+n}} \psi (\tilde{\bbeta}^{(2)}) 
&=  \frac{n}{\sqrt{k+n}} \{\psi (\tilde{\bbeta}^{(2)}) - \bar{\psi}_{k} (\tilde{\bbeta}^{(2)})\} 
 - \frac{n}{\sqrt{k+n}} \{\tilde{\psi}_{k}^{(2)} (\bbeta^{*}) - \bar{\psi}_{k} (\bbeta^{\ast})\}\\
& \quad + \frac{n}{\sqrt{k+n}} \{ (\bar{\psi}_{k} (\tilde{\bbeta}^{(2)}) - \tilde{\psi}_{k}^{(2)} (\tilde{\bbeta}^{(2)})) - (\bar{\psi}_{k} (\bbeta^{*}) - \tilde{\psi}_{k}^{(2)} (\bbeta^{*}))\} \\ 
& \eqqcolon A_n + B_n - C_n .
\end{align*}
Term $A_n$ is bias-related and is $o_{\mathbb{P}}(1)$ under Condition \eqref{eq:tbias-beta0}. Term $B_n$ is a centered $U$-statistic, determining the asymptotic distribution of $\tilde{\bbeta}^{(2)}$. Its asymptotic normality follows from the classical martingale CLT \citep{bhattacharya1992class}. Since the kernel of this $U$-statistic depends on $k$, the linear component in the Hoeffding decomposition does not always dominate and thus the order of its asymptotic variance depends on the scaling between $n$ and $k$, split into three different regimes as indicated in Theorem~\ref{th:beta-dist}.
The main technical challenge lies in analyzing the $U$-process term $C_n$, indexed by  $\tilde{\bbeta}^{(2)}$ and defined over a function class changing with the sample size $n$. Existing $U$-process maximal inequalities, such as those of \citet{chen2020jackknife}, 
are conservative for our setting. To address this, we establish a refined local maximal inequality in Supplementary Material Section~\ref{app:Maximal Inequalities}.

\begin{remark}\label{rem:cattaneo2025higher}
Edgeworth expansions offer a different notion of optimality when non-asymptotic or finite sample performance is of concern. \citet{cattaneo2025higher} study inference for the density-weighted average derivative effect (DWADE) parameter $\beta^{\ast}_0$ (see Remark~\ref{ade}) under a randomized experiment, i.e., the propensity score $\p_{T|X}$ is known, using a kernel-based estimator $\hat{\beta}_0$. Their estimator is a second-order $U$-statistic with Hoeffding decomposition $\hat{\beta}_0-\beta^{\ast}_0 = \bar{L} + \bar{Q}$, where $\bar{L}$ and $\bar{Q}$ are the linear and quadratic terms, respectively. Classical asymptotics impose bandwidth conditions that make $\bar{Q}$ negligible, yielding an asymptotic linear representation. With smaller bandwidths, $\bar{Q}$ becomes non-negligible, leading to a quadratic approximation where $\mathbb{V}[\hat{\beta}_0] = \mathbb{V}[\bar{L}] + \mathbb{V}[\bar{Q}]$ and $\mathbb{V}[\hat{\beta}_0]^{-1/2}(\hat{\beta}_0-\beta^{\ast}_0) \overset{d}{\to} N(0,1)$. Using Edgeworth expansions, they derive refined distributional approximations that improve finite-sample confidence interval coverage.

Our estimator $\hat{\psi}^{(2)}_k$ of the estimating equation is also a $U$-statistic whose asymptotic distribution depends on quadratic terms when $k \gtrsim n$. However, the goal of \citet{cattaneo2025higher} is inference refinement (improving finite-sample coverage), while ours is bias correction to achieve rate-optimal estimation under low smoothness conditions. In future work, we can also consider such more refined higher-order distributional approximation if concerned about the finite-sample performance beyond the asymptotic distribution results in Theorem~\ref{th:beta-dist}.
\end{remark}

We next derive the convergence rate of the feasible estimator $\hat{\bbeta}^{(2)}$. To this end, we need the following assumption on $\hat{\Sigma}_{k}$.
\begin{assumption}
\label{as:sigma}
We assume that the estimator $\hat{\Sigma}_{k}$ of $\Sigma_{k}$ satisfies the following condition:
\begin{align*}
\Vert \hat{\xi} - \xi \Vert_{\P,2} \cdot \left( \Vert \bbeta_{\mathsf{init}} - \bbeta^{\ast} \Vert + \Vert \hat{b}_{\bbeta_{\mathsf{init}}} - b_{\bbeta_{\mathsf{init}}} \Vert_{\P,2} \right) \cdot \Vert \hat{\Sigma}_{k} - \Sigma_{k} \Vert_{\op} = o_{\P} (n^{- 1 / 2}).
\end{align*}
\end{assumption}
Assumption~\ref{as:sigma} ensures that estimating $\Sigma_{k}$ by $\hat{\Sigma}_{k}$ does not introduce excessive bias. When $\bbeta_{\mathsf{init}}$ is estimated using only $\hat{\xi}$, we have $\Vert \bbeta_{\mathsf{init}} - \bbeta^{\ast} \Vert \lesssim \Vert \hat{\xi} - \xi \Vert_{\P,2}$. Then the condition in Assumption \ref{as:sigma} reduces to
\begin{equation}
\label{reduced sigma rate}
\left( \Vert \hat{\xi} - \xi \Vert_{\P,2}^{2} + \Vert \hat{\xi} - \xi \Vert_{\P,2} \cdot \Vert \hat{b}_{ \bbeta_{\mathsf{init}} } - b_{ \bbeta_{\mathsf{init}} } \Vert_{\P,2} \right) \cdot \Vert \hat{\Sigma}_{k} - \Sigma_{k} \Vert_{\op} = o_{\P} (n^{-1 / 2}).
\end{equation}
Under Assumption \ref{as:Holder}, the first two terms satisfy $\Vert \hat{\xi} - \xi \Vert_{\P,2}^{2} \lesssim n^{- \frac{ 2s_{1} }{d + 2 s_{1}}}$ and $\Vert \hat{\xi} - \xi \Vert_{\P,2} \cdot \Vert \hat{b}_{\bbeta_{\mathsf{init}}} - b_{\bbeta_{\mathsf{init}}} \Vert_{\P,2} \lesssim n^{- \frac{s_{1}}{d + 2 s_{1}} - \frac{s_{2}}{d + 2 s_{2}}}$. Consequently, \eqref{reduced sigma rate} requires
\begin{align*}
\Vert \hat{\Sigma}_{k} - \Sigma_{k} \Vert_{\op} \lesssim n^{- \frac{(d - 2 s_{1} ) \vee 0}{2 (d + 2 s_{1})}} \wedge n^{- \frac{(d^{2} - 4 s_{1} s_{2}) \vee 0}{2 (d + 2 s_{1}) (d + 2 s_{2})}},
\end{align*}
which holds if the joint density $\p_{X, T} (\cdot, \cdot)$ is sufficiently smooth. Similar assumptions appear in most recent papers related to HOIF frameworks \citep{kennedy2024minimax, bonvini2022fast, mcgrath2024nuisance} and can be relaxed using diverging-order $U$-statistic estimators, as detailed in Supplementary Material Section~\ref{app:truncation} following \citet{liu2017semiparametric} or \citet{robins2023minimax}.

Theorem~\ref{thm:feasible rates} below establishes the convergence rate of the feasible estimator $\hat{\bbeta}^{(2)}$ using Theorem~\ref{th:consistency} together with Assumption~\ref{as:sigma}. A proof is provided in Supplementary Material Section~\ref{app:feasible rates}.
\begin{theorem}
\label{thm:feasible rates}
Suppose that for some constant $c>0$, 
$
 \Vert \bbE [ \hat{\psi}^{(2)}_{k}(\bbeta_1) ] - \bbE [ \hat{\psi}^{(2)}_{k}(\bbeta_2) ] \Vert \geq c \|\bbeta_1 - \bbeta_2\| .$
Under the assumptions of Theorem \ref{th:consistency}, we have
\begin{align*}
\bbE[\Vert \hat{\bbeta}^{(2)} - \bbeta^{\ast} \Vert] \lesssim \tilde{r}_{n,\bbeta} + \|\xi - \hat{\xi}\|_{\P,2} \cdot \|b_{\bbeta^*} - \hat{b}_{\bbeta_{\mathsf{init}} }\|_{\P,2} \cdot \|\hat{\Sigma}_k - \Sigma_k\|_{\mathrm{op}}.
\end{align*}
If we further impose Assumption~\ref{as:sigma}, then $\bbE[\Vert \hat{\bbeta}^{(2)} - \bbeta^{\ast} \Vert] \lesssim \tilde{r}_{n,\bbeta}$.
\end{theorem}

The final result in this section characterizes the asymptotic normality of the feasible second-order estimator $\hat{\bbeta}^{(2)}$, which is a direct consequence of Theorem~\ref{th:beta-dist} and Assumption~\ref{as:sigma}.

\begin{theorem}
\label{thm:feasible beta-dist}
Under the assumptions of Theorem~\ref{th:beta-dist}, together with Assumption~\ref{as:sigma}, all conclusions in Theorem~\ref{th:beta-dist} continue to hold for $\hat{\bbeta}^{(2)}$.
\end{theorem}

\begin{remark}
\label{rem:approximation}
We briefly discuss how the rate of the initial estimator $\bbeta_{\ini}$ impacts the results in this section. Condition~\eqref{eq:tbias-beta0} in Theorem~\ref{th:beta-dist} specifies the approximation power of $\bar{\phi}_{k}$ for the residuals between estimated and true nuisance parameters. Suppose that $b_{\bbeta}$ is differentiable with respect to $\bbeta$. It is natural to assume that the second-order estimator is at least as close to the truth as the sub-optimal initial estimator, i.e., $\tilde{r}_{n,\bbeta} \lesssim \|\bbeta_{\mathsf{init}} - \bbeta^{\ast}\|$. Under this assumption, using the triangle inequality and a mean-value expansion in $\bbeta$, we can decompose the left-hand side of condition~\eqref{eq:tbias-beta0} as
\begin{align}
   & \ \sup_{\|\bbeta - \bbeta^{\ast}\| \leq \tilde{r}_{n,\bbeta} } \| \Pi^{\perp} (\hat{\xi} - \xi \mid \bar{\phi}_{k}) \|_{\P,2} \cdot \|  \Pi^{\perp} [b_{\bbeta} - \hat{b}_{\bbeta_{\mathsf{init}}} \mid \bar{\phi}_{k}]  \|_{\P,2} \label{our0} \\
   \leq & \ \| \Pi^{\perp} (\hat{\xi} - \xi \mid \bar{\phi}_{k}) \|_{\P,2} \cdot \Big\{ \sup_{\|\bbeta - \bbeta^{\ast}\| \leq \tilde{r}_{n,\bbeta} } \|  \Pi^{\perp} (b_{\bbeta} - b_{\bbeta_{\mathsf{init}}} \mid \bar{\phi}_{k})  \|_{\P,2} +  \|  \Pi^{\perp} (b_{\bbeta_{\mathsf{init}}} - \hat{b}_{\bbeta_{\mathsf{init}}} \mid \bar{\phi}_{k})  \|_{\P,2} \Big\} \notag \\
   \lesssim & \ \| \Pi^{\perp} (\hat{\xi} - \xi \mid \bar{\phi}_{k}) \|_{\P,2} \cdot \Big\{ \sup_{\bbeta^{\dag} \in \mathcal{N}} \|  \Pi^{\perp} (\nabla_{\bbeta}b_{\bbeta^{\dag}}  \mid \bar{\phi}_{k}) \|_{\P,2} \cdot \|\bbeta_{\mathsf{init}} - \bbeta^{\ast}\|  + \| \Pi^{\perp} (b_{\bbeta_{\mathsf{init}}} - \hat{b}_{\bbeta_{\mathsf{init}}} \mid \bar{\phi}_{k})  \|_{\P,2}\Big\}. \notag
\end{align}
Here, $\mathcal{N}\coloneqq \{ \bbeta: \|\bbeta-\bbeta^{*}\| \lesssim \|\bbeta_{\mathsf{init}}-\bbeta^{*}\| \}$,  which contains all line segments between $\bbeta_{\mathsf{init}}$ and any $\bbeta$ with $\|\bbeta-\bbeta^{*}\|\le \tilde r_{n,\bbeta}$.
The expression \eqref{our0} will be used in Section~\ref{sec:examples} to further clarify the rate requirement for $\bbeta_{\ini}$ in QTE with binary treatment. 
\end{remark}

\begin{remark}\label{rem:connection}
As suggested by a referee, the econometrics literature contains several related proposals for bias correction in two-step semiparametric settings, including jackknife- and bootstrap-based methods \citep{cattaneo2019two, cattaneo2018kernel}. In Supplementary Material Section~\ref{app:connection}, we use the average treatment effect (ATE) as a concrete example to compare our approach with these alternatives.
All three methods address the problem that first-stage nuisance estimation can bias the second-stage estimator. However, the sources of bias and the correction mechanisms differ. \citet{cattaneo2019two} and \citet{cattaneo2018kernel} aim to estimate the estimating equation $\psi(\bbeta)$ itself, with nuisance functions estimated by linear regression and kernel methods, respectively. The resulting estimator is affected by leave-in bias due to the reuse of the same sample in both steps, particularly when the covariate dimension is large or the bandwidth is small. By contrast, we use sample splitting and focus on estimating the approximation bias defined in \eqref{eq:bias}. Our construction accommodates generic nuisance learners, such as logistic regression, random forests, or neural networks, because we use linear approximation only for the nuisance estimation residuals $\hat{\xi}-\xi$ and $\hat{b}_{\bbeta_{\ini}}-b_{\bbeta}$, although to obtain sharp theoretical results we need to impose certain smoothness assumptions on $\hat{\xi} - \xi$ and $\hat{b}_{\bbeta_{\ini}}-b_{\bbeta}$ (see Assumption~\ref{as:Holder}).

Despite these differences, there are important connections. Specifically, we show that, as \citet{cattaneo2019two} impose (approximately) linear assumptions on the nuisance functions whereas we do so only for the residuals of the nuisance estimates, both our target parameter and that of \citet{cattaneo2019two} reduce to a bilinear functional \citep{robins2007comment, bruns2026augmented} of the form $\mathrm{B}_{\psi,k} (\bbeta)$ defined in \eqref{EB}. We show that the jackknife- and bootstrap-based (under some assumption on the scaling between $k$ and $n$) bias-corrected estimators are, respectively, exactly equivalent and asymptotically equivalent to our second-order $U$-statistic estimator $\tilde{\mathrm{B}}_{\psi,k}(\bbeta)$ of $\mathrm{B}_{\psi,k}(\bbeta)$. Moreover, the HOIF framework provides a systematic way to construct higher-order $U$-statistic estimators even when $\Sigma_k$ is unknown (see Supplementary Material Section~\ref{app:truncation}), which we believe can also motivate similar higher-order jackknife- or bootstrap-based bias reduction methods \citep{koltchinskii2022bootstrap}.
\end{remark}

\begin{remark}\label{rem:cavaliere2024bootstrap}
As also noted by one referee, both \citet{cavaliere2024bootstrap} and our paper study settings where a non-negligible asymptotic bias affects the limiting distribution.
However, the way how bias is handled differs. \citet{cavaliere2024bootstrap} consider statistics whose limiting distribution takes the form $T_n = \sqrt{n}(\hat{\theta}_n - \theta) \stackrel{d}{\to} B_n + Z_1$, where $Z_1$ is centered normal and $B_n$ is a non-vanishing asymptotic bias. They construct a bootstrap analogue $T_n^* := \sqrt{n}(\hat{\theta}_n^* - \hat{\theta}_n)$, where $\hat{\theta}_n^*$ is a bootstrap version of $\hat{\theta}_n$, such that $T_n^* - \hat{B}_n \stackrel{d^*}{\to} Z_1$, with $\hat{B}_n$ a bootstrap-based estimator of $B_n$. When the bootstrap fails to replicate the bias, i.e., $\hat{B}_n - B_n$ does not converge to zero but converges in distribution to a zero-mean random variable, the standard bootstrap $p$-value
$$
\hat{p}_n := \mathbb{P}^*(T_n^* \le T_n) = \mathbb{P}^*\bigl(T_n^* - \hat{B}_n \le T_n - \hat{B}_n\bigr) = \mathbb{P}^*\bigl(T_n^* - \hat{B}_n \le (T_n - B_n) - (\hat{B}_n - B_n)\bigr)
$$
is no longer asymptotically uniform. Instead, its limiting distribution converges to a distribution $H$ that depends on the joint limiting distribution of $(T_n - \hat{B}_n, T_n^* - \hat{B}_n)$ but not on the unknown bias $B_n$ itself. The key insight of \citet{cavaliere2024bootstrap} is to apply a prepivoting transformation: $H(\hat{p}_n)$ is asymptotically uniform. Since $H$ is unknown, they estimate it via a double bootstrap to obtain $\hat{H}_n$, and output the corrected $p$-value $\tilde{p}_n := \hat{H}_n(\hat{p}_n)$. Therefore, their procedure does not directly estimate or remove the bias $B_n$; instead, it bypasses the bias by transforming the bootstrap $p$-value and estimates the distribution function $H$.

In contrast, our approach directly estimates the bias itself. As shown in Section~\ref{sec:intuition}, we further decompose the first-order bias $ \bbE [\{\hat{\xi} (X, T) - \xi (X, T)\} \{b_{\bbeta} (X, T) - \hat{b}_{\bbeta_{\mathsf{init}}} (X, T)\}]$ into a leading bias component $\mathrm{B}_{\psi,k}(\bbeta)$ and a truncation bias $\mathrm{TB}_{\psi,k}(\bbeta)$ by projecting the residuals of the estimated nuisance functions onto a sieve basis $\bar{\phi}_k$. The leading bias $\mathrm{B}_{\psi,k}(\bbeta)$ is then unbiasedly estimated by a $U$-statistic-based correction term $\tilde{\mathrm{B}}_{\psi,k}(\bbeta)$ when $\Sigma_k$ is known (see Section~\ref{sec:intuition}). Subtracting this correction term from the first-order estimating equation yields a bias-corrected estimator $\hat{\bbeta}^{(2)}$ that achieves $\sqrt{n}$-consistency under weaker smoothness conditions than first-order methods (see Theorem~\ref{thm:QTE} and the discussion in Section~\ref{sec:examples}).
\end{remark}

\subsection{Application: QTE with a binary treatment over \Holder{} smoothness classes}
\label{sec:examples}
In this section, we revisit Section~\ref{sec:review} and specialize the general results from Theorem~\ref{th:consistency} to Theorem~\ref{thm:feasible beta-dist} to the QTE setting with binary treatment $T \in \{0,1\}$ (see Example~\ref{qte}).  To keep the exposition simple, we take the target parameter $\beta_1^*$ to be the $\tau^{th}$ quantile of the potential outcome $Y(1)$,  i.e. $\beta_1^* = \inf\{q: \mathbb{P}(Y (1) \leq q) \geq \tau\}$.  Under Assumption~\ref{as:identification}, $\beta_1^*$ is identified as the unique solution to
\begin{align*}
   \psi(\beta_1) = \bbE (\tau - \mathbbm{1} \{Y(1) \leq \beta_1\}) =  \bbE \left\{ \frac{T}{\bbP (T= 1 \mid X )} \cdot (\tau - \mathbbm{1} \{Y \leq \beta_1\}) \right\}.
\end{align*}
This example fits into our framework with the following correspondences:
\begin{align*}
& \Gamma(Y(t),t,\beta_1) = \left(\tau - \mathbbm{1} \{Y(t) \leq \beta_1\} \right) \cdot t ,\quad  \xi(x,t) = \frac{t}{\bbP (T= 1 \mid X = x )},  \\
& b_{\beta_1} (x,t) = t \cdot \bbE (\tau - \mathbbm{1} \{Y \leq \beta_1 \} \mid X = x, T=t).
\end{align*}
In the QTE setting, several general regularity conditions for GTMs from Section~\ref{sec:properties} (Assumptions~\ref{as:Holder}–\ref{as:criterion}) can be simplified. In particular, Assumption~\ref{as:bounded} holds automatically because $T$ is binary. Assumption~\ref{as:criterion} also holds for QTE under the regularity conditions stated below; see Supplementary Material Section~\ref{app:Verification} for verification details. To make this section self-contained, we therefore omit Assumptions~\ref{as:bounded} and~\ref{as:criterion} and restate the remaining  relevant conditions as follows.
\begin{assumption} \label{as:smoothness-psit QTE}
(i) The parameter space $\calB$ is compact and $\beta_1^*$ lies in its interior;
(ii) $\beta^{\ast}_1$ is the unique solution to $\psi (\beta_1)= 0$, and for every $\epsilon > 0$, $\inf_{|\beta_1 - \beta_1^{\ast}| \geq \epsilon} |\psi (\beta_1)| > 0$; and (iii) $\psi (\cdot)$ continuously differentiable  on $\calB$ and its derivative function satisfies $\nabla_{\beta_1} \psi (\beta^{\ast}_1)> 0$. 
\end{assumption}

\begin{assumption}\label{as:nuisance_est QTE}
(i)  $\sup_{\beta_1 \in \calB} \|\hat{b}_{\beta_1} (\cdot, 1)\|_{\infty} \lesssim 1$. (ii) $\bbP ( Y \le \beta_1| X = x, T = 1)$ is  Lipschitz continuous in $\beta_1$. (iii) For all $x \in \calX $,  $ \ c_5 \leq \p_{X}(x)\leq c_6$.
\end{assumption}

\begin{assumption}\label{as:Holder QTE}
$\xi (\cdot, 1)$ and $\p_{Y|X,T} (\beta_1|\cdot, T=1)$ are \Holder{} smooth with indices $s_{1}$ and $s_{2}$ respectively, for every $\beta_1 \in \calB$. Let $s = (s_1 + s_2)/2$ denote the average smoothness. The nuisance estimators $\hat{\xi}$ and $\hat{b}_{\beta_1}$ converge to $\xi$ and $b_{\beta_1}$ at the minimax optimal rates ($n^{- \frac{s_{1}}{d + 2 s_{1}}}$ and $n^{- \frac{s_{2}}{d + 2 s_{2}}}$) in $L_{2} (\P)$.   Moreover, we assume $\hat{\xi}$ and $\hat{b}_{\beta_{1,\ini}}$ belong to \Holder{} smoothness classes with smoothness indices $s_{1} > 0$ and $s_{2} > 0$, respectively. 
\end{assumption}

\begin{theorem}
\label{thm:QTE}
For the QTE parameter $\beta_1^* = \inf\{q: \mathbb{P}(Y (1) \leq q) \geq \tau\}$, suppose that Assumptions~\ref{as:identification} and \ref{as:basis}--\ref{as:Holder QTE} hold. Then
\begin{itemize}
\item [(i)] if $s / d > 1 / 4$ and $k = o (n)$,
\begin{align*}
\sqrt{n} (\hat{\beta}^{(2)}_1 - \beta^{\ast}_1) \overset{d}{\rightarrow} N \left( 0, \{\nabla_{\beta_1} \psi (\beta^{\ast}_1)\}^{-1} V_{1} (\beta^{\ast}_1) \{\nabla_{\beta_1}^{\top} \psi (\beta^{\ast}_1)\}^{-1} \right);
\end{align*}
\item [(ii)] if $0 < s / d \leq 1 / 4$ and $k \asymp n^{\frac{2}{1 + 4 s / d}}$,
\begin{equation*}
\bbE^{1 / 2} [|\hat{\beta}^{(2)}_1 - \beta^{\ast}_1|^{2}] \lesssim n^{- \frac{4 s / d}{1 + 4 s / d}}.
\end{equation*}
\end{itemize}
\end{theorem}

The convergence rates for QTE in Theorem~\ref{thm:QTE} match those for ATE established in \citet{robins2016technical}, which shows that these rates are minimax optimal for ATE. The proof of Theorem~\ref{thm:QTE} is straightforward and proceeds as follows. Under the assumptions of Theorem~\ref{thm:QTE} with an appropriate choice of basis $\bar{\phi}_{k}$, we have
\begin{align*}
\| \Pi^{\perp} (\hat{\xi} - \xi \mid \bar{\phi}_{k}) \|_{\P,2} = O(k^{-s_1/d}),\quad \| \Pi^{\perp} (b_{\beta_{1,\mathsf{init}}} - \hat{b}_{\beta_{1,\mathsf{init}}} \mid \bar{\phi}_{k})  \|_{\P,2} = O(k^{-s_2/d}),
\end{align*}
because  the \Holder{}-$s_2$ smoothness of $\p_{Y|X,T}(\beta_1 \mid \cdot, T=1)$ implies that $b_{\beta_1}(\cdot,1)$ is also \Holder{}-$s_2$ smooth for any $\beta_1 \in \calB$. Furthermore, a direct calculation shows that $\nabla_{\beta_1} b_{\beta_1} (\cdot, 1) =  - \p_{Y|X,T}(\beta_1 |\cdot, T = 1 ) $, which is also \Holder{}-$s_{2}$ smooth for any $\beta_1 \in \calB$. Combining these results with \eqref{our0}, the bias of $\hat{\beta}^{(2)}_1$ is bounded above by
\begin{align*}
k^{- 2 s / d} \cdot | \beta_{1,\ini} - \beta^{\ast}_1 | + k^{- 2 s / d} = \left\{ \begin{array}{ll}
o (n^{-1 / 2}) & s / d > 1 / 4, k = o (n), \\
O \left( n^{- \frac{4 s / d}{1 + 4 s / d}} \right) & 0 < s / d \leq 1 / 4, k \asymp n^{\frac{2}{1 + 4 s / d}}.
\end{array} \right.
\end{align*}
Therefore, somewhat surprisingly, the higher-order estimator imposes no rate requirement on $\beta_{1,\ini}$, as long as $\beta_{1,\ini}$ is bounded. This stands in contrast to the first-order LDML estimator, which, as discussed following Proposition~\ref{prop:dml}, requires $s_1/d > 1/2$. We summarize these comparisons in Table~\ref{tab:smoothness} below.
\begin{table}[htbp]
\centering\small
\begin{tabular}{l|cc}
\toprule
 &
\makecell{\textbf{Smoothness conditions}} &
\makecell{\textbf{Rate condition for $\beta_{1,\ini}$}} 
\\ \midrule
Our paper &
$\displaystyle \frac{s_1/d+s_2/d}{2}>\frac14$ &
$\displaystyle  \|\beta_{1,\ini} - \beta^{*}_1\| = O_{\P}(1)$
\\
\citet{kallus2024localized} &
$\displaystyle {s_{1}}/{d} > \frac{1}{2}$ \& $\displaystyle \frac{s_{1}/d}{1 + 2 s_{1}/d} + \frac{s_{2}/d}{1 + 2 s_{2}/d} > \frac{1}{2}$
&
$\displaystyle \|\beta_{1,\ini} - \beta^{*}_1\| = o_{\P}(n^{-1/4})$
\\
\citet{ai2021unified} &
$\displaystyle  {s_1}/{d}>2$ &
no $\beta_{1,\ini}$ required
\\ \bottomrule
\end{tabular}
\caption{Comparison of smoothness and initial estimator rate conditions for $\sqrt{n}$-consistent estimation of the QTE parameter $\beta_1^* = \inf\{q: \mathbb{P}(Y(1) \leq q) \geq \tau\}$. Here $s_1$ denotes the \Holder{} smoothness index of the propensity score $\xi(\cdot, 1) = 1/\mathbb{P}(T=1 \mid X = \cdot)$, and $s_2$ denotes the Hölder smoothness index of the conditional outcome density $\p_{Y|X,T}(y \mid \cdot, T=t)$ for every $y$.}
\label{tab:smoothness}
\end{table}

\begin{remark}
For ATE in Example~\ref{ate}, a direct calculation yields $b_{\bbeta}(x,t) = ( (\bbE(Y | X = x, T = 0 ) - \beta_0)\cdot(1-t),   ( \bbE(Y |X = x, T = 1 ) -  \beta_1)\cdot t )^{\top}$ and $\nabla_{\bbeta}b_{\bbeta} (x,t) = \text{diag}(1-t, t)$. Each component of $\nabla_{\bbeta}b_{\bbeta}$ is infinitely smooth. As a result, condition \eqref{eq:tbias-beta0} imposes no further convergence rate requirements on the initial estimator $\bbeta_{\ini}$ of $\bbeta^{\ast}$.
\end{remark}

\section{Higher-Order Estimators: Toward Infinite-Dimensional Parameters}
\label{sec:dose}

For continuous or high-dimensional treatments \citep{su2019non}, fixed-dimensional parameters in \eqref{ident} are often inadequate. Motivated by \citet{bonvini2022fast} and \citet{colangelo2026double}, we extend our framework to infinite-dimensional GTMs by setting $\omega(\cdot) \equiv \p_T(\cdot)\delta_t(\cdot)$ and letting $p \to \infty$ in \eqref{ident}. One purpose of this section is to further advocate the fruitful philosophy of turning an infinite-dimensional problem into a finite but diverging-dimensional functional estimation problem \citep{kennedy2024minimax}. For notational convenience, we denote the infinite-dimensional GTM parameter as $\beta_t^* \equiv \beta^*(t): \calT \to \bbR$, a function of the treatment value $T=t$. Then $\beta_t^*$ solves
\begin{equation*}
\int_{u\in\calT} \bbE \{\Gamma (Y (u) , u, \beta^{\ast}_{t})\} \p_T(u) \delta_{t} (u) \diff u \equiv 0 , \ \forall \ t \in \calT.
\end{equation*}
Under Assumption~\ref{as:identification}, $\beta^{\ast}_{t}$ is identified as the solution to:
\begin{equation}\label{newmodel_np} 
\psi_{t} (\beta) \equiv \psi_{t} (\theta_{\beta}) \coloneqq \bbE \left\{ \xi (X, T) \Gamma (Y, T, \beta) \delta_{t} (T) \right\}  \equiv 0,
\end{equation}
where $\xi (x,t) = \p_T(t)/\p_{T|X}(t \mid x)$. 
Here $\Gamma: \calY \times \calT \times \calB \rightarrow \bbR$ is a known one-dimensional generalized residual function. To simplify notation, we maintain the conventions of Section~\ref{sec:setup}: for any $\beta\in\calB$, we define $b_{\beta}(x,t) = \bbE \{\Gamma (Y, T, \beta) \mid X = x, T=t\}$. It is worth noting that in this section we are interested in estimating the GTM parameter at some fixed treatment value $T = t$.

\subsection{Higher-order estimators for infinite-dimensional parameters}
\label{sec:formal}

Since the Dirac delta function $\delta_{t} (\cdot)$ cannot be evaluated in practice, we approximate it with a kernel function, following \citet{bonvini2022fast}; see Remark~\ref{rem:comp} for details. Let $K (\cdot)$ be a base kernel function and define   
\begin{align*}
K_{t, h} (u) \coloneqq \frac{1}{h} K \left( \frac{u-t}{h} \right) , \quad g(x) \coloneqq \int_{u\in\calT} K_{t,h} (u) \p_{X,T} ( x, u) \diff u.
\end{align*}
We assume the kernel $K(\cdot)$ satisfies the following condition, which is also imposed by \citet{bonvini2022fast}.
\begin{assumption}
\label{as:kernel}
$K (\cdot): \calT \rightarrow \bbR $ is a  uniformly bounded kernel function of order $\alpha_1 \wedge \alpha_2$, where $\alpha_1,\alpha_2$ are positive integers specified in Assumption~\ref{as:Lip}. It satisfies $\int_{\calT} K (u) \diff u = 1$, $\int_{\calT} u^{i}K (u) \diff u = 0$ for $i=1,\ldots, \alpha_1 \wedge \alpha_2-1$, and $\int_{\calT} u^{\alpha_1 \wedge \alpha_2} K (u) \diff u \eqqcolon \kappa_{\alpha_1 \wedge \alpha_2,1} \neq 0$. 
Moreover, $g(x) \in (\underline{c}_{g}, \bar{c}_g)$ for any $x$ and  $h$, where $0 < \underline{c}_g < \bar{c}_g < \infty$ are two universal constants.
\end{assumption} 
Intuitively, $K_{t, h} (\cdot)\to \delta_{t} (\cdot)$ as $h\to 0$,  so we replace $\delta_{t}(u)$ in \eqref{newmodel_np} with $K_{t,h}(u)$. Analogous to \eqref{eq:sample_est_eq}, we construct the first-order estimating equation:
\begin{align*}
\hat{\psi}_{\beta, t}^{(1)} \coloneqq \bbU_{n, 2} \left[ K_{t,h} (T_1)  \hat{\xi} (X_1, T_1) \{ \Gamma (Y_1, T_1, \beta) - \hat{b}_{\beta_{\mathsf{init}}} (X_1, T_1) \} +  \hat{b}_{\beta_{\mathsf{init}}} (X_1, T_2)K_{t,h}(T_2) \right].
\end{align*}
Following the debiasing argument in Section~\ref{sec:intuition}, we project the localized nuisance errors $\xi(\cdot,t)-\hat{\xi}(\cdot,t)$ and $b_\beta(\cdot,t)-\hat{b}_{\beta_{\mathsf{init}}}(\cdot,t)$ onto the dictionary $\sfzbar_{k_x}(x)$ with respect to the weight $g$ (see Supplementary Material Section~\ref{app:rate-nonparametric}). This yields the oracle second-order estimating equation $\tilde{\psi}_{t,k}^{(2)}(\beta) := \hat{\psi}_t^{(1)}(\beta) - \tilde{\mathrm{B}}_{\psi,t,k}(\beta)$, where
\begin{align*}
 &\tilde{\mathrm{B}}_{\psi, t,k} (\beta)\coloneqq \bbU_{n, 3} \{K_{t,h}(T_{1}) \hat{\xi} (X_{1}, T_{1})  - K_{t,h} (T_{3} )\} \sfzbar_{k_{x}}^{\top} (X_{1}) \tilde{\Omega}_{k_x}^{-1}\\
 &\qquad\qquad\qquad\qquad\times \sfzbar_{k_{x}} (X_{2}) \{ \Gamma (Y_2, T_2, \beta) - \hat{b}_{\beta_{\mathsf{init}}} (X_{2}, T_{2})\} K_{t,h} (T_{2}),
\end{align*}
with $\tilde{\Omega}_{k_x} \coloneqq \bbE \{K_{t,h}(T) \sfzbar_{k_{x}} (X)\sfzbar_{k_{x}}(X)^{\top}\} = \int_{\calX} g(x)\sfzbar_{k_{x}} (x)\sfzbar_{k_{x}}(x)^{\top} \diff x $ treated as known. Here $\beta_{\mathsf{init}}$ is an initial estimator of $\beta^{\ast}_t$ computed from the nuisance sample.  The oracle second-order estimator $\tilde{\beta}^{(2)}_{t}$ is defined by solving $\tilde{\psi}_{t,k}^{(2)} (\beta)= 0$. 

\subsection{Convergence rates of higher-order estimators}
\label{sec:kernels}
To proceed, we impose a set of regularity conditions that closely parallel those for the finite-dimensional case. For brevity, we assume Assumptions~\ref{as:bounded} and \ref{as:nuisance_est}–\ref{as:basis} hold with appropriate notational adjustments. Assumptions~\ref{as:uniqueness} and~\ref{as:smoothness} are restated below in simplified form.
\begin{assumption}\label{as:smoothness-psit}
For any $t\in\calT$: (i) the parameter space $\calB$ is compact and  $\beta^{*}_t$ is in the interior of $\calB$.
(ii) $\beta^{\ast}_{t}$ is the unique solution to $\psi_{t} (\beta)= 0$, and $\inf_{|\beta - \beta^{\ast}_{t}| \geq \epsilon} |\psi_{t} (\beta)| > 0$ for every $\epsilon > 0$. (iii) $\psi_{t} (\beta)$ is continuously differentiable in $\beta$ with $\nabla_{\beta} \psi_{t} (\beta^{\ast}_{t})> 0$.
\end{assumption}

Additionally, we impose the following smoothness conditions on the nuisance functions and their estimators, which facilitate kernel smoothing and are standard in nonparametric estimation.
\begin{assumption}\label{as:Lip}
 The functions $t \mapsto \xi(x, t)$ and $t \mapsto \p_{T}(t)$, $t \mapsto \hat{\xi}(x, t)$ are $\alpha_{1}$-times continuously differentiable with uniformly bounded derivatives for any $x \in \calX$.  Analogously, $t \mapsto b_{\beta} (x, t)$ and  $ t \mapsto \hat{b}_{\beta} (x, t)$ are $\alpha_{2}$-times continuously differentiable with uniformly bounded derivatives for any $x \in \calX$ and $\beta\in\calB$.
\end{assumption}

Theorem~\ref{th:rate-nonparametric} is the main result of this section, characterizing the convergence rate of $\tilde{\beta}^{(2)}_t$. To reduce clutter,  define   $\Delta_{b,t}  (X;\beta) \coloneqq b_{\beta} (X,t) - \hat{b}_{\beta_{\ini}} (X,t)$ and  $\Delta_{\xi,t}  (X) \coloneqq \xi (X,t) - \hat{\xi} (X,t).$
\begin{theorem}\label{th:rate-nonparametric}
   Suppose that Assumptions
   \ref{as:bounded}, \ref{as:nuisance_est}--\ref{as:basis}, and \ref{as:kernel}--\ref{as:Lip} hold.  Let $\tilde{r}_{n,\mathsf{np}}  \to 0$  be a diminishing sequence as $n\to \infty$.  Suppose that $ h (\log n)^2 \to 0 $,  $(\tilde{r}_{n,\mathsf{np}})^{\alpha_{0}} \log n \to 0$ and
\begin{align*}
   & \left( \frac{\sqrt{k_{x}} }{nh} + \frac{ 1 }{\sqrt{nh}}\right)\log n + h^{\alpha_1 \wedge \alpha_2}  + \sup_{ \beta \in\calB } \left\|\Pi_{g}^{\perp} \{\Delta_{b,t}  (\beta) \mid \sfzbar_{k_{x}}\} \right\|_{g,2} \cdot \left\| \Pi_{g}^{\perp} (\Delta_{\xi,t}  \mid \sfzbar_{k_{x}}) \right\|_{g,2} \lesssim \tilde{r}_{n,\mathsf{np}}.
\end{align*}
Then, $\tilde{\beta}^{(2)}_{t}$ is a consistent estimator of $\beta^{*}_{t}$ and  satisfies
\begin{align*}
    \bbE (|\tilde{\beta}^{(2)}_{t} - \beta^{*}_{t}|) \lesssim &\  
      \sup_{|\beta - \beta_t^{\ast}|\leq \tilde{r}_{n,\mathsf{np}}} \left\|\Pi_{g}^{\perp} \{\Delta_{b,t}  (\beta) \mid \sfzbar_{k_{x}}\} \right\|_{g,2}  \left\| \Pi_{g}^{\perp} (\Delta_{\xi,t}  \mid \sfzbar_{k_{x}}) \right\|_{g,2} 
    + \frac{\sqrt{k_{x}} }{nh} + \frac{ 1 }{\sqrt{nh}} + h^{\alpha_1 \wedge \alpha_2}.
\end{align*}
Here, $\Pi^{\perp}_{g} (f \mid \sfzbar_{k_x}) (\cdot) \coloneqq f (\cdot) - \int_{\calX} f(x) \sfzbar_{k_x} (x)^{\top} g(x) \diff x \tilde{\Omega}_{k_x}^{-1} \sfzbar_{k_x} (\cdot)$ denotes the orthocomplement of the $g$-weighted projection onto $\sfzbar_{k_x}$, and $\|f\|_{g,2}^{2} =
\int_{\calX} f(x)^2g(x)\diff x$.
\end{theorem}
The proof of Theorem~\ref{th:rate-nonparametric} is given in Supplementary Material Section~\ref{app:rate-nonparametric} and closely follows that of Theorem~\ref{th:beta-dist}. Using a similar decomposition, we first establish consistency by showing $\bbE (|\tilde{\beta}^{(2)}_t - \beta^*_t|) \lesssim \tilde{r}_{n,\mathsf{np}} \to 0$. Building on this preliminary rate, we refine the analysis to the local region $\{\beta \in \calB: |\beta - \beta_{t}^{\ast}| \lesssim \tilde{r}_{n,\mathsf{np}} \}$  to obtain the final convergence rate. The key difference lies in the kernel-weighted $U$-processes introduced by the kernel function. To apply the maximal inequality from Supplementary Material Section~\ref{app:Maximal Inequalities}, each term in the Hoeffding decomposition must be carefully controlled with bandwidth-dependent envelopes, requiring a refined and delicate argument. These technical calculations distinguish our analysis from that of \citet{kennedy2017non}, where the target parameter is explicitly defined and does not require such $U$-process techniques.

\begin{remark}
   Following the same argument in Remark~\ref{rem:approximation}, we have 
\begin{align*}
        &\sup_{ |\beta - \beta_t^{\ast}|\leq \tilde{r}_{n,\mathsf{np}} } \|\Pi_{g}^{\perp} \{\Delta_{b,t}  (\beta) \mid \sfzbar_{k_{x}}\} \|_{g,2}  \| \Pi_{g}^{\perp} (\Delta_{\xi,t}   \mid \sfzbar_{k_{x}}) \|_{g,2}
        \lesssim \| \Pi_{g}^{\perp} \{ \Delta_{b,t}  (\beta_{\ini}) \mid \sfzbar_{k_{x}}\} \|_{g,2}  \| \Pi_{g}^{\perp} (\Delta_{\xi,t}  \mid \sfzbar_{k_{x}}) \|_{g,2}\\
        &+\ \|\Pi_{g}^{\perp} \{ \nabla_{\beta}b_{\beta^{\dag}} (\cdot,t) \mid \sfzbar_{k_{x}}\}  \|_{g,2}\cdot \|\beta_{\mathsf{init}} - \beta^{*}_t\| \cdot \| \Pi_{g}^{\perp} ( \Delta_{\xi,t}  \mid \sfzbar_{k_{x}}) \|_{g,2},
    \end{align*}
    where $\beta^\dagger$ lies between
$\beta_{\mathsf{init}}$ and $\beta^{\ast}_{t}$.
   Assume that  $\Delta_{\xi,t}(\cdot)$ is  $s_{1}$-\Holder{} smooth, while $\Delta_{b,t}  (\cdot;\beta_{\ini})$ and $\nabla_{\beta}b_{\beta^{\dag}} (\cdot,t)$ are $s_2$-\Holder{} smooth. With an appropriate basis $\sfzbar_{k_{x}}$, we have
\begin{align*}
&\| \Pi^{\perp} (\Delta_{\xi,t} \mid \sfzbar_{k_{x}}) \|_{g,2} = O(k_{x}^{-s_1/d}),\quad \left\|\Pi_{g}^{\perp} \{\nabla_{\beta}b_{\beta^{\dag}} (\cdot,t) \mid \sfzbar_{k_{x}}\}  \right\|_{g,2} = O(k_{x}^{-s_2/d}),\\
&\text{and}\  \| \Pi^{\perp} \{ \Delta_{b,t}  (\beta_{\ini})\mid \sfzbar_{k_{x}}\}  \|_{g,2} = O(k_{x}^{-s_2/d}).
\end{align*}
Thus, if $|\beta_{\mathsf{init}} - \beta^{*}_t| = O_{\P}(1)$, Theorem~\ref{th:rate-nonparametric} gives
\begin{align*}
\bbE (| \tilde{\beta}^{(2)}_{t} - \beta^{*}_{t} |) \lesssim &\  k_{x}^{-\frac{s_1 + s_2}{d}}
+ \ \frac{\sqrt{k_{x}} }{nh} + \frac{ 1 }{\sqrt{nh}} + h^{\alpha_1 \wedge \alpha_2}.
\end{align*}
Choosing $k_x \asymp nh$ makes the variance of order $(nh)^{-1/2}$, reducing the rate to $(nh)^{-(s_1+s_2)/d} + (nh)^{-1/2} + h^{\alpha_1 \wedge \alpha_2}$. If the average smoothness satisfies $(s_1+s_2)/(2d) \geq 1/4$ and $h \asymp n^{-1/(2(\alpha_1 \wedge \alpha_2)+1)}$, we obtain the rate $n^{-(\alpha_1 \wedge \alpha_2)/(2(\alpha_1 \wedge \alpha_2)+1)}$.
\end{remark}

\begin{remark}
\label{rem:comp}
Our higher-order estimator $\tilde{\beta}_t^{(2)}$ is related to two existing proposals for average dose–response function (ADRF) estimation. First, \citet{colangelo2026double} develop a kernel-based DML estimator motivated by an approximate first-order influence function obtained through kernel localization around a target treatment value. Their estimator takes the form
\begin{align*}
\hat{\beta}_{\mathrm{CL}} = & \, \bbU_{n, 1}\left\{\frac{K_{t,h}(T_1)\{Y_1-\hat{\mu}(X_1, t)\}}{ \hat{\p}_{T \mid X} (T = t \mid X_1 ) }+\hat{\mu}(X_1, t)\right\},
\end{align*}
where $\hat{\mu}(x,t) \coloneqq \bbE [Y \mid X=x, T=t]$. They establish asymptotic normality of $\hat{\beta}_{\mathrm{CL}}$ under high-level regularity conditions, including the key product-rate requirement
$\sqrt{n h}\|\hat{\p}_{T \mid X} (T = t \mid \cdot )-\p_{T \mid X} (T = t \mid \cdot )\|_{\P,2} \cdot \|\hat{\mu}(\cdot,t) - \hat{\mu}(\cdot,t)\|_{\P,2}\to 0$. Second, \citet{bonvini2022fast} propose a second-order estimator (the BK estimator) based on the truncated parameter approach of HOIFs \citep{robins2008higher}:
\begin{align*}
\tilde{\beta}_{\BK} = & \, \bbU_{n, 1}\left\{\frac{K_{t,h}(T_1)\{Y_1-\hat{\mu}( X_1, t)\}}{ \hat{\p}_{T \mid X} (T = t \mid X_1 ) }+\hat{\mu}(t, X_1)\right\} \\
& + \, \bbU_{n, 2}\left\{K_{t,h}(T_{1})\{Y_{1}-\hat{\mu}(X_{1}, T_{1})\}\sfzbar_{k_{x}}(X_{1})^{\top}\tilde{\Omega}_{k_x}^{-1}\sfzbar_{k_{x}}(X_{2})\left(\frac{K_{t,h}(T_{2})}{ \hat{ \p}_{T \mid X} (T_{2} \mid X_{2} )}-1\right)\right\}.
\end{align*} 
In essence, $\tilde{\beta}_{\mathrm{BK}}$ can be viewed as a higher-order extension of $\hat{\beta}_{\mathrm{CL}}$ for ADRF.

Our oracle second-order estimator is obtained by solving
\begin{align*}
& 0 = \bbU_{n, 2} \left[ K_{t,h} (T_1)  \hat{\xi} (X_1, T_1) \{ Y_1  - \hat{\mu} (X_1,T_1) - \beta + \beta_{\mathsf{init}} \} +  \{ \hat{\mu} (X_1, T_2) - \beta_{\mathsf{init}}\}K_{t,h} (T_2)  \right]\\
& + \bbU_{n, 3}\big[ \{ K_{t,h}(T_{1})\hat{\xi} (X_{1}, T_{1})  -   K_{t,h} (T_{3} )\}\sfzbar_{k_{x}} (X_{1})^{\top} \tilde{\Omega}_{k_x}^{-1} \sfzbar_{k_{x}} (X_{2})  \{Y_{2} - \hat{\mu} (X_{2}, T_{2}) -\beta +\beta_{\mathsf{init}}\}K_{t,h} (T_{2})\big] .
\end{align*} 
Notably, because $b_\beta(X,t) - \hat{b}_{\beta_{\mathsf{init}}}(X,t) = \mu(X,t) - \beta - \hat{\mu}(X,t) + \beta_{\mathsf{init}}$ and $\Pi_g^{\perp}[\beta - \beta_{\mathsf{init}} \mid \sfzbar_{k_x}] \equiv 0$ for any $\beta$, the convergence rate in Theorem~\ref{th:rate-nonparametric} simplifies to
\begin{align*}
\frac{\sqrt{k_{x}} }{nh} + \frac{ 1 }{\sqrt{nh}}  + h^{\alpha_1 \wedge \alpha_2} +   \left\|\Pi_{g}^{\perp} \left[\mu (\cdot,t) - \hat{\mu} (\cdot,t) \mid \sfzbar_{k_{x}}\right]  \right\|_{g,2} \cdot \left\| \Pi_{g}^{\perp} \left[ \Delta_{\xi,t}\mid \sfzbar_{k_{x}}\right] \right\|_{g,2} \notag,
    \end{align*}
This rate coincides with that of Theorem 1 in \citet{bonvini2022fast}, confirming that our framework generalizes their result to a broader class of nonparametric GTM parameters, including the QDRF (Example~\ref{qdrf}) as a special case.
\end{remark}

\begin{remark}
When the unknown matrix $\tilde{\Omega}_{k_x}$ in $\tilde{\psi}_{t,k}^{(2)}$ is replaced by an estimator $\hat{\Omega}_{k_x}$ computed from the nuisance sample, we obtain the feasible estimator $\hat{\psi}_{t,k}^{(2)}$. Correspondingly, we define the feasible estimator $\hat{\beta}_t^{(2)}$ for $\beta_t^*$ as the solution to $\hat{\psi}_{t,k}^{(2)}(\beta) = 0$. Following similar arguments as in Theorem~\ref{thm:feasible rates}, the convergence rate of $\hat{\beta}_t^{(2)}$ satisfies
    \begin{equation*}
       \bbE [| \hat{\beta}^{(2)}_{t} - \beta^{*}_{t}|]  \lesssim\  \bbE [| \tilde{\beta}^{(2)}_{t} - \beta^{*}_{t}|] +  \|\Delta_{\xi,t} \|_{g,2}\cdot \|\Delta_{b,t}(\beta^*_t)  \|_{g,2} \cdot \|\hat{\Omega}_{k_x} - \tilde{\Omega}_{k_x}\|_{\op}.
    \end{equation*}
 This bound mirrors the structure established in Theorem~\ref{thm:feasible beta-dist} for the finite-dimensional case, and can serve as a building block for establishing asymptotic normality of the feasible estimator under appropriate conditions on $\hat{\Omega}_{k_x}$.
\end{remark}

\section{Numerical Experiments}
\label{sec:empirical}


\subsection{Simulation studies}
\label{sec:simulation}

In this section, we conduct a proof-of-concept simulation study comparing the performance of our proposed higher-order estimator with two competing methods that do not utilize the HOIF framework. We focus on estimating the QTE parameter with binary treatment, $\bbeta^* = (\beta_{0}^*, \beta_{1}^*)^{\top}$, as described in Example~\ref{qte}. Recall that $\beta_{1}^*$ (resp. $\beta_{0}^*$) denotes the $\tau$-quantile of the treatment (resp. control) group. We choose $\tau = 25\%$ in this simulation. Following \cite{xu2022deepmed} and \cite{liu2017semiparametric}, we simulate data from the following two data-generating mechanisms with different numbers of baseline covariates ($d = 1$ for Case~1 and $d = 4$ for Case~2):
\begin{itemize}
    \item Case~1:
    \begin{align*}
\left\{ \begin{array}{l}
X \sim \mathrm{Uniform} ([-1, 1]), \\
T \mid X \sim \mathrm{Bernoulli} (\mathrm{expit} (\eta (0.5 \cdot X;s))), \\
Y (t) \mid X = (1 + 0.3 \cdot \eta (0.5 \cdot X;s)) \cdot t + 0.2 \cdot \varepsilon,
\end{array} \right.
\end{align*}

\item Case~2:
\begin{align*}
\left\{ \begin{array}{l}
X = (X_1,\ldots,X_4)^{\top} \sim \mathrm{Uniform} ([-1, 1]^4), \\
T \mid X \sim \mathrm{Bernoulli} (\mathrm{expit} (\eta (0.5 \cdot X_1;s))), \\
Y (t) \mid X = (1 + 0.3 \cdot \eta (0.5 \cdot X_1;s)) \cdot t + 0.2 \cdot \varepsilon,
\end{array} \right.
\end{align*}
\end{itemize}
where $\varepsilon \sim N (0, 1)$, $\mathrm{expit} (x) = 1 / (1 + \exp(-x))$, $\eta(x; s) = \sum_{j \in J, l \in \bbZ} 2
^{-j (s + 0.25)} w_{j, l}(x)$ with $J = \{0, 3, 6, 9, 10, 16\}$ and $w_{j,l}(\cdot)$ is the D6 father wavelet functions. By construction, $\eta (\cdot; s)$ lies close to the boundary of \Holder{}-smooth functions with smoothness index $s$. We note that this is only an approximation, as exact infinite basis expansions are infeasible in simulations. Cases~1 and~2 share the same true outcome regression and propensity score. The only difference is that Case~2 includes three extra non-contributing covariates. In the analysis, we do not assume such  knowledge so the basis $\sfzbar_{k}$ depends on all four covariates by stacking the univariate spline basis blocks together with an intercept term $\sfzbar_{k} (x) = (1,\sfzbar_{k_{1}} (x_{1})^{\top}, \sfzbar_{k_{2}} (x_{2})^{\top}, \sfzbar_{k_{3}} (x_{3})^{\top}, \sfzbar_{k_{4}} (x_{4})^{\top})^{\top}$, which yields a basis vector of dimension $k = \sum_{j = 1}^{4} k_{j}+1$.

We vary $s \in \{0.25, 0.4, 0.6\}$ to examine its impact on estimator performance and to validate the theoretical results derived earlier. For each combination of sample size $n$ and smoothness level $s$, we run 1000 replications. In each replication, $2n$ observations are drawn and split equally into a nuisance sample and an estimation sample. We estimate $\bbeta^*$ using three methods:
\begin{description}
  \item[HOE:] Our higher-order estimator is defined in Section~\ref{sec:intuition}. For each continuous covariate, we use a degree-1 B-spline expansion with the number of interior knots set to $\lceil n/100\rceil$. As mentioned, we construct the basis $\sfzbar_k(x)$ by stacking the univariate spline basis blocks for all continuous covariates, together with an intercept term. Specifically, in Case~2, $\sfzbar_k(x)
= (1,
\sfzbar_{k_1}(x_1)^\top,
\sfzbar_{k_2}(x_2)^\top,
\sfzbar_{k_3}(x_3)^\top,
\sfzbar_{k_4}(x_4)^\top)^\top,$ with $k_j=\lceil n/100\rceil+1$ and total dimension $k=\sum_{j=1}^4 k_j + 1$.
  On the nuisance sample, we estimate the stabilized weight function $\hat\xi$ via logistic regression, obtain the initial estimator $\bbeta_{\mathrm{init}}$ by stabilized weighting using $\hat\xi$, and estimate  $\hat b_{\bbeta_{\mathrm{init}}}$ also via logistic regression.
  \item[GOE:] The generalized optimization estimator of \citet{ai2021unified}, applied to both the nuisance and estimation samples. See Supplementary Material  Section~\ref{app:simulation} for its explicit form.

  \item[LDML:] The localized debiased machine learning estimator of \citet{kallus2024localized}, implemented with the same settings as in Section 6.1 of that paper, where all nuisance functions are estimated via logistic regression.
\end{description}

Both Figure~\ref{fig:sim_results_1} (for Case~1) and Figure~\ref{fig:sim_results_2} (for Case~2), displaying simulation results for the QTE $\beta_1^\ast - \beta_0^\ast$, show that when the smoothness parameter $s$ is low (e.g., $s = 0.25$ or more precisely $s$ is close to $0.25$),  our higher-order estimator (HOE) substantially outperforms both GOE and LDML, achieving lower MSE and absolute bias, especially when the sample size $n$ is relatively large. In particular, the higher-order estimator is empirically $\sqrt{n}$-consistent even when $s$ is near 0.25 (the lower right panels of Figure~\ref{fig:sim_results_1} and Figure~\ref{fig:sim_results_2}). This advantage stems from the difficulty of estimating nuisance functions under low smoothness, which induces large first-stage bias—due to errors in both the nuisance parameter estimates and the initial estimator $\bbeta_{\mathsf{init}}$. Consequently, GOE and LDML fail to achieve $\sqrt{n}$-consistency (see the lower‑right panel). As $s$ increases, all three methods exhibit improvements, with bias, variance, and MSE all decreasing. These empirical findings are consistent with our theoretical results, which indicate that the higher-order estimator enjoys improved statistical properties when the \Holder{} smoothness of the nuisance parameters is low.

Figures~\ref{fig:sim_QQplot_beta0_1} and~\ref{fig:sim_QQplot_beta1_1} (for Case~1) and Figures~\ref{fig:sim_QQplot_beta0_2} and~\ref{fig:sim_QQplot_beta1_2} (for Case~2) in Supplementary Material Section~\ref{app:simulation} present normal QQ-plots of our higher-order estimators for
$\beta_0^*$ and $\beta_1^*$  across all $(n,s)$ combinations. These plots support the theoretical conclusion (Theorem~\ref{thm:feasible beta-dist}) that the higher-order estimators are asymptotically normal under appropriate scaling.

Finally, we perform sensitivity analyses to examine whether our higher-order estimator is sensitive to the choice of $k$ or to the initial estimator $\bbeta_{\ini}$. Recall that, for each continuous covariate, we use a degree-1 B-spline basis with $\lceil n/100\rceil$ interior knots. To examine the sensitivity to the choice of $k$, we multiply $ n/100$ by a factor $\lambda \in \{0.75,1,1.25\}$, so that the number of knots varies over $\lceil \lambda n/100\rceil$, resulting in $k = d\cdot(\lceil \lambda n/100\rceil + 1) + 1$. Figure~\ref{fig:basis_sensitivity} in Supplementary Material Section~\ref{app:simulation} reports the empirical bias, variance, and mean squared error for HOE QTE estimator across different values of $\lambda$ for both Cases~1 and~2   with $n=2000$ and $s=0.25$, based on 1000 Monte Carlo replications.
The figure shows the expected bias--variance tradeoff: as $k$ increases, the empirical bias of the higher-order estimator decreases, whereas the empirical variance increases.  In practice, one may choose $k$ such that the point estimates do not differ as much compared to the increase in variance. Here, for both Cases~1 and~2, $\lceil n / 100 \rceil$ is a reasonable choice. A more thorough study of how to choose $k$ is left to a future work that is dedicated to the problem of constructing data-driven basis.

To assess sensitivity to the initial estimator $\bbeta_{\ini}$,  we perturb it deterministically as
$
\bbeta_{\ini}^{(\delta)}
=
\bbeta_{\ini}
+
\delta
\cdot (-1,1)^\top,
$ with $
\delta \in \{-0.5,-0.25,0,0.25,0.5\}$. Here, the perturbation level $\delta$ controls both the magnitude and the direction of the perturbation. Figure~\ref{fig:betainit_sensitivity} in Supplementary Material Section~\ref{app:simulation} reports boxplots of the estimation error of the HOE QTE estimator across different values of $\delta$ for both Cases~1 and~2, in the setting $n=2000$ and $s=0.25$, based on 1000 Monte Carlo replications. The results reported in Figure~\ref{fig:betainit_sensitivity} in Supplementary Material Section~\ref{app:simulation} suggest that HOE is fairly robust to moderate perturbations in the initial estimator. However, when the perturbation becomes too large, the estimator becomes less stable, as reflected by the increased dispersion and a slightly greater number of outliers.

\begin{figure}[htbp]
  \centering
  \includegraphics[width=1\textwidth]{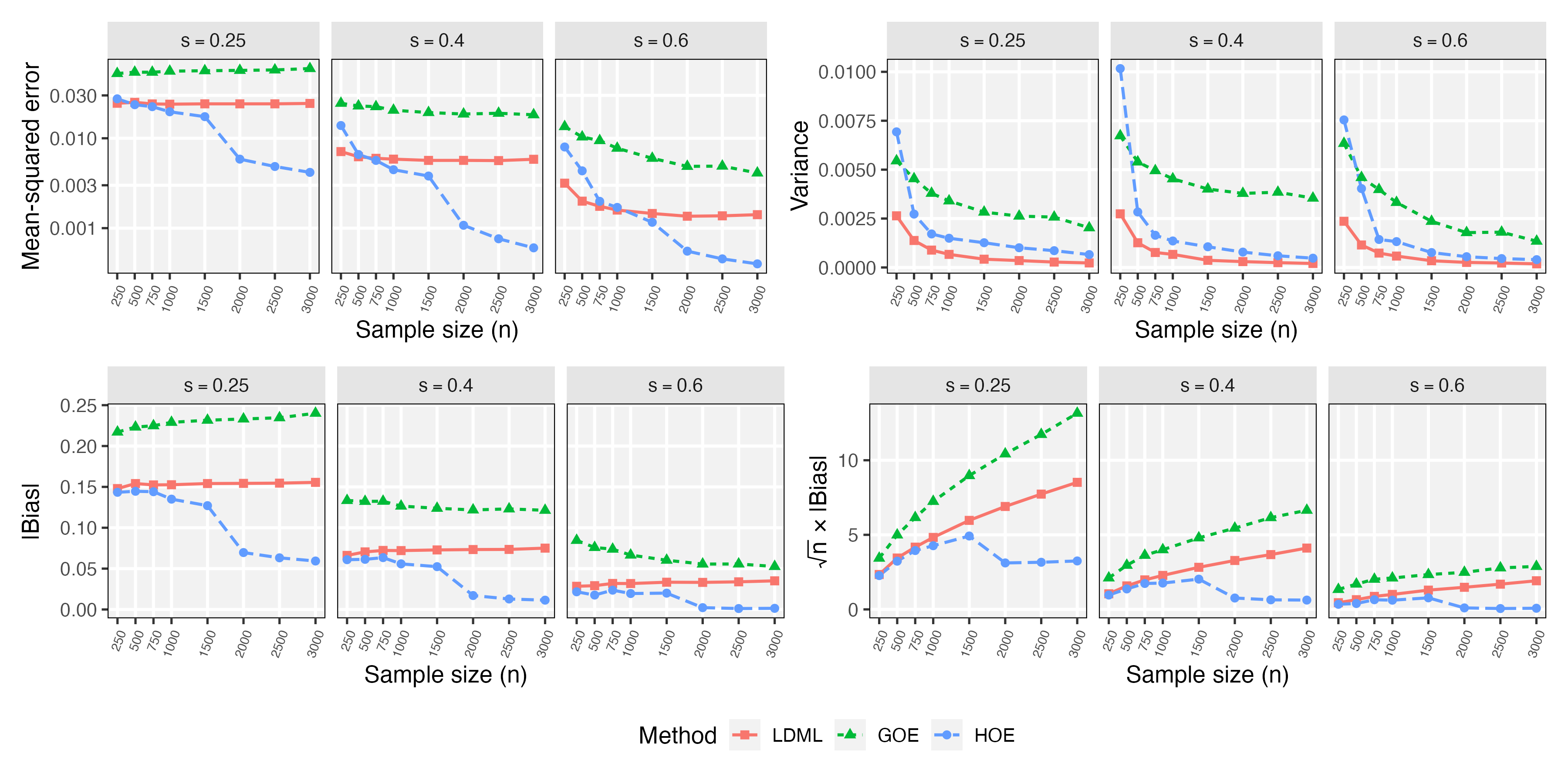} 
  \caption{Results for the simulation study Case~1 of estimating the QTE,  $\beta_1^*-\beta_0^*$, where $\beta_{1}^*$ (resp. $\beta_{0}^*$) denotes the $25\%$-quantile of the treatment (resp. control) group, using different methods, based on 1000 Monte Carlo replications. 
  }
  \label{fig:sim_results_1}
\end{figure}

\begin{figure}[htbp]
  \centering
  \includegraphics[width=1\textwidth]{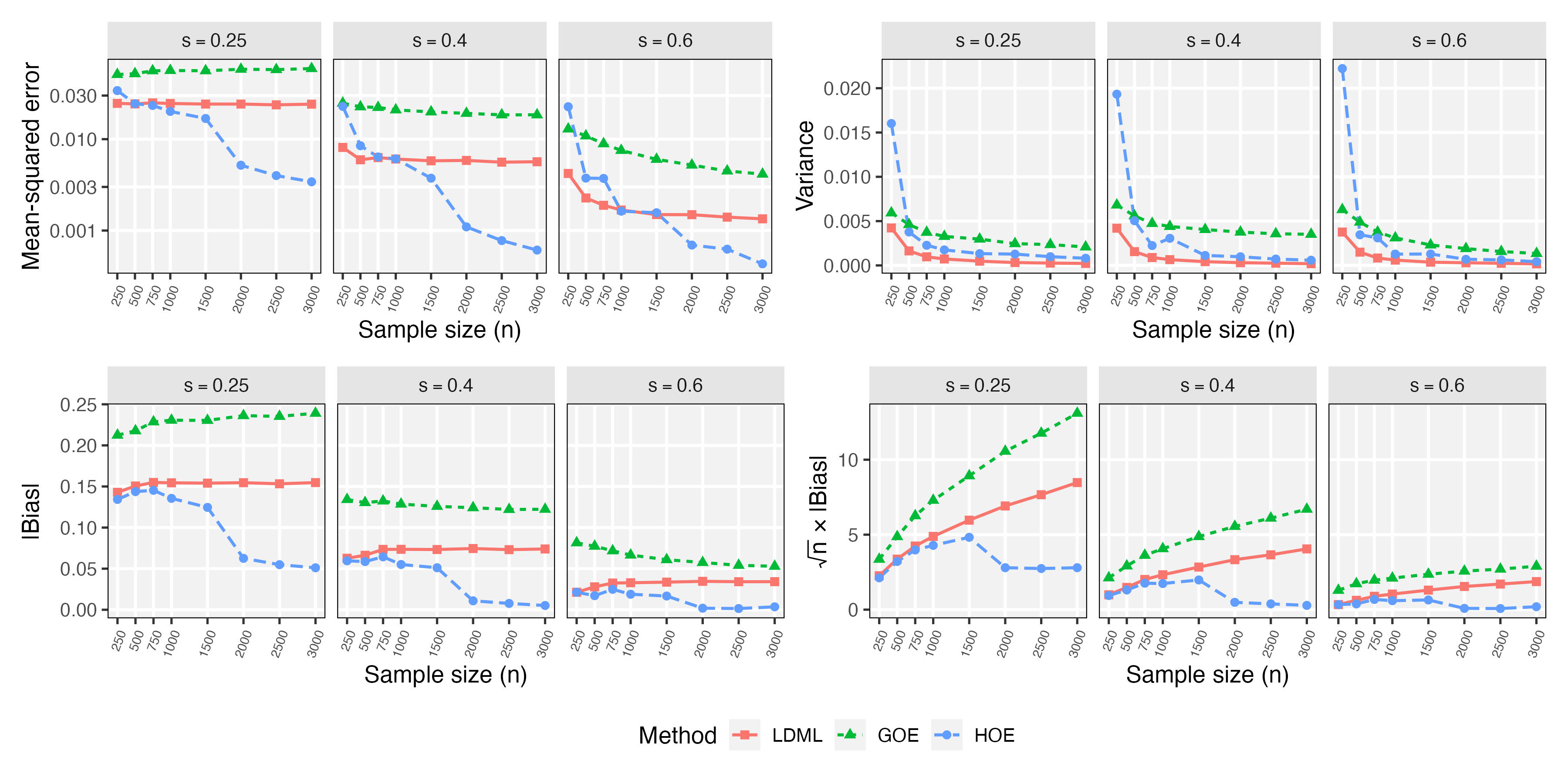} 
  \caption{Results for the simulation study Case~2 of estimating QTE,  $\beta_1^*-\beta_0^*$, where $\beta_{1}^*$ (resp. $\beta_{0}^*$) denotes the $25\%$-quantile of the treatment (resp. control) group, using different methods, based on 1000 Monte Carlo replications. 
  }
  \label{fig:sim_results_2}
\end{figure}

\subsection{Real data analysis}
\label{sec:real}

In this section, we apply our higher-order estimators to estimate the QTE of 401(k) plan eligibility on household wealth, using the Survey of Income and Program Participation (SIPP) data studied in \citet{benjamin2003does,abadie2003semiparametric, chernozhukov2004effects,kallus2024localized}.  The sample contains 9{,}915 observations. The treatment $T$ indicates eligibility for a 401(k) plan, and the outcome $Y$ is net financial assets; our goal is to assess whether 401(k) eligibility increases household saving. The dataset consists of a set of pre-treatment covariates $X$, including age, income, family size, education, marital status, two-earner status, DB pension status, IRA participation, and homeownership (see \citet{chernozhukov2004effects}, Section~3, for details). Recently, \citet{kallus2024localized} analyzed this dataset using LDML estimators.  We implement our higher-order debiased estimators on the same dataset to provide a proof-of-concept for their application in real-world settings.

As in our simulation studies, we split the sample into two folds. With the nuisance sample, we compute the initial estimator $\hat{\bbeta}_{\ini}$ by stabilized weighting, and estimate both the stabilized weights $\hat\xi$ and the generalized outcome regression model $\hat{b}_{\bbeta_{\ini}}$ using four learners: random forest, boosting, LASSO, and a one-hidden-layer neural network. Random forests are implemented using the R package \texttt{randomForest}, boosting using \texttt{gbm}, LASSO using \texttt{hdm}, and the one-hidden-layer neural network using \texttt{nnet}. 
For LASSO, we include polynomial terms of degree 6 for age, 8 for income, 4 for education, and 2 for family size, together with all binary covariates, and all pairwise interactions among these terms, resulting in a total of 275 predictors.

In the debiasing step, we use degree-2 B-spline expansions for the continuous covariates income, age, family size, and education, with interior knot numbers $(\lceil n/200\rceil,\lceil n/200\rceil,4,2)=(25,25,4,2)$, respectively. We then construct $\sfzbar_k(x)$ by stacking the resulting univariate spline basis blocks with the binary covariates and an intercept term. This choice accounts for the heterogeneous supports of the continuous covariates: income and age vary over wide ranges, while family size and education take values on smaller grids. It also ensures numerical stability of $\hat{\Sigma}_k$ by avoiding very small eigenvalues. For comparison, we also report LDML estimates under the same settings as Section~6.2 of \citet{kallus2024localized} (with $K=5$), using a fixed two-fold split without re-randomizing folds to maintain transparency.

\begin{table}[ht] 
\centering
\begin{tabular}{clcccc}
  \toprule
$\tau$ &  & Forest & Neural Net & Boosting & LASSO \\ 
  \midrule
\multirow{2}{*}{0.25} & LDML & 1.06 & 1.02 & 1.04 & 1.06 \\ 
   & HOE & 1.01 & 1.09 & 1.09 & 1.11 \\ 
   \midrule
  \multirow{2}{*}{0.50} & LDML & 4.95 & 4.46 & 4.45 & 4.75 \\ 
   & HOE & 4.74 & 4.37 & 4.40 & 4.35 \\ 
   \midrule
  \multirow{2}{*}{0.75}& LDML & 13.60 & 11.90 & 12.16 & 12.65 \\ 
   & HOE & 12.68 & 11.77 & 12.15 & 11.90 \\ 
   \bottomrule
\end{tabular}
\caption{The QTE of 401(k) eligibility in thousand dollars estimated by LDML and higher-order estimates (HOE) using diﬀerent regression methods (LASSO, neural network, boosting, and random forests). Here $\tau \in \{25\%,50\%,75\%\}$ denotes the quantile level of QTE.}
\label{tab:realdata}
\end{table}

Table~\ref{tab:realdata} reports the point estimates of 25\%, 50\%, and 75\% QTEs. Overall, QTE estimates are generally stable among different nuisance learners, suggesting that the bias of the LDML estimator may be relatively small in this application. The higher-order estimates are close to their LDML counterparts, further strengthening this conclusion.
Consequently, the difference between the higher-order estimator and the LDML estimator provides empirical evidence on the magnitude of the potential bias of the LDML estimator \citep{liu2024assumption}.

We examine sensitivity to both the choice of $k$ and the initial estimator $\bbeta_{\ini}$. To assess sensitivity to $k$, we increase the number of interior knots for age by multiplying $n/200$ by a factor $\lambda \in \{0.75,1,1.25\}$ while keeping other knot numbers fixed. The higher-order estimators perform reasonably well in this range, as shown in Table~\ref{tab:basis_sensitivity} in Supplementary Material Section~\ref{app:real}. Stability is ensured by the numerical inversion of $\hat{\Sigma}_k$ when $k$ is relatively small compared to $n$, whereas the performance degrades if $\hat{\Sigma}_k$ approaches singularity.  Therefore, a heuristic strategy  to explore in the future is to increase $k$ (via knots or degree) until numerical instability occurs. Developing such data-driven rules for selecting $k$ and constructing $\sfzbar_k$ remains an important direction. 

Sensitivity to $\bbeta_{\ini}$ is examined similarly to the simulation studies. We find that the higher-order estimator is generally robust to small perturbations of $\hat{\bbeta}_{\ini}$, as shown in Table~\ref{tab:betainit_sensitivity} in Supplementary Material Section~\ref{app:real}.

\section{Concluding Remarks}
\label{sec:conclusion}

In this paper, we have constructed higher-order estimators for implicitly defined parameters frequently encountered in econometrics, including average treatment effects and quantile treatment effects as special cases. By applying the HOIFs framework to implicitly defined parameters,  we have contributed to the existing econometrics and causal inference literature. Our proposed estimators achieve improved convergence rates compared to first-order methods, which requires only that the initial estimator be consistent rather than imposing stringent rate conditions. For binary-treatment QTE, our estimator attains $\sqrt{n}$-consistency under the Hölder smoothness condition $s/d > 1/4$, substantially relaxing the requirements of existing approaches and matching the minimal Hölder smoothness condition for the ATE to be $\sqrt{n}$-estimable. We have also extended the framework to infinite-dimensional GTMs, providing a unified approach for estimating nonparametric causal functions such as dose–response curves.

Several directions remain for future research. First, while we focused on second-order estimators assuming that $\Sigma_k$ can be estimated sufficiently accurately, relaxing this assumption would require further advances in the theory of higher-order $U$-processes. Second, extending our framework to other complexity-reducing assumptions, such as sparsity classes \citep{athey2018approximate}, remains an open question. Third, with slight modifications, our framework is also expected to be able to address endogeneity problems (including nonparametric and many weak IV models) \citep{breunig2024adaptive}. Finally, it would be interesting to compare our higher-order debiased estimator with other related bias-correction or bias-aware methods \citep{armstrong2020bias, zheng2025perturbed, bonhomme2026higher}.

\putbib[\myreferences]
\end{bibunit}

\newpage
\setcounter{page}{1}
\begin{center}
{\LARGE Supplementary Materials of ``Higher-Order Debiased Estimators for General Treatment Models''}
\end{center}

	\vspace{.4cm} 
	\centerline{Yulin Zhang\textsuperscript{1}, Lin Liu\textsuperscript{2}, Zheng Zhang\textsuperscript{1}} 
	\vspace{.4cm} 
	\begin{center}
  \textsuperscript{1}Institute of Statistics \& Big Data, Renmin University of China; 
  \textsuperscript{2}Institute of Natural Sciences, MOE--LSC, School of Mathematical Sciences, CMA--Shanghai, SJTU--Yale Joint Center for Biostatistics and Data Science, Shanghai Jiao Tong University
\end{center}
	\vspace{.55cm}
    
\appendix

\begin{bibunit}[plainnat]

\allowdisplaybreaks

\begin{appendices}

\setcounter{equation}{0}
\def\theequation{S\arabic{section}.\arabic{equation}}
\def\thesection{S\arabic{section}}

\allowdisplaybreaks

\section{Comparisons of General Treatment Models and Marginal Structural Models}
\label{app:msm}

In this section, we further discuss the similarities and differences between the \emph{General Treatment Models (GTM)} \citep{ai2021unified} and the \emph{Marginal Structural Models (MSM)} (See Example~\ref{msm}), a well-known class of semiparametric causal models developed by Robins and colleagues \citep{hernan2000marginal}. MSM was developed, in part, to overcome the curse-of-dimensionality when drawing causal conclusions from interventional or observational studies when the treatment regimes under study are complex. In its simplest form, MSM directly models the counterfactual mean $\bbE Y (t)$ as a parameterized function $g (t; \bbeta^*)$ of $t \in \calT$ with $\bbeta^* \coloneqq (\beta_{0}^*, \beta_{1}^*)^{\top}$. 

When $\calT = \{0, 1\}$, $g$ can only be a linear model, and this linear model is saturated/nonparametric: MSM \eqref{main:msm} coincides with GTM in Example~\ref{ate}. For example, one can posit the following nonparametric (or saturated) linear MSM:
\begin{equation}
\label{linear}
\bbE Y (t) = g (t; \bbeta^*) = \beta_{0}^* + \beta_{1}^* t.
\end{equation}
This model is nonparametric because the model imposed on $\bbE Y (t)$ does not induce any restrictions on the observed data distribution. 

When $t$ is a continuous variable, however, finite-dimensional parametric MSM is often specified. For example, the statistician could insist on modeling $\bbE Y (t)$ by linear model \eqref{linear} and believe that the model is the truth. Under unconfoundedness, linear MSM \eqref{linear} imposes the following restriction on the observed data distribution:
\begin{equation}
\label{linear MSM}
\bbE Y (t) = \bbE [\bbE (Y | X, T = t)] =  \bbE [\xi (X, T) Y |  T = t]  = \beta_{0}^* + \beta_{1}^* t \equiv g (t; \bbeta^{\ast}),
\end{equation}
where $\xi (x, t) = \p_{T}(x)/\p_{T|X}(t|x)$.
Under unconfoundedness, Model \eqref{linear MSM} imposes the following moment/equality constraint on the data generating process:
\begin{equation}
\bbE \left[ \xi (X, T) \left( Y - g (T; \bbeta^{\ast}) \right) \omega (T) \right] \equiv 0, \, \forall \, \omega: \calT \rightarrow \bbR,
\end{equation}
so induces a conditional moment constraint in the sense of \citet{ai2003efficient}, leading to a non-trivial orthocomplement to the model tangent space. We also refer readers to Section~3 of \citet{chen2025local} for a related discussion. It should be noted that one can also parameterize the MSM using nonparametric models when $t$ is continuous, as in \citet{bonvini2022fast} or as we have done in Section~\ref{sec:dose}.

\section{Verification of Assumption~\ref{as:criterion}  for ATE and QTE}\label{app:Verification}
In this section, we demonstrate that Assumptions~\ref{as:criterion}(i) and \ref{as:criterion}(iii) hold for both ATE and QTE from Examples~\ref{ate} and~\ref{qte}. 
We first recall the following definition from the empirical process theory \citep{van2023weak}.
\begin{definition}[VC-subgraph class]
\label{def:VC-subgraph}
A collection $\calF$ of measurable functions is called a VC-subgraph class, or a VC-class, if the collection of all subgraphs of the functions (see Section 2.6.2 of \citet{van2023weak}) in $\calF$ forms a VC-class of sets (see Section 2.6.1 of \citet{van2023weak}).
\end{definition}

It suffices to verify  Assumption~\ref{as:criterion} for the first coordinate of $\Gamma (\cdot, \cdot, \bbeta)$, which corresponds to the estimating equation for $\beta_0^*$ when $T=0$. We first verify  Assumption~\ref{as:criterion}(i):

\begin{itemize}
    \item ATE in Example~\ref{ate}. We need only to consider the function class  $\calF_{\mathsf{ATE} } \coloneqq \{  Y - \beta_0: \beta_0 \in \calB_0\}$, where $\calB_0$ is the parameter space for $\beta_0^*$. By Lemma~2.6.17 in \cite{van2023weak}, $\calF_{\mathsf{ATE} }$ is of VC-subgraph with index $1$. Theorem 2.6.7 of \citet{van2023weak} further implies that, for any probability measure $Q$,  
      $$   \sfN \left( \calF_{\mathsf{ATE} }, \| \cdot \|_{Q,2}, \epsilon \| \F_{\mathsf{ATE} } \|_{Q,2} \right) \le
        C\,(16e)\,\epsilon^{-2},\quad 0<\epsilon<1,$$
        where  $\F_{\mathsf{ATE} } = \sup_{\beta_0 \in \calB_0}|y - \beta_0|$  is the envelope and $C$ is a universal constant.  Hence, by  Definition~\ref{def:VC-type},
       $\calF_{\mathsf{ATE} }$ is of VC-type with characteristics $(4\sqrt{Ce},2\bigr)$.  

    \item QTE in Example~\ref{qte}. In this case, the relevant function class $ \calF_{\mathsf{QTE} } \coloneqq \{\tau - \mathbbm{1} (Y \leq \beta_0) : \beta_0 \in \calB_{0}\}$. Denote $\rho'(u)\coloneqq \tau -\mathbbm1(u\le0)$, for $u\in\bbR$. Then $\rho'(\cdot)$ is monotone and $\calF_{\mathsf{QTE} } = \rho' \circ \calF_{\mathsf{ATE} }\coloneqq \{\rho'(f): f\in\calF_{\mathsf{ATE}} \}$. Lemma 9.9 (viii) of \citet{kosorok2008introduction}  implies that applying a monotone function to a VC-subgraph class does not increase the VC index. Therefore, $\calF_{\mathsf{QTE} }$ is also of VC-subgraph with index $1$, and thus of VC-type with characteristics $(4\sqrt{Ce}, 2\bigr)$.
\end{itemize}

For verification of Assumption~\ref{as:criterion}(iii),  denote the first coordinate of $\Gamma (\cdot, \cdot, \bbeta)$ as $\Gamma^{(0)} (\cdot, \cdot, \beta_0)$.
\begin{itemize}
    \item ATE in Example~\ref{ate}. The difference is   $\Gamma^{(0)} (Y, T, \beta_0) - \Gamma^{(0)} (Y, T, \beta_0') = (\beta_0 - \beta_0')(  1-T )$, which is linear in $\beta_0 - \beta_0'$. Thus, $\Gamma^{(0)} (\cdot, \cdot, \beta_0)$ meets this condition with index $\alpha_{0} = 1$. 

    \item QTE in  Example~\ref{qte}. We have $\Gamma^{(0)} (Y, T, \beta_0) - \Gamma^{(0)} (Y, T, \beta_0')  =[ \mathbbm{1} ( Y \leq \beta_0') - \mathbbm{1} ( Y \leq \beta_0) ](1-T)$. Then, 
\begin{align*}
    & \ \bbE \Big[\sup_{ |\beta_{0}-\beta_{0}'| \lesssim \delta} \big| \Gamma^{(0)} (Y, T, \beta_0) - \Gamma^{(0)} (Y, T, \beta_0') \big|^2\Big] \\
    = & \ \bbE \Big[ (1-T)^2\cdot\sup_{|\beta_{0}-\beta_{0}' |\lesssim \delta} \mathbbm{1}(\beta_{0} \wedge \beta_{0}' <Y \leq \beta_{0} \vee \beta_{0}' ) \Big] \\
    = &\ \bbE \Big[ \sup_{|\beta_{0}-\beta_{0}' |\lesssim \delta} \mathbbm{1}(\beta_{0} \wedge \beta_{0}' <Y \leq \beta_{0} \vee \beta_{0}' ) \Big] \cdot \bbP(T=0)\\
    = & \  \bbE \Big[ \sup_{|\beta_{0}-\beta_{0}' |\lesssim \delta} \left|\P_{Y}(\beta_{0} ) - \P_{Y}(\beta_{0}' )\right| \Big] \cdot \bbP(T=0)\\
    \lesssim& \  \delta,
\end{align*}
where $\P_{Y} (\beta) \coloneqq \P(Y\le \beta)$ is the marginal distribution of $Y$ and the last inequality holds if $\P_{Y} (\cdot)$ is Lipschitz continuous. Thus, in this case,  $\Gamma^{(0)} (\cdot, \cdot, \beta_0)$ satisfies this condition with index $\alpha_{0} = 1/2$.
\end{itemize}

\section{Higher-Order Estimators for Parameters of GTMs}
\label{app:truncation}
\setcounter{equation}{0}

As mentioned in Remark~\ref{rem:unknown Sigma}, for the sake of completeness, the following lemma characterizes the HOIFs of $\psi (\bbeta)$ (or more precisely, of $\bar{\psi}_{k}(\bbeta)$) in the sense of \citet{robins2016technical}, which paves the way towards further reducing the bias due to estimating $\Sigma_{k}$ by $\hat{\Sigma}_{k}$. However, deriving the statistical properties of these higher-order estimators is beyond the scope of this paper, as the proof will eventually need to handle $U$-processes of order growing at a rate poly-log in sample size $n$.

\begin{lemma}
\label{lem:HOIF}
For any $\bbeta \in \calB$, under Assumption \ref{as:identification}, the $m^{th}$-order influence functions of $\mathrm{B}_{\psi, k}$ for $m \geq 2$ are of the following forms:
\begin{align*}
& \IF_{\mathrm{B}_{\psi, k}}^{(2)} (\bbeta)= - (\hat{\xi} (X_{1}, T_{1}) - \xi (X_{1}, T_{1})) \bar{\phi}_{k} (X_{1}, T_{1})^{\top} \Sigma_{k}^{-1} \bar{\phi}_{k} (X_{2}, T_{2}) (\Gamma (Y_{2}, T_{2}, \bbeta) - \hat{b}_{\bbeta_{\mathsf{init}}} (X_{2}, T_{2})), \\
& \IF_{\mathrm{B}_{\psi, k}}^{(m)} (\bbeta)- \IF_{\mathrm{B}_{\psi, k}}^{( m - 1)} (\bbeta)= \\
& + (-1)^{m + 1} (\hat{\xi} (X_{1}, T_{1}) - \xi (X_{1}, T_{1})) \bar{\phi}_{k} (X_{1}, T_{1})^{\top} \Sigma_{k}^{-1} \prod_{s = 3}^{m} \left\{ (\bar{\phi} (X_{s}, T_{s}) \bar{\phi} (X_{s}, T_{s})^{\top} - \Sigma_{k}) \Sigma_{k}^{-1} \right\} \\
& \times \bar{\phi}_{k} (X_{2}, T_{2}) (\Gamma (Y_{2}, T_{2}, \bbeta) - \hat{b}_{\bbeta_{\mathsf{init}}} (X_{2}, T_{2})).
\end{align*}
Then the $m^{th}$-order influence functions of $\bar{\psi}_{k} (\bbeta)$ for $m \geq 2$ simply take the forms
\begin{align*}
\IF_{\bar{\psi}_{k}}^{(m)} (\bbeta)= \IF_{\psi}^{(1)} (\bbeta) + \IF_{\mathrm{B}_{\psi, k}}^{(m)} (\bbeta).
\end{align*}
\end{lemma}

Since \eqref{ident} is only used for defining the treatment effect parameter $\bbeta^{\ast}$ and there is no restriction on the observed data tangent space, the derivations of HOIFs of $\bar{\psi}_{k} (\bbeta) $ will follow the strategy presented in \citet{liu2024assumption} and \citet{robins2016technical}. When referring to \citet{liu2024assumption} in the proof, we always mean the \href{https://arxiv.org/pdf/2306.10590}{arXiv version} of \citet{liu2024assumption}.

\begin{proof}[Proof of Lemma \ref{lem:HOIF}]
We first derive the HOIFs of $\mathrm{B}_{\psi, k}(\bbeta)$ and $\bar{\psi}_{k}(\bbeta)$. To ease presentation, we introduce some additional notation useful locally in the derivation.
\begin{align*}
& \alpha_{k} \coloneqq \Sigma_{k}^{-1} \E [(\hat{\xi} (X, T) - \xi (X, T)) \bar{\phi}_{k} (X, T)], \eta_{\bbeta, k} \coloneqq \Sigma_{k}^{-1} \E [(\Gamma (Y, T, \bbeta) - \hat{b}_{\bbeta_{\mathsf{init}}} (X, T)) \bar{\phi}_{k} (X, T)], \\
& \bar{\xi}_{k} (\cdot, \cdot) \coloneqq \hat{\xi} (\cdot, \cdot) - \alpha_{k}^{\top} \bar{\phi}_{k} (\cdot, \cdot), \bar{b}_{\bbeta, k} (\cdot, \cdot) \coloneqq \hat{b}_{\bbeta_{\mathsf{init}}} (\cdot, \cdot) + \eta_{\bbeta, k}^{\top} \bar{\phi}_{k} (\cdot, \cdot).
\end{align*}
So $\mathrm{B}_{\psi, k} (\bbeta)\equiv \alpha_{k}^{\top} \Sigma_{k} \eta_{\bbeta, k}$ and $\bar{\psi}_{k} (\bbeta)= \hat{\psi}^{(1)} (\bbeta) - \alpha_{k}^{\top} \Sigma_{k} \eta_{\bbeta, k}$. By Lemma A.2 of \citet{liu2024assumption}, we have
\begin{align*}
& \IF_{\alpha_{k}}^{(1)} = \Sigma_{k}^{-1} \left( \bar{\xi}_{k} (X, T) - \xi (X, T) \right) \bar{\phi}_{k} (X, T), \\
& \IF_{\eta_{\bbeta, k}}^{(1)} = \Sigma_{k}^{-1} \left( \Gamma (Y, T, \bbeta) - \bar{b}_{\bbeta, k} (X, T) \right) \bar{\phi}_{k} (X, T).
\end{align*}

It is a standard calculation to show that the influence function of $\mathrm{B}_{\psi, k} (\bbeta)$ is:
\begin{align*}
& \ \IF_{\mathrm{B}_{\psi, k}}^{(1)} (\bbeta)\\
= & \ (\hat{\xi} (X, T) - \xi (X, T)) \bar{\phi}_{k} (X, T)^{\top} \Sigma_{k}^{-1} \eta_{\bbeta, k} + \alpha_{k}^{\top} \Sigma_{k}^{-1} \bar{\phi}_{k} (X, T) (\Gamma (Y, T, \bbeta) - \hat{b}_{\bbeta_\mathsf{init}} (X, T)) \\
& - \alpha_{k}^{\top} \Sigma_{k}^{-1} \bar{\phi} (X, T) \bar{\phi} (X, T)^{\top} \Sigma_{k}^{-1} \eta_{\bbeta, k} - \alpha_{k}^{\top} \Sigma_{k}^{-1} \eta_{\bbeta, k}.
\end{align*}
By applying the same calculus rule from page 50 to page 51 of \citet{liu2024assumption}, we obtain the second-order influence function of $\mathrm{B}_{\psi, k}$ takes the form
\begin{align*}
\IF_{\mathrm{B}_{\psi, k}}^{(2)} (\bbeta) = - (\hat{\xi} (X_{1}, T_{1}) - \xi (X_{1}, T_{1})) \bar{\phi}_{k} (X_{1}, T_{1})^{\top} \Sigma_{k}^{-1} \bar{\phi}_{k} (X_{2}, T_{2}) (\Gamma (Y_{2}, T_{2}, \bbeta) - \hat{b}_{\bbeta_{\mathsf{init}}} (X_{2}, T_{2})).
\end{align*}
The HOIFs of $\mathrm{B}_{\psi, k} (\bbeta) $ then read as follows: for $j \geq 3$
\begin{align*}
& \ \IF_{\mathrm{B}_{\psi, k}}^{(j)} (\bbeta)- \IF_{\mathrm{B}_{\psi, k}}^{(j-1)} (\bbeta)\\
= & \ (-1)^{j - 1} (\hat{\xi} (X_{1}, T_{1}) - \xi (X_{1}, T_{1})) \bar{\phi}_{k} (X_{1}, T_{1})^{\top} \Sigma_{k}^{-1} \prod_{s = 3}^{j} \left\{ \left( \bar{\phi}_{k} (X_{s}, T_{s}) \bar{\phi}_{k} (X_{s}, T_{s})^{\top} \Sigma_{k} \right) \Sigma_{k}^{-1} \right\} \\
& \times \bar{\phi}_{k} (X_{2}, T_{2}) (\Gamma (Y_{2}, T_{2}, \bbeta) - \hat{b}_{\bbeta_{\mathsf{init}}} (X_{2}, T_{2})).
\end{align*}
The HOIFs of $\bar{\psi}_{k} (\bbeta)$ follows from the HOIFs of $\mathrm{B}_{\psi, k} (\bbeta)$ and the linearity of influence functions.
\end{proof}

Armed with these results, we can construct higher-order estimators of $\psi (\bbeta)$ and of $\bbeta$ itself as follows. Similar to the first-order estimator of $\psi (\bbeta)$,  which is a second-order $U$-statistic, the $m$-th order estimator of $\psi (\bbeta)$ is given by the following $(m + 1)$-th order $U$-statistic:
\begin{align*}
\hat{\psi}_{k}^{(m)} (\bbeta) \coloneqq \hat{\psi}^{(1)}(\bbeta) + \bbU_{n, m + 1} [\hat{\IF}_{\mathrm{B}_{\psi, k}}^{(m)}]
\end{align*}
where $\hat{\IF}_{\mathrm{B}_{\psi, k}}^{(m)}$ is defined in the same way as $\IF_{\mathrm{B}_{\psi, k}}^{(m)}$, except that (i) $\xi (X_{1}, T_{1}) \bar{\phi}_{k} (X_{1}, T_{1})$ is replaced by $\bar{\phi}_{k} (X_{1}, T_{3})$,  (ii) $\prod\limits_{s = 3}^{m}$ becomes $\prod\limits_{s = 4}^{m + 1}$ and (iii) $\Sigma_k$ is replaced by $\hat{\Sigma}_k$.  It is worth noting that $- \hat{\mathrm{B}}_{\psi, k} (\bbeta) \equiv \bbU_{n, 3} [\hat{\IF}_{\mathrm{B}_{\psi, k}}^{(2)} (\bbeta)]$. $\hat{\bbeta}^{(m)}$ is then defined as the solution to the system of $U$-statistic equations
\begin{align*}
\hat{\psi}_{k}^{(m)} (\bbeta) \equiv 0.
\end{align*}

Next, the following result bounds the bias and variance of $\hat{\psi}_{k}^{(m)} (\bbeta)$ as an estimator of the truncated parameter $\bar{\psi}_{k} (\bbeta)$, indicating the potential benefit of using a higher-order estimator to debias the estimation error due to estimating $\Sigma_{k}$. The proof follows immediately from \citet{liu2017semiparametric}, hence omitted.
\begin{lemma}
\label{lem:bias}
For any $\bbeta \in \calB$ and fixed positive integer $m \geq 2$, we have
\begin{equation}
\label{higher-order bias}
\left\Vert \E \left[ \hat{\psi}_{k}^{(m)} (\bbeta) - \bar{\psi}_{k} (\bbeta) \right] \right\Vert \lesssim \Vert \hat{\xi} - \xi \Vert_{\P,2} \cdot \Vert \hat{b}_{\bbeta_{\mathsf{init}}} - b_{\bbeta} \Vert_{\P,2} \cdot \Vert \hat{\Sigma}_{k} - \Sigma_{k} \Vert_{\op}^{m - 1}.
\end{equation}
As a consequence of Assumption \ref{as:nuisance_est}(ii), we in turn have
\begin{equation}
\label{reduced higher-order bias}
\left\Vert \E \left[ \hat{\psi}_{k}^{(m)} (\bbeta) - \bar{\psi}_{k} (\bbeta)\right] \right\Vert \lesssim \Vert \hat{\xi} - \xi \Vert_{\P,2} \cdot \left( \Vert \bbeta_{\mathsf{init}} - \bbeta \Vert + \Vert \hat{b}_{\bbeta_{\mathsf{init}} } - b_{\bbeta_{\mathsf{init}} } \Vert_{\P,2} \right) \cdot \Vert \hat{\Sigma}_{k} - \Sigma_{k} \Vert_{\op}^{m - 1}.
\end{equation}
If we further impose the same conditions as in Theorem \ref{th:consistency} and additionally assume (i) $m \lesssim \sqrt{\log n}$ and $k \lesssim n / \log^{3} (n)$ and (ii) $\Vert \hat{\Sigma}_{k} - \Sigma_{k} \Vert_{\op} = o_{\P} (1)$, we also have
\begin{equation}
\label{higher-order variance}
\begin{split}
\cov [\hat{\psi}_{k}^{(m)} (\bbeta)] \preceq \frac{C}{n} \left( \begin{array}{c}
\dfrac{k}{n} + \left\{ \Vert \Pi [\hat{b}_{\bbeta_{\ini}} - b_{\bbeta} | \bar{\phi}_{k}] \Vert_{\P,2}^{2} + \Vert \Pi [\hat{\xi} - \xi | \bar{\phi}_{k}] \Vert_{\P,2}^{2} \right\} \\
+ \, \underset{(p, q): p, q \geq 1, \frac{1}{p} + \frac{1}{q} = 1}{\min} \Vert \Pi [\hat{b}_{\bbeta_{\ini}} - b_{\bbeta} | \bar{\phi}_{k}] \Vert_{\P,p}^{2} \cdot \Vert \Pi [\hat{b}_{\bbeta_{\ini}} - b_{\bbeta} | \bar{\phi}_{k}] \Vert_{\P,q}^{2}
\end{array} \right),
\end{split}
\end{equation}
for some absolute constant $C > 0$. 
\end{lemma}

Lemma~\ref{lem:bias} can be served as the starting point of deriving statistical properties of $\hat{\bbeta}^{(m)}$, were results related to higher-order $U$-processes established. We leave this important technical problem to future work.

\subsection{Convergence Rates of the Empirical Gram Matrix Estimator}
\label{app:unknown}

When $\bbeta^*$ is $n^{-1/2}$-estimable, we can appeal to the methodology developed in \citet{liu2017semiparametric} and take $k = o (n)$ or more precisely $k \lesssim n / \log^{3} (n)$ without assuming regularity conditions on the marginal density $\p_{X}(x)$ of $X$. Specifically, we use the same nuisance sample for estimating the nuisance parameters to construct the following second-order $U$-statistic estimator of $\Sigma_{k}$:
\begin{align*}
\hat{\Sigma}_{k} & \coloneqq \frac{1}{n (n - 1)} \sum_{1 \leq i \neq j \leq n} \{\sfzbar_{k_{x}} (X_{i}) \sfzbar_{k_{x}} (X_{i})^{\top}\} \otimes \{\sfwbar_{k_{t}} (T_{j}) \sfwbar_{k_{t}} (T_{j})^{\top}\} \\
& = \frac{1}{n (n - 1)} \sum_{1 \leq i \neq j \leq n} \bar{\phi}_{k} (X_{i}, T_{j}) \bar{\phi}_{k} (X_{i}, T_{j})^{\top}.
\end{align*}
Obviously, $\bbE [\hat{\Sigma}_{k}] = \Sigma_{k}$ so it is an unbiased estimator of $\Sigma_{k}$. We note that throughout the article $\hat{\Sigma}_{k}$ denotes a generic estimator of $\Sigma_{k}$ and only in this section $\hat{\Sigma}_{k}$ denotes the empirical Gram matrix estimator.

Let $\bar{\phi}_{k} (X_{i}, T_{j}) \bar{\phi}_{k} (X_{i}, T_{j})^{\top} \eqqcolon Q_{i} \otimes R_{j}$ to further simplify the presentation. Corollary 3.1 of \citet{minsker2019moment}, which we cite below, provides a user-friendly moment bound on the second-order degenerate $U$-statistic component of the following Hoeffding decomposition:
\begin{align*}
& \ \hat{\Sigma}_{k} - \Sigma_{k} \\
\equiv & \ \sum_{i = 1}^{n} \frac{1}{n} \left\{ \bbE \left[ \frac{Q_{i} \otimes R_{j} + Q_{j} \otimes R_{i}}{2} | O_{i} \right] - \bbE Q \otimes \bbE R \right\} \\
& + \sum_{1 \leq i < j \leq n} \frac{1}{n (n - 1)} \left\{ (Q_{i} \otimes R_{j} + Q_{j} \otimes R_{i}) - \bbE (Q_{i} \otimes R_{j} + Q_{j} \otimes R_{i} | O_{i}) \right\} \\
\eqqcolon & \ \sum_{i = 1}^{n} G (O_{i}) + \sum_{1 \leq i < j \leq n} H (O_{i}, O_{j}),
\end{align*}
where $G$ and $H$ are appropriately chosen \emph{$k \times k$-symmetric matrix}-valued functions. For more recent development on tail bounds for matrix-valued $U$-statistics, we refer readers to \citet{bandeira2025matrix}.

\begin{lemma}[Corollary 3.1 of \citet{minsker2019moment}]
\label{lem:matrix-U}
Let $\{O_{i}^{(l)}, i = 1, \cdots, n\}$ for $l = 1, 2$ be two independent copies of $\{O_{i}, i = 1, \cdots, n\}$ and let $\bbE_{l}$ denote the expectation with respect to only the copies indexed by $l$. Then
\begin{equation}
\label{matrix U}
\begin{split}
\bbE \left\Vert \sum_{1 \leq i < j \leq n} H (O_{i}, O_{j}) \right\Vert_{\op} \lesssim & \ (\log k)^{2} \left( \sum_{i = 1}^{n} \bbE \left\Vert \sum_{j: j \neq i} \bbE_{2} H^{2} (O_{i}^{(1)}, O_{j}^{(2)}) \right\Vert_{\op} \right) \\
& + (\log k)^{3} \left( \bbE \max_{i} \left\Vert \sum_{j: j \neq i} H^{2} (O_{i}^{(1)}, O_{j}^{(2)}) \right\Vert_{\op} \right).
\end{split}
\end{equation}
\end{lemma}
Thus to obtain a useful bound of $\bbE \left\Vert \sum_{1 \leq i < j \leq n} H (O_{i}, O_{j}) \right\Vert_{\op}$, we need to control each of the two terms on the RHS of \eqref{matrix U}, respectively denoted as $(\log k)^{2} I_{1}$ and $(\log k)^{3} I_{2}$ here. We first simplify the following term that will be useful to control both $I_{1}$ and $I_{2}$. Then 
\begin{align*}
& \ \sum_{j: j \neq i} H^{2} (O_{i}^{(1)}, O_{j}^{(2)}) \\
= & \ \frac{1}{n^{2} (n - 1)^{2}} \sum_{j: j \neq i} \left\{ \begin{array}{c}
\{Q_{i}^{(1)} \otimes (R_{j}^{(2)} - \bbE R)\} \{Q_{i}^{(1) \top} \otimes (R_{j}^{(2)} - \bbE R)^{\top}\} \\
+ \, \{(Q_{j}^{(2)} - \bbE Q) \otimes R_{i}^{(1)}\} \{(Q_{j}^{(2)} - \bbE Q)^{\top} \otimes R_{i}^{(1) \top}\} \\ 
+ \, \{Q_{i}^{(1)} \otimes (R_{j}^{(2)} - \bbE R)\} \{(Q_{j}^{(2)} - \bbE Q)^{\top} \otimes R_{i}^{(1) \top}\} \\
+ \, \{(Q_{j}^{(2)} - \bbE Q) \otimes R_{i}^{(1)}\} \{Q_{i}^{(1) \top} \otimes (R_{j}^{(2)} - \bbE R)^{\top}\}
\end{array} \right\} \\
= & \ \frac{1}{n^{2} (n - 1)} \left\{ \begin{array}{c}
\{Q_{i}^{(1)} Q_{i}^{(1) \top}\} \otimes \{\frac{1}{n - 1} \sum\limits_{j: j \neq i} (R_{j}^{(2)} - \bbE R) (R_{j}^{(2)} - \bbE R)^{\top}\} \\
+ \, \{\frac{1}{n - 1} \sum\limits_{j: j \neq i} (Q_{j}^{(2)} - \bbE Q) (Q_{j}^{(2)} - \bbE Q)^{\top}\} \otimes \{R_{i}^{(1)} R_{i}^{(1) \top}\} \\ 
+ \, \frac{1}{n - 1} \sum\limits_{j: j \neq i} \{Q_{i}^{(1)} (Q_{j}^{(2)} - \bbE Q)^{\top}\} \otimes \{(R_{j}^{(2)} - \bbE R) R_{i}^{(1) \top}\} \\
+ \, \frac{1}{n - 1} \sum\limits_{j: j \neq i} \{(Q_{j}^{(2)} - \bbE Q) Q_{i}^{(1) \top}\} \otimes \{R_{i}^{(1)} (R_{j}^{(2)} - \bbE R)^{\top}\}
\end{array} \right\}.
\end{align*}
As for $I_{1}$, we have
\begin{align*}
& \ \sum_{j: j \neq i} \bbE_{2} H^{2} (O_{i}^{(1)}, O_{j}^{(2)}) \\
= & \ \frac{1}{n^{2} (n - 1)} \left\{ \begin{array}{c}
\{Q_{i}^{(1)} Q_{i}^{(1) \top}\} \otimes \bbE_{2} (R_{j}^{(2)} - \bbE R) (R_{j}^{(2)} - \bbE R)^{\top} \\
+ \bbE_{2} (Q_{j}^{(2)} - \bbE Q) (Q_{j}^{(2)} - \bbE Q)^{\top} \otimes \{R_{i}^{(1)} R_{i}^{(1) \top}\} \\ 
+ \, \bbE_{2} \{Q_{i}^{(1)} (Q_{j}^{(2)} - \bbE Q)^{\top}\} \otimes \{(R_{j}^{(2)} - \bbE R) R_{i}^{(1) \top}\} \\
+ \bbE_{2} \{(Q_{j}^{(2)} - \bbE Q) Q_{i}^{(1) \top}\} \otimes \{R_{i}^{(1)} (R_{j}^{(2)} - \bbE R)^{\top}\}
\end{array} \right\}.
\end{align*}
The first two terms in the above display are relatively easy to control. For the last two terms, we use the third term as an example and by symmetry the fourth term follows from the same analysis.
\begin{align*}
& \ \bbE_{2} \{Q_{i}^{(1)} (Q_{j}^{(2)} - \bbE Q)^{\top}\} \otimes \{(R_{j}^{(2)} - \bbE R) R_{i}^{(1) \top}\} = \bbE_{2} \{Q_{i}^{(1)} \otimes (R_{j}^{(2)} - \bbE R)\} \{(Q_{j}^{(2)} - \bbE Q) \otimes R_{i}^{(1)}\}^{\top} \\
= & \ Q_{i}^{(1)} R_{i}^{(1) \top} \otimes \bbE_{2} (R_{j}^{(2)} - \bbE R) (Q_{j}^{(2)} - \bbE R)^{\top}.
\end{align*}
Then
\begin{align*}
I_{1} \lesssim \frac{\bbE \Vert Q Q^{\top} \Vert_{\op} k_{t} + \bbE \Vert R R^{\top} \Vert_{\op} k_{x} + \bbE \Vert Q R^{\top} \Vert_{\op} + \bbE \Vert R Q^{\top} \Vert_{\op}}{n (n - 1)} \lesssim \frac{k}{n^{2}}.
\end{align*}

As for $I_{2}$, we have, by repeatedly applying the matrix Bernstein inequality as in Lemma \ref{lem:matrix Bernstein}, 
\begin{align*}
I_{2} \equiv & \ \bbE \max_{i} \left\Vert \sum_{j: j \neq i} H^{2} (O_{i}^{(1)}, O_{j}^{(2)}) \right\Vert_{\op} \\
\lesssim & \ \frac{k}{n^{3}}.
\end{align*}
Thus combining the above together, we have that the second-order degenerate $U$-statistic part is bounded by
\begin{align*}
\left( \frac{k \log k}{n} \right)^{2} + \left( \frac{\sqrt{k} \log k}{n} \right)^{2} + \left( \frac{k^{1 / 3} \log k}{n} \right)^{3}.
\end{align*}

As for the linear term in the above decomposition, one can directly apply the standard matrix Bernstein inequality \citep{tropp2015introduction}, cited below.

\begin{lemma}[Rephrasing Theorem 6.1.1 of \citet{tropp2015introduction}]
\label{lem:matrix Bernstein}
Let $M \coloneqq \sum_{i = 1}^{n} G (O_{i})$. If for each $i = 1, \cdots, n$, $\Vert G (O_{i}) \Vert_{\op} \lesssim k / n$, then
\begin{align*}
\bbE \Vert M \Vert_{\op} \lesssim \sqrt{\frac{k \log k}{n}} + \frac{k \log k}{n}.
\end{align*}
\end{lemma}

Thus in summary, we have
\begin{align*}
\bbE \Vert \hat{\Sigma}_{k} - \Sigma_{k} \Vert_{\op} \lesssim \sqrt{\frac{k \log k}{n}} + \frac{k \log k}{n} + \left( \frac{k \log k}{n} \right)^{2} + k \left\{ \left( \frac{\log k}{n} \right)^{2} + \left( \frac{\log k}{n} \right)^{3} \right\}.
\end{align*}
We can then combine this result with the empirical HOIF framework developed in \citet{liu2017semiparametric} to construct higher-order estimating equations for $\bbeta^*$ as: 
\begin{align*}
\hat{\psi}_{k}^{(m)} (\bbeta) = 0,
\end{align*}
such that the final estimator $\hat{\bbeta}^{(m)}$ with order $m \asymp \log n$ is $n^{-1/2}$-consistent and asymptotic normal with no extra smoothness assumption imposed on the marginal density of $X$.

\section{Discussion of Assumption~\ref{as:basis}}
\label{app:assumption basis}

In this section, we briefly show how the  localized conditions   in Assumption~\ref{as:basis} imply the standard uniform eigenvalue bounds for the Gram matrices: 
\begin{lemma}\label{lemma:Gram-Matrix eigenvalue bounds}
       For every $k_{x}$ and $k_{t}$, the eigenvalues of $
\bbE \left[ \sfzbar_{k_{x}} (X)^{\otimes 2} \right] $ and $\bbE [\sfwbar_{k_{t}} (T)^{\otimes 2}]$
are bounded from above and away from zero uniformly in $k_{x}$ and $k_{t}$.
\end{lemma}
\begin{proof}

We focus on $\bbE [\sfwbar_{k_{t}} (T)^{\otimes 2}]$; the argument for  $\bbE \left[ \sfzbar_{k_{x}} (X)^{\otimes 2} \right]$ is analogous. 
For any unit vector $\bmb\in\bbR^{k_t}$, by definition 
\begin{align*}
 \bmb^{\top} \bbE [\sfwbar_{k_{t}} (T)^{\otimes 2}] \bmb = \int_{\calT} \left( \sum_{i=1}^{k_{t}} b_i \sfw_{i} (t) \right)^2 \p_{T} (t) \diff t .
\end{align*}
By Assumption~\ref{as:basis}(ii), for any $t\in\calT$, the number of nonzero terms in the basis function $\sfwbar_{k_t}(t)$ is at most $c_{8}$. Denoted them
by $(\sfw_{i_1}(t),\ldots,\sfw_{i_{c_{8}}}(t))^{\top}$. Then, 
\begin{align*}
    \left( \sum_{i=1}^{k_{t}} b_i \sfw_{i} (t) \right)^2 \lesssim \sum_{j=1}^{c_{8}} b_{i_j}^2 \sfw_{i_j}(t)^2 . 
\end{align*}
It follows from Assumptions~\ref{as:basis}(i) and (iv) that
\begin{align*}
    \int_{\calT} \left( \sum_{i=1}^{k_{t}} b_i \sfw_{i} (t) \right)^2 \p_{T} (t) \diff t \lesssim & \int_{\calT} \sum_{j=1}^{c_{8}} b_{i_j}^2 \sfw_{i_j}(t)^2 \p_{T} (t) \diff t \\
    \lesssim & \sum_{j=1}^{c_{8}} b_{i_j}^2  \int_{\calS_{T,i_j}} \zeta_t(k_t)^2 \p_{T} (t) \diff t \\
    \lesssim  & \sum_{j=1}^{c_{8}} b_{i_j}^2  \cdot k_t \cdot 1/k_t \lesssim 1,
\end{align*}
which yields the desired uniform upper bound on the largest eigenvalue.

For the lower bound, note that $\sum_{i=1}^{k_t} \mathbbm{1} \{t\in\calS_{T,i}\} \le c_8 $ for all $t\in\calT$. Hence,
\begin{align*}
   \sum_{i=1}^{k_t} \int_{\calS_{T,i}} \left( \sum_{i=1}^{k_{t}} b_i \sfw_{i} (t) \right)^2 \p_{T} (t) \diff t = &  \int_{\calT} \left( \sum_{i=1}^{k_t} \mathbbm{1} \{t\in\calS_{T,i}\}\right) \left( \sum_{i=1}^{k_{t}} b_i \sfw_{i} (t) \right)^2 \p_{T} (t) \diff t \\
   \le &\  c_8 \int_{\calT}  \left( \sum_{i=1}^{k_{t}} b_i \sfw_{i} (t) \right)^2 \p_{T} (t) \diff t\\
   \lesssim &\
    \bmb^{\top} \mathbb{E}\left[ \sfwbar_{k_{t}} (T) \sfwbar_{k_{t}} (T)^{\top} \right] \bmb.
\end{align*}
Then taking sum over all $\calS_{T,i}$ and  using the assumption that $\bmb^{\top}\int_{\calS_{X,i}}  \sfwbar_{k_{t}} (t) \sfwbar_{k_{t}} (t)^{\top}  \diff
         t \bmb \gtrsim b_i^2$, we obtain the lower bound
         \begin{align*}
             \bmb^{\top} \mathbb{E}\left[ \sfwbar_{k_{t}} (T) \sfwbar_{k_{t}} (T)^{\top} \right] \bmb \gtrsim 1.
         \end{align*}
    
\end{proof}

\section{On the Relationship Between Higher-Order Estimators and Other Related Proposals in Econometrics}
\label{app:connection}

As suggested by a referee, in this section, we illustrate the connection between the higher-order estimators studied in this paper and several other bias correction methods in the recent econometrics literature, including \citet{cattaneo2019two}, \citet{cattaneo2018kernel}, and \citet{cattaneo2025higher}.

\subsection{Discussion of \texorpdfstring{\citet{cattaneo2019two}}{}}\label{app:cattaneo2019}

\citet{cattaneo2019two} study two-step semiparametric estimation and inference in settings where the first-step linear regression may include many covariates. To keep the exposition simple, we focus on the ATE case and take the target parameter to be $\beta^* = \bbE[Y (1)]$. To facilitate comparison with our construction, we  characterize $\beta^*$ via the moment condition
\begin{align}\label{eq:psi ate} 
\psi (\beta^*)  = \E \left[ \xi (X, T) \{Y - b (X,T = 1)\} + b (X,T=1) - \beta^{\ast} \right]=0,
\end{align}
where $\xi (x,t) =  t/\P(T = 1 | X = x) \ $ and $\ b(x,t) = t \bbE [Y |X = x, T=1]$ are nuisance functions. 

In this example, we first would like to emphasize that \citet{cattaneo2019two} and our paper target different parameters, which will be discussed in detail shortly after. However, \citet{cattaneo2019two} and our paper share a common feature, which is that both consider to use linear models (of a possibly nonlinear $k$-dimensional feature map/basis $\bar{\phi}_{k}$ of $(X, T)$) to approximate certain quantities related to nuisance parameters, such that the target parameters of \citet{cattaneo2019two} and ours both reduce to bilinear functionals, with their precise forms to be revealed thereafter. In particular, \citet{cattaneo2019two} use linear models to directly approximate the nuisance parameters, whereas we use a linear approximation to represent the nuisance estimation \emph{errors}; when the nuisance estimates are zeros, these two approaches collapse. For such bilinear functionals, a na\"{i}ve plug-in estimator can have non-negligible leave-in bias in high-dimensional settings when the dimension $k \gtrsim \sqrt{n}$ \citep{cattaneo2019two}. In terms of bias correction strategies, \citet{cattaneo2019two} start from the plug-in estimator and apply a jackknife-based correction, whereas we exploit the bilinear structure directly and construct an unbiased $U$-statistic estimator. We show below that, for bilinear functionals, these two strategies yield the same estimator.

To facilitate the analysis below, let $\sfzbar_k(X)\in\bbR^k$ be a dictionary for $X$. Since $T$ is binary,  define $\bar{\phi}_k (X, T) = T\sfzbar_k(X)$ and the $k\times k$ population Gram matrix $\Sigma_k \coloneqq \E \big[\bar{\phi}_k (X, T)^{\otimes 2}\big] = \E \big[T\sfzbar_k(X)^{\otimes 2}\big]$. 

Strictly following \citet{cattaneo2019two}, they are interested in solving for $\beta^*$ of the estimating equation $\psi (\beta^*) = 0$ in \eqref{eq:psi ate}, and postulate linear approximations for the nuisance components $\xi$ and $b$. Specifically, we assume that
\begin{align*}
    \xi (X_i,T_i) = \bar{\phi}_k(X_i,T_i)^\top \gamma_{\xi} + e_{1i},
\qquad 
\E[\bar{\phi}_k(X_i,T_i) e_{1i}]=0,\\
b(X_i,T_i) =  \bar{\phi}_k(X_i,T_i)^\top \gamma_{b} + e_{2i},
\qquad 
\E[\bar{\phi}_k(X_i,T_i) e_{2i}]=0,
\end{align*}
where $e_{1i}$ and $e_{2i}$  are the best linear approximation error, and $(\gamma_\xi,\gamma_b)$ are the population least-squares coefficients, i.e., 
\begin{align*}
\gamma_{\xi} & = \Sigma_k^{-1} \bbE [\xi(X,T)\sfzbar_{k} (X)] = \Sigma_k^{-1} \bbE [\sfzbar_{k} (X)]\\
\gamma_{b} & = \Sigma_k^{-1} \bbE [b(X,T)\sfzbar_{k} (X)] = \Sigma_k^{-1} \bbE [T Y\sfzbar_{k} (X)].
\end{align*}
Let $\tilde{\xi} (X, T) = \bar{\phi}_k(X,T)^\top \gamma_{\xi}$  and  $\tilde{b} (X, T) = \bar{\phi}_k(X,T)^\top \gamma_{b}$. Plugging $(\tilde\xi,\tilde b)$ into the moment equation \eqref{eq:psi ate} yields the following population quantity
\begin{align*}
   \tilde{\beta} & = \E \left[ \tilde{\xi} (X, T) \{Y - \tilde{b} (X,T = 1)\} + \tilde{b} (X,T=1)  \right]\\
    & = \E \left[ T\sfzbar_{k} (X) Y \right]^\top \gamma_{\xi} - \gamma_{\xi}^\top \E \left[ T\sfzbar_{k} (X) \sfzbar_{k} (X)^\top \right] \gamma_{b}  + \bbE [\sfzbar_{k} (X)]^\top \gamma_{b}\\
     & = \E \left[ T\sfzbar_{k} (X) Y \right]^\top \Sigma_k^{-1} \bbE [\sfzbar_{k} (X)] - \gamma_{\xi}^\top \Sigma_k \gamma_{b}  + \gamma_{\xi}^\top \Sigma_k \gamma_{b}\\
     & = \E \left[ T\sfzbar_{k} (X) Y \right]^\top \Sigma_k^{-1} \bbE [\sfzbar_{k} (X)].
\end{align*}
Hence, under this linear-approximation regime, the target parameter $\beta^{\ast}$ actually reduces to the above \emph{bilinear functional} $\tilde{\beta}$, which has exactly the same structure as $\mathrm{B}_{\psi, k}$ (see e.g. \eqref{EB}), which will also be reintroduced later specifically for ATE $\beta^{\ast}$. 

In contrast, our higher-order construction instead targets the approximation bias of first-order estimating equation $\hat{\psi}^{(1)} (\beta)\coloneqq \frac{1}{n} \sum_{i=1}^{n} \left[ \hat{\xi} (X_i, T_i) \{Y_i - \hat{b}(X_i,1)\} + \hat{b}(X_i,1) - \beta\right]$, and uses linear approximation only to approximate the nuisance residuals $\hat{\xi}-\xi$ and $\hat{b}-b$. The bias of $\hat{\psi}^{(1)} (\beta)$ has a product form involving two first-step estimation errors:
\begin{align}
\label{bias ATE}
\bbE \{\hat{\psi}^{(1)} (\beta)\} - \psi (\beta) = \bbE [\{\hat{\xi} (X,T) - \xi (X,T)\}  \{  b(X,T) - \hat{b} (X,T)\}].
\end{align}
Consistent with the main text, we employ sample splitting: the nuisance estimators  $\hat\xi$ and $\hat b$ are obtained from an independent nuisance sample and all expectations $\E[\cdot]$ heretofore in this section are taken with respect to the main sample conditional on the nuisance sample. Hence $\hat\xi$ and $\hat b$ can be treated as fixed functions.

Our framework imposes no linearity restriction on $\xi$ and $b$. We allow $\hat{\xi}$ and $\hat{b}$ to be obtained using flexible methods (e.g., kernel smoothing, series estimation, or machine learning), although to obtain sharp theoretical results, we sometimes impose certain smoothness assumptions on $\xi - \hat{\xi}$ and $b - \hat{b}$ such that the truncation bias that cannot be captured by the basis $\bar{\phi}_{k}$ is sufficiently small. We project the nuisance approximation residuals $\hat{\xi} - \xi$ and $\hat{b} - b$ onto the sieve space spanned by  $\bar{\phi}_{k}(x,t)$.  
As described in Section~\ref{sec:review}, this projection leads to a decomposition of \eqref{bias ATE} into (i) a truncation bias $\TB_{\psi,k}$ and (ii) a part of the bias \eqref{bias ATE} captured by the $k$-dimensional dictionary $\bar{\phi}_{k}$. The latter can be written as a bilinear functional
\begin{equation*}
    \mathrm{B}_{\psi,k} \equiv \alpha_{k}^{\top} \Sigma_{k}^{-1} \eta_{k},
\end{equation*}
where 
\begin{align*}
& \alpha_{k} \coloneqq   \E [(\hat{\xi} (X,T) - 1) \sfzbar_{k} (X)],\quad \eta_{k} \coloneqq \E [T (Y - \hat{b} (X,T)) \sfzbar_{k} (X)]
\end{align*}
are $k$-dimensional vectors; See Section~\ref{sec:review} and \ref{sec:intuition} for a detailed derivation. 

Next, we use the bilinear functional $\mathrm{B}_{\psi,k} = \alpha_{k}^{\top} \Sigma_{k}^{-1} \eta_{k}$ as a concrete example to show that the two constructions based on jackknife bias correction and $U$-statistics lead to the same estimator. Specifically, we take $\mathrm{B}_{\psi,k}$ as the parameter of interest and $\alpha_{k}$ and $\eta_{k}$ are high-dimensional nuisance parameters.
We also assume that $\Sigma_k \coloneqq \bbE[\bar{\phi}_{k}(X,T)^{\otimes 2} ]$ is known for convenience. Our strategy estimates $\mathrm{B}_{\psi,k}$ by leveraging its equivalent bilinear form:
\begin{align*}
    \mathrm{B}_{\psi,k} \equiv \E [\{\hat{\xi} (X,T) - 1\} \sfzbar_{k} (X)]^{\top} \Sigma_{k}^{-1} \E [T \{Y - \hat{b} (X,T)\} \sfzbar_{k} (X)],
\end{align*}
which can be unbiasedly estimated by the following second-order $U$-statistic:
\begin{align}\label{eq:Bk app}
\hat{\mathrm{B}}_{\psi,k}  \coloneqq \bbU_{n, 2} [\{\hat{\xi} (X_{i},T_i)  - 1\} \sfzbar_{k} (X_{i})^{\top} \Sigma_{k}^{-1} \sfzbar_{k} (X_{j}) T_{j} \{Y_{j} - \hat{b} (X_{j},T_j)\}].
\end{align}

The jackknife-based bias correction procedure in \citet{cattaneo2019two} starts from the plug-in estimator of $\mathrm{B}_{\psi,k}$
\begin{align*}
\hat{\mathrm{B}}_{\psi,k,\text{plug-in}} = \hat{\alpha}_{k}^{\top} \Sigma_{k}^{-1} \hat{\eta}_{k},
\end{align*}
where 
\begin{align*}
    \hat{\alpha}_{k} =  \frac{1}{n} \sum_{i=1}^n \{\hat{\xi} (X_i,T_i) - 1\} \sfzbar_{k} (X_i), \quad\hat{\eta}_{k} \coloneqq   \frac{1}{n} \sum_{i=1}^n T_i \{Y_i - \hat{b} (X_i,T_i)\} \sfzbar_{k} (X_i).
\end{align*}
The estimator $\hat{\mathrm{B}}_{\psi,k,\text{plug-in}}$, which is a $V$-statistic, will have a bias of an order greater than or equal to $n^{-1 / 2}$ when $k \gtrsim \sqrt{n}$. To implement the jackknife-based bias correction procedure in \citet{cattaneo2019two}, for each $\ell\in\{1,\dots,n\}$, we compute the leave-one-out analogs of $\hat{\mathrm{B}}_{\psi,k, \text{plug-in}}$ as
\begin{align*}
\hat{\mathrm{B}}_{\psi,k}^{(-\ell)} =  \hat{\alpha}_{k}^{(-\ell)\top} \Sigma_{k}^{-1} \hat{\eta}_{k}^{(-\ell)},
\end{align*}
where
\begin{align}
    &\hat{\alpha}_{k}^{(-\ell)} \coloneqq  \frac{1}{n-1} \sum_{i=1, i\neq\ell}^n (\hat{\xi} (X_i,T_i) - 1) \sfzbar_{k} (X_i) = \frac{n}{n - 1} \hat{\alpha}_{k} - \frac{1}{n-1}  (T_{\ell}\hat{\xi} (X_{\ell},T_{\ell}) - 1) \sfzbar_{k} (X_{\ell}),\label{eq:l1o aplha}\\
    &\hat{\eta}_{k}^{(-\ell)} \coloneqq \frac{n}{n - 1 } \hat{\eta}_{k} - \frac{1}{n-1}  T_{\ell} ( Y_{\ell} - \hat{b} (X_{\ell},T_{\ell})) \sfzbar_{k} (X_{\ell}).\label{eq:l1o eta}
\end{align}
Let  $ \hat{\mathrm{B}}_{\psi,k}^{(\cdot)} = n^{-1} \sum_{\ell=1}^n  \hat{\mathrm{B}}_{\psi,k}^{(-\ell)} = \frac{1}{n} \sum_{\ell=1}^n \hat{\alpha}_{k}^{(-\ell)\top} \Sigma_{k}^{-1} \hat{\eta}_{k}^{(-\ell)} $ be
the average of the leave-one-out estimators. Then, the jackknife bias-corrected estimator is defined by
\begin{align*}
    \hat{\mathrm{B}}_{\psi,k,\text{jackknife}} =  \hat{\mathrm{B}}_{\psi,k,\text{plug-in}} -  (n-1)(\hat{\mathrm{B}}_{\psi,k}^{(\cdot)} -  \hat{\mathrm{B}}_{\psi,k,\text{plug-in}} ).
\end{align*}
Lemma~\ref{lemma:bhat jackknife} in the end of this subsection shows that the jackknife bias-corrected estimator coincides exactly with our second-order $U$-statistic estimator, i.e.,
\begin{equation*}
 \hat{\mathrm{B}}_{\psi,k,\text{jackknife}}  = \hat{\mathrm{B}}_{\psi,k}.
\end{equation*}

The above derivation shows that, in the special case where the target functional admits the bilinear representation $\mathrm{B}_{\psi,k} = \alpha_{k}^{\top} \Sigma_{k}^{-1} \eta_{k}$, the two constructions $\hat{\mathrm{B}}_{\psi,k}$ and $\hat{\mathrm{B}}_{\psi,k, \text{jackknife}}$ yield the same $U$-statistic. The key ingredient of our approach is to first construct a representation like $\mathrm{B}_{\psi,k}$ and estimate $\mathrm{B}_{\psi,k}$ by a natural $U$-statistic. By contrast, the jackknife-based bias correction method of \citet{cattaneo2019two} can be viewed as a general bias correction scheme for parameters such as $\mathrm{B}_{\psi,k}$, for which standard plug-in estimators (such as the standard $M$-estimation) have non-negligible bias when the dimension is large relative to $\sqrt{n}$. It is worth mentioning that the approach in \citet{cattaneo2019two} can deal with parameters of a more generic nonlinear form than the bilinear form $\mathrm{B}_{\psi,k}$, but the theoretical results are restricted to $k = O (\sqrt{n})$ due to the nonlinearity of the more general problems considered. For $\mathrm{B}_{\psi,k}$, we have seen that the $U$-statistic estimator is unbiased for any $k$ when $\Sigma_{k}$ is known. Finally, if $\Sigma_k$ is unknown and is replaced by $\hat{\Sigma}_k$, our Supplementary Material Section~\ref{app:truncation} derives HOIFs of $\mathrm{B}_{\psi,k}$ to use higher-order $U$-statistics for further bias correction. To our knowledge, systematically extending jackknife-based corrections to control such higher-order remainder terms seems to be beyond the scope of \citet{cattaneo2019two}, but we suspect that constructions similar to those in HOIFs can be considered. We also refer interested readers to \citet{lin2024worthwhile} for a recent result in this direction.

Finally, we record Lemma~\ref{lemma:bhat jackknife} and its proof below.
\begin{lemma}\label{lemma:bhat jackknife}
The following algebraic identity holds:
    \begin{align*}
    \hat{\mathrm{B}}_{\psi,k,\normalfont\textrm{jackknife}}  = \hat{\mathrm{B}}_{\psi,k}.
    \end{align*}
\end{lemma}
\begin{proof}
     Substituting the explicit leave-one-out formulas for $\hat\alpha_k^{(-\ell)}$ and $\hat\eta_k^{(-\ell)}$ in \eqref{eq:l1o aplha}--\eqref{eq:l1o eta} into the definition of $\hat{\mathrm{B}}_{\psi,k}^{(\cdot)}$, we have
\begin{align*}
   &\ \hat{\mathrm{B}}_{\psi,k}^{(\cdot)}  
   =  \ \frac{1}{n} \sum_{\ell=1}^n \hat{\alpha}_{k}^{(-\ell)\top} \Sigma_{k}^{-1} \hat{\eta}_{k}^{(-\ell)}\\
    = & \ \frac{1}{n} \sum_{\ell=1}^n \left\{ \frac{n}{n - 1} \hat{\alpha}_{k} - \frac{1}{n-1}  (\hat{\xi} (X_{\ell}, T_{\ell}) - 1) \sfzbar_{k} (X_{\ell}) \right\}^{\top}  \Sigma_{k}^{-1} \left\{ \frac{n}{n - 1} \hat{\eta}_{k} - \frac{1}{n-1}  T_{\ell}( Y_{\ell} - \hat{b} (X_{\ell},T_{\ell})) \sfzbar_{k} (X_{\ell}) \right\} \\
    = & \ \frac{n^2}{(n-1)^2} \hat{\alpha}_{k}^{\top} \Sigma_{k}^{-1} \hat{\eta}_{k} - \frac{2n}{(n-1)^2} \hat{\alpha}_{k}^{\top} \Sigma_{k}^{-1} \hat{\eta}_{k} \\
    & + \frac{1}{n(n-1)^2} \sum_{\ell=1}^n  (\hat{\xi} (X_{\ell},T_{\ell}) - 1) \sfzbar_{k} (X_{\ell})^{\top} \Sigma_{k}^{-1} \sfzbar_{k} (X_{\ell}) T_{\ell} ( Y_{\ell} - \hat{b} (X_{\ell},T_{\ell}))\\
    = & \ \left\{ 1 - \frac{1}{(n-1)^2}\right\} \hat{\alpha}_{k}^{\top} \Sigma_{k}^{-1} \hat{\eta}_{k}  + \frac{1}{n(n-1)^2} \sum_{\ell=1}^n  (\hat{\xi} (X_{\ell},T_{\ell}) - 1) \sfzbar_{k} (X_{\ell})^{\top} \Sigma_{k}^{-1} \sfzbar_{k} (X_{\ell}) T_{\ell} (Y_{\ell} - \hat{b} (X_{\ell},T_{\ell})).
\end{align*}
Then, we obtain 
\begin{align*}
    &(n-1)(\hat{\mathrm{B}}_{\psi,k}^{(\cdot)} -  \hat{\mathrm{B}}_{\psi,k,\text{plug-in}} ) \\
    =&   - \frac{1}{n-1} \hat{\alpha}_{k}^{\top} \Sigma_{k}^{-1} \hat{\eta}_{k}  + \frac{1}{n(n-1)} \sum_{\ell=1}^n  (\hat{\xi} (X_{\ell},T_{\ell}) - 1) \sfzbar_{k} (X_{\ell})^{\top} \Sigma_{k}^{-1} \sfzbar_{k} (X_{\ell}) T_{\ell} (Y_{\ell} - \hat{b} (X_{\ell},T_{\ell})).
\end{align*}
Hence, using the above result, we derive
\begin{align*}
&\ \hat{\mathrm{B}}_{\psi,k,\text{jackknife}} 
=  \ \hat{\mathrm{B}}_{\psi,k,\text{plug-in}} -  (n-1)(\hat{\mathrm{B}}_{\psi,k}^{(\cdot)} -  \hat{\mathrm{B}}_{\psi,k,\text{plug-in}} )\\
= & \ \left( 1 + \frac{1}{n-1} \right) \hat{\alpha}_{k}^{\top} \Sigma_{k}^{-1} \hat{\eta}_{k}   - \frac{1}{n(n-1)} \sum_{\ell=1}^n  (\hat{\xi} (X_{\ell},T_{\ell}) - 1) \sfzbar_{k} (X_{\ell})^{\top} \Sigma_{k}^{-1} \sfzbar_{k} (X_{\ell}) T_{\ell} ( Y_{\ell} - \hat{b} (X_{\ell},T_{\ell})) \\
= & \ \frac{n}{(n-1)n^2} \sum_{i = 1}^n \sum_{j = 1}^n (\hat{\xi} (X_i,T_{i}) - 1) \sfzbar_{k} (X_i)^{\top} \Sigma_{k}^{-1} \sfzbar_{k} (X_j) T_j( Y_j - \hat{b} (X_j,T_{j})) \\
& - \frac{1}{n(n-1)} \sum_{\ell=1}^n  (\hat{\xi} (X_{\ell},T_{\ell}) - 1) \sfzbar_{k} (X_{\ell})^{\top} \Sigma_{k}^{-1} \sfzbar_{k} (X_{\ell}) T_{\ell}( Y_{\ell} - \hat{b} (X_{\ell},T_{\ell}))\\
= & \ \frac{1}{(n-1)n} \sum_{1 \leq i \neq j \leq n} (\hat{\xi} (X_i,T_{i}) - 1) \sfzbar_{k} (X_i)^{\top} \Sigma_{k}^{-1} \sfzbar_{k} (X_j) T_j( Y_j - \hat{b} (X_j,T_{j}))\\
= & \ \hat{\mathrm{B}}_{\psi,k},
\end{align*}
which completes the proof.
\end{proof}

\subsection{Discussion of \texorpdfstring{\citet{cattaneo2018kernel}}{}}\label{app:cattaneo2018}

We next discuss how the bootstrap-based bias correction approach developed in \citet{cattaneo2018kernel} is also connected to our higher-order estimators.

\citet{cattaneo2018kernel} study two-step semiparametric estimation and inference with kernel-based first-step estimators, allowing the first step to converge slowly when the bandwidth is small. To facilitate comparison with our framework, we again use the ATE example in Section~\ref{app:cattaneo2018}, where the target moment $\psi(\beta)$ is linear in the nuisance functions $\xi$ and $b$. In this linear case, the bias analyzed in \citet{cattaneo2018kernel}  is mainly driven by the leave-in bias that results from the use of the same sample in both steps, when the bandwidth is small. They show that, certain nonparametric bootstrap procedures can also recover the leading bias term and hence deliver valid  inference.

Our construction allows the nuisance estimators $\hat{\xi}$ and $\hat{b}$ to be obtained by generic learners, such as random forests, or neural networks. By using sample splitting, we focus on the approximation bias. We first approximate this bias by a bilinear functional $\mathrm{B}_{\psi,k}$ induced by a dictionary-based projection, and then estimate this functional directly via a $U$-statistic estimator.

Despite this difference, the two approaches are closely connected. As discussed in Section~\ref{app:cattaneo2018}, once the target is reduced to the bilinear functional $\mathrm{B}_{\psi,k}$, we will show that the resulting bootstrap bias-corrected estimator is asymptotically equivalent to our $U$-statistics-based estimator $\hat{\mathrm{B}}_{\psi,k}$.

Let $\{O_i=(Y_i,T_i,X_i)\}_{i=1}^n$ be the original sample and  let  $O_1^*,\ldots,O_n^*$ be a bootstrap sample drawn with replacement from $\{O_i\}_{i=1}^n$.  Define
$$
\hat\alpha_k^*=\frac1n\sum_{j=1}^n (\hat\xi(X_j^*,T_j^*)-1)\sfzbar_{k} (X_j^*),
\qquad
\hat\eta_k^*=\frac1n\sum_{j=1}^n T_j^*(Y_j^*-\hat b(X_j^*,T_j^*))\sfzbar_{k}(X_j^*).
$$
The bootstrap plug-in estimator of $\mathrm{B}_{\psi,k}=\alpha_k^\top\Sigma_k^{-1}\eta_k$ is
$$
\hat{\mathrm{B}}_{\psi,k}^*=(\hat\alpha_k^*)^\top\Sigma_k^{-1}\hat\eta_k^*.
$$

Let  $M_{ni}$ be the number of times that $O_i$ is ``redrawn'' from $\{O_1, \ldots, O_n \}$ to form $O_1^*, \ldots, O_n^*, $. By construction, the vector of counts $(M_{n1}, \ldots , M_{nn})$ is independent of
$O_1, \ldots, O_n$ and multinomially distributed with parameters $n$ and (probabilities) $1/n, \ldots,
1/n$. This representation allows us to rewrite any bootstrap empirical average as a weighted average over the original sample, with random weights $\{M_{ni}\}_{i=1}^n$. Then
\begin{align*}
    &\hat\alpha_k^*=\frac1n\sum_{j=1}^n (\hat\xi(X_j^*,T_j^*)-1)\sfzbar_{k} (X_j^*) = \frac1n\sum_{i=1}^n M_{ni}(T_i\hat\xi(X_i,T_i)-1) \sfzbar_{k} (X_i)
\end{align*}
and similarly
\begin{equation*}
\hat\eta_k^*=\frac1n\sum_{i=1}^n M_{ni} T_i(Y_i-\hat b(X_i,T_i)) \sfzbar_{k} (X_i).
\end{equation*}

The difference $\hat{\mathrm{B}}_{\psi,k}^* - \hat{\mathrm{B}}_{\psi,k,\text{plug-in}}$  provides a natural estimator of the bias of the  plug-in estimator $\hat{\mathrm{B}}_{\psi,k,\text{plug-in}}$ for $\mathrm{B}_{\psi,k}$. For simplicy, we define the oracle bootstrap bias corrected estimator as
\begin{equation}
    \hat{\mathrm{B}}_{\psi,k,\text{bootstrap}}  = \hat{\mathrm{B}}_{\psi,k,\text{plug-in}} -  \bbE^* [\hat{\mathrm{B}}_{\psi,k}^* - \hat{\mathrm{B}}_{\psi,k,\text{plug-in}}], 
\end{equation}
where $\bbE^*[]$ denotes expectation conditional on the original sample. 
We refer to this estimator as oracle because it involves the conditional expectation $\bbE^*[]$. In practice, one can approximate this conditional expectation by the Monte Carlo bootstrap average over $B$ resamples. 

Lemma~\ref{lemma:bhat boot} below shows that 
\begin{equation*}
    \hat{\mathrm{B}}_{\psi,k,\text{bootstrap}} = \hat{\mathrm{B}}_{\psi,k} + O_{\bbP} \left( \frac{k}{n^2} \right).
    \end{equation*}
Under $k= o(n^2)$, it follows that
\begin{align*}
  \frac{k}{n^2} \cdot \frac{n}{\sqrt{k+n}} =   \frac{k}{n\sqrt{k+n}} = o(1). 
\end{align*}
Hence, the remainder term is asymptotically negligible under the normalization used in our distributional results. Therefore, $\hat{\mathrm{B}}_{\psi,k,\text{bootstrap}}$ and $\hat{\mathrm{B}}_{\psi,k}$ are asymptotically equivalent up to an error term of order $k / n^{2}$, which is $o (n^{-1 / 2})$ (negligible compared to the parametric rate) if $k = o (n^{3 / 2})$.

\begin{lemma} \label{lemma:bhat boot}
The following hold under Assumptions~\ref{as:nuisance_est}--\ref{as:basis}:
    \begin{equation*}
    \hat{\mathrm{B}}_{\psi,k, \normalfont\text{bootstrap}} = \hat{\mathrm{B}}_{\psi,k} + O_{\bbP} \left( \frac{k}{n^2} \right).
    \end{equation*}
\end{lemma}
\begin{proof}
   Using the definitions of  $\hat{\mathrm{B}}_{\psi,k}^* $ and $ \hat{\mathrm{B}}_{\psi,k,\text{plug-in}}$, we have
\begin{align*}
    \hat{\mathrm{B}}_{\psi,k}^* - \hat{\mathrm{B}}_{\psi,k,\text{plug-in}} = & \ (\hat\alpha_k^*)^\top\Sigma_k^{-1}\hat\eta_k^* - \hat\alpha_k^\top\Sigma_k^{-1}\hat\eta_k\\
    = & \ \frac{1}{n^2} \sum_{i=1}^n \sum_{j=1}^n  \{ M_{ni} M_{nj} - 1\}  (\hat\xi(X_i,T_i)-1) \sfzbar_{k} (X_i)^{\top} \Sigma_{k}^{-1}  \sfzbar_{k} (X_j)  T_j(Y_j-\hat b(X_j,T_j)) \\
    = & \ \frac{1}{n^2} \sum_{i=1, i \neq j}^n \sum_{j=1}^n  \{ M_{ni} M_{nj} - 1\}  (\hat\xi(X_i,T_i)-1) \sfzbar_{k} (X_i)^{\top} \Sigma_{k}^{-1} \sfzbar_{k} k(X_j) T_j(Y_j-\hat b(X_j,T_j)) \\
    & + \frac{1}{n^2} \sum_{i=1}^n   \{ M_{ni}^2 - 1\}  (\hat\xi(X_i,T_i)-1) \sfzbar_{k} (X_i)^{\top} \Sigma_{k}^{-1} \sfzbar_{k} (X_i) T_i(Y_i-\hat b(X_i,T_i)). 
\end{align*}
Since 
\begin{align*}
    &\bbE^*[M_{ni}M_{nj}]
=\frac{n(n-1)}{n^2} = 1  - \frac{1}{n},\\
&\bbE^*[M_{ni}^2]
=\frac{n(n-1)}{n^2}+1 = \frac{n-1}{n}+1 = 2 - \frac{1}{n},
\end{align*}
taking $\bbE^*[\cdot]$ conditional on the original sample yields
\begin{align*}
    \bbE^* [\hat{\mathrm{B}}_{\psi,k}^* - \hat{\mathrm{B}}_{\psi,k,\text{plug-in}}] =& - \frac{1}{n^3} \sum_{i=1, i \neq j}^n \sum_{j=1}^n    (\hat\xi(X_i,T_i)-1) \sfzbar_{k} (X_i)^{\top} \Sigma_{k}^{-1} \sfzbar_{k} (X_j) T_j (Y_j-\hat b(X_j,T_j)) \\
    & + \frac{n-1}{n^3} \sum_{i=1}^n    (\hat\xi(X_i,T_i)-1) \sfzbar_{k} (X_i)^{\top} \Sigma_{k}^{-1} \sfzbar_{k} (X_i) T_i (Y_i-\hat b(X_i,T_i)) \\
    =& - \frac{1}{n^3} \sum_{i=1}^n \sum_{j=1}^n    (\hat\xi(X_i,T_i)-1) \sfzbar_{k} (X_i)^{\top} \Sigma_{k}^{-1} \sfzbar_{k} (X_j) T_j (Y_j-\hat b(X_j,T_j)) \\
    &+  \frac{1}{n^2} \sum_{i=1}^n    (\hat\xi(X_i,T_i)-1) \sfzbar_{k} (X_i)^{\top} \Sigma_{k}^{-1} \sfzbar_{k} (X_i) T_i(Y_i-\hat b(X_i,T_i)) \\
    =&  - \frac{1}{n} \hat{\alpha}_{k}^{\top} \Sigma_{k}^{-1} \hat{\eta}_{k} +  \frac{1}{n^2} \sum_{i=1}^n    (\hat\xi(X_i,T_i)-1) \sfzbar_{k} (X_i)^{\top} \Sigma_{k}^{-1} \sfzbar_{k} (X_i) T_i (Y_i-\hat b(X_i,T_i)).
\end{align*}
Hence,
\begin{align*}
\hat{\mathrm{B}}_{\psi,k,\text{bootstrap}}  = & \ \hat{\mathrm{B}}_{\psi,k,\text{plug-in}} -  \bbE^* [\hat{\mathrm{B}}_{\psi,k}^* - \hat{\mathrm{B}}_{\psi,k,\text{plug-in}}] \\
= & \ \frac{n + 1}{n} \hat{\alpha}_{k}^{\top} \Sigma_{k}^{-1} \hat{\eta}_{k} -  \frac{1}{n^2} \sum_{i=1}^n    (\hat\xi(X_i,T_i)-1) \sfzbar_{k} (X_i)^{\top} \Sigma_{k}^{-1}  \sfzbar_{k} (X_i) T_i(Y_i-\hat b(X_i,T_i)) \\
= & \ \frac{n + 1}{n^3} \sum_{i=1, i \neq j}^n \sum_{j=1}^n    (\hat\xi(X_i,T_i)-1) \sfzbar_{k} (X_i)^{\top} \Sigma_{k}^{-1}  \sfzbar_{k} (X_j) T_j(Y_j-\hat b(X_j,T_j)) \\
 & + \left(\frac{n + 1}{n^3} - \frac{1}{n^2}  \right) \sum_{i=1}^n    (\hat\xi(X_i,T_i)-1) \sfzbar_{k} (X_i)^{\top} \Sigma_{k}^{-1}  \sfzbar_{k} (X_i) T_i(Y_i-\hat b(X_i,T_i)) \\
 = & \ \hat{\mathrm{B}}_{\psi,k} - \frac{1}{n^3(n-1)} \sum_{i=1, i \neq j}^n \sum_{j=1}^n (\hat\xi(X_i,T_i)-1) \sfzbar_{k} (X_i)^{\top} \Sigma_{k}^{-1} \sfzbar_{k} (X_j) T_j (Y_j-\hat b(X_j,T_j)) \\
 & +\frac{1}{n^3}   \sum_{i=1}^n    (\hat\xi(X_i,T_i)-1) \sfzbar_{k} (X_i)^{\top} \Sigma_{k}^{-1} \sfzbar_{k} (X_i) T_i(Y_i-\hat b(X_i,T_i)) \\
 =& \ \hat{\mathrm{B}}_{\psi,k} + O_{\bbP} \left( \frac{k}{n^2} \right),
\end{align*}
where the last equality holds by   Lemma~\ref{lemma:V remainder}.  
\end{proof}

\begin{lemma}\label{lemma:V remainder}
The following hold under Assumptions~\ref{as:nuisance_est}--\ref{as:basis}:
    \begin{align*}
         & \frac{1}{(n-1)n^2} \sum_{i = 1}^n \sum_{i \neq j, j = 1}^n (\hat{\xi} (X_i,T_i) - 1) \sfzbar_{k} (X_i)^{\top} \Sigma_{k}^{-1} \sfzbar_{k} (X_j) T_j( Y_j - \hat{b} (X_j,T_j)) = O_{\P} \left(1\right), \\
         & \frac{1}{n^2} \sum_{i = 1}^n (\hat{\xi} (X_i,T_i) - 1) \sfzbar_{k} (X_i)^{\top} \Sigma_{k}^{-1} \sfzbar_{k} (X_i) T_i( Y_i - \hat{b} (X_i,T_i)) =  O_{\bbP} \left( \frac{k}{n}  \right).
    \end{align*}
\end{lemma}

\begin{proof}
Denote
    \begin{align*}
      I \coloneqq &\  \frac{1}{n^2} \sum_{i = 1}^n \sum_{i \neq j, j = 1}^n (\hat{\xi} (X_i,T_i) - 1) \sfzbar_{k} (X_i)^{\top} \Sigma_{k}^{-1} \sfzbar_{k} (X_j) T_j( Y_j - \hat{b} (X_j,T_j)),\\
      II \coloneqq  & \ \frac{1}{n^2} \sum_{i = 1}^n (\hat{\xi} (X_i,T_i) - 1) \sfzbar_{k} (X_i)^{\top} \Sigma_{k}^{-1} \sfzbar_{k} (X_i) T_i( Y_i - \hat{b} (X_i,T_i)).
    \end{align*}
For the term $I$, the same variance calculation as in Lemma~\ref{eq:var-F3} yields $\var(I) = O (1/n + k/n^2)$. For the mean, 
\begin{align*}
  \bbE[I] = \bbE [(\hat{\xi} (X,T) - 1) \sfzbar_{k} (X)]^{\top} \Sigma_{k}^{-1} \bbE[\sfzbar_{k} (X) T( Y - \hat{b} (X,T))]  = \alpha_k^{\top} \Sigma_{k} \eta_k  = O\left(1\right),
\end{align*}
    where the last equality holds by the definitions of $\alpha_k$ and $\eta_k$, and the condition that the largest eigenvalue of $\Sigma_{k}$ is bounded. Combining mean and variance bounds,
    \begin{align*}
        I=O_{\P} \left(\frac{1}{\sqrt{n}}+\frac{\sqrt{k}}{n}+1\right)
=O_{\P} \left(1\right).
    \end{align*}

For the term $II$, by $\sfzbar_{k}(x)^\top\Sigma_k^{-1}\bar\sfzbar_{k}(x)\le \|\Sigma_k^{-1}\|_{\mathrm{op}}\| \sfzbar_{k} (x)\|^2
\lesssim \zeta_k^2 $,
\begin{align*}
    &\bbE \left[ \left| \frac{1}{n^2} \sum_{i = 1}^n (\hat{\xi} (X_i,T_i) - 1) \sfzbar_{k} (X_i)^{\top} \Sigma_{k}^{-1} \sfzbar_{k} (X_i) T_i( Y_i - \hat{b} (X_i,T_i)) \right| \right]\\
    \lesssim & \  \frac{\zeta_k^2}{n} \bbE \left[ \left| (\hat{\xi} (X_i,T_i) - 1) T_i( Y_i - \hat{b} (X_i,T_i)) \right| \right] = O \left( \frac{\zeta_k^2}{n} \right) = O \left( \frac{k}{n}  \right). 
\end{align*}
Thus, by Markov's inequality,  we have 
\begin{align*}
  II =   \frac{1}{n^2} \sum_{i = 1}^n (\hat{\xi} (X_i,T_i) - 1) \sfzbar_{k} (X_i)^{\top} \Sigma_{k}^{-1} \sfzbar_{k} (X_i) T_i( Y_i - \hat{b} (X_i,T_i))  =  O_{\bbP} \left( \frac{k}{n}  \right).
\end{align*}

Combining the results for terms $I$ and $II$, we complete the proof.
\end{proof}

\section{Technical Results Related to \texorpdfstring{$U$}{}-Processes}
\label{app:U-processes}
\setcounter{equation}{0}
In this section, we establish refined local maximal inequalities for $U$-processes that can be used to establish convergence rates and asymptotic normality of higher-order estimators considered in this paper. These are key technical tools for the main results in Sections~\ref{sec:properties} and~\ref{sec:kernels}, and may also be of independent interest for other related applications. We begin by introducing the necessary notation.

\subsection{Notation and Review of Concepts from \texorpdfstring{$U$}{}-Process Theory}\label{app:notation}
  We first recall the key notation introduced in Section~\ref{sec:notation}, and then introduce additional terminologies for $U$-statistics and $U$-processes that will be used in the technical lemmas that follow.

Recall that $\bbU_{n, m}$ is the $m$-th order $U$-statistic operator, i.e., 
\begin{align*}
\bbU_{n, m} h \equiv \bbU_{n, m} h (O_{1}, \cdots, O_{m}) = \frac{(n - m)!}{n!} \sum_{1 \leq i_{1} \neq \cdots \neq i_{m} \leq n} h (O_{i_{1}}, \cdots, O_{i_{m}})
\end{align*}
for any measurable function (kernel) $h: \calO^{m} \rightarrow \bbR$.  So, $\bbU_{n, 1} \equiv \bbP_{n}$ is understood to be the empirical mean operator.  Let $S_m$ denote the symmetrization of a function of $m$ variables, 
\begin{align*}
    S_m h (O_{1}, \cdots, O_{m}) = \frac{1}{m!} \sum_{1 \leq i_{1} \neq \cdots \neq i_{m} \leq m} h (O_{i_{1}}, \cdots, O_{i_{m}}).
\end{align*}
Because a $U$-statistic averages over unordered subsets, $\bbU_{n,m} h = \bbU_{n,m} \bigl(S_{m}h\bigr)$. 
Accordingly, all subsequent notation and results are stated for
symmetric kernels; when a kernel is asymmetric, they apply by default to
its symmetrized version.

For any symmetric measurable function $h:\calO^{m} \to \bbR$ and $j = 1,\ldots,m$, let $\P^{m-j}h$ denote the function on $\calO^{j}$ defined by
\begin{align}\label{recall notation}
\begin{split}
    \P^{m-j}h(o_1,...,o_j) = & \ \bbE[h(o_1,\ldots,o_j, O_{j+1},\ldots, O_{m})]\\
= & \int\cdots \int h(o_1,\ldots,o_j,o_{j+1},\ldots,o_{m})\diff \P( o_{j+1})\cdots \diff \P(o_m),
\end{split}
\end{align}
where we always take $O_{1}, \cdots, O_{m}$ to be i.i.d. drawn from a common law $\bbP$ in this paper. Thus, $\P^{0} h \equiv h$ and $\P^{m} h \equiv \bbE [h(O_1,\ldots,O_m)]$. We also introduce the following shorthand notation for Hoeffding projection of a symmetric measurable function $h$, for $j \leq m$:
\begin{equation}
\label{hajek}
\pi_{j}  h (o_1 ,\ldots, o_j) \coloneqq (\delta_{o_1} -\P)\cdots(\delta_{o_j} -\P)\P^{m-j} h(o_1 ,\ldots,o_m ).
\end{equation}
For a non-symmetric kernel $h$ we work with its symmetrised form
$S_m h$; all the above definitions are understood with that replacement, i.e., $\P^{m-j} h \equiv \P^{m-j} S_m h$ and $\pi_{j} h \equiv \pi_{j} S_m h$.  Then the Hoeffding decomposition of $\bbU_{n, m} h$ of any measurable functions $h$ is given by
\begin{align*}
    \bbU_{n, m} h - \P^m h = \sum_{j=1}^{m} \binom{m}{j} \bbU_{n, j} (\pi_{j} h).
\end{align*}

 We next review some basic terminologies about $U$-processes used throughout. 
Let $\calF$ be a class of symmetric measurable functions $f : \calO^{m} \rightarrow  \bbR$. We assume that there is a symmetric measurable envelope $\F$ for $\calF$, that is $\sup_{f\in\calF} |f| \leq \F $, 
such that $\P^{m} \F^2 < \infty$.  We will assume certain uniform covering number conditions for the function class $\calF$. Let $\|\cdot\|_{\calF}$ denote the supremum norm over $\calF$, i.e.  $\|f\|_{\calF}\coloneqq \sup_{f\in\calF}|f|$.
For a probability measure $Q$ and a constant $p > 1$, let $\| \cdot \|_{Q,p}$ denote the $L_{p} (Q)$-seminorm, i.e., $\|f\|_{Q, p} \coloneqq \{\int |f (x)|^{p} \diff Q (x)\}^{1/p}$ for finite q while $\|f\|_{Q,\infty}$ denotes the essential supremum of $|f|$ with respect to $Q$. For any function class $\calF$ with $\| \F \|_{Q,p} > 0$, we use $\sfN ( \calF,\|\cdot\|_{Q,p} , \epsilon\|\F\|_{Q,p})$ to denote the minimal number of $\|\cdot\|_{Q,p}$-balls of radius $\epsilon\|\F\|_{Q,p}$ needed to cover $\calF$. See Section 2.1 of \cite{van2023weak} for details. For $j = 1,\ldots,m$, we define the uniform entropy integral
\begin{align*}
   J_{j}(\delta,\calF,\F ) \coloneqq   \int_{0}^{\delta} \sup_{Q}  \left\{ 1 + \log \sfN (\P^{m-j}\calF, \|\cdot\|_{Q,2} , \tau \| \P^{m-j} \F\|_{Q,2}) \right\}^{j/2} \diff  \tau,
\end{align*}
where $\P^{m-j}\calF = \{\P^{m-j}f : f \in \calF\}$ and $\sup_Q$ is taken over all finitely discrete distributions on $\calO^j$. It is straightforward to verify that $\P^{m-j} \F$ is an envelope for $\P^{m-j} \calF$. To avoid measurability difficulties, we will assume that $\F$ is pointwise measurable. If $\F$ is pointwise measurable and $\P^m \F < \infty$ (which we have assumed), then $\pi_{j}\calF \coloneqq\{\pi_{j} f : f \in\calF\}$ and $\P^{m-j}  \F$ for $j = 1,\ldots, m$ are all pointwise measurable by the dominated convergence theorem.  

For completeness, we recall Definitions~\ref{def:VC-type} and~\ref{def:VC-subgraph}.
\setcounter{definition}{0}
\begin{definition}[VC-subgraph class]
A collection $\calF$ of measurable functions is called a VC-subgraph class, or a VC-class, if the collection of all subgraphs of the functions (see Section 2.6.2 of \citet{van2023weak}) in $\calF$ forms a VC-class of sets (see Section 2.6.1 of \citet{van2023weak}).
\end{definition}

\begin{definition}[VC type class]
A function class $\calF$ with envelope $\F$ is said to be of VC type with characteristic $(A, v)$ if $\sup_{Q} \sfN \left( \calF, \| \cdot \|_{Q, 2}, \epsilon \| \F \|_{Q, 2} \right) \leq (A / \epsilon)^{v}$ for all $\epsilon \in (0, 1]$, where $\sup_{Q}$ denotes the supremum over all finitely discrete probability measures.
\end{definition}

\subsection{Local Maximal Inequalities for \texorpdfstring{$U$}{}-Processes} \label{app:Maximal Inequalities}

In this section we allow the function class to depend explicitly on the sample size $n$. Throughout, we denote a generic function class by $\calF_n$
to emphasize this dependence and derive non-asymptotic bounds that hold for each fixed $n$.

\cite{chen2020jackknife} establish local maximal inequalities for $U$-processes indexed by a VC-type class $\calF_n \coloneqq \{f: \calO^2 \rightarrow \bbR\}$, whose envelope is denoted by $\F_n$. Their bounds are expressed in terms of the VC characteristics $(A_n,v_n)$ and quantities including $\sup_{f\in\calF_n}\|\P^{2-k} f\|_{\P^{k},2}$ and $\|\P^{2-k} \F_n\|_{\P^{k},2}$ for $k=1,2$ (see their Corollary 5.5). 
Because their bound uses the projected envelope $\P \F_n (\cdot)$ (i.e. the conditional mean of $\F_{n}$ in one of the arguments), it can be unnecessarily loose when $\P \F_n$ is of greater magnitude in order than the quantity of actual interest, $\sup_{f\in\calF_n}|\P f|$. To sharpen the analysis, we work directly with $\sup_{f\in\calF_n} |\P f|$ and derive a refined maximal inequality tailored specifically for our problem. The next lemma records this result; as indicated, it is particularly useful when $\sup_{f\in\calF_n}|\P f|$ is much less than $\P \F_n$.

\begin{lemma}\label{lemma:Umaximal_inequality_order2} 
Suppose that (1) $\calF_n \coloneqq \{f: \calO^2 \rightarrow \bbR\}$ is a collection of symmetric pointwise measurable functions from $\calO^{2}$ to $\bbR$ with envelope function $\F_n$, and (2) $\calF_n$ is of VC type with characteristics $(A_n,v_n)$ with $A_n \geq e$ and $v_n \geq 1$ (recall from Definition~\ref{def:VC-type}). Then 
\begin{enumerate}[label=(\roman*)]
    \item $ \displaystyle
        \bbE \left[ \|\bbU_{n,2}(f) - \P^{2} f \|_{\calF_n}  \right] \lesssim \frac{ v_n\log A_n \cdot \|\F_n\|_{\P^2,2} }{n} + \frac{(v_n\log A_n)^{1/2}\cdot\|\P\F_{n}\|_{\P^2,2}}{\sqrt{n}}.$

        \item If we further assume (3) $\sup_{f\in \calF_n} \|\P f (\cdot) \|_{\infty} \leq m_n$, and $\|\P \F_n (\cdot) \|_{\infty} \leq M_n$, then
    \begin{align*}
        \bbE \left[ \|\bbU_{n,2}(f) -  \P^{2} f \|_{\calF_n}  \right] \lesssim & \ \frac{ \left( \|\F_n\|_{\P^2,2} +  m_n \right) \cdot  v_n\log( A_{n}^{\prime} \vee n )}{n} \\
        & + \frac{ \sup_{f\in \calF_n} \|\P f \|_{\P,2} \cdot\sqrt{v_n\log (A_n^{\prime} \vee n) }}{\sqrt{n}}\, ,
    \end{align*}
    where $A_{n}^{\prime} \coloneqq \sqrt{A_n}M_n/m_n$.
\end{enumerate}
\end{lemma}

\begin{remark}
For the class $\calF_n = \calH_{2,n}$ defined in \eqref{eq:def-H2} in Section~\ref{app:beta-consistency}, $\calF_n$ is of VC type with fixed characteristics $(4A,v)$ and envelope $\F_n = \H_{2,n}$.  In this case, the following hold: $\|\P \F_n \|_{\P,2} \le \|\F_n\|_{\P^2,2} \lesssim \sqrt{k}$, $\|\F_n (\cdot) \|_{\infty}\lesssim k$, $\sup_{f\in\calF_n} \|f\|_{\P^2,2} \lesssim \sqrt{k}$ with $k\lesssim n^2$. Moreover, because each $\P f (\cdot)$ is a conditional expectation (recall from the notation \eqref{recall notation}), \eqref{eq:supnorm_proj}  implies $\sup_{f\in\calF_n} \|\P f (\cdot) \|_{\infty} \lesssim 1$ and $\sup_{f\in\calF_n} \|\P f\|_{\P,2} \lesssim~1$. Plugging these bounds into Corollary 5.5 of \cite{chen2020jackknife} renders all higher–order projection terms negligible, leaving the first order term:
    \begin{align*}
          \bbE \left[ \|\bbU_{n,2}(f) -  \P^{2} f \|_{\calF_n}  \right] \lesssim & \ \frac{\sup_{f\in
          \calF_{n}}\|\P f\|_{\P,2}\sqrt{v\log(A \vee n) } +  \|\P \F_n\|_{\P,2}\cdot v\log(A \vee n) }{\sqrt{n}} \\
          \lesssim &\ \frac{\sqrt{k} \log(n)}{\sqrt{n} }\,.
    \end{align*}

Instead, our maximal inequality replaces the projected envelope $\P \F_n$ by the tighter quantity
$\sup_{f\in\calF_n}|\P f|$, and $\P \F_n $ only influences the bound through the VC characteristics. Thus, $\P \F_n$ only appears in a logarithmic factor, which in turn is dominated by $\log n$. Specifically, Lemma~\ref{lemma:Umaximal_inequality_order2} (ii) implies that:
    \begin{align*}
         \bbE \left[ \|\bbU_{n,2}(f) -  \P^{2} f \|_{\calF_n}  \right] \lesssim&\  \frac{\log(  n)}{\sqrt{n}} + \frac{\sqrt{k} \log(n) }{n}\,.
    \end{align*}
    Even when $n \lesssim k \lesssim n^2$, the right–hand side of the above display still shrinks to zero.  Hence, the refined inequality remains informative 
in the regime where the bound by directly applying results in \cite{chen2020jackknife} is not sufficiently sharp.
\end{remark}

\begin{proof}[Proof of Lemma~\ref{lemma:Umaximal_inequality_order2}] 
For any $f$ defined on $\calO^2$, we have the following decomposition 
    \begin{align*}
      \bbU_{n,2}(f) -  \P^{2} f  =   \bbU_{n, 2} \left( \pi_2 f   \right) + 2\bbU_{n, 1} \left(  	 \pi_1  f   \right).
    \end{align*}
    Here, $\pi_1$ is the \Hajek{} projections defined in Section~\ref{app:notation}.
      Since the function class $\calF_n$ is of VC type with  characteristics $(A_n,v_n)$ satisfying $A_n\geq e$ and $v_n  \geq 1$,  
 applying Corollary 5.6 of \citet{chen2020jackknife} to  function class $\calF_n$ with their $k =1,2$ and $p = 1$, we obtain 
\begin{align}\label{eq:pi2fbound}
     \bbE \left[ \left\|  \bbU_{n,2} (\pi_2 f) \right\|_{\calF_n} \right] \lesssim & \ \frac{J_{2}(1, \calF_n, \F_n)\cdot \|\F_n\|_{\P^2,2}}{n} \lesssim \frac{v_n\log A_n\cdot\|\F_n\|_{\P^2,2}}{n},\\
     \bbE \left[ \left\|  \bbU_{n,1} (\pi_1 f) \right\|_{\calF_n} \right] \lesssim & \ \frac{J_{1}(1, \calF_n, \F_n )\cdot \|\P\F_n \|_{\P,2}}{\sqrt{n}} \lesssim \frac{(v_n\log A_n)^{1/2}\cdot\|\P\F_n \|_{\P,2}}{\sqrt{n}}.\label{eq:pi1fbound}
\end{align}
Combining \eqref{eq:pi2fbound}--\eqref{eq:pi1fbound}, we prove part~(i) of this lemma. 

If we further assume condition (3) stated in this lemma, we can sharpen the bound for the first–order
term $\bbU_{n, 1} \left(  	 \pi_1  f   \right)$.  Define the function class $\pi_1 \calF_n \coloneqq\{f = \pi_1 f : f \in\calF_n\}$. For any $f\in\calF_n$, by definition of $\pi_1$, we have $\pi_1 f = \P f - \P^2 f$. Up to a universal constant, centering does not alter the VC-type characteristics.  Consequently, using Lemma~\ref{lemma:PcalF-VC}, we obtain  $\pi_1 \calF_n$ is still of VC type with characteristics of order $(4\sqrt{A_n}M_n/m_n,2v_n)$ and envelope $m_n$.  Since $\sup_{f\in \calF_n} \|\P f\|_{\P,2}/m_n \le 1$, applying Lemma~\ref{lemma:sup-p1} to the function class $\pi_{1}\calF_n$, we obtain
\begin{align}\label{eq:pi1fbound-new}
        \bbE \left[ \left\| \bbU_{n,1} (f)  \right\|_{\pi_1\calF_n} \right] \lesssim \frac{ \sup_{f\in \calF_n} \|\P f \|_{\P,2} \cdot\sqrt{v_n\log ( A_{n}^{\prime} \vee n) }}{\sqrt{n}} + \frac{m_n v_n\log( A_{n}^{\prime} \vee n ) }{n} ,
    \end{align}
    where $A_{n}^{\prime} = \sqrt{A_n}M_n/m_n$.
Since $A_{n}^{\prime} \geq A_n$, combining \eqref{eq:pi2fbound} and \eqref{eq:pi1fbound-new}, we obtain
\begin{align*}
     \bbE \left[ \|\bbU_{n,2}(f) -  \P^{2} f \|_{\calF_n}  \right] \lesssim & \ \bbE \left[ \left\|  \bbU_{n,2} (\pi_2 f) \right\|_{\calF_n} \right]  + \bbE \left[ \left\|  \bbU_{n,1} (\pi_1 f) \right\|_{\calF_n} \right]  \\
    \lesssim & \ \frac{ \left( \|\F_n\|_{\P^2,2} +  m_n \right) \cdot  v_n\log( A_{n}^{\prime} \vee n )}{n} + \frac{ \sup_{f\in \calF_n} \|\P f \|_{\P,2} \cdot\sqrt{v_n\log (A_n^{\prime} \vee n) }}{\sqrt{n}} ,
\end{align*}
which completes the proof.
\end{proof}

The function class encountered in this paper generally has a product structure  $\calH_{n}= \bigl\{\,f \cdot g_{n} : f\in\calF \bigr\}$,
where $\calF$ is of a fixed VC-type class with characteristics
$(A,v)$, and
$g_{n}$ is a single function that depends on $n$. By Lemma~\ref{lemma:VC-sum-prod} below,
$\calH_{n}$ is again of a VC-type class whose characteristics remain independent of $n$, whereas its envelope, $\F_{n} =\sup_{f\in\calF} |f|\cdot|g_{n}|$, can vary with $n$.  For such a function class, we develop a new and readily useful maximal inequality for third-order $U$-processes, which are heavily used in our higher-order estimators.  The proof parallels that of Lemma~\ref{lemma:Umaximal_inequality_order2}, with additional steps to handle the third-order Hoeffding projection term $\bbU_{n,3}(\pi_3 f)$.

\begin{lemma}\label{lemma:Umaximal_inequality_order3}
Suppose that (1) $\calF_n \coloneqq \{f: \calO^3 \rightarrow \bbR\}$ is a collection of symmetric functions from $\calO^{3}$ to $\bbR$ with envelope function $\F_n$, (2) $\calF_n$ is of VC type with bounded characteristics $(A,v)$ (recall from Definition~\ref{def:VC-type}), and further (3) $\sup_{f\in \calF_n} \|\P^2 f (\cdot) \|_{\infty} \leq m_n$, and $\|\P^2 \F_n (\cdot) \|_{\infty} \leq M_n$.  Then
    \begin{align*}
        \bbE \left[ \|\bbU_{n,3}(f) - \P^3 f \|_{\calF_n}  \right] \lesssim \frac{ \left( \|\F_n\|_{\P^3,2} + m_n \right) \cdot \log( A_{n}^{\prime} \vee n ) }{n} + \frac{ \sup_{f\in \calF_n} \|\P^2 f \|_{\P,2} \cdot \sqrt{\log ( A_{n}^{\prime} \vee n) } }{\sqrt{n}},
    \end{align*}
    where $A_{n}^{\prime} \coloneqq M_n/ m_n$.
\end{lemma}
\begin{proof} 
Similar to the proof of Lemma~\ref{lemma:Umaximal_inequality_order2},   we first apply the Hoeffding decomposition:
    \begin{align*}
      \bbU_{n,3}(f) -\P^3 f =   \bbU_{n, 3} \left( \pi_3  f   \right) + 3\bbU_{n, 2} \left( \pi_2 f   \right) + 3\bbU_{n, 1} \left(  	 \pi_1  f   \right),
    \end{align*}
    where $\pi_j$ for $j=1,2,3$ represent the Hoeffding projections as defined in Section~\ref{app:notation}. By Jensen's inequality, we have $ \|
      \P^{3-2} \F_n \|_{\P^2 ,2}  = O( \|\F_n\|_{\P^3,2})$. 
      Since the function class $\calF_n$ is of VC type with bounded characteristics $(A,v)$ independent of the sample size $n$, we have $J_{3}(1, \calF_n, \F_n)\lesssim 1$ and $J_{2}(1, \calF_n, \F_n)\lesssim 1$. Thus,
 applying Corollary 5.6 of \citet{chen2020jackknife} to  function class $\calF$ with their $k =2,3$ and $p = 1$, we obtain 
\begin{align}\label{eq:pi3fbound-3}
    \bbE \left[ \left\|  \bbU_{n,3} (\pi_3 f) \right\|_{\calF_n} \right] \lesssim \frac{J_{3}(1, \calF_n, \F_n)\cdot \|\F_n\|_{\P^3,2}}{n^{3/2} } \lesssim \frac{ \|\F_n\|_{\P^3,2}}{n^{3/2} }
\end{align}
and
\begin{align}\label{eq:pi2fbound-3}
     \bbE \left[ \left\|  \bbU_{n,2} (\pi_2 f) \right\|_{\calF_n} \right] \lesssim \frac{J_{2}(1, \calF_n, \F_n)\cdot \|\P^{3-1}\F_n\|_{\P^2,2}}{n} \lesssim \frac{\|\F_n\|_{\P^3,2}}{n}.
\end{align}
Following the same analysis for the empirical process $\pi_1 \calF_n \coloneqq\{f = \pi_1 f : f \in\calF_n\}$ as in Lemma~\ref{lemma:Umaximal_inequality_order2}, we obtain
 \begin{align}\label{eq:pi1fbound-3}
     \bbE \left[ \left\|   \bbU_{n,1} ( f) \right\|_{ \pi_{1}\calF_n } \right] 
     \lesssim & \ \frac{ \sup_{f\in \calF_n} \|\P f \|_{\P,2} \cdot \sqrt{\log ( A_{n}^{\prime} \vee n) } }{\sqrt{n}} + \frac{m_n\log( A_{n}^{\prime} \vee n ) }{n},
 \end{align}
 where $A_{n}^{\prime} = M_n/m_n$.
 
Combining \eqref{eq:pi3fbound-3}--\eqref{eq:pi1fbound-3}, we conclude
\begin{align*}
    \bbE \left[ \|\bbU_{n,3}(f) \|_{\calF_n}  \right] & \lesssim \bbE \left[ \left\|  \bbU_{n,3} (\pi_3 f) \right\|_{\calF_n} \right] 
 + \bbE \left[ \left\|  \bbU_{n,2} (\pi_2 f) \right\|_{\calF_n} \right]  + \bbE \left[ \left\|  \bbU_{n,1} (\pi_1 f) \right\|_{\calF_n} \right]  \\
    & \lesssim  \frac{ \|\F_n\|_{\P^3,2} + m_n\log( A_{n}^{\prime} \vee n ) }{n} + \frac{ \sup_{f\in \calF_n} \|\P f \|_{\P,2} \cdot \sqrt{\log ( A_{n}^{\prime} \vee n) } }{\sqrt{n}} ,
\end{align*}
 which completes the proof.
\end{proof}

We next quote a result from Theorem 3 of \citet{andrews1994empirical} that will be useful in the proof.
\begin{lemma}[Uniform entropy for sums and products] \label{lemma:entropy+*}
Let $\calF$ and $\calG$ be two classes of functions with envelopes $\F$ and $\G$, respectively. Then, for any $\epsilon>0$, the uniform entropy numbers of $\calF\calG = \{fg:f\in\calF, g\in\calG\}$ and $\calF + \calG = \{f+g:f\in\calF, g\in\calG\}$ satisfy
\begin{gather*}
\sup_Q \sfN (\calF+\calG, \|\cdot\|_{Q,2},2\epsilon\|\F+\G\|_{Q,2}) \leq \sup_Q \sfN (\calF,\|\cdot\|_{Q,2},\epsilon\|\F\|_{Q,2})\cdot \sup_Q \sfN (\calG,\|\cdot\|_{Q,2},\epsilon\|\G\|_{Q,2}),\\
\sup_Q \sfN (\calF\cdot\calG, \|\cdot\|_{Q,2}, 2\epsilon\|\F\cdot\G\|_{Q,2})\leq \sup_Q \sfN (\calF, \|\cdot\|_{Q,2}, \epsilon\|\F\|_{Q,2},) \cdot \sup_Q \sfN (\calG, \|\cdot\|_{Q,2}, \epsilon\|\G\|_{Q,2}).
\end{gather*}
\end{lemma}

\begin{lemma}[VC--type characteristics of sums and products]\label{lemma:VC-sum-prod}
Suppose $\mathcal F$ and $\mathcal G$ are pointwise measurable measurable and VC type with characteristics $(A_{\calF},v_{\calF})$ and $(A_{\calG},v_{\calG})$
, and envelopes $\F$ and $\G$, respectively. Then, $\calF + \calG = \{f+g:f\in\calF, g\in\calG\}$ and $\calF\calG = \{fg:f\in\calF, g\in\calG\}$ are VC type with characteristics $A' \coloneqq  2( A_{\calF}^{v_{\calF}} A_{\calG}^{v_{\calG}} )^{1/v'}$ and  $ v' \coloneqq v_{\calF} + v_{\calG}$.  Specially, when $\calG=\{g\}$, $\calF + \calG$ and $\calF \calG$ will be VC type with characteristics $(2A_{\calF} , v_{\calF})$.
\end{lemma}
\begin{proof} Using Lemma~\ref{lemma:entropy+*}, we have
    \begin{align*}
  \sup_Q
  \sfN\bigl(\mathcal F+\mathcal G,\,
         \|\cdot\|_{Q,2},\,
         \epsilon\|\F+\G\|_{Q,2}\bigr)
 \le & \
  \sup_Q
  \sfN\bigl(\mathcal F,\|\cdot\|_{Q,2},\epsilon/2\|\F\|_{Q,2}\bigr)\;
  \sup_Q
  \sfN\bigl(\mathcal G,\|\cdot\|_{Q,2},\epsilon/2\|\G\|_{Q,2}\bigr)\\
  \leq & \ \left(\frac{2 A_{\calF}}{\epsilon}\right)^{v_{\calF}} \left(\frac{2 A_{\calG}}{\epsilon}\right)^{v_{\calG}} \\
  = & \ \left( \frac{ 2( A_{\calF}^{v_{\calF}} A_{\calG}^{v_{\calG}} )^{1/(v_{\calF} + v_{\calG})} }{\epsilon}\right)^{v_{\calF} + v_{\calG}},
\end{align*}
which implies $\calF + \calG$ is of VC type with characteristics $A' =  2(A_{\calF}^{v_{\calF}} A_{\calG}^{v_{\calG}} )^{1/v'}$ and  $ v' = v_{\calF} + v_{\calG}$.
\end{proof}

\begin{lemma}\label{lemma:sup-p1}
    Suppose that function class  $\calF_{n}: \calO \rightarrow \bbR$ is pointwise measurable with envelope function $\F_n$ and is of VC type with  characteristics $(A_n,v_n)$, satisfying $A_n \geq e$ and $v_n \geq 1$. Suppose that  $\log (\sup_{f\in\calF}\|f_n\|_{\P,2}/\|\F_n\|_{\P,2}) \lesssim \log n$ and $\| \F_n (\cdot) \|_{\infty}\lesssim M_n$. Then
    \begin{align*}
        \bbE \left[ \left\| \bbU_{n,1} f - \bbP f \right\|_{\calF_n} \right] \lesssim \frac{ \sup_{f\in \calF_n} \|f \|_{\P,2} \cdot\sqrt{v_n\log(A_n \vee n) }}{\sqrt{n}} + \frac{M_n \cdot v_n\log( A_n \vee n) }{n} .
    \end{align*}
\end{lemma}
 \begin{proof}
     For any constant sequence satisfying $\log a_{n}/b_{n} \lesssim \log n$,  we have the following inequality holds: 
	\begin{align*}
		J_{1}(a_{n}/b_{n} , \calF,\F) = & \ \int_{0}^{a_{n}/b_{n}}  \sup_{Q}  \left\{ 1 + \log \sfN (\calF_n, \|\cdot\|_{Q,2} , \epsilon \| \F_n\|_{Q,2}) \right\}^{1/2} \diff  \epsilon \\
\leq &\ \int_{0}^{a_{n}/b_{n}} \sqrt{ 1 + v_n\log (A_n/\epsilon) } \diff \epsilon \\
			\leq & \ \left(\int_{0}^{a_{n}/b_{n}} 1^2 \diff \epsilon \right)^{1/2}\cdot \left(\int_{0}^{a_{n}/b_{n}}  \{ 1 + v_n\log (A_n/\epsilon) \} \diff \epsilon \right)^{1/2}\notag\\
			= & \ \sqrt{\frac{a_n}{b_n}}\cdot\left\{ \frac{a_n}{b_n}(1 + v_n\log A_n) +  v_n\left( -\epsilon \log \epsilon + \epsilon \right)\big|_{0}^{a_{n}/b_{n}} \right\}^{1/2} \notag\\
			= & \ \sqrt{\frac{a_n}{b_n}}\cdot \left\{ \frac{a_n}{b_n}(1 + v_n\log A_n) + v_n\frac{a_n}{b_n} \{\log(a_{n}/b_{n}) +  1\} \right\}^{1/2} \notag\\
			\lesssim & \ \frac{a_n}{b_n}\cdot\sqrt{v_n\log(A_na_{n}/b_{n})} \lesssim \frac{a_n}{b_n}\cdot\sqrt{ v_n\log(A_n 
 \vee n)},\notag
		\end{align*}
   where the second inequality holds by Cauchy–Schwarz inequality,  the penultimate inequality by $A_n \geq e$, $v_n \geq 1$,  and the last by  $\log a_{n}/b_{n} \lesssim \log n$.     
Then, applying this result and Theorem~5.1 in \cite{chen2020jackknife}, we obtain
  \begin{align*}
        \bbE \left[ \left\| \bbU_{n,1} f - \bbP f \right\|_{\calF_n}  \right] = & \ \frac{J_{1}( \sup_{f\in \calF_n} \|f \|_{\P,2}/ \|  \F_n \|_{\P,2} , \calF_n,\F_n )\cdot \| \F_n \|_{\P,2} }{\sqrt{n}} \\
        &+ \frac{J_{1}( \sup_{f\in \calF_n} \|f \|_{\P,2}/ \|  \F_n \|_{\P,2} , \calF_n,\F_n )^2  \cdot \| \F_n (\cdot) \|_{\infty}}{n\cdot \sup_{f\in \calF_n }\|f \|_{\P,2}^2/  \| \F_n\|_{\P,2}^2}\notag\\
     \lesssim & \ \frac{ \sup_{f\in \calF_n} \|f \|_{\P,2} \cdot\sqrt{v_n\log(A_n \vee n) }}{\sqrt{n}} + \frac{ M_n \cdot v_n \log( A_n \vee n) }{n}.\notag
    \end{align*}
 \end{proof}

\begin{lemma}\label{lemma:PcalF-VC}
 If $\calF_n \coloneqq \{f: \calO^r \rightarrow \bbR\}$ is of VC type with characteristics $(A_n,v_n)$ and  envelope $\F_n$, then for every $k = 1,\ldots,r-1$, 
 \begin{enumerate}[label=(\roman*)]
     \item $\P^{r-k}\calF_n$ is also of VC type with characteristics $(4\sqrt{A_n},2v_n)$  and envelope $\P^{r-k} \F_n$.
     
     \item Suppose further that  $\sup_{f\in \calF_n} \|\P^{r-k}f (\cdot) \|_{\infty} \leq m_n$ and $\|\P^{r-k} \F_n (\cdot) \|_{\infty} \leq M_n$. Then, $\P^{r-k}\calF_n$ is of VC type with characteristics $(4\sqrt{A_n} M_n/m_n, 2v_n)$ and  envelope  $m_n$, i.e.,
 \begin{align*}
  \sup_{Q} \sfN (\P^{r-k}\calF_n, \|\cdot\|_{Q,2}, \tau m_n )\leq  \left( \frac{4 M_n\sqrt{A_n}}{m_n\tau} \right)^{2v_n},\ \text{for all}\  0 < \tau \leq  1.  
\end{align*}
 \end{enumerate}
\end{lemma}
\begin{proof}
Lemma 5.4 of \citet{chen2020jackknife} states that for any VC-type function class $\calF$
with characteristics $(A_n,v_n)$, $\P^{r-k}\calF$ satisfies 
\begin{align*}
  \sup_{Q} \sfN(\P^{r-k}\calF_n, \|\cdot\|_{Q,2}, \epsilon\|\P^{r-k} \F_n\|_{Q,2} )\leq  (4 \sqrt{A_n}/\epsilon)^{2v_n}, 0 < \forall \epsilon \leq  1,  
\end{align*} 
which establishes part $(i)$.

For part $(ii)$, set
$\varepsilon=\tau m_n/\|\P^{r-k} \F_n\|_{Q,2}$ in the display above. Then, since $\|\P^{r-k} \F_n\|_{Q,2} \le \|\P^{r-k} \F_n (\cdot) \|_\infty \le M_n$, we obtain
\begin{align*}
     \sup_{Q} \sfN(\P^{r-k}\calF_n, \|\cdot\|_{Q,2}, \tau m_n ) \le \left( \frac{4 \|\P^{r-k}\F_n\|_{Q,2}\sqrt{A_n}}{m_n\tau} \right)^{2v_n} \leq \left( \frac{4 M_n\sqrt{A_n}}{m_n\tau} \right)^{2v_n}.
\end{align*}
\end{proof}

\section{Proof of Results in Section \ref{sec:properties}}\label{app:properties}
\setcounter{equation}{0}
Before starting the proof, we first summarize some frequently used results. Under Assumptions~\ref{as:nuisance_est}(iii) and~\ref{as:basis}(i)--(iii), Lemma~16 of \cite{liu2017semiparametric} implies the following uniform bound: for any function $f \in \mathcal{C}(\mathcal{X}\times\mathcal{T})$,
\begin{align}\label{eq:supnorm_proj}
 \left\|\Pi [f \mid \bar{\phi}_{k}] (\cdot, \cdot) \right\|_{\infty} \leq  \|f (\cdot, \cdot) \|_{\infty}. 
\end{align}
Recall that $\Pi [f\mid \bar{\phi}_{k}]$ is the  projection of $f$ onto the linear span of basis $\bar{\phi}_{k}$, defined in Section~\ref{sec:notation}. 

Because the proof carefully tracks changes in $\bbeta$, we now make this dependence explicit by moving $\bbeta$ out of subscripts into argument parentheses. For notational simplicity, hereafter we sometimes write $\Gamma_{\bbeta}\coloneqq \Gamma(\cdot,\cdot,\bbeta)$ and denote $\hat{\Xi}_{\bbeta} \coloneqq \hat{\Xi}(O_1,O_2;\bbeta)$ if this causes no ambiguity.

\subsection{Proof of Theorem \ref{th:consistency}}\label{app:beta-consistency}

\begin{proof}[Proof of Theorem \ref{th:consistency}]

First, we will show, for any $\bbeta \in \calB$,
\begin{align}\label{eq:psi>beta}
\|\psi (\bbeta )\|   \gtrsim  \|\bbeta - \bbeta^*\| \wedge c
\end{align}
for some constant $c$. Given \eqref{eq:psi>beta}, we obtain  $\bbE[\|\tilde{\bbeta}^{(2)} -\bbeta^{*}\|] \lesssim \tilde{r}_{n,\bbeta}$, if we  show that $\|\psi(\tilde{\bbeta}^{(2)})\| \lesssim \tilde{r}_{n,\bbeta} $ with probability approaching one.

By Assumption~\ref{as:smoothness}(ii), for any $\epsilon \leq c_3/2p$  ($c_3$ is as defined in Assumption~\ref{as:smoothness}(ii)), there exists $\delta>0$ such that for any $\bbeta$ satisfying $\|\bbeta - \bbeta^*\|\leq \delta$, we have $\|\nabla_{\bbeta} \psi (\bbeta ) - \nabla_{\bbeta} \psi (\bbeta^* )\|_{\infty}\leq \epsilon $.
Then, for  any $\bbeta$ satisfying $\|\bbeta - \bbeta^*\|\leq \delta$, since $\psi (\bbeta^{\ast}) = 0$, by Taylor expansion, we have 
 \begin{align*}
    \psi (\bbeta)  = \nabla_{\bbeta} \psi (\bbeta^{\dag}) (\bbeta -  \bbeta^{\ast}),
 \end{align*}
 where $\bbeta^{\dag}$ lies between $\bbeta^{\ast}$ and $\bbeta$.   Then, by Assumption~\ref{as:smoothness}(iii), we have 
 \begin{align*}
   \| \nabla_{\bbeta} \psi (\bbeta^{\dag}) ( \tilde{\bbeta}^{(2)} -  \bbeta^{\ast}) \| \geq & \
   \| \nabla_{\bbeta} \psi (\bbeta^{\ast}) ( \tilde{\bbeta}^{(2)} -  \bbeta^{\ast}) \| -\|\{ \nabla_{\bbeta} \psi (\bbeta^{\dag}) -\nabla_{\bbeta} \psi (\bbeta^{\ast}) \} ( \tilde{\bbeta}^{(2)} -  \bbeta^{\ast}) \|\\
   \geq & \ \| \nabla_{\bbeta} \psi (\bbeta^{\ast}) ( \tilde{\bbeta}^{(2)} -  \bbeta^{\ast}) \| - \epsilon p \cdot\|\tilde{\bbeta}^{(2)} -  \bbeta^{\ast}\|\\
    \geq & \ \| \nabla_{\bbeta} \psi (\bbeta^{\ast}) ( \tilde{\bbeta}^{(2)} -  \bbeta^{\ast}) \| -  \frac{1}{2}\|\nabla_{\bbeta} \psi (\bbeta^{\ast})(\tilde{\bbeta}^{(2)} -  \bbeta^{\ast} )\| \\
   = & \ \frac{1}{2}\| \nabla_{\bbeta} \psi ( \bbeta^{\ast}) \cdot ( \tilde{\bbeta}^{(2)} -  \bbeta^{\ast}) \| \geq \frac{c_3}{2} \| \tilde{\bbeta}^{(2)} -  \bbeta^{\ast} \|.
 \end{align*}
For any $\|\bbeta - \bbeta^*\| \geq \delta$, by Assumption~\ref{as:uniqueness}, $\|\psi(\bbeta)\| \geq c$ for some constant $c$. Hence, combining these two cases, we obtain \eqref{eq:psi>beta}. 
 
Next, we show $\|\psi (\tilde{\bbeta}^{(2)} )\| \lesssim \tilde{r}_{n,\bbeta}$ with probability approaching one. Since $\tilde{\psi}^{(2)}_{k}(\tilde{\bbeta}^{(2)}) \equiv 0 $ and $\psi (\bbeta^{\ast})=0$, we decompose $\|\psi (\tilde{\bbeta}^{(2)} )\|$ as follows
\begin{align*}
     \|\psi (\tilde{\bbeta}^{(2)} )\|  
     \leq & \| \psi (\tilde{\bbeta}^{(2)} ) - \bar{\psi}_{k}(\tilde{\bbeta}^{(2)}) \| +  \| \bar{\psi}_{k}(\tilde{\bbeta}^{(2)}) - \tilde{\psi}^{(2)}_{k}(\tilde{\bbeta}^{(2)})\| + \| \tilde{\psi}^{(2)}_{k}(\tilde{\bbeta}^{(2)}) - \tilde{\psi}^{(2)}_{k}(\bbeta^{*})\| \\
     & + \| \tilde{\psi}^{(2)}_{k}(\bbeta^{*}) - \bar{\psi}_{k}(\bbeta^*) \|  + \|\bar{\psi}_{k}(\bbeta^*) - \psi (\bbeta^{*} ) \|.
\end{align*}
Since $\tilde{\psi}^{(2)}_{k}(\tilde{\bbeta}^{(2)}) \equiv 0 $ by the definition of $\tilde{\bbeta}^{(2)}$, we derive 
\begin{align*}
    &\| \tilde{\psi}^{(2)}_{k}(\tilde{\bbeta}^{(2)}) - \tilde{\psi}^{(2)}_{k}(\bbeta^{*})\|  =  \|\tilde{\psi}^{(2)}_{k}(\bbeta^{*})\| 
     \leq \| \tilde{\psi}^{(2)}_{k}(\bbeta^{*})  - \bar{\psi}_{k}(\bbeta^*) \|  + \| \bar{\psi}_{k}(\bbeta^*)  - \psi(\bbeta^*) \|.
\end{align*}
Thus 
\begin{align*}
     \|\psi (\tilde{\bbeta}^{(2)} )\| & \leq  3 \sup_{\bbeta\in\calB} \| \psi(\bbeta) - \bar{\psi}_{k}(\bbeta)\| 
     + 3\sup_{\bbeta\in\calB} \| \tilde{\psi}^{(2)}_{k}(\bbeta) - \bar{\psi}_{k}(\bbeta) \| \eqqcolon  3 \uppercase\expandafter{\romannumeral1} + 3 \uppercase\expandafter{\romannumeral2}.
\end{align*}

\textbf{Convergence Rate for $\uppercase\expandafter{\romannumeral1}$:} \begin{align}\label{eq:tildepsi-psi}
        & \psi(\bbeta) -\bar{\psi}_{k}(\bbeta)\notag\\
        = & \ \bbE \left[ \xi (X, T)  \Gamma (Y, T, \bbeta)  \right] 
        - \bbE \left[ \hat{\xi} (X, T) \left\{ \Gamma (Y, T, \bbeta) - \hat{b}_{\bbeta_{\mathsf{init}}} (X, T) \right\} + \xi(X,T) \hat{b}_{\bbeta_{\mathsf{init}}} (X, T)  \right] \notag\\
        & - \bbE \left[  \left\{ \xi (X, T) - \hat{\xi} (X, T) \right\} \bar{\phi}_{k} (X, T) ^{\top} \right]\Sigma_{k}^{-1} \bbE \left[ \bar{\phi}_{k} (X, T) \left\{ b_{\bbeta} (X, T) - \hat{b}_{\bbeta_{\mathsf{init}}} (X, T)\right\} \right] \notag\\
        = & \ \bbE \left[ \left\{\xi (X, T) -  \hat{\xi} (X, T) \right\}\cdot \left\{b_{\bbeta} (X, T) - \hat{b}_{\bbeta_{\mathsf{init}}} (X, T) \right\} \right]\notag\\
        & - \bbE \left[  \left\{ \xi (X, T) - \hat{\xi} (X, T) \right\} \bar{\phi}_{k} (X, T)^{\top} \right]\Sigma_{k}^{-1} \bbE \left[ \bar{\phi}_{k} (X, T) \left\{ b_{\bbeta} (X, T) - \hat{b}_{\bbeta_{\mathsf{init}}} (X, T) \right\} \right] \notag\\
        = & \ \bbE \left[ \Pi^{\bot} \left[ \xi-\hat{\xi} \mid \bar{\phi}_{k}\right] (X, T) \cdot \Pi^{\bot} \left[ b_{\bbeta}-\hat{b}_{\bbeta_{\mathsf{init}}}  \mid \bar{\phi}_{k} \right] (X, T) \right] = \mathrm{TB}_{\psi, k} (\bbeta).
    \end{align}	
    Hence, we obtain
    \begin{align}\label{eq:rate-tildepsi}
        \uppercase\expandafter{\romannumeral1} \leq \ &\sup_{\bbeta\in\calB} \| \mathrm{TB}_{\psi, k} (\bbeta)\| \lesssim  \| \Pi^{\perp} (\hat{\xi} - \xi | \bar{\phi}_{k}) \|_{\P,2} \cdot \sup_{\bbeta\in\calB}\|  \Pi^{\perp} [b_{\bbeta} - \hat{b}_{\bbeta_{\mathsf{init}}} | \bar{\phi}_{k}]  \|_{\P,2}. 
    \end{align} 
  
\textbf{Convergence Rate for $\uppercase\expandafter{\romannumeral2}$:}
Recalling the definition of $\IF_{\bar{\psi}_{k}}^{(2)}$ from Lemma~\ref{lem:HOIF} in Section~\ref{app:truncation},  we consider the following decomposition:
\begin{align} \label{eq:hatpsi-tildepsi}		 
		 \tilde{\psi}^{(2)}_{k}(\bbeta) - \bar{\psi}_{k}(\bbeta)\notag
		= & \ \bbU_{n, 3} \IF_{\bar{\psi}_{k}}^{(2)} - \bar{\psi}_{k}(\bbeta)  \notag\\
		= & \ \bbU_{n, 3} \left(   \IF_{\mathrm{B}_{\psi, k}}^{(2)}  - \bbE[\IF_{\mathrm{B}_{\psi, k}}^{(2)}] \right)   + \bbU_{n, 2} \left( \hat{\Xi}_{\bbeta}   - \bbE[\hat{\Xi}_{\bbeta}] \right). 
	\end{align}
We will derive the convergence rate uniformly in $\bbeta$ for each term in \eqref{eq:hatpsi-tildepsi}. 
For the second term in~\eqref{eq:hatpsi-tildepsi}, consider the function class $\calH_1$ defined as follows
	\begin{align}\label{eq:def-H1}
		\calH_1\coloneqq&\left\{ h(O_1, O_2,\bbeta) =
	\hat{\Xi} (O_1, O_2;\bbeta) : \bbeta\in \calB\right\}.
	\end{align}
 By Assumption \ref{as:criterion}, $\{\Gamma (y,t,\bbeta):\bbeta\in\calB\}$ is of VC type with  characteristics $(A,v)$.  By Lemma \ref{lemma:VC-sum-prod},  $\calH_1$ is of VC type with bounded characteristics $(4A,v)$ and  envelope   $ \H_{1} (O_1,O_2)= \sup_{\bbeta \in \calB } \left\| 		\hat{\Xi} (O_1, O_2;\bbeta)\right\|$. By  Assumptions~\ref{as:nuisance_est} and \ref{as:criterion}, we have $\| \H_{1} \|_{\P^2,2}  \lesssim 1 $. By Jensen's inequality, we have $ \|
      \P \H_1 \|_{\P ,2}  \lesssim \|\H_1\|_{\P^2,2} \lesssim 1$. Then, using Lemma~\ref{lemma:Umaximal_inequality_order2} for function class $\calH_1$, we obtain
\begin{align}\label{eq:supH1}
    \bbE \left[ \sup_{\bbeta \in \calB} \left\|  \bbU_{n, 2} ( \hat{\Xi}_{\bbeta}  )  - \P^2 (\hat{\Xi}_{\bbeta}) \right\| \right] \lesssim \frac{1}{\sqrt{n}}.
\end{align}

To control the first  term in \eqref{eq:hatpsi-tildepsi}, we need to consider the function class
\begin{align}\label{eq:def-H2}
\calH_{2,n} \coloneqq \left\{ h(O_1, O_2, O_3,\bbeta)=\IF_{\mathrm{B}_{\psi, k}}^{(2)} (O_1, O_2, O_3;\bbeta): \bbeta\in \calB\right\}.
\end{align}
Then, by Assumption \ref{as:criterion} and Lemma \ref{lemma:VC-sum-prod},  $\calH_{2,n}$ is of VC type with bounded characteristics $(4A,v)$ and  envelope  $ \H_{2,n} $, which is defined as follows
\begin{align*}
\H_{2,n} (O_1,O_2,O_3) \coloneqq& \left| \left\{ \hat{\xi} (X_{1}, T_{1}) \bar{\phi}_{k} (X_{1}, T_{1}) - \sfzbar_{k_{x}} (X_{1}) \otimes  \sfwbar_{k_{t}} (T_{3}) \right\}^{\top} \Sigma_{k}^{-1} \bar{\phi}_{k} (X_{2}, T_{2}) \right| \\
&\qquad\qquad\qquad\times\sup_{\bbeta\in \calB} \left\|  \Gamma (Y_{2}, T_{2},\bbeta) - \hat{b}_{\bbeta_{\mathsf{init}}}  (X_{2}, T_{2}) \right\| .
\end{align*}
 Using   Assumptions~\ref{as:nuisance_est}, \ref{as:criterion} and \ref{as:basis}, we have $\|\P^2 \H_{2,n} (\cdot) \|_{\infty}\lesssim \|\H_{2,n} (\cdot,\cdot,\cdot)\|_{\infty}\lesssim \zeta(k)^2 \lesssim k$, and
 \begin{align*}
   \|\H_{2,n}\|_{\P^3,2} =&  \ \left\{ \bbE \left[ \| \H_{2,n}(O_1,O_2,O_3)\|^2\right] \right\}^{1/2} \\
   \lesssim & \ \left\{\bbE \bigg[ \bar{\phi}_{k} (X_{2}, T_{2})^{\top} \Sigma_{k}^{-1} 
 \bar{\phi}_{k} (X_{2}, T_{2}) \sup_{\bbeta \in \calB } \left\| \Gamma (Y_{2}, T_{2},\bbeta)  - \hat{b}_{\bbeta_{\mathsf{init}}}  (X_{2}, T_{2}) \right\|^2\bigg] \right\}^{1/2}\\
 \lesssim & \ \zeta(k) \lesssim  \sqrt{k}.
\end{align*}
For any $h\in\calH_{2,n}$,  $\P^2 h$  can be written explicitly as
\begin{equation}\label{eq:pi1_IF}
    \begin{split}
         \P^2 h (O) = & \ \frac{1}{3} \hat{\xi} (X, T) \Pi[b_{\bbeta}  - \hat{b}_{\bbeta_{\mathsf{init}}} |\bar{\phi}_{k}](X, T)  
    - \frac{1}{3} \int_{t\in\calT} \Pi[ b_{\bbeta}  - \hat{b}_{\bbeta_{\mathsf{init}}} |\bar{\phi}_{k}](X, t)  \p_{T} (t) \diff t\\ 
    &+\frac{1}{3}  \Pi[\hat{\xi} - \xi| \bar{\phi}_{k}](X,T)  \left(\Gamma_{\bbeta}  - \hat{b}_{\bbeta_{\mathsf{init}}}
	\right)
 -\frac{1}{3}\int_{x\in\calX} \Pi[ b_{\bbeta}  - \hat{b}_{\bbeta_{\mathsf{init}}} |\bar{\phi}_{k}](x, T)  \p_{X} (x) \diff x ,
    \end{split}
\end{equation}
which is a combination of projection terms.  By  Assumptions~\ref{as:nuisance_est} and \ref{as:criterion}, we have $\|\xi (\cdot, \cdot)\|_{\infty} \lesssim 1$, $\|\hat{\xi} (\cdot, \cdot) \|_{\infty}\lesssim 1$, $\sup_{\bbeta\in\calB}\|\Gamma (\cdot, \cdot,\bbeta) \|_{\infty}\lesssim 1$, and $\sup_{\bbeta\in\calB}\|\hat{b}_{\bbeta} (\cdot, \cdot) \|_{\infty}\lesssim 1$.
Thus, using the result in \eqref{eq:supnorm_proj}, we obtain
\begin{align*}
\sup_{h\in\calH_{2,n}}\|\P^2 h (\cdot) \|_{\infty} \lesssim&\  \sup_{\bbeta \in \calB}\left\|\Pi[ b_{\bbeta}  - \hat{b}_{\bbeta_{\mathsf{init}}} |\bar{\phi}_{k}](\cdot, \cdot) \right\|_{\infty} + \| \Pi[\hat{\xi} - \xi| \bar{\phi}_{k}](\cdot,\cdot) \|_{\infty}\\
\lesssim &\ \sup_{\bbeta \in \calB}\| (b_{\bbeta}  - \hat{b}_{\bbeta_{\mathsf{init}}} ) (\cdot,\cdot) \|_{\infty} +  \|(\hat{\xi} - \xi) (\cdot,\cdot) \|_{\infty} \lesssim 1.
\end{align*}
This also implies $\sup_{h\in \calH_{2,n}} \| \P^2 h \|_{\P,2} \lesssim 1$.
 Hence, applying Lemma~\ref{lemma:Umaximal_inequality_order3} to $\calH_{2,n}$, we obtain
 \begin{align}\label{eq:supH2}
    \bbE \left[ \sup_{\bbeta\in\calB}\left\|\bbU_{n, 3} \left(   \IF_{\mathrm{B}_{\psi, k}}^{(2)}\right)   - \P^{3} \left(\IF_{\mathrm{B}_{\psi, k}}^{(2)}\right) \right\|\right] \lesssim \frac{\sqrt{k}\log n}{n} + \frac{\log n }{\sqrt{n}}.
 \end{align}
Hence, by \eqref{eq:hatpsi-tildepsi}, \eqref{eq:supH1}, \eqref{eq:supH2}, we have 
\begin{align}\label{eq:rate_hatpsi}
     \uppercase\expandafter{\romannumeral2} = \sup_{\bbeta \in \calB}\left\| \tilde{\psi}_{k}^{(2)}(\bbeta) -\bar{\psi}_{k}(\bbeta) \right\|  = O_{\P} \left( \frac{\sqrt{k}\log n}{n}\vee\frac{\log n}{\sqrt{n}}\right).
\end{align}

Combining \eqref{eq:rate-tildepsi} and \eqref{eq:rate_hatpsi}, we obtain
\begin{align*}
\left\|\psi (\tilde{\bbeta}^{(2)}) \right\| = & \ O_{\P} \left( \frac{\sqrt{k}\log n}{n} \vee \frac{\log n}{\sqrt{n}} +  \| \Pi^{\perp} (\hat{\xi} - \xi | \bar{\phi}_{k}) \|_{\P,2} \cdot \sup_{\bbeta\in\calB}\|  \Pi^{\perp} [b_{\bbeta} - \hat{b}_{\bbeta_{\mathsf{init}}} | \bar{\phi}_{k}]  \|_{\P,2}  \right) \\
= & \ O_{\P} (\tilde{r}_{n,\bbeta}), 
\end{align*}
which completes the proof.

\end{proof}

\subsection{Proof of Theorem \ref{thm:feasible rates}}\label{app:feasible rates}

\begin{proof}[Proof of Theorem \ref{thm:feasible rates}]
By construction, we have $\hat{\psi}^{(2)}_{k}(\hat{\bbeta}^{(2)}) = 0 $ and $\psi(\bbeta^{*}) = 0$.  Consider the following decomposition: 
\begin{align*}
    0 = \hat{\psi}^{(2)}_{k}(\hat{\bbeta}^{(2)}) 
    =&\  \hat{\psi}^{(2)}_{k}(\hat{\bbeta}^{(2)}) - \bbE [\hat{\psi}^{(2)}_{k}(\hat{\bbeta}^{(2)})] + \bbE [\hat{\psi}^{(2)}_{k}(\hat{\bbeta}^{(2)})] - \bbE [\hat{\psi}^{(2)}_{k}( \bbeta^{*} )] \\
    &+\  \bbE [\hat{\psi}^{(2)}_{k}(\bbeta^{*})] - \bbE [ \tilde{\psi}^{(2)}_{k}(\bbeta^{*}) ] +\bbE [ \tilde{\psi}^{(2)}_{k}(\bbeta^{*}) ] - \psi(\bbeta^{*})\\
    \le & \  \ \sup_{\bbeta\in\calB} \left\Vert \hat{\psi}^{(2)}_{k}(\bbeta) - \bbE [ \hat{\psi}^{(2)}_{k}(\bbeta) ]  \right\Vert + \left\Vert \bbE [\hat{\psi}^{(2)}_{k}(\hat{\bbeta}^{(2)})] - \bbE [\hat{\psi}^{(2)}_{k}(\bbeta^{*})] \right\Vert \\
    &\  + \Vert \bbE [\hat{\psi}^{(2)}_{k}(\bbeta^{*})] - \bbE [ \tilde{\psi}^{(2)}_{k}(\bbeta^{*}) ] \Vert + \Vert \bbE [ \tilde{\psi}^{(2)}_{k}(\bbeta^{*}) ] - \psi(\bbeta^{*}) \Vert.
\end{align*}
Using the condition that for some constant $c>0$,  $ \lVert \bbE [ \hat{\psi}^{(2)}_{k}(\bbeta_1) ] - \bbE [ \hat{\psi}^{(2)}_{k}(\bbeta_2) ] \rVert \geq c \|\bbeta_1 - \bbeta_2\| $, it follows that
\begin{align*}
     \|\hat{\bbeta}^{(2)} - \tilde{\bbeta}^{(2)}\| 
     \lesssim & \  \ \sup_{\bbeta\in\calB} \left\Vert \hat{\psi}^{(2)}_{k}(\bbeta) - \bbE [ \hat{\psi}^{(2)}_{k}(\bbeta) ]  \right\Vert + \left\| \bbE [\hat{\psi}^{(2)}_{k}(\bbeta^{*})] - \bbE [ \tilde{\psi}^{(2)}_{k}(\bbeta^{*}) ] \right\| \\
     & + \left\|\bbE [ \tilde{\psi}^{(2)}_{k}(\bbeta^{*}) ] - \psi(\bbeta^{*})\right\|.
\end{align*}
Conditioning on the nuisance sample, the analysis of $\ \sup_{\bbeta\in\calB} \lVert \hat{\psi}^{(2)}_{k}(\bbeta) - \bbE [ \hat{\psi}^{(2)}_{k}(\bbeta) ]  \rVert$ closely parallels that of the term $\sup_{\bbeta\in\calB}\|\tilde{\psi}^{(2)}_{k}(\bbeta)-\bbE[\tilde{\psi}^{(2)}_{k}(\bbeta)]\|$, so both terms attain the same convergence rate. Thus, by the result in \eqref{eq:rate_hatpsi}, we have, with probability approaching one,
\begin{align*}
    \sup_{\bbeta\in\calB} \left\Vert \hat{\psi}^{(2)}_{k}(\bbeta) - \bbE [ \hat{\psi}^{(2)}_{k}(\bbeta) ]  \right\Vert \lesssim  \frac{\sqrt{k}\log n}{n}\vee\frac{\log n}{\sqrt{n}}  \lesssim \tilde{r}_{n,\bbeta}.
\end{align*}

By \eqref{eq:tildepsi-psi}, we have $\bbE [ \tilde{\psi}^{(2)}_{k}(\bbeta^{*}) ] - \psi(\bbeta^{*}) = \mathrm{TB}_{\psi, k} (\bbeta^*)$. Hence $\|\bbE [ \tilde{\psi}^{(2)}_{k}(\bbeta^{*}) ] - \psi(\bbeta^{*})\| =  \| \mathrm{TB}_{\psi, k}(\bbeta^*)\| \lesssim \tilde{r}_{n,\bbeta}$.
It remains to control the term $\bbE [\hat{\psi}^{(2)}_{k}(\bbeta^*)] - \bbE [ \tilde{\psi}^{(2)}_{k}(\bbeta^*) ]$. 
\begin{align*}
    & \ \left\Vert \bbE [\hat{\psi}^{(2)}_{k}(\bbeta^*)] - \bbE [ \tilde{\psi}^{(2)}_{k}(\bbeta^*) ] \right\Vert\\
    =&\  \left\Vert \bbE \left[  \left\{ \xi (X, T) - \hat{\xi} (X, T) \right\} \bar{\phi}_{k} (X, T) ^{\top} \right]\{ \hat{\Sigma}_{k}^{-1} -\Sigma_{k}^{-1}\}\bbE \left[ \bar{\phi}_{k} (X, T) \left\{ b_{\bbeta^*} (X, T) - \hat{b}_{\bbeta_{\ini}} (X, T)\right\} \right] \right\Vert \\
    \leq&\   \|\xi - \hat{\xi}\|_{\P,2}\cdot \|b_{\bbeta^*} - \hat{b}_{\bbeta_{\ini}} \|_{\P,2} \cdot \|\hat{\Sigma}_{k} - \Sigma\|_{\op} . 
\end{align*} 
Combining all terms, we obtain the desired bound:
\begin{align*}
   \bbE [ \|\hat{\bbeta}^{(2)} - \tilde{\bbeta}^{(2)}\| ]
    \lesssim \tilde{r}_{n,\bbeta} 
    + \|\xi - \hat{\xi}\|_{\P,2} \cdot \|b_{\bbeta^*} - \hat{b}_{\bbeta_{\ini} }\|_{\P,2} \cdot \|\hat{\Sigma}_k - \Sigma_k\|_{\mathrm{op}}.
\end{align*}
By Assumption~\ref{as:nuisance_est}(ii), we have $$\|b_{\bbeta^*} - \hat{b}_{\bbeta_{\ini} }\|_{\P,2} \le \|b_{\bbeta^*} - b_{\bbeta_{\ini} }\|_{\P,2} + \|b_{\bbeta_{\ini}} - \hat{b}_{\bbeta_{\ini} }\|_{\P,2} \lesssim \|\bbeta^* - \bbeta_{\ini}\| + \|b_{\bbeta_{\ini}} - \hat{b}_{\bbeta_{\ini} }\|_{\P,2}.$$ 
Hence, when Assumption~\ref{as:sigma} holds, 
\begin{align*}
    \|\xi - \hat{\xi}\|_{\P,2} \cdot \|b_{\bbeta^*} - \hat{b}_{\bbeta_{\ini} }\|_{\P,2} \cdot \|\hat{\Sigma}_k - \Sigma_k\|_{\mathrm{op}} = o(n^{-1/2}) \le r_{n,\tilde{\bbeta}^{(2)} },
\end{align*}
which implies $\bbE[\Vert \hat{\bbeta}^{(2)} - \bbeta^{\ast} \Vert] \lesssim \tilde{r}_{n,\bbeta}$.
\end{proof}

\subsection{Proof of Theorem~\ref{th:beta-dist}}
\label{app:beta-dist}
The matrices $V_1 ( \bbeta )$ and $V_2 ( \bbeta )$ that appear in
Theorem~\ref{th:beta-dist} are defined by
\begin{align}
\label{eq:V1}
	V_1 ( \bbeta ) &\coloneqq \lim\limits_{n \to \infty} \frac{n}{n+k} \bbE \left[ \left\{ 2 \pi_1 
 \hat{\Xi}_{\bbeta}  + 3 \pi_1  \IF_{\mathrm{B}_{\psi, k}}^{(2)} (\bbeta)  \right\}^{\otimes 2} \right] \\
   \label{eq:V2}
\text{ and }	V_2 (\bbeta) &\coloneqq  \lim\limits_{n \to \infty} \frac{1}{k + n} \bbE \left[ \left\{ \pi_2 
 \hat{\Xi}_{\bbeta}  +  3\pi_{2} \IF_{\mathrm{B}_{\psi, k}}^{(2)} (\bbeta) \right\} ^{\otimes 2} \right].
\end{align}

\begin{proof}[Proof of Theorem \ref{th:beta-dist}]
    By Assumption \ref{as:smoothness},	 $\psi(\bbeta)$ is differentiable with respect to $\bbeta$. Using Mean Value Theorem, we obtain
	\begin{align*}
		0=\psi(\bbeta^{*})=\psi(\tilde{\bbeta}^{(2)}) +\nabla_{\bbeta}\psi(\bbeta^{\dag})\cdot(\tilde{\bbeta}^{(2)}-\bbeta^{*}  ),
	\end{align*}
where $\bbeta^{\dag}$ lies between $\bbeta^{*}$ and $\tilde{\bbeta}^{(2)}$. Because $ \nabla_{\bbeta}\psi (\bbeta)$ is  continuous in $\bbeta$ at $\bbeta^{*}$, and $\|\tilde{\bbeta}^{(2)}-\bbeta^{*}\|=o_{\P} (1)$, we have
\begin{align} \label{eq:betadist-initial}
	{n}/{\sqrt{k+n}}( \tilde{\bbeta}^{(2)} -\bbeta^{*} ) = - \left[\nabla_{\bbeta}\psi (\bbeta^{*}) \right]^{-1} \cdot {n}/{\sqrt{k+n}} \cdot \psi(\tilde{\bbeta}^{(2)})+ o_{\P} (1) .
\end{align}
Since $0=\psi(\bbeta^{*})$ and $ \tilde{\psi}_{k}^{(2)} (\tilde{\bbeta}^{(2)})=0$, we have the following decomposition
\begin{align}
  \frac{n}{\sqrt{k+n}}\psi (\tilde{\bbeta}^{(2)}) 
  = & \ \frac{n}{\sqrt{k+n}} \left\{ \psi (\tilde{\bbeta}^{(2)}) -   \bar{\psi}_{k} ( \tilde{\bbeta}^{(2)})  \right\}\label{eq:localbeta-dist-1}\\
  & + \frac{n}{\sqrt{k+n}} \left\{ \bar{\psi}_{k} (\tilde{\bbeta}^{(2)})   -   \tilde{\psi}_{k}^{(2)} (\tilde{\bbeta}^{(2)})  - \left( \bar{\psi}_{k} (\bbeta^{\ast}) -   \tilde{\psi}_{k}^{(2)} (\bbeta^{\ast})  \right) \right\} \label{eq:localbeta-dist-2}\\
  &- \frac{n}{\sqrt{k+n}} \left\{ \tilde{\psi}_{k}^{(2)} (\bbeta^{\ast}) - \bar{\psi}_{k} (\bbeta^{\ast}) \right\} . \label{eq:localbeta-dist-3}  
\end{align}
By \eqref{eq:tildepsi-psi}, \eqref{eq:tbias-beta0} and $\|\tilde{\bbeta}^{(2)}-\bbeta^{*}\| = O_{\P}(\tilde{r}_{n,\bbeta})$, the right-hand side of \eqref{eq:localbeta-dist-1} is $o_{\P}  (1)$. By \eqref{eq:tbias-beta0} and  Lemma~\ref{th:psi-dist},  \eqref{eq:localbeta-dist-3} converges to a normal distribution.  Combining this and \eqref{eq:betadist-initial}, we obtain the desired distribution of $\tilde{\bbeta}^{(2)}$ as stated in Theorem~\ref{th:beta-dist}. To complete the proof, it remains to show \eqref{eq:localbeta-dist-2}, which is a $U$-process indexed by $\tilde{\bbeta}^{(2)}$, is $o_{\P} (1)$. Since we have shown in Section~\ref{app:beta-consistency} that $\|\tilde{\bbeta}^{(2)} - \bbeta^{*}\|\lesssim \tilde{r}_{n,\bbeta} $ with probability $1 - o(1)$,  it remains to show that, conditioned on this event, we have
\begin{align}
& \bbP \left( \sup_{\|\bbeta-\bbeta^{\ast}\|\lesssim \tilde{r}_{n,\bbeta} } \frac{n}{\sqrt{k+n}} \cdot \left\| \left\{ \bar{\psi}_{k} (\bbeta)-\tilde{\psi}_{k}^{(2)}(\bbeta) \right\} - \left\{\bar{\psi}_{k} (\bbeta^{\ast}) - \tilde{\psi}_{k}^{(2)} (\bbeta^{\ast}) \right\}  \right\| > \epsilon \right) \notag \\
& = o(1). \label{eq:uprocess}
\end{align}

Consider the following decomposition
\begin{align} \label{eq:hatpsi-tildepsi-delta}		 
		  &\left\{ \tilde{\psi}_{k}^{(2)}(\bbeta)- \bar{\psi}_{k} (\bbeta) \right\} - \left\{ \tilde{\psi}_{k}^{(2)} (\bbeta^{\ast}) -  \bar{\psi}_{k} (\bbeta^{\ast})  \right\}\\
		= & \ \bbU_{n, 3} \left(   \IF_{\mathrm{B}_{\psi, k}}^{(2)} (\bbeta)
 - \IF_{\mathrm{B}_{\psi, k}}^{(2)} (\bbeta^*) - \bbE [\IF_{\mathrm{B}_{\psi, k}}^{(2)} (\bbeta)  
 - \IF_{\mathrm{B}_{\psi, k}}^{(2)} (\bbeta^*) ] \right)  \notag + \bbU_{n, 2} \left( \hat{\Xi}_{\bbeta}  -\hat{\Xi}_{\bbeta^{\ast}}  - \bbE [\hat{\Xi}_{\bbeta}  -\hat{\Xi}_{\bbeta^{\ast}} ] \right). \notag
	\end{align}
First, we define the function class 
 \begin{align}\label{eq:def-H1delta}
		\calH_{1,\delta}\coloneqq\bigg\{ h(O,\bbeta) =\  
	\hat{\Xi}_{\bbeta}  -\hat{\Xi}_{\bbeta^{\ast}} 
     =\  \hat{\xi}(X,T) [ \Gamma (Y,T,\bbeta) - \Gamma (Y,T,\bbeta^*)] : \bbeta\in \calB, \|\bbeta - \bbeta^{\ast}\|\lesssim \tilde{r}_{n,\bbeta} \bigg\}.
	\end{align}
By Assumption~\ref{as:criterion} and Lemma~\ref{lemma:VC-sum-prod},  the function class $\calH_{1,\delta}$ of VC type with bounded characteristics $(8A,v)$ and envelope $\H_{1, \delta} (O) \coloneqq \sup_{ \|\bbeta - \bbeta^{\ast}\|\lesssim \tilde{r}_{n,\bbeta}}\|\hat{\xi}(X,T) [ \Gamma (Y,T,\bbeta) - \Gamma (Y,T,\bbeta^*)] \|$.  Since $\|\hat{\xi} (\cdot,\cdot) \|_{\infty}\lesssim 1$ under Assumption~\ref{as:nuisance_est},  we obtain
\begin{align*}
     \| \H_{1,\delta} \|_{\P^2,2} \lesssim  \bbE^{1/2} \bigg[ \sup_{ \| \bbeta - \bbeta^{\ast}\| \lesssim \tilde{r}_{n,\bbeta}   } \|  \Gamma (Y,T,\bbeta) - \Gamma (Y,T,\bbeta^*) \|^2 \bigg]  \lesssim (\tilde{r}_{n,\bbeta})^{\alpha_{0}} ,
\end{align*}
where the last inequality holds by  Assumptions~\ref{as:criterion}.
 Applying Theorem~2.14.1 of \citet{van2023weak} to  $\calH_{1,\delta}$, we obtain 
 \begin{align}\label{eq:supH1delta}
      \bbE \left[ \left\|  \bbU_{n,1} (h) - \P^2 (h)\right\|_{\calH_{1,\delta}} \right] \lesssim \frac{J_{1}(1, \calH_{1,\delta}, \H_{1,\delta})\cdot \|\H_{1,\delta}\|_{\P,2}}{\sqrt{n}}\lesssim \frac{ (\tilde{r}_{n,\bbeta})^{\alpha_{0}} }{\sqrt{n}}.
\end{align}
Similarly, we define the following function class
\begin{align}\label{eq:def-H2delta}
\calH_{2,n,\delta} \coloneqq & \left\{ \begin{array}{c} h(O_1,O_2,O_3,\bbeta)= \IF_{\mathrm{B}_{\psi, k}}^{(2)} (\bbeta) 
 - \IF_{\mathrm{B}_{\psi, k}}^{(2)} (\bbeta^*) : 
 \bbeta \in \calB,  \|\bbeta -\bbeta^{\ast}  \|\lesssim \tilde{r}_{n,\bbeta}
 \end{array} \right\}.
\end{align}
Using arguments analogous to those used for $\calH_{1,\delta}$,  the function class is also VC type with bounded characteristics, and its envelope function $ \H_{2,n,\delta} $ is defined as
	\begin{align*}
\H_{2,n,\delta} (O_1,O_2,O_3) \coloneqq&\  \left| \left\{ \hat{\xi} (X_{1}, T_{1}) \bar{\phi}_{k} (X_{1}, T_{1}) - \sfzbar_{k_{x}} (X_{1}) \otimes  \sfwbar_{k_{t}} (T_{3}) \right\}^{\top} \Sigma_{k}^{-1} \bar{\phi}_{k} (X_{2}, T_{2}) \right| \\
&\qquad\qquad\qquad\qquad\times\sup_{\|\bbeta -\bbeta^{\ast}  \|\lesssim \tilde{r}_{n,\bbeta} } \left\| \Gamma (Y_2,T_2,\bbeta) - \Gamma (Y_2,T_2,\bbeta^*)  \right\| .
	\end{align*}
 Note that, $\|\xi (\cdot,\cdot)\|_{\infty} \lesssim 1$ and $\|\hat{\xi} (\cdot,\cdot) \|_{\infty}\lesssim 1$ under Assumption~\ref{as:nuisance_est}.  Using Assumptions~\ref{as:criterion} and \ref{as:basis}, we have $\|\P^2 \H_{2,n,\delta} (\cdot) \|_{\infty} \lesssim \| \H_{2,n,\delta} (\cdot,\cdot,\cdot)\|_{\infty} \lesssim \zeta(k)^{2} \lesssim k$ and 
\begin{align*}
    & \ \bbE \left[ \| \H_{2,n,\delta}(O_1,O_2,O_3)\|^2\right]\\ 
   \lesssim & \ \bbE \bigg[ \bar{\phi}_{k} (X_{2}, T_{2})^{\top} \Sigma_{k}^{-1} 
 \bar{\phi}_{k} (X_{2}, T_{2}) \sup_{\|\bbeta -\bbeta^{\ast}  \|\leq \tilde{r}_{n,\bbeta} } \left\| \Gamma (Y_2,T_2,\bbeta) - \Gamma  (Y_2,T_2,\bbeta^*) \right\|^2\bigg] \\
 \lesssim & \ \zeta(k)^{2} \cdot \bbE \bigg[ \sup_{\|\bbeta -\bbeta^{\ast}  \|\lesssim \tilde{r}_{n,\bbeta} } \left\|  \Gamma (Y_2,T_2,\bbeta) - \Gamma  (Y_2,T_2,\bbeta^*) \right\|^2\bigg]\\
 \lesssim & \ k \cdot (\tilde{r}_{n,\bbeta})^{2 \alpha_{0}}  ,
\end{align*}
which implies $\|\H_{2,n,\delta}\|_{\P^3,2} \lesssim \sqrt{k}\cdot (\tilde{r}_{n,\bbeta})^{ \alpha_{0}}$. 

For any $h\in\calH_{2,n,\delta}$, the explicit expression of $\P^2 h$  is given by 
\begin{align*}
 \P^2 h (O) = & \ \frac{1}{3} \hat{\xi} (X, T) \Pi[\Gamma_{\bbeta}  - \Gamma_{\bbeta^{\ast}} |\bar{\phi}_{k}](X, T)  
    - \frac{1}{3} \int_{t\in\calT} \Pi[ \Gamma_{\bbeta}  - \Gamma_{\bbeta^{\ast}} |\bar{\phi}_{k}](X, t)  \p_{T} (t) \diff t\\ 
    &+\frac{1}{3}  \Pi[\hat{\xi} - \xi| \bar{\phi}_{k}](X,T)  \left(\Gamma_{\bbeta}  - \Gamma_{\bbeta^{\ast}}
	\right)
 -\frac{1}{3}\int_{x\in\calX} \Pi[ \Gamma_{\bbeta}  - \Gamma_{\bbeta^{\ast}} |\bar{\phi}_{k}](x, T)  \p_{X} (x) \diff x .
\end{align*}
Since  $\|\xi (\cdot,\cdot) \|_{\infty} \lesssim 1$, $\|\hat{\xi} (\cdot,\cdot) \|_{\infty}\lesssim 1$ and projections do not increase length, using Jensen's inequality,  we obtain 
\begin{align*}
 \sup_{h\in \calH_{2,n,\delta}} \| \P^2 h \|_{\P,2} \lesssim \bbE^{1/2} \bigg[ \sup_{ \| \bbeta - \bbeta^{\ast}\| \lesssim \tilde{r}_{n,\bbeta}   } \|  \Gamma (Y,T,\bbeta) - \Gamma (Y,T,\bbeta^*) \|^2 \bigg] \lesssim (\tilde{r}_{n,\bbeta})^{\alpha_{0}} ,
\end{align*}
where the last inequality holds by Assumption~\ref{as:criterion}.
By Assumptions~\ref{as:nuisance_est}, \ref{as:criterion}, and \eqref{eq:supnorm_proj}, we have $$\sup_{\bbeta } \left\| \Pi [\Gamma_{\bbeta} | \bar{\phi}_{k}] (\cdot, \cdot)\right\|_{\infty} \leq  \sup_{\bbeta } \left\| \Gamma_{\bbeta} (\cdot,\cdot)  \right\|_{\infty} \lesssim 1,\ \text{and}\ \left\| \Pi[\hat{\xi} - \xi| \bar{\phi}_{k}] (\cdot,\cdot) \right\|_{\infty} \lesssim  \|(\hat{\xi} - \xi) (\cdot,\cdot)\|_{\infty} \lesssim 1 ,$$
which implies $ \sup_{h\in \calH_{2,n,\delta}} \| \P^2 h (\cdot)\|_{\infty} \lesssim 1 $. 
Hence, applying Lemma~\ref{lemma:Umaximal_inequality_order3}, we obtain
\begin{align}\label{eq:supH2delta}
    \bbE \left[ \left\|  \bbU_{n,3} (h)  - \P^3 h \right\|_{\calH_{2,n,\delta}} \right] \lesssim \frac{\sqrt{k}\cdot (\tilde{r}_{n,\bbeta})^{\alpha_{0}} \log n + \log n }{n} + \ \frac{ (\tilde{r}_{n,\bbeta})^{\alpha_{0}} \sqrt{\log n } }{\sqrt{n}}.
\end{align}

 Hence, combining \eqref{eq:hatpsi-tildepsi-delta}, \eqref{eq:supH1delta}, and \eqref{eq:supH2delta},  we obtain
\begin{align*}
    & \ \bbE \left[ \sup_{\|\bbeta-\bbeta^{\ast}\|\lesssim \tilde{r}_{n,\bbeta} }  \frac{n}{\sqrt{k+n}} \cdot \left\| \left( \bar{\psi}_{k} (\bbeta) -  \tilde{\psi}_{k}^{(2)} (\bbeta) \right) - \left( \bar{\psi}_{k} (\bbeta^{\ast})  -   \tilde{\psi}_{k}^{(2)} (\bbeta^{\ast})  \right)  \right\| \right] \\
    \lesssim & \ \frac{n}{\sqrt{k+n}} 
 \cdot \left[   \frac{ 1 }{\sqrt{n}}+  + \frac{ \sqrt{k}    }{ n }  \right]\cdot (\tilde{r}_{n,\bbeta})^{\alpha_{0}} \log n \\
 = & \ \left( \sqrt{\frac{n}{k+n}} + \sqrt{\frac{k}{k+n}} \right)\cdot (\tilde{r}_{n,\bbeta})^{\alpha_{0}} \log n =o_{\P} (1),
\end{align*}
where  the last equality holds since we assume $(\tilde{r}_{n,\bbeta})^{\alpha_{0}}\log n \rightarrow 0 $. Hence, we have shown that \eqref{eq:uprocess} is $o (1)$, which completes the proof.
\end{proof}

\subsection{Technical Lemmas for Results in Section \ref{sec:properties}}

\begin{lemma}\label{eq:var-F3}
Suppose that  Assumptions~ \ref{as:nuisance_est}--\ref{as:basis} hold. 
Then, for a given $\bbeta\in\calB$, conditioning on the nuisance sample, 
	\begin{equation*}
	\left\| \cov \left[ \bbU_{n, 3} \left( \IF_{\mathrm{B}_{\psi, k}}^{(2)} (\bbeta) \right) \right] \right\|_{\op}^{2} =O\left(\frac{1}{n}+\frac{k}{n^2}\right),
	\end{equation*}
    where $\IF_{\mathrm{B}_{\psi, k}}^{(2)} (\bbeta)$ is defined in Lemma~\ref{lem:HOIF}.
\end{lemma}
\begin{proof}
    Taking a unit-norm vector $v \in \bbR^{q}$ ($\Vert v \Vert = 1$), we need to show
    \begin{align*}
        v^{\top} \cov\left[ \bbU_{n, 3} \left( \IF_{\mathrm{B}_{\psi, k}}^{(2)} (\bbeta) \right) 
 \right] v = \var \left[ \bbU_{n, 3} \left( v^{\top}\IF_{\mathrm{B}_{\psi, k}}^{(2)} (\bbeta) \right)  \right] = O\left(\frac{1}{n}+\frac{k}{n^2}\right).
    \end{align*}
    By using Lemma 15 in \cite{liu2017semiparametric}, 
    \begin{align*}
        & \ \var \left[ \bbU_{n, 3} \left( v^{\top}\IF_{\mathrm{B}_{\psi, k}}^{(2)} (\bbeta) \right)  \right]  = \var \left[ \bbU_{n,3} \left\{ S_{3} \left( v^{\top}\IF_{\mathrm{B}_{\psi, k}}^{(2)} (\bbeta) \right) \right\} \right] \\
        \lesssim & \ \sum_{j=1}^{3} \frac{1}{n^{j}} \bbE \left[ \left\{ \bbE \left[  S_{3} \left( v^{\top}\IF_{\mathrm{B}_{\psi, k}}^{(2)} (O_1,O_2,O_3;\bbeta)  \right) \mid O_1, \cdots, O_j\right] \right\}^2\right] \eqqcolon \sum_{j=1}^{3} \zeta_{j},
    \end{align*}
    where $S_3$ is the symmetrization operator defined in Section~\ref{app:notation}. 
For $\zeta_{1}$, first note that 
\begin{align*}
	& \ \bbE [\IF_{\mathrm{B}_{\psi, k}}^{(2)} (O_1,O_2,O_3;\bbeta)^{\top} v \mid O_1]\\
 = & \ \left\{ \hat{\xi} (X_{1}, T_{1}) \bar{\phi}_{k} (X_{1}, T_{1})-\sfzbar_{k_{x}} (X_{1}) \otimes  \bbE\left[ \sfwbar_{k_{t}} (T) \right]  
	\right\}^{\top} \Sigma_k^{-1} \bbE \left[\bar{\phi}_{k} (X, T) \left\{\hat{b}_{\bbeta} (X, T) - b_{\bbeta}(X, T) \right\}^{\top} v\right]  \\
	= & \ \hat{\xi} (X_{1}, T_{1}) \Pi[\{\hat{b}_{\bbeta}-b_{\bbeta}\}^{\top} v|\bar{\phi}_{k}] (X_1,T_1) -\int_{t\in\calT} \Pi[\{\hat{b}_{\bbeta}-b_{\bbeta}\}^{\top} v|\bar{\phi}_{k}] (X_1,t)  \p_{T} (t) \diff t.
\end{align*}
Then
\begin{align*}
   & \ \bbE \left[ \left|  \hat{\xi} (X_{1}, T_{1}) \Pi[\{\hat{b}_{\bbeta}-b_{\bbeta}\}^{\top} v|\bar{\phi}_{k}] (X_1,T_1) \right|^2\right]\\
    \lesssim & \ \bbE \left[ \left| \Pi[\{\hat{b}_{\bbeta}-b_{\bbeta}\}^{\top} v|\bar{\phi}_{k}] (X_1,T_1) \right|^2\right] \lesssim  \bbE \left[ \left| \{\hat{b}_{\bbeta}-b_{\bbeta}\}^{\top} v \right|^2 \right] \lesssim 1,
\end{align*}
where  the second inequality follows by  the $L_{2}$-norm contraction property of $L_{2}$-projection and the last inequality follows by Assumption~\ref{as:nuisance_est}. Similarly, by Jensen's inequality, we have
\begin{align*}
     & \ \bbE \left[ \left|  \int_{\calT} \Pi[\{\hat{b}_{\bbeta}-b_{\bbeta}\}^{\top} v|\bar{\phi}_{k}] (X_1,t)  \p_{T} (t) \diff t \right|^2\right]\\
     \leq & \iint_{\calX\times\calT} \left| \Pi[\{\hat{b}_{\bbeta}-b_{\bbeta}\}^{\top} v|\bar{\phi}_{k}] (x,t) \right|^2  \p_{T} (t)   \p_{X} (x) \diff t \diff x\\
     = & \ \bbE \left[  \xi (X_{1}, T_{1})  \left| \Pi[\{\hat{b}_{\bbeta}-b_{\bbeta}\}^{\top} v|\bar{\phi}_{k}] (X_1,T_1) \right|^2\right] \\
     \lesssim & \ \bbE \left[ \left| \Pi[\{\hat{b}_{\bbeta}-b_{\bbeta}\}^{\top} v|\bar{\phi}_{k}] (X_1,T_1) \right|^2\right]  \lesssim 1.
\end{align*}
Hence,
\begin{align*}
    \bbE \left[ \left\{ \bbE [\IF_{\mathrm{B}_{\psi, k}}^{(2)} (O_1,O_2,O_3;\bbeta)^{\top} v \mid O_1] \right\}^2\right] = O(1).
\end{align*}
By symmetry, we also have
\begin{align*}
   \bbE \left[ \left\{ \bbE [\IF_{\mathrm{B}_{\psi, k}}^{(2)} (O_1,O_2,O_3;\bbeta)^{\top} v \mid O_2] \right\}^2\right]  = O(1),\\
   \bbE \left[ \left\{ \bbE [\IF_{\mathrm{B}_{\psi, k}}^{(2)}(O_1,O_2,O_3;\bbeta)^{\top} v \mid O_3] \right\}^2\right]  = O(1).
\end{align*}
 Taking these together implies $\zeta_1 \lesssim \frac{1}{n}$.

For $\zeta_2$, first note that 
\begin{align*}
    	&\bbE [\IF_{\mathrm{B}_{\psi, k}}^{(2)} (O_1,O_2,O_3;\bbeta)^{\top} v \mid O_1,O_2]\\
     =&\big\{ \hat{\xi} (X_{1}, T_{1}) \bar{\phi}_{k} (X_{1}, T_{1})-\sfzbar_{k_{x}} (X_{1}) \otimes  \int_{\calT} \sfwbar_{k_{t}} (t)  \p_{T} (t) \diff t
	\big\}^{\top} \Sigma_k^{-1}\bar{\phi}_{k} (X_{2}, T_{2}) \left\{\hat{b}_{\bbeta} (X_{2}, T_{2}) - \Gamma (Y_2,T_2,\bbeta)\right\}^{\top} v.
\end{align*}
By Assumptions~\ref{as:nuisance_est}--\ref{as:basis}, we have 
   \begin{align*}
         & \ \bbE \left[ \left|\left\{ \sfzbar_{k_{x}} (X_{1}) \otimes  \int_{\calT} \sfwbar_{k_{t}} (t)   \p_{T} (t) \diff t \right\}^{\top} \Sigma_k^{-1}\bar{\phi}_{k} (X_{2}, T_{2}) \right|^2 v^{\top} \left\{\hat{b}_{\bbeta} (X_{2}, T_{2}) - \Gamma (Y_2,T_2,\bbeta)\right\}^{\otimes 2} v \right] \\
    \lesssim & \ \bbE \left[\left\{\sfzbar_{k_{x}} (X_{1}) \otimes  \int_{\calT} \sfwbar_{k_{t}} (t)   \p_{T} (t) \diff t
    \right\}^{\top} \Sigma_k^{-1}\left\{ \sfzbar_{k_{x}} (X_{1}) \otimes  \int_{\calT} \sfwbar_{k_{t}} (t)  \p_{T} (t) \diff t
    \right\} \right]\\
    \lesssim & \ \lambda_{\max}(\Sigma_k^{-1})\bbE\left[\left( \sfzbar_{k_{x}} (X_{i_{1}}) \otimes \int_{\calT} \sfwbar_{k_{t}} (t) \p_{T} (t) \diff t \right)^{\top} \left( \sfzbar_{k_{x}} (X_{i_{1}}) \otimes \int_{\calT} \sfwbar_{k_{t}} (t) \p_{T} (t) \diff t \right)\right]\\
    \lesssim & \ \lambda_{\max}(\Sigma_k^{-1}) \bbE\left[\sfzbar_{k_{x}} (X)^{\top} \sfzbar_{k_{x}} (X)\right]\bbE\left[ \sfwbar_{k_{t}} (T)^{\top}\sfwbar_{k_{t}} (T)  \right] \lesssim k_{x}k_{t}= k.
    \end{align*}
where $\lambda_{\max}(\Sigma_k^{-1})$  is the largest eigenvalue of matrix $\Sigma_k^{-1}$. Similarly, we derive 
   \begin{align*}
         &\ \bbE \left[ \left|\bar{\phi}_{k} (X_{1}, T_{1})
    ^{\top} \Sigma_k^{-1}\bar{\phi}_{k} (X_{2}, T_{2}) \right|^2 v^{\top} \left\{\hat{b}_{\bbeta} (X_{2}, T_{2}) - \Gamma (Y_2,T_2,\bbeta)\right\}^{\otimes 2} v \right]\\
    \lesssim & \ \bbE \left[ \bar{\phi}_{k} (X_{1}, T_{1})^{\top} \Sigma_k^{-1} \bar{\phi}_{k} (X_{1}, T_{1}) \right]\\
   \lesssim & \ \lambda_{\max}(\Sigma_k^{-1})  \bbE \left[ \bar{\phi}_{k} (X, T)^{\top}\bar{\phi}_{k} (X, T) \right] \\
	= & \ \lambda_{\max}(\Sigma_k^{-1})  \iint_{\calX\times\calT} \xi(x, t)^{-1}  \left( \sfzbar_{k_{x}} (x)^{\top} \sfzbar_{k_{x}} (x) \right)\otimes \left(  \sfwbar_{k_{t}} (t)^{\top}  \sfwbar_{k_{t}} (t) \right)  \p_{T} (t) \p_{X} (x) \diff t \diff x\\
	\lesssim & \ \lambda_{\max}(\Sigma_k^{-1})\bbE\left[\sfzbar_{k_{x}} (X)^{\top} \sfzbar_{k_{x}} (X)\right]\bbE\left[ \sfwbar_{k_{t}} (T)^{\top}\sfwbar_{k_{t}} (T)  \right]\lesssim k.
    \end{align*}
Thus 
\begin{align*}
    \bbE \left[\left| \bbE [\IF_{\mathrm{B}_{\psi, k}}^{(2)} (O_1,O_2,O_3;\bbeta)^{\top} v \mid O_1,O_2] \right|^2\right] = O(k).
\end{align*}
Similarly, we have
\begin{align*}
        \bbE \left[\left| \bbE [\IF_{\mathrm{B}_{\psi, k}}^{(2)} (O_1,O_2,O_3;\bbeta)^{\top} v \mid O_1,O_3] \right|^2\right] = O(k),\\
            \bbE \left[\left| \bbE [\IF_{\mathrm{B}_{\psi, k}}^{(2)} (O_1,O_2,O_3;\bbeta)^{\top} v \mid O_2,O_3] \right|^2\right] = O(k).
\end{align*}
    Hence,
    \begin{align*}
        \zeta_2\lesssim \frac{k}{n^2}.
    \end{align*}
    
    For $\zeta_{3}$, we have
    \begin{align*}
        & \ \zeta_{3} \lesssim \frac{1}{n^3} \bbE \left[ \left( \IF_{\mathrm{B}_{\psi, k}}^{(2)} (O_1,O_2,O_3;\bbeta)^{\top} v \right) ^2\right] \\
        = & \ \frac{1}{n^3}\bbE \left[ \left| \left\{ \hat{\xi} (X_{1}, T_{1}) \bar{\phi}_{k} (X_{1}, T_{1})-\sfzbar_{k_{x}} (X_{1}) \otimes  \sfwbar_{k_{t}} (T_{3})
	\right\}^{\top} \Sigma_k^{-1}\bar{\phi}_{k} (X_{2}, T_{2}) \right|^2 \right.\\
    &\qquad\times\left.v^{\top} \left\{\hat{b}_{\bbeta} (X_{2}, T_{2}) - \Gamma (Y_2,T_2,\bbeta)\right\}^{\otimes 2} v \right] \\
	\lesssim & \ \frac{1}{n^3}\bbE \left[\left\{ \hat{\xi}  (X, T) - \xi  (X, T) \right\}^2 \bar{\phi}_{k} (X, T)^{\top} \Sigma_k^{-1} \bar{\phi}_{k} (X, T) \right]\\
	\lesssim & \ \frac{1}{n^3} \lambda_{max}(\Sigma_k^{-1})  \bbE \left[ \bar{\phi}_{k} (X, T)^{\top}\bar{\phi}_{k} (X, T) \right] \lesssim \frac{k}{n^3}.
    \end{align*}
    Combining the above analyses, we complete the proof.  
\end{proof}

Our next result characterizes the  asymptotic distribution of $\tilde{\psi}_{k}^{(2)} (\bbeta^{\ast})$ (defined in  Section~\ref{app:properties}). Before stating the lemma, recall the definitions in \eqref{eq:V1} and \eqref{eq:V2}:
\begin{align*}
	V_1 ( \bbeta^{\ast} ) &\coloneqq \lim\limits_{n \to \infty} \frac{n}{n+k} \bbE \left[ \left\{ 2 \pi_1 
 \hat{\Xi}_{\bbeta^{\ast}}  + 3 \pi_1  \IF_{\mathrm{B}_{\psi, k}}^{(2)} (\bbeta^*)  \right\}^{\otimes 2} \right]\\
	V_2 (\bbeta^*) &\coloneqq  \lim\limits_{n \to \infty} \frac{1}{k + n} \bbE \left[ \left\{ \pi_2 
 \hat{\Xi}_{\bbeta^{\ast}}  +  3\pi_{2} \IF_{\mathrm{B}_{\psi, k}}^{(2)} (\bbeta^*) \right\} ^{\otimes 2} \right].
\end{align*}
\begin{lemma}\label{th:psi-dist}
Conditioning on the nuisance sample, suppose that Assumptions~\ref{as:nuisance_est}--\ref{as:basis} hold.\\
 If $k \ll n$, 
	\begin{align*}
		\sqrt{n}\left\{ \tilde{\psi}_{k}^{(2)} (\bbeta^{\ast}) -\bar{\psi}_{k} (\bbeta^{\ast}) \right\} = & \ \sqrt{n} \cdot \bbU_{n,1} \left(  2 \pi_1 
 \hat{\Xi}_{\bbeta^{\ast}}  + 3 \pi_1  \IF_{\mathrm{B}_{\psi, k}}^{(2)} (\bbeta^{*})  \right) +o_{\P} (1)
  \stackrel{d}{\to}  \ N(0,V_1 ( \bbeta^{\ast} )).
	\end{align*}
If $k \gg n$, 
	\begin{align*}
		 \sqrt{\frac{ n^2}{ k } } \left\{ \tilde{\psi}_{k}^{(2)} (\bbeta^{\ast}) - \bar{\psi}_{k} (\bbeta^{\ast}) \right\} = \sqrt{\frac{ n^2}{ k } } \cdot   \bbU_{n,2} \left( \pi_2 
 \hat{\Xi}_{\bbeta^{\ast}}  +  3\pi_{2} \IF_{\mathrm{B}_{\psi, k}}^{(2)} (\bbeta^{*}) \right) +o_{\P} (1)\stackrel{d}{\to} N(0,V_2 ( \bbeta^{\ast} ) /2).
	\end{align*}
If $k / n \rightarrow \tau$ for some $\tau \in (0, \infty)$, 
	\begin{align*}
		\sqrt{n}\left\{ \tilde{\psi}_{k}^{(2)} (\bbeta^{\ast}) - \bar{\psi}_{k} (\bbeta^{\ast})  \right\}&= \sqrt{n} \cdot \bbU_{n,1} \left(  2 \pi_1 
 \hat{\Xi}_{\bbeta^{\ast}}  + 3 \pi_1  \IF_{\mathrm{B}_{\psi, k}}^{(2)} (\bbeta^{*})  \right) \\
  &\qquad\qquad\qquad+ \sqrt{n} \cdot \bbU_{n,2} \left( \pi_2 
 \hat{\Xi}_{\bbeta^{\ast}}  +  3\pi_{2} \IF_{\mathrm{B}_{\psi, k}}^{(2)} (\bbeta^{*})  \right)+o_{\P} (1)\\
		&\stackrel{d}{\to} N\left(0,(1+\tau)\left\{ V_{1} ( \bbeta^{\ast} )  + \frac{1}{2}  V_2 ( \bbeta^{\ast} ) \right\} \right).
	\end{align*}Here, $\Pi^{\perp}[f \mid \bar{\phi}_{k}]$ denotes the orthocomplement of the projection of $f$ onto the linear span of $\bar{\phi}_{k}$, as defined in Section~\ref{sec:notation}.  $\tilde{\psi}_{k}^{(2)} (\bbeta^{\ast}) $ and $ \bar{\psi}_{k} (\bbeta^{\ast})$  are defined in Section~\ref{app:properties}.
\end{lemma}

\begin{remark}\label{rem:betainit dependence}
Since this lemma is established conditional on the nuisance sample, and both $\hat{\Xi}_{\bbeta^*}$ and $\IF_{\mathrm{B}_{\psi, k}}^{(2)}$ depend on the initial estimator $\bbeta_{\ini}$, the asymptotic distributions of $\tilde{\psi}_{k}^{(2)} (\bbeta^{\ast}) - \bar{\psi}_{k} (\bbeta^{\ast})$ established in this lemma also implicitly depend on $\bbeta_{\ini}$. 
\end{remark}

\begin{remark}
\label{rem:bounded}
In Assumptions~\ref{as:bounded} and \ref{as:criterion}, we have assumed that the observations are bounded. For the CLT established in Lemma \ref{th:psi-dist}, the boundedness assumption can be relaxed if one uses Corollary 1.4 in \citet{bhattacharya1992class} instead of their Corollary 1.5. The former only requires certain moment conditions to be satisfied, avoiding making the stronger boundedness assumption. We decide to stick to the boundedness assumption, following the related literature on HOIFs \citep{liu2017semiparametric, robins2023minimax}.
\end{remark}

\begin{proof}
Take a unit-norm vector $v \in \bbR^{q}$ ($\Vert v \Vert = 1$) and define
\begin{gather*}
g_n ( O_1 ) \coloneqq \frac{ a^{\top} \left\{ 2 \pi_1 
 \hat{\Xi}_{\bbeta^{\ast}}  + 3 \pi_1  \IF_{\mathrm{B}_{\psi, k}}^{(2)} (\bbeta^{*})  \right\} }{\sqrt{k+n}},\\
 h_n ( O_1 , O_2 ) \coloneqq \frac{ a^{\top} \left\{ \pi_2 
 \hat{\Xi}_{\bbeta^{\ast}}  +  3\pi_{2} \IF_{\mathrm{B}_{\psi, k}}^{(2)} (\bbeta^{*})  \right\} }{n\sqrt{k+n}}.
\end{gather*}
Then, $\bbE [ g_n (O_1) ]=0 $, and $ h_n ( O_1 , O_2 )$ is degenerate in view of the definition of $\pi_{2}$ given in \eqref{hajek} in Section~\ref{app:notation}, and 
\begin{align*}
	&\lim\limits_{n \to \infty} n^2 \bbE [ h_n ( O_1 , O_2 )h_n ( O_1 , O_2 )^{\top} ] = a^{\top} V_{2} ( \bbeta^{\ast} ) a ,\\
&   \lim\limits_{n \to \infty} n \bbE [ g_n  ( O_1 )g_n(O_1)^{\top} ] =  a^{\top} V_{1} ( \bbeta^{\ast} ) a .
\end{align*}
By Corollary 1.5 in \citet{bhattacharya1992class}, if  the following conditions hold
\begin{subequations}
\begin{align}
        &\sup_{o} |g_n (o)|\to 0, \label{eq:normality cond1} \\
	&\sup_{o_1,o_2}  | h_n ( o_1 , o_2 ) |  \to 0, \label{eq:normality cond2} \\
	&n \sup_{o_1} \bbE | h_n ( o_1 , O_2 ) | \to 0 , \label{eq:normality cond3}
\end{align}
\end{subequations}
then we obtain the following central limit theorem jointly for $\sum_{j} g_n (O_j)$ and $\sum_{1 \leq i_{1} \neq i_{2} \leq n} h_n  ( O_{i_1} , O_{i_2} )$:
\begin{equation*}
	\left( \begin{matrix} \sum_{j} g_n (O_j) \\
    \sum_{1 \leq i_{1} \neq i_{2} \leq n} h_n  ( O_{i_1} , O_{i_2} )
    \end{matrix} \right) \stackrel{d}{\to}  N \left( \left( \begin{matrix}
    0 \\
    0
    \end{matrix} \right), \left[ \begin{matrix}
	 a^{\top} V_{1} ( \bbeta^{\ast} ) a & 0\\
		0 & \frac{1}{2} a^{\top} V_{2} ( \bbeta^{\ast} ) a
	\end{matrix}
	\right]\right).
\end{equation*}
By the Hoeffding decomposition (recall from Section~\ref{app:notation}), we have
\begin{align*}
    \tilde{\psi}_{k}^{(2)} (\bbeta^{\ast}) - \bar{\psi}_{k} (\bbeta^{\ast}) = & \ \bbU_{n,3} \pi_3  \IF_{\mathrm{B}_{\psi, k}}^{(2)} (\bbeta^{*}) + 3 \bbU_{n,2} \pi_2  \IF_{\mathrm{B}_{\psi, k}}^{(2)} (\bbeta^{*}) \\
    & + 3 \bbU_{n,1} \pi_1  \IF_{\mathrm{B}_{\psi, k}}^{(2)} (\bbeta^{*}) + \bbU_{n,2} \pi_2\hat{\Xi}_{\bbeta^{\ast}} + 2\bbU_{n,1} \pi_1 \hat{\Xi}_{\bbeta^{\ast}}.
\end{align*}
Then, in view of ${n}/{\sqrt{k+n}} \cdot \bbU_{n, 3}\pi_3  \IF_{\mathrm{B}_{\psi, k}}^{(2)} (\bbeta^{*})=o_{\P} (1)$, we have
\begin{align*}
    \frac{n}{\sqrt{k+n}} \cdot a^{\top} \left\{\tilde{\psi}_{k}^{(2)} (\bbeta^{\ast}) - \bar{\psi}_{k} (\bbeta^{\ast}) \right\} 
    = \sum_{1 \leq i_{1} \neq i_{2} \leq n} h_n  ( O_{i_1} , O_{i_2} )+\sum_{j} g_n (O_j) + o_{\P}(1).
\end{align*}
If $k \ll n$, then $V_2= 0 $ and ${n}/{\sqrt{k+n}} = \sqrt{n} + o(1)$, which implies that
\begin{align*}
	\sqrt{n} \cdot a^{\top} \left\{ \tilde{\psi}_{k}^{(2)} (\bbeta^{\ast}) - \bar{\psi}_{k} (\bbeta^{\ast}) 
 \right\}
	=& \sum_{j} g_n (O_j) +o_{\P} (1)\stackrel{d}{\to} N \left(0, 
a^{\top} V_1 ( \bbeta^{\ast} ) a \right).
\end{align*}
If $k \gg n$, then $V_1= 0 $ and ${n}/{\sqrt{k+n}} = n/\sqrt{k} +o(1) $, which implies that
\begin{align*}
	\frac{n}{\sqrt{k} } \cdot a^{\top} \left\{ \tilde{\psi}_{k}^{(2)} (\bbeta^{\ast}) - \bar{\psi}_{k} (\bbeta^{\ast})  \right\}
	= & \ \sum_{1 \leq i_{1} \neq i_{2} \leq n} h_n  ( O_{i_1} , O_{i_2} ) +o_{\P} (1)\stackrel{d}{\to} N \left(0, a^{\top}V_2 ( \bbeta^{\ast} ) a /2 \right).
\end{align*}
If $k = \tau n$ for some constant $\tau$, then ${n}/{\sqrt{k+n}} = \sqrt{n}/\sqrt{1+\tau}$ and
\begin{align*}
	\sqrt{n} \cdot a^{\top}\left\{ \tilde{\psi}_{k}^{(2)} (\bbeta^{\ast})  - \bar{\psi}_{k} (\bbeta^{\ast})  \right\}
	= & \ \sqrt{1+\tau} \left\{ \sum_{j} g_n (O_j) + \sum_{1 \leq i_{1} \neq i_{2} \leq n} h_n  ( O_{i_1} , O_{i_2} )  \right\} +o_{\P} (1)\\
 \stackrel{d}{\to} & \ N \left( 0, (1+\tau)a^{\top} \left\{ V_{1} ( \bbeta^{\ast} )  + \frac{1}{2}  V_2 ( \bbeta^{\ast} ) \right\} a \right).
\end{align*}
It remains to verify whether Conditions~\eqref{eq:normality cond1}--\eqref{eq:normality cond3} hold.

\noindent\textbf{\emph{Condition~\eqref{eq:normality cond1}}:} Using \eqref{eq:pi1_IF}, Condition~\eqref{eq:normality cond1} holds by Assumption~\ref{as:nuisance_est} and Assumption~\ref{as:criterion}, and the result in \eqref{eq:supnorm_proj}.

\noindent\textbf{\emph{Condition~\eqref{eq:normality cond2}}:} By Assumption~\ref{as:nuisance_est} and Assumption~\ref{as:criterion}, we have
\begin{align*}
    &\sup_{o_1,o_2} \left\vert a^{\top}  \pi_{2}  \hat{\Xi}_{\bbeta^{\ast}}  ( o_1,o_2 )\right\vert \leq  \sup_{o_1,o_2} \left\vert a^{\top}  \hat{\Xi}_{\bbeta^{\ast}}  ( o_1,o_2 )\right\vert \lesssim 1
\end{align*}
and
\begin{align*}
	&\sup_{o_1,o_2} \left\vert a^{\top}  \pi_{2}  \IF_{\mathrm{B}_{\psi, k}}^{(2)}   ( o_1,o_2; \bbeta^* )\right\vert \\
 \lesssim &	\sup_{o_1,o_2,o_3} \left\vert 
a^{\top}  \IF_{\mathrm{B}_{\psi, k}}^{(2)}(o_1,o_2,o_3; \bbeta^*)\right\vert \\
	\lesssim & \ \left\|(\hat{b}_{\bbeta_{\ini}}  - \Gamma_{\bbeta^{\ast}}) (\cdot,\cdot) \right\|_{\infty}\cdot \sup_{o_1,o_2,o_3}\left| \left\{ \hat{\xi}(x_{1}, t_{1}) \bar{\phi}_{k} (x_{1}, t_{1}) -  \sfzbar_{k_{x}} (x_1) \otimes  \sfwbar_{k_{t}} (t_3)  \right\}^{\top} \Sigma_k^{-1}  \bar{\phi}_{k} (x_{2}, t_{2})\right|\\
	\lesssim & \  \zeta(k)^2,
\end{align*}
where the last inequality holds by Assumption~\ref{as:basis}. Thus we have 
\begin{align*}
	\sup_{o_1,o_2}  | h_n ( o_1 , o_2 ) | \lesssim \frac{ \zeta(k)^2}{ n \sqrt{k+n}} \to 0.
\end{align*}

\noindent\textbf{\emph{Condition \eqref{eq:normality cond3}}:}  
Denote $\bar{\phi}_{k} \coloneqq \{\phi_{1}, \cdots, \phi_{k}\}$. By Assumption~\ref{as:basis}, there are at most $c_8$ ($c_8$ is a positive integer that does not depend on $n$) nonzero terms in the sum $ \sum_{i=1}^{k} \left|\phi_{i} (x, t) \phi_{i} (x_{2}, t_{2}) \right| \cdot \left| a^{\top} \left\{ \hat{b}_{\bbeta_{\ini}} (x_{2}, t_{2}) - \Gamma (y_{2}, t_{2}, \bbeta^{\ast}) \right\} \right|$ for any $(o,o_2)\in\calO^2$ and all nonzero terms are bounded by $$\zeta(k)^2 \left\| a^{\top} \left\{ \hat{b}_{\bbeta_{\ini}} (\cdot, \cdot)  - \Gamma (\cdot, \cdot,\bbeta^{\ast}) \right\} \right\|_{\infty}.$$  
In addition, for any $i$, the  probability of the set of $o_2 \in \calO$ such that  $|\phi_{i} (x, t)  \phi_{i} (x_{2}, t_{2}) | \cdot | a^{\top} \{ \hat{b}_{\bbeta_{\ini}} (x_{2}, t_{2}) - \Gamma (y_{2}, t_{2},\bbeta^{\ast}) \} |\neq 0$ is bounded by $1/k$.
Then, for any $(x,t)\in\mathcal{X}\times\mathcal{T}$, 
\begin{align*}
	& \ \bbE\left[\left| \bar{\phi}_{k} (x, t)^{\top} \Sigma_k^{-1} \bar{\phi}_{k} (X_{2}, T_{2})  a^{\top}  \left\{\hat{b}_{\bbeta_{\ini}} (X_{2}, T_{2}) - \Gamma (Y_2,T_2,\bbeta^*)\right\} \right|\right]\\
	\leq & \ \lambda_{\max}\left(\Sigma_{k}^{-1}\right)\bbE \left[\left|\bar{\phi}_{k} (x, t)^{\top} \bar{\phi}_{k} (X_{2}, T_{2}) \right| \cdot \left| a^{\top} \left\{ \hat{b}_{\bbeta_{\ini}} (X_{2}, T_{2}) - \Gamma (Y_2,T_2,\bbeta^*) \right\} \right|  \right]\\
    \leq & \ \lambda_{\max}\left(\Sigma_{k}^{-1}\right)\bbE \left[ \sum_{i=1}^{k}\left|\phi_{i} (x, t)\phi_{i} (X_{2}, T_{2}) \right| \cdot \left| a^{\top} \left\{ \hat{b}_{\bbeta_{\ini}} (X_{2}, T_{2}) - \Gamma (Y_2,T_2,\bbeta^*) \right\} \right|  \right]\\
    \leq & \ \frac{1}{k} \lambda_{\max}\left(\Sigma_{k}^{-1}\right) c_8 \cdot \sup_{(o,o_2)\in\calO^2} |\phi_{i} (x, t) \phi_{i} (x_{2}, t_{2})|  \cdot \left| a^{\top} \left\{ \hat{b}_{\bbeta_{\ini}} (x_{2}, t_{2}) - \Gamma (y_2,t_2,\bbeta^*) \right\} \right| \cdot \|\p_{X,T} (\cdot, \cdot) \|_{\infty} \\ 
	\leq & \ \frac{1}{k} \lambda_{\max}\left(\Sigma_{k}^{-1}\right) \zeta(k)^2\left\| a^{\top} (\hat{b}_{\bbeta_{\ini}}  - \Gamma_{\bbeta^{\ast}}) (\cdot, \cdot) \right\|_{\infty}=O(1),
\end{align*}
where $\lambda_{\max}(\Sigma_{k}^{-1})$ denotes the largest eigenvalue of $\Sigma_{k}^{-1}$ and the last equality holds by Assumptions~\ref{as:nuisance_est}--\ref{as:basis}. As a result of Assumption \ref{as:nuisance_est} and Jensen's inequality, 
\begin{align*}
	& \sup_{o} \bbE \left[\left|\big\{ \hat{\xi} (x, t) \bar{\phi}_{k} (x, t)-\sfzbar_{k_{x}} (x) \otimes  \int_{\calT} \sfwbar_{k_{t}} (u)   \p_T (u) \diff u
	\big\}^{\top} \Sigma_k^{-1}\bar{\phi}_{k} (X_{2}, T_{2}) a^{\top}\big\{\hat{b}_{\bbeta_{\ini}} (X_{2}, T_{2}) - \Gamma (Y_2,T_2,\bbeta^*)\big\} \right\|\right]\\
    & \lesssim \sup_{o}\bbE\left[\left|\left( \sfzbar_{k_{x}} (x) \otimes  \sfwbar_{k_{t}} (t) \right)^{\top} \Sigma_k^{-1} \bar{\phi}_{k} (X_{2}, T_{2})   a^{\top}\left\{\hat{b}_{\bbeta_{\ini}} (X_{2}, T_{2}) - \Gamma (Y_2,T_2,\bbeta^*) \right\} \right|\right]=O(1).
\end{align*} 
Similarly, analogous results can be obtained for the remaining terms in $ \pi_{2} \IF_{\mathrm{B}_{\psi, k}}^{(2)}(o_1, O_2; \bbeta^*)$.
Taking these together, we have:
\begin{align*}
  n \sup_{o_1} \bbE \left[ \left| \frac{ a^{\top} \left\{  \pi_{2} \IF_{\mathrm{B}_{\psi, k}}^{(2)} ( o_1 , O_2 ;\bbeta^*) \right\} }{n\sqrt{k+n}} \right| \right]  \leq\frac{n}{n\sqrt{k+n}}\cdot O(1) = o (1).
\end{align*}
Finally, since $\sup_{o_1,o_2} \left\vert a^{\top}  \pi_{2}  \hat{\Xi}_{\bbeta^{\ast}}  ( o_1,o_2 )\right\vert \lesssim 1$, we obtain the following:
\begin{align*}
	n \sup_{o_1} \bbE \| h_n ( o_1 , O_2 )\| \leq\frac{n}{n\sqrt{k+n}}\cdot O(1) = o (1).
\end{align*}

\end{proof}

\section{Proof of Results in Section \ref{sec:dose}}
\setcounter{equation}{0}
To aid the exposition, we define the following quantities:
\begin{align}
1/\lambda(x) \coloneqq & \ \frac{\p_{X}(x)}{g(x)} = \int_{u\in\calT} K_{t,h} (u) \p_{T|X}(u | x) \diff u,\label{eq:lambdax}\\
    \Delta_1 (x) \coloneqq &  \int_{u\in\calT} \{ \hat{\xi} (x, u) - \xi (x, u) - \hat{\xi} (x, t) + \xi (x, t) \} K_{t,h}(u)  \p_{T|X}( u | x) \diff u \label{eq:def_Delta1},\\
    \Delta_2 (x) \coloneqq & \int_{u\in\calT} \{ b_{\beta} (x, u) - \hat{b}_{\beta_{\mathsf{init}}} (x, u)  - b_{\beta} (x, t) + \hat{b}_{\beta_{\mathsf{init}}} (x, t) \} K_{t,h}(u)  \p_{T|X}( u | x) \diff u\label{eq:def_Delta2}.
\end{align}
Using Assumption~\ref{as:identification}(ii) and Assumption~\ref{as:Lip}, and a similar argument to that used in the proof of Theorem~1 in \cite{kennedy2017non}, we have
\begin{align}\label{eq:kernal-resluts}
 0 < \underline{c}_{\lambda} \le \lambda(x)  \le \bar{c}_{\lambda} < \infty \ \text{for any } x, \text{ and }  \sup_{x\in\calX}| \Delta_j (x)| \lesssim h^{\alpha_1\wedge \alpha_2} \text{, for } j= 1,2,   
\end{align}
where $\underline{c}_{\lambda},\bar{c}_{\lambda}$ are two fixed constants. 
 By the definitions of $\lambda(x)$ and $g(x)$ (defined in Section~\ref{sec:formal}), we have, for any function  $h(x)$,
\begin{align}\label{eq:lambdax_prop}
     &\bbE [h(X)] =  \int_{x\in\calX} h(x)  \p_{X}(x) \diff x = \int_{x\in\calX}  h(x) \lambda(x) g(x) \diff x,\\
     & \bbE [h(X)K_{t,h}(T)] = \int_{x\in\calX} h(x) g(x) \diff x. \label{eq:Kth_prop}
\end{align}

By definition of $\lambda(x)$ and \eqref{eq:kernal-resluts}, we have
\begin{align*}
 a^{\top} \tilde{\Omega}_{k_x}\, a = a^{\top} \int_{x\in\calX} g(x)\sfzbar_{k_{x}} (x)\sfzbar_{k_{x}}(x)^{\top} \diff x \,a = a^{\top} \bbE \left[ \lambda(X)\sfzbar_{k_{x}} (X)^{\otimes 2} \right] \, a \asymp a^{\top} \bbE \left[ \sfzbar_{k_{x}} (X)^{\otimes 2} \right] \,  a.
\end{align*}
Hence, Assumption~\ref{as:basis}(iii) implies that the eigenvalues of $ \tilde{\Omega}_{k_x}$ are bounded from above and away from zero uniformly in $k_x$. Consequently, analogues to \eqref{eq:supnorm_proj}, using Assumptions~~\ref{as:basis}, we have: 
for any function $f\in \calC( \calX)$, 
\begin{align}\label{eq:supnorm_proj_g}
 \left\| \Pi_{g} [f \mid \sfzbar_{k_{x}}] (\cdot) \right\|_{\infty} \leq  \|f(\cdot)\|_{\infty},
\end{align} 
where $\Pi_{g} [f\mid \sfzbar_{k_{x}}] (\cdot) = \left( \int_{\calX} f (x) \sfzbar_{k_{x}} (x) g (x) \diff x \right)^{\top} \tilde{\Omega}_{k_x}^{-1} \sfzbar_{k_{x}} (\cdot)$. 

Under Assumption~\ref{as:nuisance_est}(iii), for any $t\in\calT$ and $\delta>0$, we have 
\begin{align*}
& \bbE \big[\sup_{|\beta_{1}-\beta_{2}|\leq\delta} \left| \Gamma (Y,T,\beta_1)- \Gamma (Y,T,\beta_2) \right|^2 \mid T=t\big] \\
\le&\  \frac{1}{c_5} \bbE \big[\sup_{|\beta_{1}-\beta_{2}|\leq\delta} \left| \Gamma (Y,T,\beta_1)- \Gamma (Y,T,\beta_2) \right|^2 \mid T=t\big] \p_{T}(t) \\
\leq &\  \frac{1}{c_5} \bbE \big[\sup_{|\beta_{1}-\beta_{2}|\leq\delta} \left| \Gamma (Y,T,\beta_1)- \Gamma (Y,T,\beta_2) \right|^2 \big]. 
\end{align*}
Hence, applying  Assumption~\ref{as:criterion}(iii), we obtain 
\begin{align}\label{eq:local L2 Lip}
    \bbE^{1/2} \big[\sup_{|\beta_{1}-\beta_{2}|\leq\delta} \left| \Gamma (Y,T,\beta_1)- \Gamma (Y,T,\beta_2) \right|^2 \mid T=t\big] \lesssim \delta^{\alpha_{0}}.
\end{align}
\subsection{Proof of Theorem~\ref{th:rate-nonparametric}}\label{app:rate-nonparametric}
Because the proof carefully tracks dependence on $\beta$, we now make that dependence explicit by moving $\beta$ out of subscripts into argument parentheses. 
Denote population counterpart of $\tilde{\psi}_{t,k}^{(2)} (\beta) $ as given by $\bar{\psi}_{t,k} (\beta) \coloneqq \bbE [\tilde{\psi}_{t,k}^{(2)} (\beta)]$.
\begin{proof}[Proof of Theorem \ref{th:rate-nonparametric}]
Assumption~\ref{as:smoothness-psit},  together with an argument similar to that leading to
\eqref{eq:psi>beta} in Section~\ref{app:beta-consistency}, implies 
\begin{align}\label{eq:psit>beta}
|\psi_{t} (\beta )|   \gtrsim  |\beta - \beta^*_t | \wedge c
\end{align}
for some constant $c>0$. Consequently, to bound
$\vert\tilde{\beta}^{(2)}_{t}-\beta^{\ast}_{t}\vert$,
it suffices to control the quantity
$\vert \psi_{t}(\tilde{\beta}^{(2)}_{t}) \vert$.

First, we establish the consistency of $\tilde{\beta}^{(2)}_{t}$ by showing that $\vert \psi_{t}(\tilde{\beta}^{(2)}_{t}) \vert \lesssim \tilde{r}_{n,\mathsf{np}} \rightarrow 0$ with probability approaching one, where $\tilde{r}_{n,\mathsf{np}}$ is defined in Theorem~\ref{th:rate-nonparametric}. Note that $\psi_{t} (\beta^*_t) \equiv 0 $ and $\tilde{\psi}_{t,k} (\tilde{\beta}_t^{(2)}) \equiv 0$.  Consider the following decomposition: 
\begin{align}
    \left| \psi_{t} (\tilde{\beta}_t^{(2)})\right| \lesssim & \ \left|\psi_{t} (\tilde{\beta}_t^{(2)}) - \bar{\psi}_{t,k} (\tilde{\beta}_t^{(2)}) \right|  + \left|\bar{\psi}_{t,k} (\beta^*_t)  - \psi_{t} (\beta^*_t) \right|  \notag\\
    & + \left| \left\{ \bar{\psi}_{t,k} (\tilde{\beta}_t^{(2)}) 
    - \tilde{\psi}_{t,k}^{(2)} (\tilde{\beta}_t^{(2)}) \right\} - \left\{  \bar{\psi}_{t,k} (\beta^*_t)  - \tilde{\psi}_{t,k}^{(2)} (\beta^*_t)  \right\}\right| \notag\\
    & + \left| \tilde{\psi}_{t,k}^{(2)} (\tilde{\beta}_t^{(2)}) - \tilde{\psi}_{t,k}^{(2)} (\beta^*_t) \right|\notag\\
    \lesssim & \  \sup_{ \beta \in \calB } \left| \psi_{t} (\beta) - \bar{\psi}_{t,k} (\beta) \right|  \label{eq:kernel-consist-1}\\
    &+ \sup_{ \beta \in \calB }\left| \left\{ \bar{\psi}_{t,k} (\beta) 
    - \tilde{\psi}_{t,k}^{(2)} (\beta) \right\} - \left\{  \bar{\psi}_{t,k} (\beta^*_t)  - \tilde{\psi}_{t,k}^{(2)} (\beta^*_t)  \right\}\right|  \label{eq:kernel-consist-2}\\
    &+\left|  \tilde{\psi}_{t,k}^{(2)} (\beta^*_t) \right| . \label{eq:kernel-consist-3}
\end{align}
Here, the first term \eqref{eq:kernel-consist-1} contributes to the bias, the third term \eqref{eq:kernel-consist-3} captures  the variance, and the second term \eqref{eq:kernel-consist-2} is a $U$-process indexed by $\beta$. 

First, we analyze the term \eqref{eq:kernel-consist-1}. 
Applying  Lemmas~\ref{lemma:psith-kernelbias} and \ref{lemma:EB-kernelbias}, 
 we obtain the following bias decomposition: \allowdisplaybreaks
\begin{align}
        & \ \psi_{t}(\beta) -     \bar{\psi}_{t,k} (\beta) \notag \\
    = & \ \bbE [ \xi(X,T) b_{\beta}(X,T) |T =t]\p_{T}(t) - 
\bbE\left[ K_{t,h} (T_1)  \hat{\xi} (X_1, T_1) \{ \Gamma (Y_1,T_1,\beta) - \hat{b}_{\beta_{\mathsf{init}}} (X_1, T_1) \} \right] \notag\\
& + \bbE \left[ \hat{b}_{\beta_{\mathsf{init}}} (X_1, T_2)K_{t,h}(T_2) \right] \notag \\
& - \bbE \left[ \{K_{t,h}(T_{1}) \hat{\xi} (X_{1}, T_{1})  - K_{t,h} (T_{3} )\} \sfzbar_{k_{x}}^{\top} (X_{1}) \tilde{\Omega}_{k_x}^{-1} \sfzbar_{k_{x}} (X_{2}) \{ \Gamma (Y_2,T_2,\beta) - \hat{b}_{\beta_{\mathsf{init}}} (X_{2}, T_{2})\} K_{t,h} (T_{2})\right] \notag \\
=& \ \bbE \left[K_{t,h}(T) \left\{  \xi(X,t) - \hat{\xi}(X,t) \right\} \left\{ b_{\beta}(X,t)  - \hat{b}_{\beta_{\mathsf{init}}}(X,t) \right\}\right] \notag 
\\
&- \int_{x\in\calX} \Pi_{g}[ \hat{\xi} (\cdot, t) - \xi (\cdot, t) | \sfzbar_{k_{x}}] (x) \times \Pi_{g} [b_{\beta} (\cdot, t) - \hat{b}_{\beta_{\mathsf{init}}} (\cdot, t) | \sfzbar_{k_{x}} ] (x)  g(x) \diff x  + O(h^{\alpha_1 \wedge \alpha_2})\notag\\
=& \ \int_{x\in\calX}  \left\{  \xi(x,t) - \hat{\xi}(x,t) \right\} \left\{ b_{\beta}(x,t)  - \hat{b}_{\beta_{\mathsf{init}}}(x,t) \right\} g(x) \diff x \notag
\\
&- \int_{x\in\calX} \Pi_{g}[ \hat{\xi} (\cdot, t) - \xi (\cdot, t) | \sfzbar_{k_{x}}] (x) \times \Pi_{g} [b_{\beta} (\cdot, t) - \hat{b}_{\beta_{\mathsf{init}}} (\cdot, t) | \sfzbar_{k_{x}}] (x)  g(x) \diff x  + O(h^{\alpha_1 \wedge  \alpha_2}) \notag \\
=&  \int_{x\in\calX} \Pi_{g}^{\perp} [b_{\beta} (\cdot, t) - \hat{b}_{\beta_{\mathsf{init}}} (\cdot, t) \mid \sfzbar_{k_{x}}] (x) \times   \Pi_{g}^{\perp} [\xi (\cdot, t) - \hat{\xi} (\cdot, t) \mid \sfzbar_{k_{x}}] (x)  g(x) \diff x + O(h^{\alpha_1 \wedge \alpha_2}), \label{eq:bias-kernel}
\end{align}
where $\Pi_{g}[f \mid \sfzbar_{k_x}](\cdot) \coloneqq \int_{\calX} f(x) \sfzbar_{k_x}(x)^{\top} g(x) \diff x\, \tilde{\Omega}_{k_x}^{-1}\,\sfzbar_{k_x}(\cdot)$ denotes  the projection (with respect to the weight $g$) of $f$ onto the linear span of $\sfzbar_{k_x}$ and $\Pi^{\perp}_{g}[f \mid \sfzbar_{k_x}](\cdot) \coloneqq f (\cdot) - \Pi_{g}[f \mid \sfzbar_{k_x}](\cdot) $ denotes the orthocomplement of this projection. 
This shows that the leading bias is determined by the product of the projection errors of $b_{\beta} (\cdot, t) - \hat{b}_{\beta_{\mathsf{init}}} (\cdot, t)$ and $\xi (\cdot, t) - \hat{\xi} (\cdot, t)$ onto the span of $\sfzbar_{k_x}$, plus a kernel approximation error of order $O(h^{\alpha_1 \wedge \alpha_2})$.
Thus, we conclude that
\begin{align}\label{eq:kernelrate-1-r}
    \eqref{eq:kernel-consist-1} \lesssim  & \sup_{ \beta\in\calB  }   \left\|\Pi_{g}^{\perp} [b_{\beta}(\cdot, t)- \hat{b}_{\beta_{\mathsf{init}}}(\cdot, t) \mid \sfzbar_{k_{x}} ]  \right\|_{g,2} \cdot  \left\| \Pi_{g}^{\perp} [ \xi(\cdot, t) - \hat{\xi}(\cdot, t) \mid \sfzbar_{k_{x}} ]  \right\|_{g,2} + O(h^{\alpha_1 \wedge \alpha_2}).
\end{align}
Here,  $\|\cdot\|_{g,2}$ denotes the norm $\|f\|_{g,2} =
\{ \int_{\calX} f(x)^2g(x)\diff x\}^{1/2}$ for any function $f(x)$.

Next, we analyze the $U$-process term \eqref{eq:kernel-consist-2}. Since $\calB$ is compact, we may write $\sup_{\beta\in\calB}$ as $\sup_{|\beta - \beta^*_t| \leq c}$ for some constant $c$ sufficiently large to cover $\calB$.
Applying Lemma~\ref{lemma:Uprocess-kernel} with $\epsilon = c$, we obtain
\begin{align}\label{eq:kernel-consist-2-r}
    \eqref{eq:kernel-consist-2} = O_{\P}  \left( \frac{\sqrt{k} }{nh} + \frac{ 1  }{\sqrt{nh}}\right)\cdot \log n.
\end{align}

Last, for the term \eqref{eq:kernel-consist-3}, applying the decomposition in \eqref{eq:bias-kernel} together with Lemma~\ref{lemma:var-kernel}, we have
\begin{align*}
    |\tilde{\psi}_{t,k}^{(2)} (\beta^{\ast}_t) | = & \ O\left(  \left\|\Pi_{g}^{\perp} [b_{\beta_t^{\ast}} (\cdot, t) - \hat{b}_{\beta_{\mathsf{init}}}  (\cdot, t) \mid \sfzbar_{k_{x}}] \right\|_{g,2} \cdot \left\| \Pi_{g}^{\perp} [\xi (\cdot, t) - \hat{\xi} (\cdot, t) \mid \sfzbar_{k_{x}}] \right\|_{g,2}  \right) \\
    & + O(h^{\alpha_1 \wedge \alpha_2}) + O_{\P}  \left( \frac{\sqrt{k_{x}} }{nh} + \frac{ 1 }{\sqrt{nh}}\right).
\end{align*}

Combining the bounds for terms \eqref{eq:kernel-consist-1}–\eqref{eq:kernel-consist-3}, we have, with probability approaching one,
\begin{align*}
        \left| \psi_{t} (\tilde{\beta}_t^{(2)})\right| \lesssim & \ \eqref{eq:kernel-consist-1} + \eqref{eq:kernel-consist-2} + \eqref{eq:kernel-consist-3}
        \lesssim \ \tilde{r}_{n,\mathsf{np}} .
\end{align*}
This establishes the consistency of~$\tilde{\beta}_t^{(2)}$.

Building on this consistency result, we now restrict our analysis to the region ${ |\beta - \beta_t^*| \lesssim \tilde{r}_{n,\mathsf{np}} }$.
This refinement allows us to improve upon the initial rate $\tilde{r}_{n,\mathsf{np}}$ and obtain the final convergence rate stated in Theorem~\ref{th:rate-nonparametric}. Repeating the decomposition from the consistency argument, but now over the restricted region, we have
\begin{align*}
    \left| \psi_{t} (\tilde{\beta}_t^{(2)})\right|  
    \lesssim & \  \sup_{ | \beta - \beta^{\ast}_{t} |\lesssim \, \tilde{r}_{n,\mathsf{np}}  } \left| \psi_{t} (\beta) - \bar{\psi}_{t,k} (\beta) \right|  +\left|  \tilde{\psi}_{t,k}^{(2)} (\beta^*_t) \right| \\
    &+ \sup_{ |\beta - \beta^{\ast}_{t}|\lesssim\, \tilde{r}_{n,\mathsf{np}}  }\left| \left\{ \bar{\psi}_{t,k} (\beta) 
    - \tilde{\psi}_{t,k}^{(2)} (\beta) \right\} - \left\{  \bar{\psi}_{t,k} (\beta^*_t)  - \tilde{\psi}_{t,k}^{(2)} (\beta^*_t)  \right\}\right|  . 
\end{align*}
The only change from the earlier argument is that the supremum is now taken over a shrinking neighborhood of $\beta_t^*$.   Applying the bias decomposition~\eqref{eq:bias-kernel}, Lemma~\ref{lemma:Uprocess-kernel} with $\epsilon = \tilde{r}_{n,\mathsf{np}}$, and Lemma~\ref{lemma:var-kernel}, we obtain, with probability approaching one,
\begin{align*}
   \left|\psi_{t} (\tilde{\beta}^{(2)}_t) \right| \lesssim & \   \sup_{ | \beta  - \beta^*_t | \lesssim\, \tilde{r}_{n,\mathsf{np}}  }   \left\|\Pi_{g}^{\perp} [b_{\beta}(\cdot, t)- \hat{b}_{\beta_{\mathsf{init}}}(\cdot, t) \mid \sfzbar_{k_{x}} ]  \right\|_{g,2} \cdot  \left\| \Pi_{g}^{\perp} [ \xi(\cdot, t) - \hat{\xi}(\cdot, t) \mid \sfzbar_{k_{x}} ]  \right\|_{g,2} \notag \\
   &+ h^{\alpha_1 \wedge \alpha_2} + \frac{\sqrt{k_{x}} }{nh} + \frac{ 1 }{\sqrt{nh}} +    \left( \frac{\sqrt{k_{x}} }{nh} + \frac{ 1 }{\sqrt{nh}}\right)\cdot \{ (\tilde{r}_{n,\mathsf{np}})^{\alpha_{0}}  \vee \sqrt{h} \} \cdot \log n.   
\end{align*}
Given the additional conditions  $ h (\log n)^2 \to 0 $ and $(\tilde{r}_{n,\mathsf{np}})^{\alpha_{0}} \log n \to 0$, this further simplifies to
\begin{align*}
  \left|\psi_{t} (\tilde{\beta}^{(2)}_t) \right| \lesssim & \  \frac{\sqrt{k_{x}} }{nh} + \frac{ 1 }{\sqrt{nh}} + h^{\alpha_1 \wedge \alpha_2}  \\
    &+   \sup_{ | \beta  - \beta^*_t | \lesssim\, \tilde{r}_{n,\mathsf{np}}  }   \left\|\Pi_{g}^{\perp} [b_{\beta}(\cdot, t)- \hat{b}_{\beta_{\mathsf{init}}}(\cdot, t) \mid \sfzbar_{k_{x}} ]  \right\|_{g,2} \cdot  \left\| \Pi_{g}^{\perp} [ \xi(\cdot, t) - \hat{\xi}(\cdot, t) \mid \sfzbar_{k_{x}} ]  \right\|_{g,2} \notag
\end{align*}
with probability approaching one. Combining this with \eqref{eq:psit>beta} yields the final convergence rate for $\tilde{\beta}_t^{(2)}$, thus completing the proof.

\end{proof}

\subsection{Technical Lemmas for Proving Theorem~\ref{th:rate-nonparametric}}

\begin{lemma}\label{lemma:psith-kernelbias}
Under  Assumptions~\ref{as:kernel} and \ref{as:Lip}, we have 
\begin{align*}
    & \ \bbE [ \xi(X,T) b_{\beta}(X,T) |T =t] \p_{T}(t) = \bbE [ \xi(X,t) b_{\beta}(X,t) K_{t,h}(T)]  +  O(h^{\alpha_1 \wedge \alpha_2}),\\
    & \ \bbE\left[ K_{t,h} (T_1)  \hat{\xi} (X_1, T_1) \{ \Gamma (Y_1,T_1,\beta) - \hat{b}_{\beta_{\mathsf{init}}} (X_1, T_1) \} +  \hat{b}_{\beta_{\mathsf{init}}} (X_1, T_2)K_{t,h}(T_2) \right]\\
    = & \ \bbE\left[ K_{t,h} (T)  \hat{\xi} (X, t) \{ b_{\beta}(X,t) - \hat{b}_{\beta_{\mathsf{init}}} (X, t) \} + \xi (X,t) \hat{b}_{\beta_{\mathsf{init}}} (X, t)K_{t,h}(T) \right] +  O(h^{\alpha_1 \wedge \alpha_2}).
\end{align*}
\end{lemma}

\begin{proof} Using a similar argument for proving \eqref{eq:kernal-resluts}, we have 
\begin{align*}
    \sup_{x\in\calX}\left|\int_{u\in\calT}  K_{t,h}(u) \p_{T|X}(u|x) \diff u - \p_{T|X}(t|x)\right| = O(h^{\alpha_1 \wedge \alpha_2}) .
\end{align*}
Thus 
\begin{align*}
    \bbE [ \xi(X,t) b_{\beta}(X,t) K_{t,h}(T) ] =& \int_{x\in\calX} \xi(x,t) b_{\beta}(x,t) \int_{u\in\calT}  K_{t,h}(u) \p_{T|X}(u|x) \diff u \p_{X}(x)\diff x\\
    =& \int_{x\in\calX} \xi(x,t)  b_{\beta}(x,t)  \p_{T|X}(t|x)  \p_{X}(x) \diff x + O(h^{\alpha_1 \wedge \alpha_2})\\
    =& \int_{x\in\calX} b_{\beta}(x,t)\p_{T}(t)\p_{X}(x) \diff x + O(h^{\alpha_1 \wedge \alpha_2})\\
       =&\  \bbE [ \xi(X,T) b_{\beta}(X,T) |T =t] \p_{T}(t) + O(h^{\alpha_1 \wedge \alpha_2}).
\end{align*} 
Hence, we obtain the first result.

Using the definition of $\xi(x,t)$, we have {\small
\begin{align*}
    &\bbE\left[ K_{t,h} (T_1)  \hat{\xi} (X_1, T_1) \{ \Gamma (Y_1,T_1,\beta) - \hat{b}_{\beta_{\mathsf{init}}} (X_1, T_1) \} +  \hat{b}_{\beta_{\mathsf{init}}} (X_1, T_2)K_{t,h}(T_2) \right]\\
    =&\int_{x\in\calX} \int_{u\in\calT} K_{t,h} (u) \left\{ \hat{\xi} (x, u) ( b_{\beta}(x,u) - \hat{b}_{\beta_{\mathsf{init}}} (x, u) ) + \xi (x,u) \hat{b}_{\beta_{\mathsf{init}}} (x, u) \right\} \p_{T|X}(u|x) \diff u \p_{X}(x) \diff x\\
    =& \int_{x\in\calX} \int_{u\in\calT} K_{t,h} (u) \left\{ \hat{\xi} (x, t) ( b_{\beta}(x,t) - \hat{b}_{\beta_{\mathsf{init}}} (x, t) ) + \xi (x,t) \hat{b}_{\beta_{\mathsf{init}}} (x, t) \right\} \p_{T|X}(u|x) \diff u \p_{X}(x) \diff x  + O(h^{\alpha_1 \wedge \alpha_2})\\
    =&\  \bbE\left[ K_{t,h} (T) \left\{ \hat{\xi} (X, t) ( b_{\beta}(X,t) - \hat{b}_{\beta_{\mathsf{init}}} (X, t) ) + \xi (X,t) \hat{b}_{\beta_{\mathsf{init}}} (X, t) \right\} \right]  + O(h^{\alpha_1 \wedge \alpha_2}),
\end{align*} }
where the second equality uses a similar argument for proving \eqref{eq:kernal-resluts}.
\end{proof}

\begin{lemma}\label{lemma:EB-kernelbias}
    Under  Assumptions~\ref{as:nuisance_est}--\ref{as:basis}, \ref{as:kernel} and \ref{as:Lip}, we have 
    \begin{align*}
        & \ \bigg | \bbE \left[ \{\hat{\xi} (X, T) - \xi (X, T) \} \sfzbar_{k_{x}} (X)^{\top} K_{t,h}(T)\right]  \tilde{\Omega}_{k_x}^{-1} \bbE \left[\sfzbar_{k_{x}} (X) \{ \Gamma (Y,T,\beta) - \hat{b}_{\beta_{\mathsf{init}}} (X, T)\}  K_{t,h}(T)\right] \notag\\
    &- \int_{x\in\calX} \Pi_{g}[ \hat{\xi} (\cdot, t) - \xi (\cdot, t) | \sfzbar_{k_{x}} ](x) \times \Pi_{g} [b_{\beta} (\cdot, t) - \hat{b}_{\beta_{\mathsf{init}}} (\cdot, t) | \sfzbar_{k_{x}}  ](x)  g(x) \diff x \bigg | = O(h^{\alpha_1\wedge \alpha_2}).
    \end{align*}
\end{lemma}

\begin{proof} By the definition of $\xi(X,T)$, 
\begin{align*}
   & \ \bbE \left[ \{K_{t,h}(T_{1}) \hat{\xi} (X_{1}, T_{1})  - K_{t,h} (T_{3} )\} \sfzbar_{k_{x}}^{\top} (X_{1}) \tilde{\Omega}_{k_x}^{-1} \sfzbar_{k_{x}} (X_{2}) \{ \Gamma (Y_2,T_2,\beta) - \hat{b}_{\beta_{\mathsf{init}}} (X_{2}, T_{2})\} K_{t,h} (T_{2})\right]\\
    = & \ \bbE \left[ \{\hat{\xi} (X, T) - \xi (X, T) \} \sfzbar_{k_{x}} (X)^{\top} K_{t,h}(T)\right]  \tilde{\Omega}_{k_x}^{-1} \bbE \left[\sfzbar_{k_{x}} (X) \{ b_{\beta} (X, T) - \hat{b}_{\beta_{\mathsf{init}}} (X, T)\}  K_{t,h}(T)\right] . 
\end{align*}
Recalling the definitions of  $\Delta_1(X)$ and $\Delta_2(X)$ in \eqref{eq:def_Delta1} and \eqref{eq:def_Delta2}, we have 
\begin{align*}
     \bbE \left[ \{ b_{\beta} (X, T) - \hat{b}_{\beta_{\mathsf{init}}} (X, T)  - b_{\beta} (X, t) + \hat{b}_{\beta_{\mathsf{init}}} (X, t)\}  K_{t,h}(T)  \sfzbar_{k_{x}} (X) \right] 
     =  \bbE [\Delta_2 (X) \sfzbar_{k_{x}} (X)^{\top}]
\end{align*}
and 
\begin{align*}
    \bbE [ \{\hat{\xi} (X, T) - \xi (X, T) - \hat{\xi} (X, t) + \xi (X, t)\} \sfzbar_{k_{x}} (X)^{\top} K_{t,h}(T)]
    =    \bbE [\Delta_1 (X) \sfzbar_{k_{x}} (X)^{\top}]. 
\end{align*}
Then, we obtain
\begin{align*}
    &\bbE [ \{\hat{\xi} (X, T) - \xi (X, T) \} \sfzbar_{k_{x}} (X)^{\top} K_{t,h}(T)]  \tilde{\Omega}_{k_x}^{-1} \bbE \left[\sfzbar_{k_{x}} (X) \{ \Gamma (Y,T,\beta) - \hat{b}_{\beta_{\mathsf{init}}} (X, T)\}  K_{t,h}(T)\right] \\
    =&\left(\bbE [ \{\hat{\xi} (X, t) - \xi (X, t) \} \sfzbar_{k_{x}} (X)^{\top} K_{t,h}(T)] +  \bbE [\Delta_1 (X) \sfzbar_{k_{x}} (X)^{\top}] \right) \times \tilde{\Omega}_{k_x}^{-1} \\
    &\times \left( \bbE [ \{ b_{\beta} (X, t) - \hat{b}_{\beta_{\mathsf{init}}} (X, t) \} \sfzbar_{k_{x}} (X) K_{t,h}(T) ] + \bbE [\Delta_2 (X) \sfzbar_{k_{x}} (X)] \right)\\
    = &\left[\int_{x\in\calX} \{\hat{\xi} (x, t) - \xi (x, t) \} \sfzbar_{k_{x}} (x)^{\top} g(x) \diff x + \int_{x\in\calX}  \Delta_{1}(x)  \lambda(x)\sfzbar_{k_{x}} (x)^{\top} g(x) \diff x \right]\\
    &\times \tilde{\Omega}_{k_x}^{-1} \times \left[ \int_{x\in\calX} \{ b_{\beta} (x, t) - \hat{b}_{\beta_{\mathsf{init}}} (x, t) \}\sfzbar_{k_{x}} (x) g(x) \diff x + \int_{x\in\calX}  \Delta_{2}(x) \lambda(x)\sfzbar_{k_{x}} (x)^{\top} g(x) \diff x \right] \\
    = & \int_{x\in\calX} \left\{ \Pi_{g}[ \hat{\xi} (\cdot, t) - \xi (\cdot, t) | \sfzbar_{k_{x}} ](x) + \Pi_{g}[ \Delta_{1} \lambda | \sfzbar_{k_{x}}  ](x) \right\} \\
    &\times \left\{ \Pi_{g} [b_{\beta} (\cdot, t) - \hat{b}_{\beta_{\mathsf{init}}} (\cdot, t) | \sfzbar_{k_{x}}  ](x) + \Pi_{g} [  \Delta_{2} \lambda|\sfzbar_{k_{x}} ](x) \right\} g(x) \diff x. 
\end{align*}
where the second equality holds by using  \eqref{eq:lambdax_prop} and \eqref{eq:Kth_prop}.
Hence
\begin{align}\label{eq:ifjj_t}
    & \ \bigg | \bbE [ \{\hat{\xi} (X, T) - \xi (X, T) \} \sfzbar_{k_{x}} (X)^{\top} K_{t,h}(T)]  \tilde{\Omega}_{k_x}^{-1} \bbE \left[\sfzbar_{k_{x}} (X) \{ \Gamma (Y,T,\beta) - \hat{b}_{\beta_{\mathsf{init}}} (X, T)\}  K_{t,h}(T)\right] \notag\\
    &- \int_{x\in\calX} \Pi_{g} [ \hat{\xi} (\cdot, t) - \xi (\cdot, t) | \sfzbar_{k_{x}}  ](x) \times \Pi_{g} [b_{\beta} (\cdot, t) - \hat{b}_{\beta_{\mathsf{init}}} (\cdot, t) | \sfzbar_{k_{x}} ] (x)  g(x) \diff x \bigg |\notag\\
    \leq & \ \bigg | \int_{x\in\calX}  \Pi_{g} [ \hat{\xi} (\cdot, t) - \xi (\cdot, t) | \sfzbar_{k_{x}} ](x) \times \Pi_{g} [   \Delta_{2} \lambda|\sfzbar_{k_{x}} ](x) g(x) \diff x \bigg |\notag\\
    & + \bigg | \int_{x\in\calX} \Pi_{g}[ \Delta_{1} \lambda | \sfzbar_{k_{x}}  ](x) \times \Pi_{g} [b_{\beta} (\cdot, t) - \hat{b}_{\beta_{\mathsf{init}}} (\cdot, t) | \sfzbar_{k_{x}}  ](x)    g(x) \diff x \bigg |\notag\\
    &+ \bigg | \int_{x\in\calX} \Pi_{g}[ \Delta_{1} \lambda | \sfzbar_{k_{x}}  ] (x) \times \Pi_{g} [   \Delta_{2} \lambda|\sfzbar_{k_{x}} ] (x) g(x) \diff x \bigg | \notag\\
    \lesssim & \ \|   \Pi_{g}[ \hat{\xi} (\cdot, t) - \xi (\cdot, t) | \sfzbar_{k_{x}} ] \|_{g,2} \cdot \| \Pi_{g} [   \Delta_{2} \lambda |\sfzbar_{k_{x}} ] \|_{g,2}\notag \\
    & + \left\|\Pi_{g}[ \Delta_{1} \lambda | \sfzbar_{k_{x}}  ] \right\|_{g,2}\cdot \left\| \Pi_{g} [b_{\beta} (\cdot, t) - \hat{b}_{\beta_{\mathsf{init}}} (\cdot, t) | \sfzbar_{k_{x}}  ]  \right\|_{g,2}\notag\\
     & + \left\| \Pi_{g}[ \Delta_{1} \lambda | \sfzbar_{k_{x}}  ]  \right\|_{g,2}\cdot \left\| \Pi_{g} [   \Delta_{2} \lambda|\sfzbar_{k_{x}} ] \right\|_{g,2}\notag\\
     \lesssim & \ \|   \hat{\xi} (X, t) - \xi (X, t) \|_{g,2} \cdot \|   \Delta_{2}(X) \lambda(X) \|_{g,2}\notag \\
    & + \left\|  \Delta_{1}(X) \lambda(X) \right\|_{g,2}\cdot \left\| b_{\beta} (X, t) - \hat{b}_{\beta_{\mathsf{init}}} (X, t)  \right\|_{g,2}\notag\\
     & + \left\|  \Delta_{1}(X) \lambda(X)   \right\|_{g,2}\cdot \left\|   \Delta_{2}(X) \lambda(X) \right\|_{g,2} \notag\\
     = & \ O(h^{\alpha_1 \wedge \alpha_2}),
\end{align}
where the last equality holds by using Assumptions~\ref{as:nuisance_est}, \ref{as:Lip}, and \eqref{eq:kernal-resluts}.
\end{proof}

\begin{lemma}\label{lemma:H3_L2}   Under Assumptions~\ref{as:nuisance_est}, \ref{as:criterion}, \ref{as:kernel}, and  \ref{as:Lip}, 
we have 
    \begin{gather*}
        \ \bbE \bigg[\sup_{|\beta_1 - \beta_2|\lesssim \epsilon } \left| \{ \Gamma (Y,T,\beta_1) - \Gamma (Y,T,\beta_2) \}K_{t,h}\left(T\right) \right|^2 \bigg] = O \left( \frac{\epsilon^{2\alpha_{0}}\vee h}{h}  \right), \\
       a^{\top} \bbE \left[ K_{t,h}(T)^2 \{ \hat{\xi} (X, T)  -  1\}^2 \xi(X, T) \sfzbar_{k_{x}} (X) \sfzbar_{k_{x}} (X)^{\top} \right] a \lesssim \frac{1}{h} a^{\top} \tilde{\Omega}_{k_x}a,
    \end{gather*}
    where $a\in\bbR^{k_x}$ is a unit vector.
\end{lemma}
\begin{proof}
    By Taylor expansion, 
Assumptions \ref{as:nuisance_est} and \eqref{eq:local L2 Lip}, we have 
\begin{align}
 & \ \bbE \bigg[\sup_{|\beta_1 - \beta_2|\lesssim \epsilon } \left| \{ \Gamma (Y,T,\beta_1) - \Gamma (Y,T,\beta_2) \}K_{t,h}\left(T\right) \right|^2 \bigg]\\
    =& \ \frac{1}{h^2}\int_{u\in\calT} \bbE \bigg[\sup_{|\beta_1 - \beta_2|\lesssim \epsilon } \left| \hat{\xi}(X,T) \{ \Gamma (Y,T,\beta_1) - \Gamma (Y,T,\beta_2) \} \right|^2 \mid T=u  \bigg] \left| K\left(\frac{u -t }{h}\right) \right|^2 \p_{T}(u) \diff u\notag\\
    =& \ \frac{1}{h}\int_{u\in\calT} \bbE \bigg[ \sup_{|\beta_1 - \beta_2|\leq \epsilon } \left| \hat{\xi}(X,T) \{ \Gamma (Y,T,\beta_1) - \Gamma (Y,T,\beta_2) \} \right|^2\mid T=t+hu \bigg] \left| K\left(u\right) \right|^2 \p_{T}(t + h u) \diff u \notag\\
    = &\ \frac{1}{h}  \bbE \bigg[ \sup_{|\beta_1 - \beta_2|\lesssim \epsilon } \left| \hat{\xi}(X,T) \{ \Gamma (Y,T,\beta_1) - \Gamma (Y,T,\beta_2) \} \right|^2 \mid T=t \bigg]\p_{T}(T=t)\cdot\kappa_{0,2} \notag\\
    &+ \frac{1}{h} \int_{u\in\calT}  \nabla_{t}\bbE \left[ \left. \sup_{|\beta_1 - \beta_2|\lesssim \epsilon } \left| \hat{\xi}(X,T) \{ \Gamma (Y,T,\beta_1) - \Gamma (Y,T,\beta_2) \} \right|^2 \right| T= t + \tau hu \right] hu K(u)^2 \diff u \notag \\
    \lesssim & \ \frac{1}{h} \left\{ \bbE \left[ \sup_{|\beta_1 - \beta_2|\lesssim \epsilon } \left|   \Gamma (Y,T,\beta_1) - \Gamma (Y,T,\beta_2)  \right|^2 \mid T=t \right]\p_{T}(T=t)\cdot\kappa_{0,2} + O(h)\right\} = O \left( \frac{\epsilon^{2\alpha_{0}}\vee h}{h}  \right), \notag
\end{align}
where $\tau\in(0,1)$ and $\kappa_{0,2} = \int K\left(u\right)^2\diff u$. Thus, we obtain the first result. Similarly, we have $ \bbE [K_{t,h}(T)^2] = O\left( {1}/{h}  \right)$. Then, we obtain 
 \begin{align*}
     & \ a^{\top} \bbE \left[ K_{t,h}(T)^2 \{ \hat{\xi} (X, T)  -  1\}^2 \xi(X, T) \sfzbar_{k_{x}} (X) \sfzbar_{k_{x}} (X)^{\top} \right] a \\
     \lesssim & \ a^{\top} \bbE \left[ K_{t,h}(T)^2 \xi(X, T) \sfzbar_{k_{x}} (X) \sfzbar_{k_{x}} (X)^{\top} \right] a\\
      = &\  \bbE \left[ K_{t,h}(T)^2 \right] \cdot a^{\top} \bbE \left[ \sfzbar_{k_{x}} (X) \sfzbar_{k_{x}} (X)^{\top} \right] a \\
      = & \ \bbE \left[ K_{t,h}(T)^2 \right] \cdot a^{\top} \int_{x\in\calX}  \sfzbar_{k_{x}} (x) \sfzbar_{k_{x}} (x)^{\top} g(x) \cdot \frac{\p_{X}(x)}{g(x)} \diff x a\\
      \le & \ \sup_{x\in\calX} \left|  \frac{\p_{X}(x)}{g(x)} \right|\cdot \bbE \{K_{t,h}(T)^2\} \cdot a^{\top} \tilde{\Omega}_{k_x}\,a \lesssim \frac{1}{h} a^{\top} \tilde{\Omega}_{k_x}\,a, 
\end{align*}
where the first inequality holds since we have $\|\hat{\xi}\| \lesssim 1$; the first equality holds by the definition of $\xi(x,t)$.
\end{proof}

\begin{lemma}\label{lemma:kernel-pih-order} Under Assumptions~\ref{as:nuisance_est}--\ref{as:basis}, \ref{as:kernel}, and  \ref{as:Lip}, we have the following results:
\begin{align*}
\mathrm{(i)} & \sup\limits_{|\beta - \beta^{\ast}_{t}|\lesssim \epsilon } \bbE \big[ \big|\bbE [ \{ \hat{\xi} (X, T)  -  \xi (X, T)\}K_{t,h}(T)\sfzbar_{k_{x}} (X)^{\top} ]  \tilde{\Omega}_{k_x}^{-1} \sfzbar_{k_{x}} (X) K_{t,h} (T) ( \Gamma_{\beta}  - \Gamma_{\beta^{\ast}_{t}} ) (Y,T) \big|^2\big] \\
& \lesssim \frac{\epsilon^{2\alpha_{0}}\vee h}{h}; \\
\mathrm{(ii)} & \sup\limits_{|\beta - \beta^{\ast}_{t}|\lesssim \epsilon } \bbE \big[ \big| K_{t,h} (T ) \bbE [ \sfzbar_{k_{x}} (X)^{\top}]\tilde{\Omega}_{k_x}^{-1} \bbE [ \sfzbar_{k_{x}} (X) K_{t,h} (T)  \{ b_{\beta}(X,T)  - b_{\beta^{\ast}_{t}}(X,T) \} ] \big|^2\big] \\
& \lesssim \frac{\epsilon^{2\alpha_{0}}\vee h^{2(\alpha_1\wedge \alpha_2)}}{h}; \\
\mathrm{(iii)} & \sup\limits_{|\beta - \beta^{\ast}_{t}|\lesssim \epsilon } \bbE \big[ \big|  K_{t,h}(T) \hat{\xi} (X, T) \sfzbar_{k_{x}} (X)^{\top} \tilde{\Omega}_{k_x}^{-1} \bbE [ \sfzbar_{k_{x}} (X) K_{t,h} (T)  \{ b_{\beta}(X,T)  - b_{\beta^{\ast}_{t}}(X,T) \} ] \big|^2\big] \\
& \lesssim \frac{\epsilon^{2\alpha_{0}}\vee h^{2(\alpha_1\wedge \alpha_2)}}{h}; \\
\mathrm{(iv)} & \sup\limits_{|\beta - \beta^{\ast}_{t}|\lesssim \epsilon } \bbE \big[ \big|  \bbE [ K_{t,h} (T ) ]  \sfzbar_{k_{x}} (X)^{\top} \tilde{\Omega}_{k_x}^{-1} \bbE [ \sfzbar_{k_{x}} (X) K_{t,h} (T)  \{ b_{\beta}(X,T)  - b_{\beta^{\ast}_{t}}(X,T) \} ] \big|^2\big] \\
& \lesssim \epsilon^{2\alpha_{0}}\vee h^{2(\alpha_1\wedge \alpha_2)}.
\end{align*}
\end{lemma}

\begin{proof}
To aid the exposition,  we define the following quantities:
\begin{align}
    \Delta_3 (x) \coloneqq & \ \int_{u\in\calT} \{ \hat{\xi} (x, u)  -  \xi (x, u)\} K_{t,h}(u)  \p_{T|X}( u | x) \diff u,\label{eq:def_Delta3}\\ 
     \Delta_4 (x) \coloneqq & \ \int_{u\in\calT} \{ b_{\beta} (x, u)  -  b_{\beta^{\ast}_{t}} (x, u)\} K_{t,h}(u)  \p_{T|X}( u | x) \diff u - \{b_{\beta} (x, t)  -  b_{\beta^{\ast}_{t}} (x, t)\}\p_{T|X}( t | x) \label{eq:def_Delta4}\\
     \Delta_5 (x) \coloneqq & \  \Delta_4 (x)  + \{b_{\beta} (x, t)  -  b_{\beta^{\ast}_{t}} (x, t)\}\p_{T|X}( t | x) \notag\\
     =&  \  \int_{u\in\calT} \{ b_{\beta} (x, u)  -  b_{\beta^{\ast}_{t}} (x, u)\} K_{t,h}(u)  \p_{T|X}( u | x) \diff u\label{eq:def_Delta5}
\end{align}
Note that both $\Delta_4(x)$ and $\Delta_5(x)$ depend on $\beta$, but, for notational convenience, we suppress this argument. Assumption \ref{as:Lip} states that the functions $b_\beta$ share the same degree of smoothness for every $\beta$, we can repeat the argument used to establish \eqref{eq:kernal-resluts} and obtain the following results: 
\begin{align}\label{eq:kernel-results2}
    \sup_{x\in\calX}| \Delta_3 (x)| \lesssim 1, \ \sup_{\beta}\sup_{x\in\calX}| \Delta_4 (x)| \lesssim h^{\alpha_1\wedge \alpha_2},\ \text{and}\ \sup_{\beta}\sup_{x\in\calX}| \Delta_5 (x)| \lesssim 1.
\end{align}
 First, we have
 \begin{align}\label{eq:pi-Delta5}
     & \ \sup_{|\beta - \beta^{\ast}_{t}|\lesssim \epsilon } \bbE \left[ \left|    \Pi_{g}[\Delta_5(\cdot)\lambda(\cdot)| \sfzbar_{k_{x}} ](X)  \right|^2\right]\notag\\
     = & \ \sup_{|\beta - \beta^{\ast}_{t}|\lesssim \epsilon }  \int_{x\in\calX} \left| \Pi_{g}[\Delta_5(\cdot)\lambda(\cdot)| \sfzbar_{k_{x}} ](x) \right|^2 \lambda(x)g(x)\diff x \notag \\ 
     \lesssim & \ \sup_{|\beta - \beta^{\ast}_{t}|\lesssim \epsilon } \int_{x\in\calX} \left| \Delta_5(x)\lambda(x) \right|^2 g(x) \diff x \notag\\
     \lesssim & \ \sup_{|\beta - \beta^{\ast}_{t}|\lesssim \epsilon }\sup_{x\in\calX}  |\Delta_4 (x)|^2  + \sup_{|\beta - \beta^{\ast}_{t}|\lesssim \epsilon } \bbE \left[\xi(X,T)\big\{b_{\beta} (X, T)  -  b_{\beta^{\ast}_{t}} (X, T)\big\}^2 |T = t\right] \notag\\
     \lesssim & \ h^{2(\alpha_1\wedge \alpha_2)}  + \bbE \left[ \sup_{|\beta - \beta^{\ast}_{t}|\lesssim \epsilon }[\Gamma (Y,T,\beta)   -  \Gamma (Y,T,\beta^{\ast}_{t}) ]^2 | T= t\right] \notag\\
     =& \ O \left( \epsilon^{2\alpha_{0}}\vee h^{2(\alpha_1\wedge \alpha_2)}  \right), 
 \end{align}
 where the first equality holds by \eqref{eq:lambdax_prop}; the first inequality by \eqref{eq:kernal-resluts} and the properties of  projections;  the second inequality holds by definition of $\xi(x,t)$ and $\sup_{x\in\calX} |\lambda(x)| \lesssim 1$; the third inequality holds by \eqref{eq:kernel-results2} and Jensen's inequality; the last equality by \eqref{eq:local L2 Lip}.
 
   Recalling the definition of $\Delta_3 (x)$ in \eqref{eq:def_Delta3} and using \eqref{eq:lambdax_prop},
we have
\begin{align}\label{eq:Delta3-prop}
  &\bbE [ \{ \hat{\xi} (X, T)  -  \xi (X, T)\} \sfzbar_{k_{x}} (X) K_{t,h}(T)]   
  = \bbE [ \Delta_3(X) \sfzbar_{k_{x}} (X) ] = \int_{x\in\calX} \Delta_3 (x) 
  \lambda(x)\sfzbar_{k_{x}} (x) g(x) \diff x,\\
&\bbE [ \{ b_{\beta}(X,T)  - b_{\beta^{\ast}_{t}}(X,T) \}  \sfzbar_{k_{x}} (X) K_{t,h} (T)  ] = \bbE [ \Delta_5 (X) \sfzbar_{k_{x}} (x) ] = \int_{x\in\calX} \Delta_5 (x) 
  \lambda(x)\sfzbar_{k_{x}} (x) g(x) \diff x \label{eq:Delta5-prop}.
 \end{align}
Then, we have 
\begin{align}\label{eq:|pi1h_H4|}
     & \ \bbE \left[ \sup_{|\beta - \beta^{\ast}_{t}|\lesssim \epsilon }\left|\bbE [ \{ \hat{\xi} (X, T)  -  \xi (X, T)\}K_{t,h}(T)\sfzbar_{k_{x}} (X)^{\top} ]  \tilde{\Omega}_{k_x}^{-1} \sfzbar_{k_{x}} (X) K_{t,h} (T) [ \Gamma (Y,T,\beta)  - \Gamma (Y,T,\beta^{\ast}_{t}) ] \right|^2\right]\\
     = & \ \bbE \left[ \sup_{|\beta - \beta^{\ast}_{t}|\lesssim \epsilon }\left|   \int_{x\in\calX} \Delta_3 (x) \lambda(x) \sfzbar_{k_{x}} (x)^{\top} g(x) \diff x \tilde{\Omega}_{k_x}^{-1} \sfzbar_{k_{x}} (X) K_{t,h} (T) [ \Gamma (Y,T,\beta)  - \Gamma (Y,T,\beta^{\ast}_{t}) ] \right|^2\right]\notag\\
     = & \ \bbE \left[\sup_{|\beta - \beta^{\ast}_{t}|\lesssim \epsilon } \left| \Pi_{g} \left[ \Delta_3 (X) \lambda(X) \big| \sfzbar_{k_{x}} (X) \right]  K_{t,h}(T) [ \Gamma (Y,T,\beta)  - \Gamma (Y,T,\beta^{\ast}_{t}) ] \right|^2 \right] \notag \\
     \lesssim & \ \sup_{x\in\calX} \left| \Pi_{g} \left[ \Delta_3 (\cdot) \lambda(\cdot) \big| \sfzbar_{k_{x}}  \right](x) \right|^2 \cdot \bbE \left[ \sup_{|\beta - \beta^{\ast}_{t}|\lesssim \epsilon }  \left| K_{t,h}(T) [ \Gamma (Y,T,\beta)  - \Gamma (Y,T,\beta^{\ast}_{t}) ] \right|^2 \right] \notag\\
      \lesssim & \ \sup_{x\in\calX} \left| \Delta_3 (x)\lambda(x) \right|^2 \cdot \bbE \left[ \sup_{|\beta - \beta^{\ast}_{t}|\lesssim \epsilon }  \left| K_{t,h}(T) [ \Gamma (Y,T,\beta)  - \Gamma (Y,T,\beta^{\ast}_{t}) ] \right|^2 \right] = O \left( \frac{\epsilon^{2\alpha_{0}}\vee h}{h}  \right),\notag
 \end{align}
 where the last inequality holds by \eqref{eq:supnorm_proj_g} and the last equality holds by Lemma~\ref{lemma:H3_L2} and \eqref{eq:kernal-resluts}. Hence, we obtain the result $(i)$.

 For the result $(ii)$, 
 \begin{align*}
     & \ \sup_{|\beta - \beta^{\ast}_{t}|\lesssim \epsilon } \bbE \left[ \left| K_{t,h} (T ) \bbE [ \sfzbar_{k_{x}} (X)^{\top}]\tilde{\Omega}_{k_x}^{-1} \bbE [ \sfzbar_{k_{x}} (X) K_{t,h} (T)  \{ b_{\beta}(X,T)  - b_{\beta^{\ast}_{t}}(X,T) \} ] \right|^2\right]  \\
    = & \ \bbE \left[ \left| K_{t,h} (T ) \right|^2 \right] \cdot \sup_{|\beta - \beta^{\ast}_{t}|\lesssim \epsilon }  \left|\bbE [ \sfzbar_{k_{x}} (X)^{\top}]\tilde{\Omega}_{k_x}^{-1} \bbE [ \sfzbar_{k_{x}} (X) K_{t,h} (T)  \{ b_{\beta}(X,T)  - b_{\beta^{\ast}_{t}}(X,T) \} ] \right|^2\\
     = & \ \bbE \left[ \left| K_{t,h} (T ) \right|^2 \right] \cdot \sup_{|\beta - \beta^{\ast}_{t}|\lesssim \epsilon }  \left|\bbE [ \sfzbar_{k_{x}} (X)^{\top}]\tilde{\Omega}_{k_x}^{-1} \int_{x\in\calX} \sfzbar_{k_{x}} (x) \Delta_5(x) \lambda(x) g(x) \diff x  \right|^2 \\
     = & \ \bbE \left[ \left| K_{t,h} (T ) \right|^2 \right] \cdot \sup_{|\beta - \beta^{\ast}_{t}|\lesssim \epsilon }  \left| \bbE \left[     \Pi_{g}[\Delta_5(\cdot)\lambda(\cdot)| \sfzbar_{k_{x}} ](X)  \right] \right|^2 \\
     \lesssim  & \ \bbE \left[ \left| K_{t,h} (T ) \right|^2 \right] \cdot \sup_{|\beta - \beta^{\ast}_{t}|\lesssim \epsilon } \bbE \left[ \left|    \Pi_{g}[\Delta_5(\cdot)\lambda(\cdot)| \sfzbar_{k_{x}} ](X)  \right|^2\right]\\
     \lesssim & \ \frac{ \epsilon^{2\alpha_{0}}\vee h^{2(\alpha_1\wedge \alpha_2)}}{h},
 \end{align*}
where the second equality uses \eqref{eq:Delta5-prop} and the last inequality holds by \eqref{eq:pi-Delta5}.

 For the result $(iii)$,
 \begin{align*}
     & \ \sup_{|\beta - \beta^{\ast}_{t}|\lesssim \epsilon } \bbE \left[ \left|  K_{t,h}(T) \hat{\xi} (X, T) \sfzbar_{k_{x}} (X)^{\top} \tilde{\Omega}_{k_x}^{-1} \bbE [ \sfzbar_{k_{x}} (X) K_{t,h} (T)  \{ b_{\beta}(X,T)  - b_{\beta^{\ast}_{t}}(X,T) \} ] \right|^2\right] \\
     = & \ \sup_{|\beta - \beta^{\ast}_{t}|\lesssim \epsilon } \bbE \left[ \left|  K_{t,h}(T) \hat{\xi} (X, T) \sfzbar_{k_{x}} (X)^{\top} \tilde{\Omega}_{k_x}^{-1} \int_{x\in\calX} \sfzbar_{k_{x}} (x) \Delta_5(x) \lambda(x) g(x) \diff x \right|^2\right] \\
     \lesssim & \ \sup_{|\beta - \beta^{\ast}_{t}|\lesssim \epsilon } \bbE \left[ \left|  K_{t,h}(T)  \Pi_{g}[\Delta_5(\cdot)\lambda(\cdot)| \sfzbar_{k_{x}} ](X)  \right|^2 \xi (X, T)\right] \\
     =& \ \bbE \left[ \left|  K_{t,h}(T)  \right|^2\right] \cdot\sup_{|\beta - \beta^{\ast}_{t}|\lesssim \epsilon } \bbE \left[ \left|    \Pi_{g}[\Delta_5(\cdot)\lambda(\cdot)| \sfzbar_{k_{x}} ](X)  \right|^2\right] \\
      \lesssim & \ \frac{ \epsilon^{2\alpha_{0}}\vee h^{2(\alpha_1\wedge \alpha_2)}}{h},
 \end{align*}
where the first inequality holds since $\xi(x,t)$ and $\hat{\xi}(x,t)$ are uniformly bounded, and the last inequality holds by \eqref{eq:pi-Delta5}.

 Similarly, for the result $(iv)$,  we have
 \begin{align*}
&\sup_{|\beta - \beta^{\ast}_{t}|\lesssim \epsilon } \bbE \left[ \left|  \bbE [ K_{t,h} (T ) ]  \sfzbar_{k_{x}} (X)^{\top} \tilde{\Omega}_{k_x}^{-1} \bbE [ \sfzbar_{k_{x}} (X) K_{t,h} (T)  \{ b_{\beta}(X,T)  - b_{\beta^{\ast}_{t}}(X,T) \} ] \right|^2\right] \\
= & \left|\bbE [ K_{t,h} (T ) ] \right|^2\sup_{|\beta - \beta^{\ast}_{t}|\lesssim \epsilon } \bbE \left[ \left|   \sfzbar_{k_{x}} (X)^{\top} \tilde{\Omega}_{k_x}^{-1} \int_{x\in\calX} \sfzbar_{k_{x}} (x) \Delta_5(x) \lambda(x) g(x) \diff x \right|^2\right] \\
= & \left|\bbE [ K_{t,h} (T ) ] \right|^2\sup_{|\beta - \beta^{\ast}_{t}|\lesssim \epsilon } \bbE \left[ \left|    \Pi_{g}[\Delta_5(\cdot)\lambda(\cdot)| \sfzbar_{k_{x}} ] (X) \right|^2\right] \lesssim \epsilon^{2\alpha_{0}}\vee h^{2(\alpha_1 \wedge \alpha_2)},
 \end{align*}
 where the last inequality uses $\bbE [ K_{t,h} (T ) ] \lesssim 1$ and  \eqref{eq:pi-Delta5}.
\end{proof}

\begin{lemma}\label{lemma:H4_pi1}
    Under Assumptions~\ref{as:nuisance_est}--\ref{as:basis}, \ref{as:kernel}, and  \ref{as:Lip}, for any $h\in\calH_{4,n}$, we have
\begin{align*}
    \sup_{\beta\in\calB}\sup_{o\in\calO}|\P^2 h(o;\beta)| \lesssim \frac{1}{h}.
\end{align*}
\end{lemma}

\begin{proof}

For any $h\in\calH_{4,n}$, the explicit expression of $\P^{2} h$  is given by
\begin{align}\label{eq:pi1h_H4}
     & \ \P^2 h (O;\beta) \\
     = & \ \frac{1}{3} \{ K_{t,h}(T) \hat{\xi} (X, T) - \bbE [ K_{t,h} (T ) ]\}  \sfzbar_{k_{x}} (X)^{\top} \tilde{\Omega}_{k_x}^{-1} \bbE [ \sfzbar_{k_{x}} (X) K_{t,h} (T)  \{ b_{\beta}(X,T)  - b_{\beta^{\ast}_{t}}(X,T) \} ] \notag \\ 
    & + \frac{1}{3} \bbE [ \{ \hat{\xi} (X, T)  -  \xi (X, T)\}K_{t,h}(T)\sfzbar_{k_{x}} (X)^{\top} ]  \tilde{\Omega}_{k_x}^{-1} \sfzbar_{k_{x}} (X) K_{t,h} (T) [ \Gamma (Y,T,\beta)  - \Gamma (Y,T,\beta^{\ast}_{t}) ] \notag \\
    & - \frac{1}{3}  K_{t,h} (T ) \bbE [ \sfzbar_{k_{x}} (X)^{\top}]\tilde{\Omega}_{k_x}^{-1} \bbE [ \sfzbar_{k_{x}} (X) K_{t,h} (T)  \{ b_{\beta}(X,T)  - b_{\beta^{\ast}_{t}}(X,T) \} ] .\notag
\end{align}
    Recalling the definitions of $\Delta_3(x)$ and $\Delta_4(x)$ in \eqref{eq:def_Delta3} and \eqref{eq:def_Delta5}, for any $h\in\calH_{4,n}$, we have:
\begin{align*}
    \P^2 h (O;\beta) =&\  \frac{1}{3} \{ K_{t,h}(T) \hat{\xi} (X, T)  -  \bbE [ K_{t,h} (T ) ] \} \Pi_{g}[\Delta_5(\cdot) \lambda(\cdot)| \sfzbar_{k_{x}} ](X) \\ 
   &+ \frac{1}{3} \Pi_{g} \left[ \Delta_3 (\cdot) \lambda(\cdot) \big| \sfzbar_{k_{x}}  \right](X)  K_{t,h}(T) [ \Gamma (Y,T,\beta)  - \Gamma (Y,T,\beta^{\ast}_{t}) ] \\
   &  - \frac{1}{3}K_{t,h} (T ) \bbE \left[ \Pi_{g}[\Delta_5(\cdot)\lambda(\cdot)| \sfzbar_{k_{x}} ](X) \right].
\end{align*}
By \eqref{eq:supnorm_proj_g}, we have 
\begin{align*}
\sup_{x\in\calX} \left| \Pi_{g}[ \Delta_3(\cdot) \lambda(\cdot) | \sfzbar_{k_{x}} ](x) \right| \lesssim &\sup_{x\in\calX} \left|\Delta_3(x)\lambda(x)\right| \lesssim 1,\\
  \text{and}\  \sup_{x\in\calX} \left| \Pi_{g}[\Delta_5(\cdot) \lambda(\cdot) | \sfzbar_{k_{x}} ](x) \right| \lesssim & \sup_{x\in\calX} \left|\Delta_5(x)\lambda(x)\right| \lesssim 1.
\end{align*}
Since  $  \|\hat{\xi}(\cdot,\cdot)\|_{\infty} \lesssim 1$, $\|\Gamma (\cdot,\cdot,\beta)\|_{\infty} \lesssim 1$, and $\|K_{t,h}(\cdot)\|_{\infty} \lesssim 1/h$, it follows that 
\begin{align*}
\sup_{\beta\in\calB}\sup_{o\in\calO}|\pi_1 h(o;\beta)| \lesssim \frac{1}{h},
\end{align*}
which completes the proof.
\end{proof}

\begin{lemma}   
\label{lemma:var-kernel}
Suppose that Assumptions~\ref{as:nuisance_est}--\ref{as:basis} and  \ref{as:kernel} hold and  $k_{x}, n \to \infty$ with $k_{x} = o(n^2)$. Then, for a given $\beta$, conditioning on the nuisance sample, 
\begin{align*}
    \var\left( \tilde{\psi}_{t,k}^{(2)} (\beta) \right) =  O\left( \frac{1}{nh} + \frac{k_{x}}{n^2h^2}\right).
\end{align*}
\end{lemma}
\begin{proof}
 First, we have: 
    \begin{align*}
       \var\left( \tilde{\psi}_{t,k}^{(2)} (\beta) \right) \lesssim \var\left( \bbU_{n,2}  \hat{\Upsilon}_{\beta} \right) +  \var\left( \bbU_{n, 3} \left( \IF_{\mathrm{B}_{\psi,t, k}}^{(2)} (\beta)   \right) \right),
    \end{align*}
    where 
    \begin{align*}
        &\hat{\Upsilon}_{\beta}(O_1,O_2) \coloneqq K_{t,h} (T_1)  \hat{\xi} (X_1, T_1) \{ \Gamma_{\beta} - \hat{b}_{\beta_{\mathsf{init}}} (X_1, T_1) \} +  \hat{b}_{\beta_{\mathsf{init}}} (X_1, T_2)K_{t,h}(T_2) \text{ and }\\
        &\IF_{\mathrm{B}_{\psi,t, k}}^{(2)} (\beta) \coloneqq \{K_{t,h}(T_{1}) \hat{\xi} (X_{1}, T_{1})  - K_{t,h} (T_{3} )\} \sfzbar_{k_{x}}^{\top} (X_{1}) \tilde{\Omega}_{k_x}^{-1} \sfzbar_{k_{x}} (X_{2}) \{ \Gamma (Y_2,T_2,\beta) - \hat{b}_{\beta_{\mathsf{init}}} (X_{2}, T_{2})\} K_{t,h} (T_{2}).
    \end{align*}
     Using Lemma 15 in \cite{liu2017semiparametric} and a Taylor expansion argument similar to that used in the proof of Lemma~\ref{lemma:H3_L2}, we obtain 
     \begin{align*}
        & \ \var\left( \bbU_{n,2}  \hat{\Upsilon}_{\beta} \right) = \var \left[ \bbU_{n,2} \left\{ S_{2} \left( \hat{\Upsilon}_{\beta}   \right) \right\} \right] \\
        \lesssim & \ \frac{1}{n} \bbE \left[ \left\{ \bbE \left[  S_{2} \left(  \hat{\Upsilon}_{\beta}  \right) \mid O_1\right] \right\}^2\right] + \frac{1}{n^{2}} \bbE \left[ \left\{ \bbE \left[  S_{2} \left( \hat{\Upsilon}_{\beta}  \right) \mid O_1, O_2\right] \right\}^2\right] \lesssim \frac{1}{nh},
    \end{align*}
    where $S_2$ is the symmetrization operator defined in Section~\ref{app:notation}.  
 Similarly, for $\var\left\{ \bbU_{n, 3} \left( \IF_{\mathrm{B}_{\psi,t, k}}^{(2)} (\beta)   \right) \right\}$,  we have 
     \begin{align*}
        & \ \var\left( \bbU_{n, 3} \left( \IF_{\mathrm{B}_{\psi,t, k}}^{(2)} (\beta)   \right) \right) = \var \left[ \bbU_{n,3} \left\{ S_{3} \left( \IF_{\mathrm{B}_{\psi,t, k}}^{(2)} (\beta)   \right) \right\} \right] \\
        \lesssim & \ \sum_{j=1}^{3} \frac{1}{n^{j}} \bbE \left[ \left\{ \bbE \left[  S_{3} \left( \IF_{\mathrm{B}_{\psi,t, k}}^{(2)} (\beta)  \right) \mid O_1, \cdots, O_j\right] \right\}^2\right] \eqqcolon \sum_{j=1}^{3} \zeta_{j}.
    \end{align*}
 Following \eqref{eq:pi1h_H4} and Lemma~\ref{lemma:kernel-pih-order}, we have  
\begin{align*}
	\bbE \left[ \left\{\bbE [ S(\IF_{\mathrm{B}_{\psi,t, k}}^{(2)} (\beta))  \mid O_1]\right\}^2 \right] \lesssim \frac{1}{h}, \text{ and hence } \zeta_1 \lesssim \frac{1}{nh}.
\end{align*}

For $\zeta_3$, following the proof of \eqref{eq:EH4^2}, we conclude 
  \begin{align*}
        &\zeta_{3} \lesssim \frac{1}{n^3} \bbE \left[ \left(  \IF_{\mathrm{B}_{\psi,t, k}}^{(2)} (\beta)  \right) ^2\right] \lesssim \frac{k_{x}}{n^3h^2}.
    \end{align*}
    For $\zeta_2$, by Jensen's inequality, we obtain:
    \begin{align*}
        \zeta_{2} \lesssim \frac{1}{n^2} \bbE \left[ \left(  \IF_{\mathrm{B}_{\psi,t, k}}^{(2)} (\beta)  \right) ^2\right] \lesssim \frac{k_{x}}{n^2h^2}.
    \end{align*} 
    Combining all the analysis above, we obtain 
    \begin{align*}
        \var\left( \bbU_{n, 3} \left( \IF_{\mathrm{B}_{\psi,t, k}}^{(2)} (\beta)    \right) \right) \lesssim \sum_{j=1}^{3} \zeta_{j} \lesssim \frac{1}{nh} + \frac{k_{x}}{n^2h^2},
    \end{align*}
    which completes the proof.
\end{proof}

\begin{lemma}\label{lemma:Uprocess-kernel} 
Under Assumptions~\ref{as:nuisance_est}--\ref{as:basis}, \ref{as:kernel}, and  \ref{as:Lip}, for any $\epsilon>0$,  
    \begin{align*}
      \bbE \bigg[  \sup_{ | \beta - \beta^{\ast}_{t}|\lesssim \epsilon  }\left| \left\{ \bar{\psi}_{t,k} (\beta) 
    - \tilde{\psi}_{t,k}^{(2)} (\beta) \right\} - \left\{  \bar{\psi}_{t,k} (\beta^*_t)  - \tilde{\psi}_{t,k}^{(2)} (\beta^*_t)  \right\}\right| \bigg] \lesssim \left( \frac{\sqrt{k} }{nh} + \frac{ 1  }{\sqrt{nh}} \right)\cdot ( \epsilon^{\alpha_{0}}  \vee \sqrt{h} )\cdot\log n .
    \end{align*}
\end{lemma}
\begin{proof} 
We begin by decomposing the supremum as follows:
\begin{align}
    &\sup_{|\beta - \beta^{\ast}_{t}|\lesssim \epsilon } \left|   \bar{\psi}_{t,k}(\beta)   -   \tilde{\psi}_{t,k}^{(2)} (\beta) - \bar{\psi}_{t,k} (\beta^{\ast}_{t})  + \tilde{\psi}_{t,k}^{(2)} (\beta^{\ast}_{t}) \right|\notag\\
    =& \sup_{|\beta - \beta^{\ast}_{t}|\lesssim \epsilon } \left|   \bar{\psi}_{t,k} (\beta) - \hat{\psi}^{(1)}_{t}(\beta)  - \bbU_{n,3} [\IF_{\mathrm{B}_{\psi,t, k}}^{(2)} (\beta) ]  - \bar{\psi}_{t,k} (\beta^{\ast}_{t}) +  \hat{\psi}^{(1)}_{t}(\beta_t^*)  + \bbU_{n,3} [ \IF_{\mathrm{B}_{\psi,t, k}}^{(2)} (\beta_{t}^{*}) ] \right| \notag\\
    \leq & \  \sup_{|\beta - \beta^{\ast}_{t}|\lesssim \epsilon }  \left|   \bbU_{n, 3} \left\{ \IF_{\mathrm{B}_{\psi,t, k}}^{(2)} (\beta) - \IF_{\mathrm{B}_{\psi,t, k}}^{(2)} (\beta_{t}^{*})   -  \bbE \left( \IF_{\mathrm{B}_{\psi,t, k}}^{(2)} (\beta) - \IF_{\mathrm{B}_{\psi,t, k}}^{(2)} (\beta_{t}^{*})    \right) \right\} \right| \label{eq:Uprocess-kernel-order2}\\
    & +  \sup_{|\beta - \beta^{\ast}_{t}|\lesssim \epsilon }  \left| \hat{\psi}^{(1)}_{t} (\beta) -\hat{\psi}^{(1)}_{t} (\beta_t^*) - \bbE \left[ 		\hat{\psi}^{(1)}_{t}(\beta) - \hat{\psi}^{(1)}_{t}(\beta_t^*) \right] \right|\label{eq:Uprocess-kernel-order1}. 
\end{align}
We first analyze the term~\eqref{eq:Uprocess-kernel-order1}. Consider the function class
\begin{align*}
     \calH_3 \coloneqq \left\{ h(O,\beta) =K_{t,h} (T) \hat{\xi} (X,T) [ \Gamma (Y,T,\beta)  - \Gamma (Y,T,\beta^*_t) ] : \beta\in\calB , |\beta - \beta^{\ast}_{t}|\lesssim \epsilon  \right\}.
\end{align*}
with envelope  function $\H_3 (O) \coloneqq \sup_{|\beta - \beta^{\ast}_{t}|\lesssim \epsilon } | K_{t,h} (T) \hat{\xi} (X,T) [ \Gamma (Y,T,\beta)  - \Gamma (Y,T,\beta^*_t) ] | $. By an argument similar to that for  $\calH_{1,\delta}$ (see \eqref{eq:def-H1delta}), we find that $\calH_3$ is also of VC type with bounded characteristics. Since $\|\hat{\xi} (\cdot, \cdot) \|_{\infty}\lesssim 1$,  Lemma~\ref{lemma:H3_L2} gives 
\begin{align}\label{eq:EK2}
    \bbE [ |\H_{3} (O)|^2] = & \ \bbE \left[\sup_{ |\beta - \beta^{\ast}_{t} |\lesssim \epsilon } \left|\hat{\xi}(X,T) \{ \Gamma (Y,T,\beta)  - \Gamma (Y,T,\beta^*_t) \}K_{t,h}\left(T\right) \right|^2 \right] \notag\\
    \lesssim &\ \bbE \left[\sup_{|\beta - \beta^{\ast}_{t}|\lesssim \epsilon } \left| \{ \Gamma (Y,T,\beta)  - \Gamma (Y,T,\beta^*_t) \}K_{t,h}\left(T\right) \right|^2 \right]  = O \left( \frac{\epsilon^{2\alpha_{0}}\vee h}{h}  \right).  
\end{align}
 Thus, $  \| \H_3 \|_{\P,2} \lesssim (\epsilon^{\alpha_{0}} \vee \sqrt{h})/\sqrt{h}$.  Applying Theorem~2.14.1 of \citet{van2023weak} to  $\calH_{3}$  yields
 \begin{align}\label{eq:supH3delta}
   \bbE \left[ \big| \eqref{eq:Uprocess-kernel-order1} \big| \right] = \bbE \left[ \left\|  \bbU_{n} ( h) \right\|_{\calH_{3}} \right] &\lesssim \frac{J_{1}(1, \calH_{3}, \H_{3})\cdot \|\H_{3}\|_{\P,2}}{\sqrt{n}}\lesssim \frac{ \epsilon^{\alpha_{0}}  \vee \sqrt{h}}{\sqrt{nh}}.
\end{align}

We now turn to the term  \eqref{eq:Uprocess-kernel-order2}.  Consider the function class 
\begin{align*}
    \calH_{4,n}\coloneqq&\left\{ h(O_1,O_2,O_3,\beta) =   \IF_{\mathrm{B}_{\psi,t, k}}^{(2)} (\beta) -   \IF_{\mathrm{B}_{\psi,t, k}}^{(2)} (\beta_{t}^{*}) : \beta\in \calB, |\beta - \beta^{\ast}_{t}|\lesssim \epsilon \right\}.
\end{align*}
Similarly, we can show that $\calH_{4,n}$ is also of VC type with bounded characteristics and envelope 
\begin{align*}
    \H_{4,n} (O_1,O_2,O_3) \coloneqq & \ |\{ K_{t,h}(T_{1})\hat{\xi} (X_{1}, T_{1})  -  K_{t,h} (T_{3} )\}\sfzbar_{k_{x}} (X_{1})^{\top}\tilde{\Omega}_{k_x}^{-1} \sfzbar_{k_{x}} (X_{2}) K_{t,h} (T_{2})|\\
    &\qquad\qquad\qquad\times\sup_{|\beta - \beta^{\ast}_{t}|\lesssim \epsilon} | \Gamma (Y_2,T_2,\beta)  - \Gamma (Y_2,T_2,\beta^*_t) | .
\end{align*}
Under Assumptions~\ref{as:nuisance_est}, \ref{as:basis} and \ref{as:kernel}, we have $\|\H_{4,n} (\cdot,\cdot,\cdot)\|_{\infty} \lesssim h^2 \zeta(k_x)^2 \lesssim k_x h^2$.  In addition, 
\begin{align}\label{eq:EH4^2}
    &\bbE [ |\H_{4,n} (O_1,O_2,O_3) |^2] \\
    = & \ \bbE \bigg[ \sup_{|\beta - \beta^{\ast}_{t}|\lesssim \epsilon} | \Gamma (Y,T,\beta)  - \Gamma (Y,T,\beta^*_t) |^2 K_{t,h} (T)^2 \sfzbar_{k_{x}} (X)^{\top}\tilde{\Omega}_{k_x}^{-1}\notag\\
    & \times \bbE \left[ K_{t,h}(T)^2 \{ \hat{\xi} (X, T)  -  1\}^2 \xi(X, T) \sfzbar_{k_{x}} (X) \sfzbar_{k_{x}} (X)^{\top} \right] \tilde{\Omega}_{k_x}^{-1} \sfzbar_{k_{x}} (X) \bigg] \notag\\
    \lesssim & \ \frac{1}{h}\bbE \bigg[ \sup_{ |\beta - \beta^{\ast}_{t} |\lesssim \epsilon} | \Gamma (Y,T,\beta)  - \Gamma (Y,T,\beta^*_t) |^2 K_{t,h} (T)^2 \sfzbar_{k_{x}} (X)^{\top}\tilde{\Omega}_{k_x}^{-1}\cdot \tilde{\Omega}_{k_x} \cdot \tilde{\Omega}_{k_x}^{-1} \sfzbar_{k_{x}} (X) \bigg]\notag\\
   =  & \ \frac{1}{h} \bbE \bigg[\sup_{ |\beta - \beta^{\ast}_{t} |\lesssim \epsilon} | \Gamma (Y,T,\beta)  - \Gamma (Y,T,\beta^*_t) |^2   \big| K_{t,h}(T)^2 \sfzbar_{k_{x}}  (X)^{\top} 
\tilde{\Omega}_{k_x}^{-1} \sfzbar_{k_{x}} (X)\big|\bigg] \notag\\
    \lesssim & \ \frac{\zeta(k_x)^2}{h} \bbE \left[\sup_{ |\beta - \beta^{\ast}_{t} |\lesssim \epsilon} | \Gamma (Y,T,\beta)  - \Gamma (Y,T,\beta^*_t) |^2 \cdot  \big| K_{t,h}(T) \big|^2\right] \notag\\
\lesssim & \ \frac{ k_x \cdot (\epsilon^{2\alpha_{0}} \vee h) }{h^2},\notag
\end{align}
where the inequalities follow from Lemma~\ref{lemma:H3_L2} and Assumption~\ref{as:basis}. Thus, $  \| \H_{4,n} \|_{\P^3,2} \lesssim \sqrt{k_x}(\epsilon^{\alpha_{0}} \vee \sqrt{h})/h$. Further, by Lemma~\ref{lemma:kernel-pih-order}, we have $ \sup_{h\in \calH_{4,n}}\| \P^2 h \|_{\P,2} \lesssim (\epsilon^{\alpha_{0}} \vee \sqrt{h}) /\sqrt{h} $ and by Lemma~\ref{lemma:H4_pi1}, $\sup_{h\in \calH_{4,n}} \|  \P^2 h (\cdot) \|_{\infty} \lesssim 1/h$. 
Thus, applying Lemma~\ref{lemma:Umaximal_inequality_order3} yields 
 \begin{align}\label{eq:supH4delta}
        \bbE \left[ \big| \eqref{eq:Uprocess-kernel-order1} \big| \right] =  \bbE \left[  \left\|\bbU_{n,3} ( h )\right\|_{\calH_{4,n} } \right] \lesssim  \frac{\sqrt{k_x} \cdot (\epsilon^{\alpha_{0}} \vee \sqrt{h})\cdot\log n }{nh} + \frac{(\epsilon^{\alpha_{0}} \vee \sqrt{h} ) \cdot\sqrt{\log n} }{\sqrt{nh}}.
 \end{align}
 
Therefore, combining \eqref{eq:supH3delta} and \eqref{eq:supH4delta} completes the proof.
\end{proof}


\section{Supplementary Information on the Numerical Experiments}
\setcounter{equation}{0}

\subsection{Supplementary simulation results}
\label{app:simulation}

In this section, we collect supplementary information on the simulation study conducted in Section~\ref{sec:simulation}. In particular, we display the normal-QQ plots of our proposed higher-order estimators over various settings in the simulation (i.e. varying the sample size $n$ and the smoothness parameter $s$): see Figure~\ref{fig:sim_QQplot_beta0_1} for $\beta_{0}^*$ and Figure~\ref{fig:sim_QQplot_beta1_1} for $\beta_{1}^*$ below for Case~1 and Figure~\ref{fig:sim_QQplot_beta0_2} for $\beta_{0}^*$ and Figure~\ref{fig:sim_QQplot_beta1_2} for $\beta_{1}^*$ for Case~2.

It is quite clear that the higher-order estimator proposed here for QTE with binary treatment is asymptotically normally distributed across different sample sizes $$n \in \{250, 500, 750, 1000, 1500, 2000, 2500, 3000\}$$ and different smoothness levels of the nuisance parameters with $s = 0.25, 0.4, 0.6$.

The GOE estimator \citep{ai2021unified} is obtained by solving \eqref{eq:sample_est_eq} with   $\hat{\xi}$ estimated through the following quadratic programming problem:   
\begin{align*} \min_{\boldsymbol{\xi}} 
\sum_{i=1}^{n} \frac{(\xi_i-1)^2}{2} \text{ s.t. }\ \frac{1}{n}\sum_{i=1}^{n} \xi_i \bar{\phi}_{k}(X_i,T_i)= \frac{1}{n (n-1)} \sum_{1 \leq i \neq j \leq n} \bar{\phi}_{k}(X_i, T_j). 
\end{align*}
Here, $\boldsymbol{\xi} = (\xi_1, \ldots, \xi_n)$ with $\xi_i \equiv \xi(X_i, T_i)$, and $\bar{\phi}_{k}$ is a power series basis of length $k$. The parameter $k$ is selected by minimizing a cross-validation criterion.

Sensitivity analyses of our proposed higher-order estimators are also conducted, specifically on the choice of $k$ and on the initial estimator $\bbeta_{\ini}$.

\begin{figure}[htbp]
  \centering
  \includegraphics[width=\textwidth]{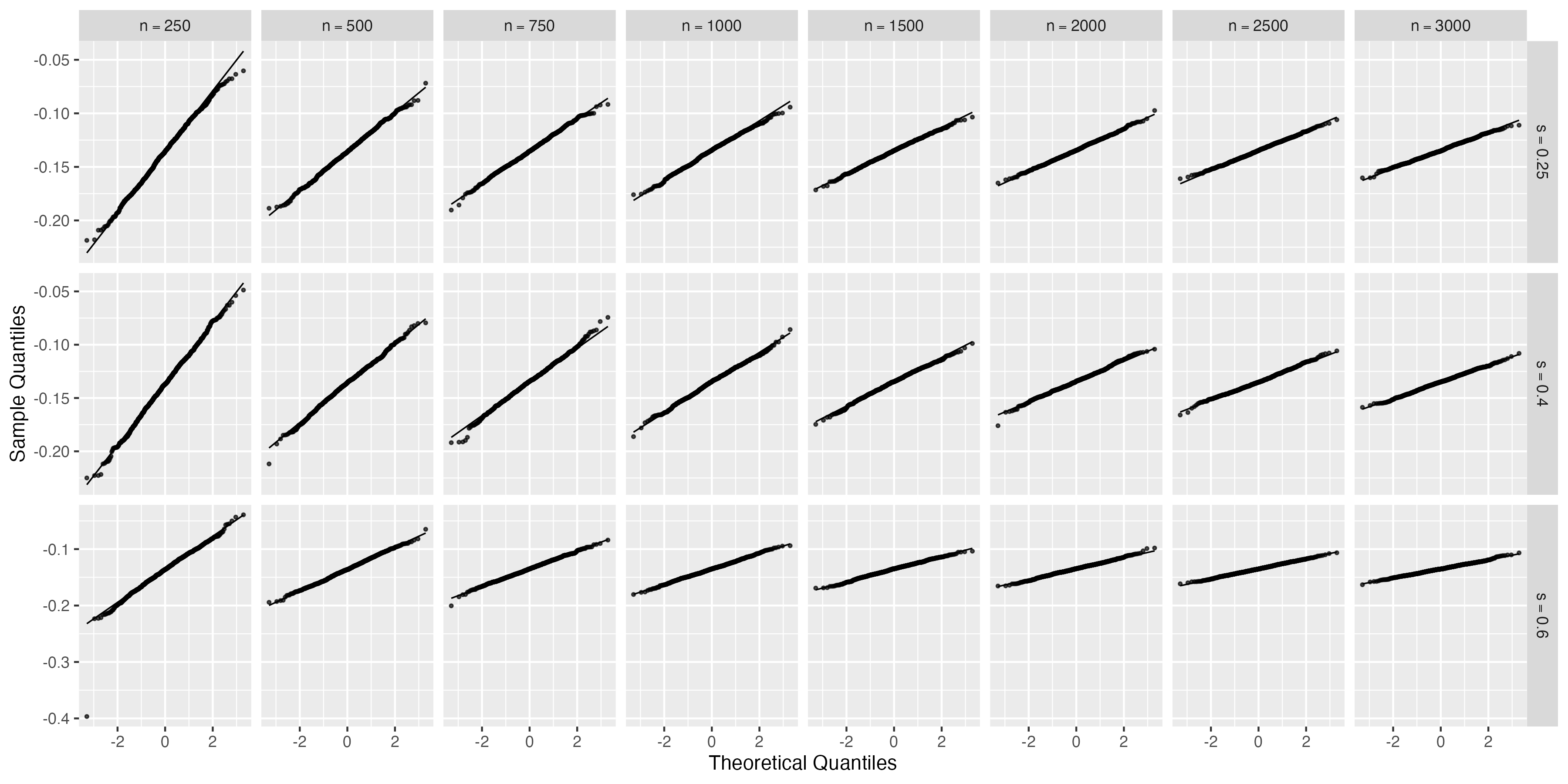} 
  \caption{Normal-QQ plots for the first coordinate of $\hat{\bbeta}^{(2)}$ across all simulation settings in Case~1, based on 1000 Monte Carlo replications.}
  \label{fig:sim_QQplot_beta0_1}
\end{figure}

\begin{figure}[htbp]
  \centering
  \includegraphics[width=\textwidth]{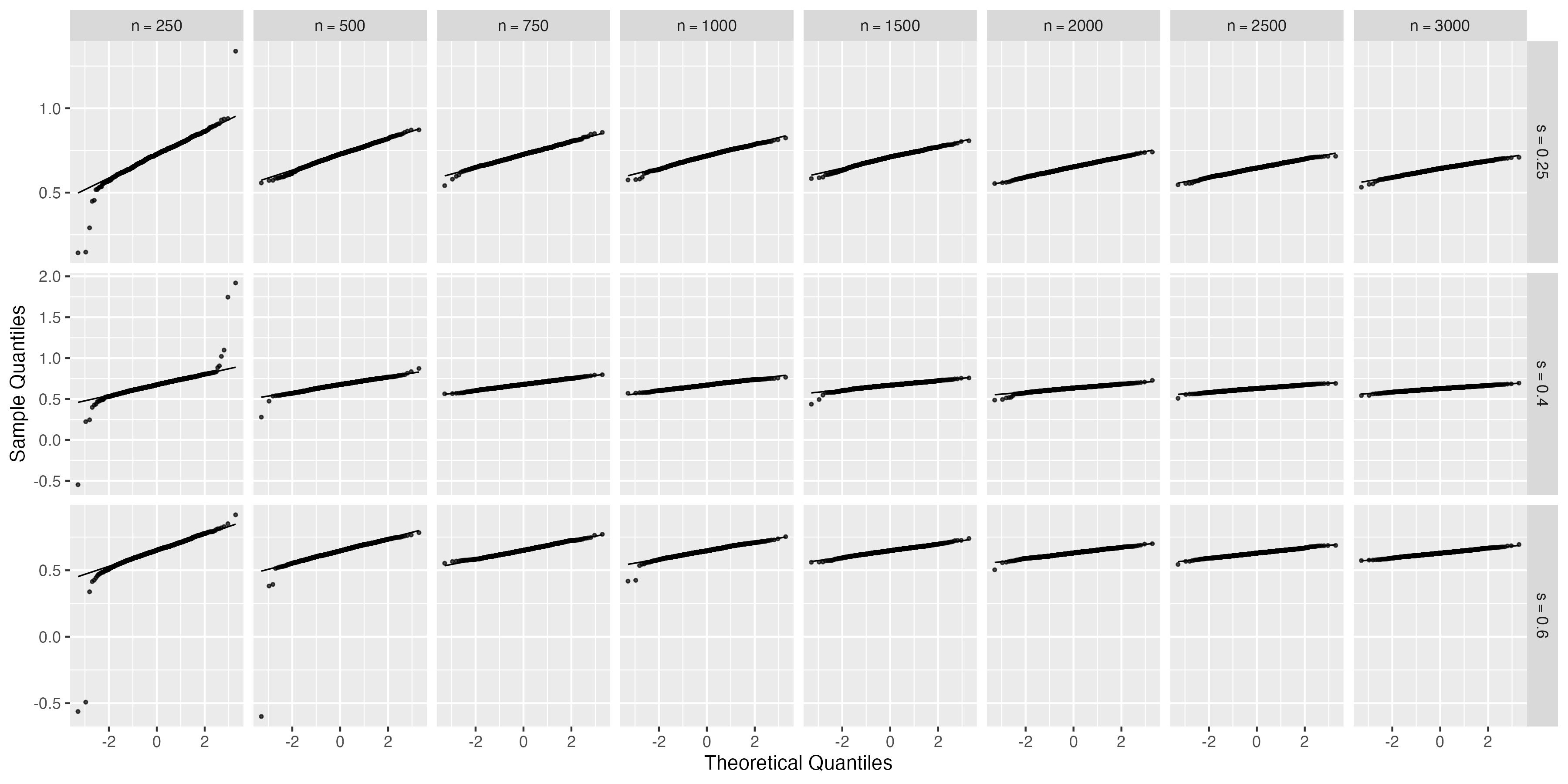} 
  \caption{Normal-QQ plots for the second coordinate of $\hat{\bbeta}^{(2)}$ across all simulation settings in Case~1, based on 1000 Monte Carlo replications.}
  \label{fig:sim_QQplot_beta1_1}
\end{figure}

\begin{figure}[htbp]
  \centering
  \includegraphics[width=\textwidth]{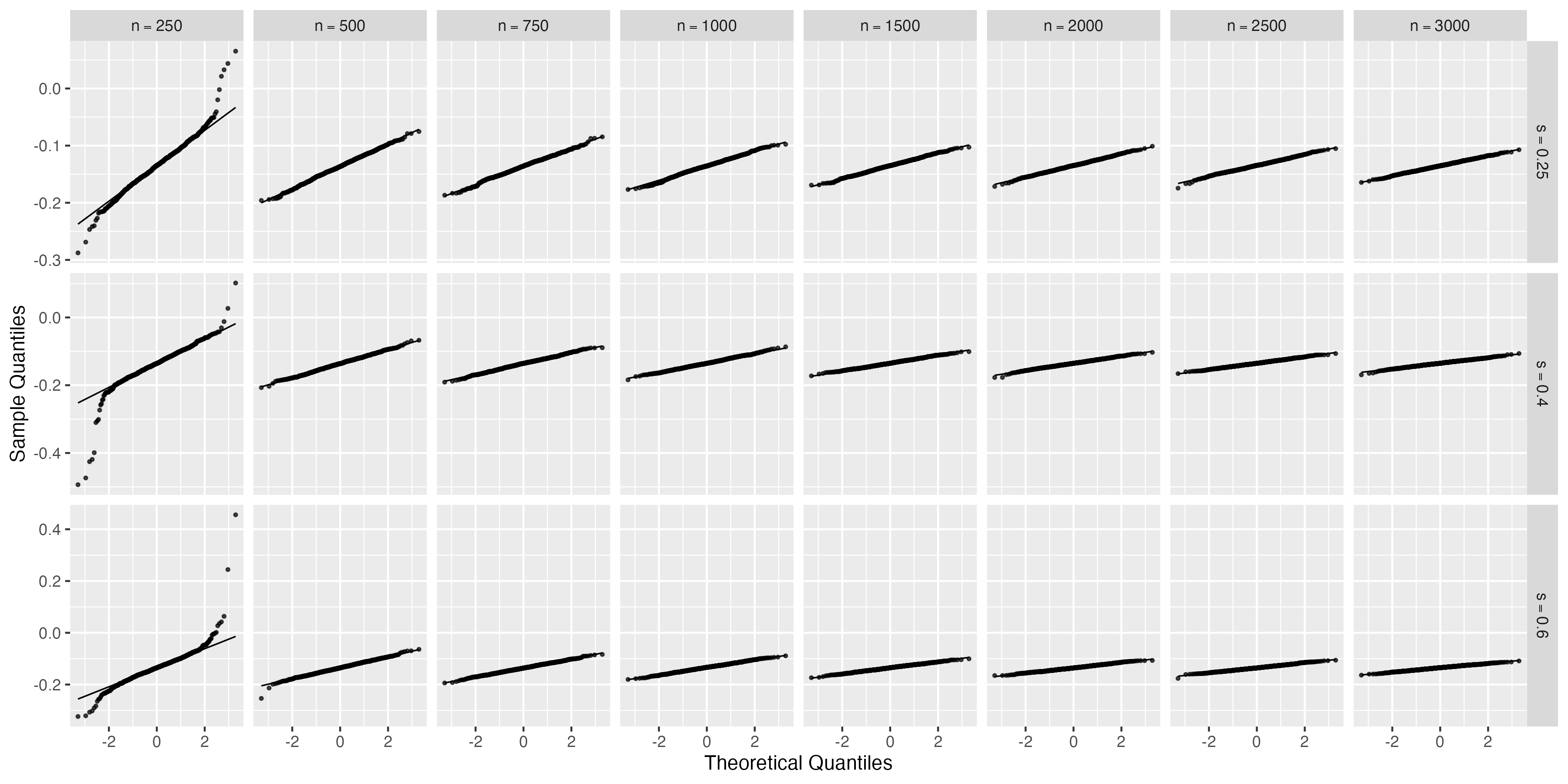} 
  \caption{Normal-QQ plots for the first coordinate of $\hat{\bbeta}^{(2)}$ across all simulation settings in Case~2, based on 1000 Monte Carlo replications.}
  \label{fig:sim_QQplot_beta0_2}
\end{figure}

\begin{figure}[htbp]
  \centering
  \includegraphics[width=\textwidth]{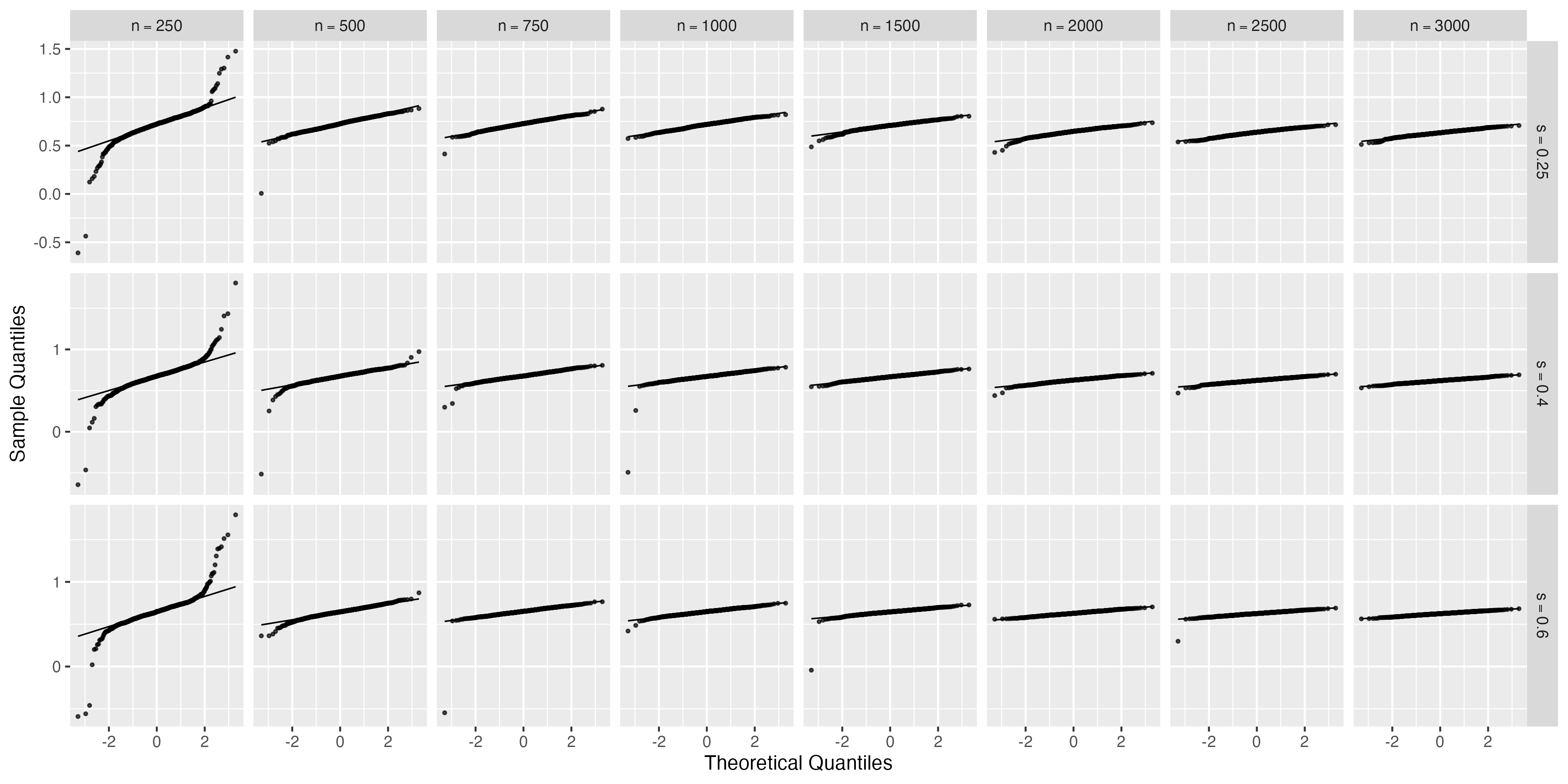} 
  \caption{Normal-QQ plots for the second coordinate of $\hat{\bbeta}^{(2)}$ across all simulation settings in Case~2, based on 1000 Monte Carlo replications.}
  \label{fig:sim_QQplot_beta1_2}
\end{figure}

\begin{figure}[htbp]
  \centering
  \includegraphics[width=\textwidth]{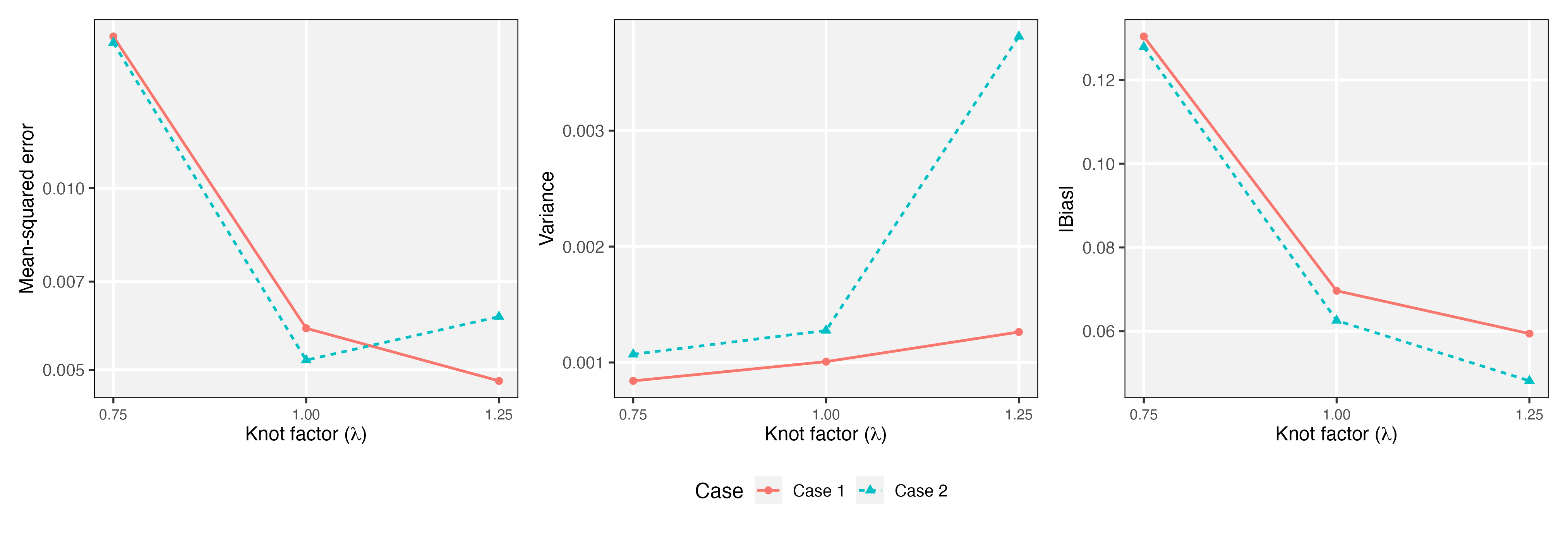} 
 \caption{Sensitivity of the higher-order estimator for the QTE, $\beta_1^*-\beta_0^*$, to the choice of spline basis dimension $k$. The horizontal axis reports the knot factor $\lambda$, where the number of interior knots for each continuous covariate is set to $\lceil \lambda n/100\rceil$. The three panels report the empirical mean squared error, variance, and absolute bias across 1000 Monte Carlo replications for Cases~1 and~2 with $n=2000$ and $s=0.25$.}
\label{fig:basis_sensitivity}
\end{figure}

\begin{figure}[htbp]
  \centering
  \includegraphics[width=\textwidth]{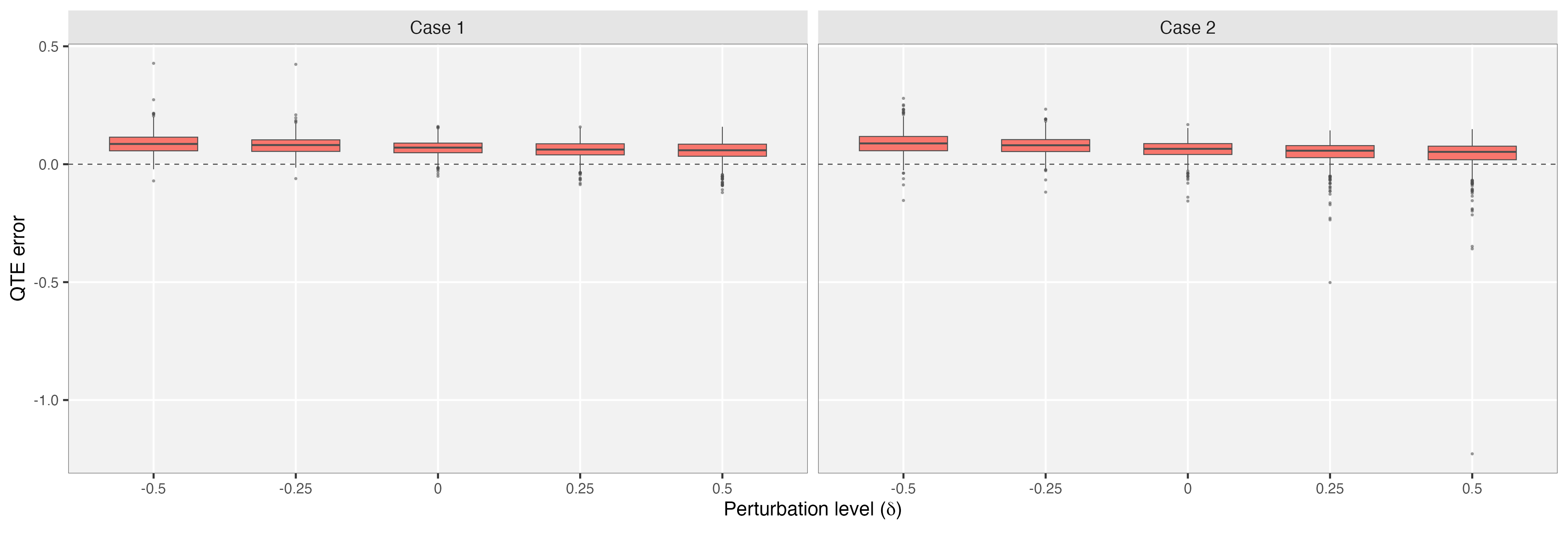} 
 \caption{Sensitivity of the higher-order estimator for the QTE, $\beta_1^*-\beta_0^*$, to perturbations in the initial estimator. The horizontal axis reports the perturbation level $\delta$, where the perturbed initial estimator is defined by $\bbeta_{\ini}^{(\delta)}=\bbeta_{\ini}+\delta(-1,1)^\top$. The boxplots display the resulting QTE estimation errors across 1000 Monte Carlo replications for Cases~1 and~2 with $n=2000$ and $s=0.25$.}
\label{fig:betainit_sensitivity}
\end{figure}

\subsection{Supplementary results on real data analysis}
\label{app:real}

This subsection provides additional numerical results for Section~\ref{sec:real}. Specifically, Tables~\ref{tab:basis_sensitivity}--\ref{tab:betainit_sensitivity} report sensitivity analyses of our higher-order estimator with respect to the choice of $k$ and to the quality of the initial estimator $\bbeta_{\ini}$. A summary of the results is provided at the end of Section~\ref{sec:real}.

\begin{table}[ht]
\centering
\begin{tabular}{cccccc}
  \toprule
$\tau$ & Knot Scale ($\lambda$)& Forest & Neural Net & Boosting & LASSO \\  
  \hline
\multirow{3}{*}{0.25} & 0.75   & 1.00 & 1.09 & 1.09 & 1.11 \\ 
   & 1.00   & 1.01 & 1.09 & 1.09 & 1.11 \\ 
   & 1.25   & 1.00 & 1.09 & 1.06 & 1.11 \\ 
   \hline
  \multirow{3}{*}{0.50} & 0.75   & 4.67 & 4.47 & 4.50 & 4.45 \\ 
   & 1.00   & 4.74 & 4.37 & 4.40 & 4.35 \\ 
   & 1.25   & 4.52 & 4.24 & 4.35 & 4.22 \\ 
   \hline
  \multirow{3}{*}{0.75} & 0.75   & 13.08 & 11.70 & 12.04 & 11.79 \\ 
   & 1.00   & 12.68 & 11.77 & 12.15 & 11.90 \\ 
   & 1.25   & 12.81 & 11.75 & 12.15 & 11.90 \\ 
   \bottomrule
\end{tabular}
\caption{Sensitivity of the HOE estimates of the QTE of 401(k) eligibility (in thousand dollars) to the choice of spline basis size. The knot scale $\lambda$, defined in Section~\ref{sec:real}, controls the size of the spline basis, with a larger value corresponding to a larger basis. Here $\tau\in\{25\%,50\%,75\%\}$ denotes the quantile level of the QTE.}
\label{tab:basis_sensitivity}
\end{table}

\begin{table}[H]
\centering

\begin{tabular}{cccccc}
  \toprule
$\tau$ & Perturbation Level ($\delta$) & Forest & Neural Net & Boosting & LASSO \\ 
  \hline
 \multirow{5}{*}{0.25} & -0.50 & 1.10 & 1.12 & 1.11 & 1.12 \\ 
   & -0.25 & 1.05 & 1.12 & 1.10 & 1.11 \\ 
   & 0.00 & 1.01 & 1.09 & 1.09 & 1.11 \\ 
   & 0.25 & 1.00 & 1.11 & 1.09 & 1.11 \\ 
   & 0.50 & 1.00 & 1.09 & 1.06 & 1.11 \\ 
   \hline
   \multirow{5}{*}{0.50} & -0.50 & 4.75 & 4.40 & 4.56 & 4.36 \\ 
   & -0.25 & 4.85 & 4.30 & 4.53 & 4.32 \\ 
   & 0.00 & 4.74 & 4.37 & 4.40 & 4.35 \\ 
   & 0.25 & 4.67 & 4.26 & 4.42 & 4.35 \\ 
   & 0.50 & 4.30 & 4.25 & 4.45 & 4.35 \\ 
   \hline
   \multirow{5}{*}{0.75} & -0.50 & 12.68 & 11.90 & 12.10 & 11.90 \\ 
   & -0.25 & 12.68 & 11.80 & 12.15 & 11.90 \\ 
   & 0.00 & 12.68 & 11.77 & 12.15 & 11.90 \\ 
   & 0.25 & 12.43 & 11.77 & 12.15 & 11.90 \\ 
   & 0.50 & 12.63 & 11.77 & 12.15 & 11.90 \\ 
   \bottomrule
\end{tabular}
\caption{Sensitivity of the HOE estimates of the QTE of 401(k) eligibility (in thousand dollars) to deterministic perturbations of the initial estimator. The perturbation follows the same scheme as in the simulation study, namely $\bbeta_{\ini}^{(\delta)}=\bbeta_{\ini}+\delta(-1,1)^\top$, where $\delta$ denotes the perturbation level. Here $\tau\in\{25\%,50\%,75\%\}$ denotes the quantile level of the QTE.}
\label{tab:betainit_sensitivity}
\end{table}
\end{appendices}

\putbib[\myreferences]
\end{bibunit}
	
\end{document}

%% file: preamble.tex
\usepackage{xcolor}
\definecolor{darkblue}{rgb}{0.0, 0.0, 0.55}
\definecolor{chicago-maroon}{RGB}{128,0,0}
\usepackage[bookmarks, colorlinks=true, plainpages = false, citecolor=darkblue, linkcolor=teal, anchorcolor=red, urlcolor=chicago-maroon, filecolor=chicago-maroon]{hyperref}
\usepackage{url}
\usepackage{amsmath, amsfonts, amsthm, amssymb, bm, bbm, verbatim, dsfont, mathtools, mathrsfs}
\usepackage{color, graphicx, appendix}
\usepackage{etoolbox}
\usepackage{almendra}
\usepackage{authblk}
\usepackage{array}
\usepackage{multirow}
\usepackage{bigints}
\usepackage{ulem}
\usepackage{cancel}

\usepackage{float}
\usepackage{pgf,tikz}
\usetikzlibrary{arrows.meta,shapes.arrows,shapes.geometric,shapes.multipart,fit,automata,decorations.pathmorphing,positioning}
\usetikzlibrary{calc}
\tikzset{
	>=Latex,
	true/.style={
		rectangle,
		draw=black, very thick,
		text width=6.5em,
		minimum height=2em,
		text centered,
		fill=gray, opacity = 0.5},
	punkt/.style={
		rectangle,
		rounded corners,
		draw=black, very thick,
		text width=6.5em,
		minimum height=2em,
		text centered},
	est/.style={
		circle,
		draw=black, very thick,
		text centered},
	shade/.style={
		circle,
		draw=black, very thick, fill=gray!50,
		text centered},
	weight/.style={
		circle,
		draw=black, very thick,
		text width=6.5em,
		minimum height=2em,
		text centered},
	pil/.style={
		->,
		thick,
		shorten <=2pt,
		shorten >=2pt,},
	double/.style={
		<->,
		thick,
		shorten <=2pt,
		shorten >=2pt,},
	dash/.style={
		dashed,
		thick,
		shorten <=2pt,
		shorten >=2pt,},
	dashdouble/.style={
		<->,
		dashed,
		thick,
		shorten <=2pt,
		shorten >=2pt,}
}

\usepackage{tikz-cd}
\makeatletter 
\newcolumntype{C}[1]{>{\centering\arraybackslash}p{#1}}

\usepackage{epstopdf}
\usepackage{fullpage}
\usepackage{algorithm}
\usepackage{algorithmicx}
\usepackage{subcaption}

\usepackage[nameinlink]{cleveref}
\usepackage[round]{natbib}

\usepackage[utf8]{inputenc} 
\usepackage[T1]{fontenc}    
\usepackage{url}            
\usepackage{booktabs}     
\usepackage{nicefrac}       
\usepackage{microtype}      
\usepackage{xcolor}
\usepackage{pdfpages}
\usepackage{scalerel}
\usepackage{dashrule}
\usepackage{lmodern}
\usepackage{fancyhdr}
\usepackage{tabu}
\usepackage{enumitem}

\def\IIFF{\mathbb{IF}}
\def\IF{\mathsf{IF}}

\def\E{\mathsf{E}}

\def\var{\mathrm{var}}
\def\cov{\mathrm{cov}}

\def\Holder{\text{H\"{o}lder}}

\def\op{\mathrm{op}}

\renewcommand{\[}{\left[}

\renewcommand{\hat}{\widehat}
\renewcommand{\tilde}{\widetilde}

\newcommand{\myreferences}{Master.bib}

\usepackage{xr,xspace}
\usepackage{todonotes}

\usepackage{caption,subcaption,soul}
\usepackage{algpseudocode}
\makeatletter
\usepackage{romanbar}

\theoremstyle{plain}
\newtheorem{theorem}{Theorem}
\newtheorem{lemma}{Lemma}
\newtheorem{proposition}{Proposition}

\theoremstyle{definition}
\newtheorem{definition}{Definition}
\newtheorem{example}{Example}

\newtheorem{assumption}{Assumption}

\newtheorem{problem}{Problem}

\newtheorem{remark}{Remark}
\AtBeginEnvironment{remark}{%
  \pushQED{\qed}%
}
\AtEndEnvironment{remark}{\popQED\endexample}

\newtheorem*{remark*}{Remark}

\usepackage{algorithm}
\usepackage{algpseudocode}

\usepackage{xspace, prettyref}

\newcommand{\diff}{{\mathrm d}}

\newcommand\indep{\protect\mathpalette{\protect\independenT}{\perp}}
\def\independenT#1#2{\mathrel{\rlap{$#1#2$}\mkern2mu{#1#2}}}

\def\E{\mathbb{E}}

\newrefformat{eq}{(\ref{#1})}
\newrefformat{chap}{Chapter~\ref{#1}}
\newrefformat{sec}{Section~\ref{#1}}
\newrefformat{alg}{Algorithm~\ref{#1}}
\newrefformat{fig}{Fig.~\ref{#1}}
\newrefformat{tab}{Table~\ref{#1}}
\newrefformat{rmk}{Remark~\ref{#1}}
\newrefformat{clm}{Claim~\ref{#1}}
\newrefformat{def}{Definition~\ref{#1}}
\newrefformat{cor}{Corollary~\ref{#1}}
\newrefformat{lmm}{Lemma~\ref{#1}}
\newrefformat{prop}{Proposition~\ref{#1}}
\newrefformat{prob}{Problem~\ref{#1}}
\newrefformat{app}{Appendix~\ref{#1}}
\newrefformat{hyp}{Hypothesis~\ref{#1}}
\newrefformat{thm}{Theorem~\ref{#1}}

\newcommand{\sfu}{{\mathsf{u}}}

\newcommand{\sfw}{{\mathsf{w}}}

\newcommand{\sfz}{{\mathsf{z}}}

\newcommand{\sfN}{{\mathsf{N}}}

\newcommand{\calB}{{\mathcal{B}}}
\newcommand{\calC}{{\mathcal{C}}}

\newcommand{\calF}{{\mathcal{F}}}
\newcommand{\calG}{{\mathcal{G}}}
\newcommand{\calH}{{\mathcal{H}}}

\newcommand{\calM}{{\mathcal{M}}}

\newcommand{\calO}{{\mathcal{O}}}

\newcommand{\calS}{{\mathcal{S}}}
\newcommand{\calT}{{\mathcal{T}}}
\newcommand{\calU}{{\mathcal{U}}}

\newcommand{\calX}{{\mathcal{X}}}
\newcommand{\calY}{{\mathcal{Y}}}

\def\TB{\mathrm{TB}}

\renewcommand{\tilde}{\widetilde}
\renewcommand{\hat}{\widehat}

\newcommand{\bma}{{\bm{a}}}
\newcommand{\bmb}{{\bm{b}}}

\newcommand{\bbE}{{\mathbb{E}}}

\newcommand{\bbM}{{\mathbb{M}}}

\newcommand{\bbP}{{\mathbb{P}}}

\newcommand{\bbR}{{\mathbb{R}}}

\newcommand{\bbU}{{\mathbb{U}}}

\newcommand{\bbZ}{{\mathbb{Z}}}

\newcommand{\sfzbar}{\bar{\sfz}}

\usepackage{tcolorbox}

\def\proj{\text{proj}}

\makeatletter
\def\ubar#1{\underline{\sbox\tw@{#1}\dp\tw@\z@\box\tw@}}
\makeatother

\def\Hajek{Haj\'{e}k}

\usepackage{caption}

%% file: HOIF_GTM_arXiv.bbl
\begin{thebibliography}{75}
\providecommand{\natexlab}[1]{#1}
\providecommand{\url}[1]{\texttt{#1}}
\expandafter\ifx\csname urlstyle\endcsname\relax
  \providecommand{\doi}[1]{doi: #1}\else
  \providecommand{\doi}{doi: \begingroup \urlstyle{rm}\Url}\fi

\bibitem[Abadie(2003)]{abadie2003semiparametric}
Alberto Abadie.
\newblock Semiparametric instrumental variable estimation of treatment response
  models.
\newblock \emph{Journal of Econometrics}, 113\penalty0 (2):\penalty0 231--263,
  2003.

\bibitem[Ai et~al.(2021)Ai, Linton, Motegi, and Zhang]{ai2021unified}
Chunrong Ai, Oliver Linton, Kaiji Motegi, and Zheng Zhang.
\newblock A unified framework for efficient estimation of general treatment
  models.
\newblock \emph{Quantitative Economics}, 12\penalty0 (3):\penalty0 779--816,
  2021.

\bibitem[Ao et~al.(2021)Ao, Calonico, and Lee]{ao2021multivalued}
Wallice Ao, Sebastian Calonico, and Ying-Ying Lee.
\newblock Multivalued treatments and decomposition analysis: An application to
  the {WIA} program.
\newblock \emph{Journal of Business \& Economic Statistics}, 39\penalty0
  (1):\penalty0 358--371, 2021.

\bibitem[Armstrong et~al.(2020)Armstrong, Koles{\'a}r, and
  Kwon]{armstrong2020bias}
Timothy~B Armstrong, Michal Koles{\'a}r, and Soonwoo Kwon.
\newblock Bias-aware inference in regularized regression models.
\newblock \emph{arXiv preprint arXiv:2012.14823}, 2020.

\bibitem[Athey et~al.(2018)Athey, Imbens, and Wager]{athey2018approximate}
Susan Athey, Guido~W Imbens, and Stefan Wager.
\newblock Approximate residual balancing: Debiased inference of average
  treatment effects in high dimensions.
\newblock \emph{Journal of the Royal Statistical Society Series B: Statistical
  Methodology}, 80\penalty0 (4):\penalty0 597--623, 2018.

\bibitem[Balakrishnan et~al.(2026)Balakrishnan, Kennedy, and
  Wasserman]{balakrishnan2026fundamental}
Sivaraman Balakrishnan, Edward~H Kennedy, and Larry Wasserman.
\newblock The fundamental limits of structure-agnostic functional estimation.
\newblock \emph{Statistical Science}, 2026.

\bibitem[Belloni et~al.(2015)Belloni, Chernozhukov, Chetverikov, and
  Kato]{belloni2015some}
Alexandre Belloni, Victor Chernozhukov, Denis Chetverikov, and Kengo Kato.
\newblock Some new asymptotic theory for least squares series: Pointwise and
  uniform results.
\newblock \emph{Journal of Econometrics}, 186\penalty0 (2):\penalty0 345--366,
  2015.

\bibitem[Benjamin(2003)]{benjamin2003does}
Daniel~J Benjamin.
\newblock Does 401 (k) eligibility increase saving? {E}vidence from propensity
  score subclassification.
\newblock \emph{Journal of Public Economics}, 87\penalty0 (5-6):\penalty0
  1259--1290, 2003.

\bibitem[Bhattacharya and Ghosh(1992)]{bhattacharya1992class}
Rabi~N Bhattacharya and Jayanta~K Ghosh.
\newblock A class of ${U}$-statistics and asymptotic normality of the number of
  $k$-clusters.
\newblock \emph{Journal of Multivariate Analysis}, 43\penalty0 (2):\penalty0
  300--330, 1992.

\bibitem[Bickel et~al.(1998)Bickel, Klaassen, Ritov, and
  Wellner]{bickel1998efficient}
Peter~J Bickel, Chris A~J Klaassen, Ya'acov Ritov, and Jon~A Wellner.
\newblock \emph{Efficient and Adaptive Estimation for Semiparametric Models}.
\newblock Johns Hopkins Series in the Mathematical Sciences. Springer New York,
  1998.
\newblock ISBN 9780387984735.

\bibitem[Bonhomme et~al.(2026)Bonhomme, Jochmans, Newey, and
  Weidner]{bonhomme2026higher}
St{\'e}phane Bonhomme, Koen Jochmans, Whitney~K Newey, and Martin Weidner.
\newblock Higher-order {N}eyman orthogonality in moment-condition models.
\newblock \emph{arXiv preprint arXiv:2605.10842}, 2026.

\bibitem[Bonvini and Kennedy(2022)]{bonvini2022fast}
Matteo Bonvini and Edward~H Kennedy.
\newblock Fast convergence rates for dose-response estimation.
\newblock \emph{arXiv preprint arXiv:2207.11825}, 2022.

\bibitem[Bonvini et~al.(2022)Bonvini, Kennedy, Ventura, and
  Wasserman]{bonvini2022sensitivity}
Matteo Bonvini, Edward Kennedy, Valerie Ventura, and Larry Wasserman.
\newblock Sensitivity analysis for marginal structural models.
\newblock \emph{arXiv preprint arXiv:2210.04681}, 2022.

\bibitem[Bonvini et~al.(2024)Bonvini, Kennedy, Dukes, and
  Balakrishnan]{bonvini2024doubly}
Matteo Bonvini, Edward~H Kennedy, Oliver Dukes, and Sivaraman Balakrishnan.
\newblock Doubly-robust inference and optimality in structure-agnostic models
  with smoothness.
\newblock \emph{arXiv preprint arXiv:2405.08525}, 2024.

\bibitem[Breunig and Chen(2024)]{breunig2024adaptive}
Christoph Breunig and Xiaohong Chen.
\newblock Adaptive, rate-optimal hypothesis testing in nonparametric {IV}
  models.
\newblock \emph{Econometrica}, 92\penalty0 (6):\penalty0 2027--2067, 2024.

\bibitem[Bruns-Smith et~al.(2026)Bruns-Smith, Dukes, Feller, and
  Ogburn]{bruns2026augmented}
David Bruns-Smith, Oliver Dukes, Avi Feller, and Elizabeth~L Ogburn.
\newblock Augmented balancing weights as linear regression.
\newblock \emph{Journal of the Royal Statistical Society Series B: Statistical
  Methodology}, 2026.

\bibitem[Cattaneo(2010)]{cattaneo2010efficient}
Matias~D Cattaneo.
\newblock Efficient semiparametric estimation of multi-valued treatment effects
  under ignorability.
\newblock \emph{Journal of Econometrics}, 155\penalty0 (2):\penalty0 138--154,
  2010.

\bibitem[Cattaneo and Jansson(2018)]{cattaneo2018kernel}
Matias~D Cattaneo and Michael Jansson.
\newblock Kernel-based semiparametric estimators: Small bandwidth asymptotics
  and bootstrap consistency.
\newblock \emph{Econometrica}, 86\penalty0 (3):\penalty0 955--995, 2018.

\bibitem[Cattaneo et~al.(2019)Cattaneo, Jansson, and Ma]{cattaneo2019two}
Matias~D Cattaneo, Michael Jansson, and Xinwei Ma.
\newblock Two-step estimation and inference with possibly many included
  covariates.
\newblock \emph{The Review of Economic Studies}, 86\penalty0 (3):\penalty0
  1095--1122, 2019.

\bibitem[Cattaneo et~al.(2020)Cattaneo, Farrell, and Feng]{cattaneo2020large}
Matias~D Cattaneo, Max~H Farrell, and Yingjie Feng.
\newblock Large sample properties of partitioning-based series estimators.
\newblock \emph{The Annals of Statistics}, 48\penalty0 (3):\penalty0
  1718--1741, 2020.

\bibitem[Cattaneo et~al.(2025)Cattaneo, Farrell, Jansson, and
  Masini]{cattaneo2025higher}
Matias~D Cattaneo, Max~H Farrell, Michael Jansson, and Ricardo~P Masini.
\newblock Higher-order refinements of small bandwidth asymptotics for
  density-weighted average derivative estimators.
\newblock \emph{Journal of Econometrics}, 252\penalty0 (Part B):\penalty0
  105855, 2025.

\bibitem[Cavaliere et~al.(2024)Cavaliere, Gon{\c{c}}alves, Nielsen, and
  Zanelli]{cavaliere2024bootstrap}
Giuseppe Cavaliere, S{\'\i}lvia Gon{\c{c}}alves, Morten~{\O}rregaard Nielsen,
  and Edoardo Zanelli.
\newblock Bootstrap inference in the presence of bias.
\newblock \emph{Journal of the American Statistical Association}, 119\penalty0
  (548):\penalty0 2908--2918, 2024.

\bibitem[Chakrabortty and Kuchibhotla(2025)]{chakrabortty2025tail}
Abhishek Chakrabortty and Arun~K Kuchibhotla.
\newblock Tail bounds for canonical ${U}$-statistics and ${U}$-processes with
  unbounded kernels.
\newblock \emph{arXiv preprint arXiv:2504.01318}, 2025.

\bibitem[Chen and Xie(2025)]{chen2025local}
Xiaohong Chen and Haitian Xie.
\newblock Local overidentification and efficiency gains in modern causal
  inference and data combination.
\newblock \emph{arXiv preprint arXiv:2510.16683}, 2025.

\bibitem[Chen et~al.(2003)Chen, Linton, and Van~Keilegom]{chen2003estimation}
Xiaohong Chen, Oliver Linton, and Ingrid Van~Keilegom.
\newblock Estimation of semiparametric models when the criterion function is
  not smooth.
\newblock \emph{Econometrica}, 71\penalty0 (5):\penalty0 1591--1608, 2003.

\bibitem[Chen et~al.(2024{\natexlab{a}})Chen, Liu, Ma, and
  Zhang]{chen2024causal}
Xiaohong Chen, Ying Liu, Shujie Ma, and Zheng Zhang.
\newblock Causal inference of general treatment effects using neural networks
  with a diverging number of confounders.
\newblock \emph{Journal of Econometrics}, 238\penalty0 (1):\penalty0 105555,
  2024{\natexlab{a}}.

\bibitem[Chen and Kato(2020)]{chen2020jackknife}
Xiaohui Chen and Kengo Kato.
\newblock Jackknife multiplier bootstrap: Finite sample approximations to the
  ${U}$-process supremum with applications.
\newblock \emph{Probability Theory and Related Fields}, 176\penalty0
  (3):\penalty0 1097--1163, 2020.

\bibitem[Chen et~al.(2024{\natexlab{b}})Chen, Liu, and
  Mukherjee]{chen2024method}
Xingyu Chen, Lin Liu, and Rajarshi Mukherjee.
\newblock Method-of-moments inference for {GLM}s and doubly-robust functionals
  under proportional asymptotics.
\newblock \emph{arXiv preprint arXiv:2408.06103}, 2024{\natexlab{b}}.

\bibitem[Chernozhukov and Hansen(2004)]{chernozhukov2004effects}
Victor Chernozhukov and Christian Hansen.
\newblock The effects of 401 (k) participation on the wealth distribution: An
  instrumental quantile regression analysis.
\newblock \emph{Review of Economics and Statistics}, 86\penalty0 (3):\penalty0
  735--751, 2004.

\bibitem[Chernozhukov and Hansen(2005)]{chernozhukov2005iv}
Victor Chernozhukov and Christian Hansen.
\newblock An {IV} model of quantile treatment effects.
\newblock \emph{Econometrica}, 73\penalty0 (1):\penalty0 245--261, 2005.

\bibitem[Chernozhukov et~al.(2013)Chernozhukov, Fern{\'a}ndez-Val, and
  Melly]{chernozhukov2013inference}
Victor Chernozhukov, Iv{\'a}n Fern{\'a}ndez-Val, and Blaise Melly.
\newblock Inference on counterfactual distributions.
\newblock \emph{Econometrica}, 81\penalty0 (6):\penalty0 2205--2268, 2013.

\bibitem[Chernozhukov et~al.(2018)Chernozhukov, Chetverikov, Demirer, Duflo,
  Hansen, Newey, and Robins]{chernozhukov2018double}
Victor Chernozhukov, Denis Chetverikov, Mert Demirer, Esther Duflo, Christian
  Hansen, Whitney Newey, and James Robins.
\newblock Double/debiased machine learning for treatment and structural
  parameters.
\newblock \emph{The Econometrics Journal}, 21\penalty0 (1):\penalty0 C1--C68,
  2018.

\bibitem[Chernozhukov et~al.(2022)Chernozhukov, Escanciano, Ichimura, Newey,
  and Robins]{chernozhukov2022locally}
Victor Chernozhukov, Juan~Carlos Escanciano, Hidehiko Ichimura, Whitney~K
  Newey, and James~M Robins.
\newblock Locally robust semiparametric estimation.
\newblock \emph{Econometrica}, 90\penalty0 (4):\penalty0 1501--1535, 2022.

\bibitem[Chernozhukov et~al.(2025)Chernozhukov, Deaner, Gao, Hausman, and
  Newey]{chernozhukov2025linear}
Victor Chernozhukov, Ben Deaner, Ying Gao, Jerry Hausman, and Whitney~K Newey.
\newblock Fisher-{S}chultz lecture: Linear estimation of structural and causal
  effects for nonseparable panel data.
\newblock Technical Report 33325, National Bureau of Economic Research, 2025.

\bibitem[Colangelo and Lee(2026)]{colangelo2026double}
Kyle Colangelo and Ying-Ying Lee.
\newblock Double debiased machine learning nonparametric inference with
  continuous treatments.
\newblock \emph{Journal of Business \& Economic Statistics}, 44\penalty0
  (1):\penalty0 67--79, 2026.

\bibitem[Donald and Hsu(2014)]{donald2014estimation}
Stephen~G Donald and Yu-Chin Hsu.
\newblock Estimation and inference for distribution functions and quantile
  functions in treatment effect models.
\newblock \emph{Journal of Econometrics}, 178:\penalty0 383--397, 2014.

\bibitem[Fan et~al.(2025)Fan, Qi, and Xu]{fan2025policy}
Yanqin Fan, Yuan Qi, and Gaoqian Xu.
\newblock Policy learning with $\alpha$-expected welfare.
\newblock \emph{arXiv preprint arXiv:2505.00256}, 2025.

\bibitem[Firpo(2007)]{firpo2007efficient}
Sergio Firpo.
\newblock Efficient semiparametric estimation of quantile treatment effects.
\newblock \emph{Econometrica}, 75\penalty0 (1):\penalty0 259--276, 2007.

\bibitem[Galvao and Wang(2015)]{galvao2015uniformly}
Antonio~F Galvao and Liang Wang.
\newblock Uniformly semiparametric efficient estimation of treatment effects
  with a continuous treatment.
\newblock \emph{Journal of the American Statistical Association}, 110\penalty0
  (512):\penalty0 1528--1542, 2015.

\bibitem[Hahn(1998)]{hahn1998role}
Jinyong Hahn.
\newblock On the role of the propensity score in efficient semiparametric
  estimation of average treatment effects.
\newblock \emph{Econometrica}, 66\penalty0 (2):\penalty0 315--331, 1998.

\bibitem[Hirano et~al.(2003)Hirano, Imbens, and Ridder]{hirano2003efficient}
Keisuke Hirano, Guido~W Imbens, and Geert Ridder.
\newblock Efficient estimation of average treatment effects using the estimated
  propensity score.
\newblock \emph{Econometrica}, 71\penalty0 (4):\penalty0 1161--1189, 2003.

\bibitem[Jin and Syrgkanis(2025)]{jin2025sharp}
Jikai Jin and Vasilis Syrgkanis.
\newblock Sharp structure-agnostic lower bounds for general linear functional
  estimation.
\newblock \emph{arXiv preprint arXiv:2512.17341}, 2025.

\bibitem[Kallus et~al.(2024)Kallus, Mao, and Uehara]{kallus2024localized}
Nathan Kallus, Xiaojie Mao, and Masatoshi Uehara.
\newblock Localized debiased machine learning: Efficient inference on quantile
  treatment effects and beyond.
\newblock \emph{Journal of Machine Learning Research}, 25:\penalty0 1--59,
  2024.

\bibitem[Kennedy(2019)]{kennedy2019nonparametric}
Edward~H Kennedy.
\newblock Nonparametric causal effects based on incremental propensity score
  interventions.
\newblock \emph{Journal of the American Statistical Association}, 114\penalty0
  (526):\penalty0 645--656, 2019.

\bibitem[Kennedy et~al.(2017)Kennedy, Ma, McHugh, and Small]{kennedy2017non}
Edward~H Kennedy, Zongming Ma, Matthew~D McHugh, and Dylan~S Small.
\newblock Non-parametric methods for doubly robust estimation of continuous
  treatment effects.
\newblock \emph{Journal of the Royal Statistical Society Series B: Statistical
  Methodology}, 79\penalty0 (4):\penalty0 1229--1245, 2017.

\bibitem[Kennedy et~al.(2024)Kennedy, Balakrishnan, Robins, and
  Wasserman]{kennedy2024minimax}
Edward~H Kennedy, Sivaraman Balakrishnan, James~M Robins, and Larry Wasserman.
\newblock Minimax rates for heterogeneous causal effect estimation.
\newblock \emph{The Annals of Statistics}, 52\penalty0 (2):\penalty0 793--816,
  2024.

\bibitem[Koltchinskii(2022)]{koltchinskii2022bootstrap}
Vladimir Koltchinskii.
\newblock Estimation of smooth functionals in high-dimensional models:
  Bootstrap chains and {G}aussian approximation.
\newblock \emph{The Annals of Statistics}, 50\penalty0 (4):\penalty0
  2386--2415, 2022.

\bibitem[Liu et~al.(2017)Liu, Mukherjee, Newey, and
  Robins]{liu2017semiparametric}
Lin Liu, Rajarshi Mukherjee, Whitney~K Newey, and James~M Robins.
\newblock Semiparametric efficient empirical higher order influence function
  estimators.
\newblock \emph{arXiv preprint arXiv:1705.07577}, 2017.

\bibitem[Liu et~al.(2024)Liu, Mukherjee, and Robins]{liu2024assumption}
Lin Liu, Rajarshi Mukherjee, and James~M Robins.
\newblock Assumption-lean falsification tests of rate double-robustness of
  double-machine-learning estimators.
\newblock \emph{Journal of Econometrics}, 240\penalty0 (2):\penalty0 105500,
  2024.

\bibitem[Manski(2004)]{manski2004statistical}
Charles~F Manski.
\newblock Statistical treatment rules for heterogeneous populations.
\newblock \emph{Econometrica}, 72\penalty0 (4):\penalty0 1221--1246, 2004.

\bibitem[McGrath and Mukherjee(2024)]{mcgrath2024nuisance}
Sean McGrath and Rajarshi Mukherjee.
\newblock Nuisance function tuning and sample splitting for optimal doubly
  robust estimation.
\newblock \emph{arXiv preprint arXiv:2212.14857}, 2024.

\bibitem[Newey(1990)]{newey1990semiparametric}
Whitney~K Newey.
\newblock Semiparametric efficiency bounds.
\newblock \emph{Journal of Applied Econometrics}, 5\penalty0 (2):\penalty0
  99--135, 1990.

\bibitem[Newey and McFadden(1994)]{newey1994large}
Whitney~K Newey and Daniel McFadden.
\newblock Large sample estimation and hypothesis testing.
\newblock \emph{Handbook of Econometrics}, 4:\penalty0 2111--2245, 1994.

\bibitem[Newey and Robins(2018)]{newey2018cross}
Whitney~K Newey and James~M Robins.
\newblock Cross-fitting and fast remainder rates for semiparametric estimation.
\newblock \emph{arXiv preprint arXiv:1801.09138}, 2018.

\bibitem[Newey et~al.(2004)Newey, Hsieh, and Robins]{newey2004twicing}
Whitney~K Newey, Fushing Hsieh, and James~M Robins.
\newblock Twicing kernels and a small bias property of semiparametric
  estimators.
\newblock \emph{Econometrica}, 72\penalty0 (3):\penalty0 947--962, 2004.

\bibitem[Patton et~al.(2019)Patton, Ziegel, and Chen]{patton2019dynamic}
Andrew~J Patton, Johanna~F Ziegel, and Rui Chen.
\newblock Dynamic semiparametric models for expected shortfall (and
  value-at-risk).
\newblock \emph{Journal of Econometrics}, 211\penalty0 (2):\penalty0 388--413,
  2019.

\bibitem[Powell(2020)]{powell2020quantile}
David Powell.
\newblock Quantile treatment effects in the presence of covariates.
\newblock \emph{Review of Economics and Statistics}, 102\penalty0 (5):\penalty0
  994--1005, 2020.

\bibitem[Richardson and Rotnitzky(2014)]{richardson2014causal}
Thomas~S Richardson and Andrea Rotnitzky.
\newblock Causal etiology of the research of {J}ames {M}. {R}obins.
\newblock \emph{Statistical Science}, 29\penalty0 (4):\penalty0 459--484, 2014.

\bibitem[Robins et~al.(2007)Robins, Sued, Lei-Gomez, and
  Rotnitzky]{robins2007comment}
James Robins, Mariela Sued, Quanhong Lei-Gomez, and Andrea Rotnitzky.
\newblock Comment: Performance of double-robust estimators when ``inverse
  probability'' weights are highly variable.
\newblock \emph{Statistical Science}, 22\penalty0 (4):\penalty0 544--559, 2007.

\bibitem[Robins et~al.(2008)Robins, Li, Tchetgen~Tchetgen, and van~der
  Vaart]{robins2008higher}
James Robins, Lingling Li, Eric Tchetgen~Tchetgen, and Aad van~der Vaart.
\newblock Higher order influence functions and minimax estimation of nonlinear
  functionals.
\newblock In \emph{Probability and Statistics: Essays in Honor of David A.
  Freedman}, pages 335--421. Institute of Mathematical Statistics, 2008.

\bibitem[Robins et~al.(2009)Robins, {Tchetgen Tchetgen}, Li, and {van der
  Vaart}]{robins2009semiparametric}
James Robins, Eric {Tchetgen Tchetgen}, Lingling Li, and Aad {van der Vaart}.
\newblock Semiparametric minimax rates.
\newblock \emph{Electronic Journal of Statistics}, 3:\penalty0 1305--1321,
  2009.

\bibitem[Robins et~al.(2016)Robins, Li, {Tchetgen Tchetgen}, and {van der
  Vaart}]{robins2016technical}
James Robins, Lingling Li, Eric {Tchetgen Tchetgen}, and Aad {van der Vaart}.
\newblock Technical report: Higher order influence functions and minimax
  estimation of nonlinear functionals.
\newblock \emph{arXiv preprint arXiv:1601.05820}, 2016.

\bibitem[Robins and Ritov(1997)]{robins1997toward}
James~M Robins and Ya'acov Ritov.
\newblock Toward a curse of dimensionality appropriate ({CODA}) asymptotic
  theory for semi-parametric models.
\newblock \emph{Statistics in Medicine}, 16\penalty0 (3):\penalty0 285--319,
  1997.

\bibitem[Robins et~al.(1994)Robins, Rotnitzky, and Zhao]{robins1994estimation}
James~M Robins, Andrea Rotnitzky, and Lue~Ping Zhao.
\newblock Estimation of regression coefficients when some regressors are not
  always observed.
\newblock \emph{Journal of the American Statistical Association}, 89\penalty0
  (427):\penalty0 846--866, 1994.

\bibitem[Robins et~al.(2000)Robins, Hern{\'a}n, and
  Brumback]{hernan2000marginal}
James~M Robins, Miguel~A Hern{\'a}n, and Babette Brumback.
\newblock Marginal structural models and causal inference in epidemiology.
\newblock \emph{Epidemiology}, 11\penalty0 (5):\penalty0 550--560, 2000.

\bibitem[Robins et~al.(2023)Robins, Li, Liu, Mukherjee, Tchetgen~Tchetgen, and
  van~der Vaart]{robins2023minimax}
James~M Robins, Lingling Li, Lin Liu, Rajarshi Mukherjee, Eric
  Tchetgen~Tchetgen, and Aad van~der Vaart.
\newblock Minimax estimation of a functional on a structured high-dimensional
  model ({C}orrected version).
\newblock \emph{arXiv preprint arXiv:1512.02174}, 2023.

\bibitem[Rotnitzky et~al.(2021)Rotnitzky, Smucler, and
  Robins]{rotnitzky2021characterization}
Andrea Rotnitzky, Ezequiel Smucler, and James~M Robins.
\newblock Characterization of parameters with a mixed bias property.
\newblock \emph{Biometrika}, 108\penalty0 (1):\penalty0 231--238, 2021.

\bibitem[Schick(1986)]{schick1986asymptotically}
Anton Schick.
\newblock On asymptotically efficient estimation in semiparametric models.
\newblock \emph{The Annals of Statistics}, 14\penalty0 (3):\penalty0
  1139--1151, 1986.

\bibitem[Shah and Peters(2020)]{shah2020hardness}
Rajen~D Shah and Jonas Peters.
\newblock The hardness of conditional independence testing and the generalised
  covariance measure.
\newblock \emph{The Annals of Statistics}, 48\penalty0 (3):\penalty0
  1514--1538, 2020.

\bibitem[Stone(1982)]{stone1982optimal}
Charles~J Stone.
\newblock Optimal global rates of convergence for nonparametric regression.
\newblock \emph{The Annals of Statistics}, 10\penalty0 (4):\penalty0
  1040--1053, 1982.

\bibitem[Su et~al.(2019)Su, Ura, and Zhang]{su2019non}
Liangjun Su, Takuya Ura, and Yichong Zhang.
\newblock Non-separable models with high-dimensional data.
\newblock \emph{Journal of Econometrics}, 212\penalty0 (2):\penalty0 646--677,
  2019.

\bibitem[Tchetgen~Tchetgen et~al.(2010)Tchetgen~Tchetgen, Robins, and
  Rotnitzky]{tchetgen2010doubly}
Eric~J Tchetgen~Tchetgen, James~M Robins, and Andrea Rotnitzky.
\newblock On doubly robust estimation in a semiparametric odds ratio model.
\newblock \emph{Biometrika}, 97\penalty0 (1):\penalty0 171--180, 2010.

\bibitem[van~der Vaart and Wellner(2023)]{van2023weak}
Aad~W van~der Vaart and Jon Wellner.
\newblock \emph{Weak Convergence and Empirical Processes: with Applications to
  Statistics}.
\newblock Springer Science \& Business Media, 2023.

\bibitem[Xu et~al.(2022)Xu, Liu, and Liu]{xu2022deepmed}
Siqi Xu, Lin Liu, and Zhonghua Liu.
\newblock Deep{M}ed: Semiparametric causal mediation analysis with debiased
  deep learning.
\newblock In \emph{Proceedings of the 36th International Conference on Neural
  Information Processing Systems}, pages 28238--28251, 2022.

\bibitem[Zheng et~al.(2025)Zheng, Bonvini, and Guo]{zheng2025perturbed}
Mengchu Zheng, Matteo Bonvini, and Zijian Guo.
\newblock Perturbed double machine learning: Nonstandard inference beyond the
  parametric length.
\newblock \emph{arXiv preprint arXiv:2511.01222}, 2025.

\end{thebibliography}


\begin{thebibliography}{22}
\providecommand{\natexlab}[1]{#1}
\providecommand{\url}[1]{\texttt{#1}}
\expandafter\ifx\csname urlstyle\endcsname\relax
  \providecommand{\doi}[1]{doi: #1}\else
  \providecommand{\doi}{doi: \begingroup \urlstyle{rm}\Url}\fi

\bibitem[Ai and Chen(2003)]{ai2003efficient}
Chunrong Ai and Xiaohong Chen.
\newblock Efficient estimation of models with conditional moment restrictions
  containing unknown functions.
\newblock \emph{Econometrica}, 71\penalty0 (6):\penalty0 1795--1843, 2003.

\bibitem[Ai et~al.(2021)Ai, Linton, Motegi, and Zhang]{ai2021unified}
Chunrong Ai, Oliver Linton, Kaiji Motegi, and Zheng Zhang.
\newblock A unified framework for efficient estimation of general treatment
  models.
\newblock \emph{Quantitative Economics}, 12\penalty0 (3):\penalty0 779--816,
  2021.

\bibitem[Andrews(1994)]{andrews1994empirical}
Donald~WK Andrews.
\newblock Empirical process methods in econometrics.
\newblock \emph{Handbook of Econometrics}, 4:\penalty0 2247--2294, 1994.

\bibitem[Bandeira et~al.(2025)Bandeira, Lucca, Nizic-Nikolac, and van
  Handel]{bandeira2025matrix}
Afonso~S Bandeira, Kevin Lucca, Petar Nizic-Nikolac, and Ramon van Handel.
\newblock Matrix chaos inequalities and chaos of combinatorial type.
\newblock In \emph{Proceedings of the 57th Annual ACM Symposium on Theory of
  Computing}, pages 795--805, 2025.

\bibitem[Bhattacharya and Ghosh(1992)]{bhattacharya1992class}
Rabi~N Bhattacharya and Jayanta~K Ghosh.
\newblock A class of ${U}$-statistics and asymptotic normality of the number of
  $k$-clusters.
\newblock \emph{Journal of Multivariate Analysis}, 43\penalty0 (2):\penalty0
  300--330, 1992.

\bibitem[Bonvini and Kennedy(2022)]{bonvini2022fast}
Matteo Bonvini and Edward~H Kennedy.
\newblock Fast convergence rates for dose-response estimation.
\newblock \emph{arXiv preprint arXiv:2207.11825}, 2022.

\bibitem[Cattaneo and Jansson(2018)]{cattaneo2018kernel}
Matias~D Cattaneo and Michael Jansson.
\newblock Kernel-based semiparametric estimators: Small bandwidth asymptotics
  and bootstrap consistency.
\newblock \emph{Econometrica}, 86\penalty0 (3):\penalty0 955--995, 2018.

\bibitem[Cattaneo et~al.(2019)Cattaneo, Jansson, and Ma]{cattaneo2019two}
Matias~D Cattaneo, Michael Jansson, and Xinwei Ma.
\newblock Two-step estimation and inference with possibly many included
  covariates.
\newblock \emph{The Review of Economic Studies}, 86\penalty0 (3):\penalty0
  1095--1122, 2019.

\bibitem[Cattaneo et~al.(2025)Cattaneo, Farrell, Jansson, and
  Masini]{cattaneo2025higher}
Matias~D Cattaneo, Max~H Farrell, Michael Jansson, and Ricardo~P Masini.
\newblock Higher-order refinements of small bandwidth asymptotics for
  density-weighted average derivative estimators.
\newblock \emph{Journal of Econometrics}, 252\penalty0 (Part B):\penalty0
  105855, 2025.

\bibitem[Chen and Xie(2025)]{chen2025local}
Xiaohong Chen and Haitian Xie.
\newblock Local overidentification and efficiency gains in modern causal
  inference and data combination.
\newblock \emph{arXiv preprint arXiv:2510.16683}, 2025.

\bibitem[Chen and Kato(2020)]{chen2020jackknife}
Xiaohui Chen and Kengo Kato.
\newblock Jackknife multiplier bootstrap: Finite sample approximations to the
  ${U}$-process supremum with applications.
\newblock \emph{Probability Theory and Related Fields}, 176\penalty0
  (3):\penalty0 1097--1163, 2020.

\bibitem[Kennedy et~al.(2017)Kennedy, Ma, McHugh, and Small]{kennedy2017non}
Edward~H Kennedy, Zongming Ma, Matthew~D McHugh, and Dylan~S Small.
\newblock Non-parametric methods for doubly robust estimation of continuous
  treatment effects.
\newblock \emph{Journal of the Royal Statistical Society Series B: Statistical
  Methodology}, 79\penalty0 (4):\penalty0 1229--1245, 2017.

\bibitem[Kosorok(2008)]{kosorok2008introduction}
Michael~R Kosorok.
\newblock \emph{Introduction to Empirical Processes and Semiparametric
  Inference.}
\newblock Springer, 2008.

\bibitem[Lin et~al.(2024)Lin, Su, Mou, Ding, and Wainwright]{lin2024worthwhile}
Licong Lin, Fangzhou Su, Wenlong Mou, Peng Ding, and Martin Wainwright.
\newblock When is it worthwhile to jackknife? {B}reaking the quadratic barrier
  for ${Z}$-estimators.
\newblock \emph{arXiv preprint arXiv:2411.02909}, 2024.

\bibitem[Liu et~al.(2017)Liu, Mukherjee, Newey, and
  Robins]{liu2017semiparametric}
Lin Liu, Rajarshi Mukherjee, Whitney~K Newey, and James~M Robins.
\newblock Semiparametric efficient empirical higher order influence function
  estimators.
\newblock \emph{arXiv preprint arXiv:1705.07577}, 2017.

\bibitem[Liu et~al.(2024)Liu, Mukherjee, and Robins]{liu2024assumption}
Lin Liu, Rajarshi Mukherjee, and James~M Robins.
\newblock Assumption-lean falsification tests of rate double-robustness of
  double-machine-learning estimators.
\newblock \emph{Journal of Econometrics}, 240\penalty0 (2):\penalty0 105500,
  2024.

\bibitem[Minsker and Wei(2019)]{minsker2019moment}
Stanislav Minsker and Xiaohan Wei.
\newblock Moment inequalities for matrix-valued ${U}$-statistics of order 2.
\newblock \emph{Electronic Journal of Probability}, 24\penalty0 (133):\penalty0
  1--32, 2019.

\bibitem[Robins et~al.(2016)Robins, Li, {Tchetgen Tchetgen}, and {van der
  Vaart}]{robins2016technical}
James Robins, Lingling Li, Eric {Tchetgen Tchetgen}, and Aad {van der Vaart}.
\newblock Technical report: Higher order influence functions and minimax
  estimation of nonlinear functionals.
\newblock \emph{arXiv preprint arXiv:1601.05820}, 2016.

\bibitem[Robins et~al.(2000)Robins, Hern{\'a}n, and
  Brumback]{hernan2000marginal}
James~M Robins, Miguel~A Hern{\'a}n, and Babette Brumback.
\newblock Marginal structural models and causal inference in epidemiology.
\newblock \emph{Epidemiology}, 11\penalty0 (5):\penalty0 550--560, 2000.

\bibitem[Robins et~al.(2023)Robins, Li, Liu, Mukherjee, Tchetgen~Tchetgen, and
  van~der Vaart]{robins2023minimax}
James~M Robins, Lingling Li, Lin Liu, Rajarshi Mukherjee, Eric
  Tchetgen~Tchetgen, and Aad van~der Vaart.
\newblock Minimax estimation of a functional on a structured high-dimensional
  model ({C}orrected version).
\newblock \emph{arXiv preprint arXiv:1512.02174}, 2023.

\bibitem[Tropp(2015)]{tropp2015introduction}
Joel~A Tropp.
\newblock An introduction to matrix concentration inequalities.
\newblock \emph{Foundations and Trends{\textregistered} in Machine Learning},
  8\penalty0 (1-2):\penalty0 1--230, 2015.

\bibitem[van~der Vaart and Wellner(2023)]{van2023weak}
Aad~W van~der Vaart and Jon Wellner.
\newblock \emph{Weak Convergence and Empirical Processes: with Applications to
  Statistics}.
\newblock Springer Science \& Business Media, 2023.

\end{thebibliography}
